\documentclass[runningheads, a4paper]{llncs}

\usepackage{amsmath}
\usepackage{amssymb}
\usepackage{stmaryrd}
\usepackage[]{algorithm2e}
\usepackage[final,colorlinks=true,linktocpage=true,hyperfootnotes]{hyperref}
\usepackage{csquotes}
\usepackage{tabularx}
\usepackage{enumerate}
\usepackage{pgf}
\usepackage{tikz}
\usepackage[version=3,arrows=pgf]{mhchem}
\usetikzlibrary{decorations.pathmorphing}
\usetikzlibrary{arrows,automata}

\usepackage{csquotes}

\allowdisplaybreaks

\newcommand{\qedhere}{\qed}
\newcommand{\TECHNICALREPORT}{\cite{technicalReport}}
\newcommand{\SEEAPPENDIXPLEASE}{
        Most proofs as well as detailed constructions
        are provided in the appendix for the readers convenience.
}

\usepackage{todonotes}
\usepackage{verbatim} 

\usepackage{changepage} 

\usepackage{paralist}

\newcommand{\N}{\ensuremath{\mathbb{N}}} 
\newcommand{\POSN}{\ensuremath{\mathbb{N}_{> 0}}} 
\newcommand{\DEFEQ}{\ensuremath{\triangleq}} 
\newcommand{\subst}[2]{\ensuremath{\left[#1 / #2\right]}} 
\newcommand{\remap}[2]{\ensuremath{\left[#1 \mapsto #2\right]}} 
\newcommand{\restr}[2]{\ensuremath{#1\upharpoonright_{#2}}} 

\newcommand{\PTO}{\ensuremath{\rightharpoonup}} 
\newcommand{\DOM}{\ensuremath{\textnormal{dom}}} 

\newcommand{\SUBTREE}[2]{#1|_{#2}} 

\newcommand{\EMPTYSEQ}{\ensuremath{\varepsilon}} 
\newcommand{\T}[1]{\ensuremath{\mathbf{#1}}} 
\newcommand{\SIZE}[1]{\ensuremath{\|#1\|}} 
\newcommand{\PROJ}[2]{\ensuremath{#1\!\left[#2\right]}} 

\newcommand{\TARULE}[3]{\ensuremath{#1 \xrightarrow{#2} #3}} 

\newcommand{\NIL}{\ensuremath{\textnormal{\textbf{null}}}} 
\newcommand{\EMP}{\ensuremath{\textnormal{emp}}} 
\newcommand{\PT}[2]{\ensuremath{#1 \mapsto (#2)}} 
\newcommand{\PTS}[2]{\ensuremath{#1 \mapsto #2}} 
\newcommand{\SEP}{\ensuremath{*}} 

\newcommand{\VAL}{\ensuremath{\textit{Val}}} 
\newcommand{\VAR}{\ensuremath{\textit{Var}}} 
\newcommand{\LOC}{\ensuremath{\textit{Loc}}} 
\newcommand{\PRED}{\ensuremath{\textnormal{Pred}}} 

\newcommand{\ARITY}{\ensuremath{\textnormal{ar}}} 
\newcommand{\NOFV}[1]{\ensuremath{\SIZE{\FV{0}{#1}}}} 
\newcommand{\NOCALLS}[1]{\ensuremath{\SIZE{\CALLS{#1}}}} 

\newcommand{\sh}{\ensuremath{\varphi}} 
\newcommand{\sha}{\ensuremath{\psi}} 
\newcommand{\shb}{\ensuremath{\vartheta}} 
\newcommand{\rsh}{\ensuremath{\tau}} 
\newcommand{\rsha}{\ensuremath{\sigma}} 
\newcommand{\rshb}{\ensuremath{\zeta}} 
\newcommand{\SL}[2]{\ensuremath{\textnormal{SH}^{#1}_{#2}}} 
\newcommand{\RSL}[2]{\ensuremath{\textnormal{RSH}^{#1}_{#2}}} 
\newcommand{\SPATIAL}[1]{\ensuremath{\Sigma^{#1}}} 
\newcommand{\PS}{\ensuremath{P}} 
\newcommand{\FV}[2]{\ensuremath{\T{x}_{#1}^{#2}}} 
\newcommand{\BV}[1]{\ensuremath{\T{z}^{#1}}} 
\newcommand{\PURE}[1]{\ensuremath{\Pi^{#1}}} 
\newcommand{\CALLS}[1]{\ensuremath{\Gamma^{#1}}} 
\newcommand{\CALLN}[2]{\ensuremath{\PS_{#1}^{#2}\FV{#1}{#2}}} 
\newcommand{\SYMBOLICHEAP}[1]{\ensuremath{ \exists \BV{#1} \,.\, \SPATIAL{#1} \SEP \CALLS{#1} \,:\, \PURE{#1}}}

\newcommand{\heap}{\ensuremath{h}} 
\newcommand{\stack}{\ensuremath{s}} 
\newcommand{\HEAPUNION}{\ensuremath{\uplus}} 
\newcommand{\STATES}{\ensuremath{\textit{States}}} 
\newcommand{\CALLSEM}[2]{\ensuremath{\mathbb{U}_{#2}(#1)}} 
\newcommand{\MODELS}[1]{\ensuremath{\textit{Models}({#1})}} 

\newcommand{\UTREES}[2]{\ensuremath{\mathbb{T}_{#1}(#2)}} 
\newcommand{\UNFOLD}[1]{\ensuremath{\llbracket #1 \rrbracket}} 

\newcommand{\SRDARROW}{\Leftarrow}
\newcommand{\SRDRULE}[2]{\ensuremath{#1 \SRDARROW #2}} 

\newcommand{\SRD}{\ensuremath{\Phi}} 
\newcommand{\SRDALT}{\ensuremath{\Psi}} 
\newcommand{\SETSRD}[1]{\ensuremath{\textnormal{SID}_{#1}}} 

\newcommand{\SAT}[1]{\ensuremath{\models_{#1}}} 
\newcommand{\ENTAIL}[1]{\ensuremath{\models_{#1}}} 
\newcommand{\HEAPMODELS}[2]{\ensuremath{H_{#1}^{#2}}}

\newcommand{\SRDCLASS}{\ensuremath{\mathcal{C}}}
\newcommand{\SRDCLASSFV}[1]{\ensuremath{\textrm{FV}^{\leq #1}}}
\newcommand{\CENTAIL}[1]{\ensuremath{\langle #1 \rangle}}
\newcommand{\SHCLASS}{\SL{}{\ensuremath{\mathcal{C}}}}
\newcommand{\SHCLASSFV}[1]{\SL{}{\SRDCLASSFV{#1}}}
\newcommand{\RSHCLASS}{\RSL{}{\ensuremath{\mathcal{C}}}}
\newcommand{\RSHCLASSFV}[1]{\RSL{}{\SRDCLASSFV{#1}}}
\newcommand{\SHCENTAIL}[1]{\SL{}{\ensuremath{\langle #1 \rangle}}}

\newcommand{\HA}[1]{\ensuremath{\mathfrak{#1}}} 
\newcommand{\MOVE}[4]{\ensuremath{#4 \xrightarrow{#3}_{\HA{#1}} #2}} 
\newcommand{\OMEGA}[3]{\ensuremath{\EMPTYSEQ \xrightarrow{#3}_{\HA{#1}} #2}} 

\newcommand{\ALLOC}[1]{\ensuremath{\textnormal{\textit{alloc}}(#1)}} 
\newcommand{\MPT}[1]{\ensuremath{\mapsto_{#1}}} 
\newcommand{\MEQ}[1]{\ensuremath{=_{#1}}} 
\newcommand{\MSIM}[1]{\ensuremath{\sim_{#1}}} 
\newcommand{\MNEQ}[1]{\ensuremath{\neq_{#1}}} 
\newcommand{\REACH}[3]{\ensuremath{#1 \rightsquigarrow_{#3} #2}} 
\newcommand{\CHECK}{\ensuremath{\textnormal{\textit{check}}}} 

\newcommand{\SQUEEZE}[2]{\ensuremath{\left[#1 \downarrow #2\right]}} 

\newcommand{\TRACK}{\ensuremath{\textnormal{\texttt{TRACK}}}} 
\newcommand{\SATPROP}{\ensuremath{\textnormal{\texttt{SAT}}}} 
\newcommand{\ESTPROP}{\ensuremath{\textnormal{\texttt{EST}}}} 
\newcommand{\RPROP}{\ensuremath{\textnormal{\texttt{REACH}}}} 
\newcommand{\GARBAGEPROP}{\ensuremath{\textnormal{\texttt{GFREE}}}} 
\newcommand{\CYCLEPROP}{\ensuremath{\textnormal{\texttt{ACYCLIC}}}} 

\newcommand{\USET}[3]{\ensuremath{H_{#1,#2}^{#3}}}

\DeclareMathOperator{\diff}{diff}
\DeclareMathOperator{\qeq}{eq}
\DeclareMathOperator{\qfst}{fst}
\DeclareMathOperator{\qsnd}{snd}
\DeclareMathOperator{\qrev}{rev}

\DeclareMathOperator{\fid}{id}

\newcommand{\PI}{I}
\newcommand{\EE}[1]{\cong_{#1}}

\newcommand{\IFV}[1]{\PROJ{\FV{0}{}}{#1}} 
\newcommand{\PIFV}[1]{\PROJ{\T{y}_{0}}{#1}} 
\newcommand{\APP}[2]{#1\,#2}              
\newcommand{\PCDS}[2]{\mathtt{#1}\,#2}    
\newcommand{\PERM}{\rho}
\newcommand{\PERMA}{\tilde{\rho}}
\newcommand{\PERMS}{\mathcal{S}_{\FV{}{}}}
\newcommand{\PURES}{\textnormal{Pure}(\FV{0}{})}
\newcommand{\PUREA}{\Lambda}
\newcommand{\NUMTLL}{8} 
\newcommand{\ITLLC}[1]{P_{#1}^{\PERM}} 
\newcommand{\PCTXT}[3]{\sh^{#2,#1}_{#3}}
\newcommand{\CTXT}[2]{\PCTXT{\PERM}{#1}{#2}}
\newcommand{\LTLLC}[1]{\CALLSEM{\ITLLC{#1}\,\T{x}}{}}
\newcommand{\LIFT}[1]{\textnormal{lift}^{\PERM}(\PI\, #1)}

\newcommand{\DPROBLEM}[1]{{\textnormal{\textsc{#1}}}} 
\newcommand{\DENTAIL}[2]{\ensuremath{\DPROBLEM{SL-ENTAIL}_{#1}^{#2}}} 
\newcommand{\CCLASS}[1]{\textsc{#1}} 
\newcommand{\COMPLEMENT}[1]{\ensuremath{\overline{\DPROBLEM{#1}}}} 
\newcommand{\BIGO}[1]{\ensuremath{\mathcal{O}\left(#1\right)}}

\newcommand{\SHRINK}[1]{\textnormal{\textit{compress}}(#1)}
\newcommand{\KERNEL}{\textnormal{\textit{core}}}
\newcommand{\REDUCE}[1]{\textnormal{\textit{reduce}}(#1)}
\newcommand{\SSIGMA}[2]{\textit{kernel}(#1,#2)}
\newcommand{\STRIP}[1]{\textnormal{\textit{strip}}(#1)}
\newcommand{\EC}[2]{\ensuremath{[#1]_{\EE{#2}}}}

\newcommand{\HATRACK}{\HA{A}_{\TRACK}}
\newcommand{\HASAT}{\HA{A}_{\SATPROP}}
\newcommand{\HAREACH}{\HA{A}_{\RPROP}}
\newcommand{\HACYCLE}{\HA{A}_{\CYCLEPROP}}
\newcommand{\HAEST}{\HA{A}_{\ESTPROP}}
\newcommand{\HAGARBAGE}{\HA{A}_{\GARBAGEPROP}}
\newcommand{\HASCHEME}[3]{\HA{A}_{\texttt{SCHEME}}(#1,#2,#3)}
\newcommand{\HASH}[1]{\HA{A}_{\textnormal{$#1$}}}

\begin{document}




\title{Unified Reasoning about Robustness Properties of Symbolic-Heap Separation Logic} 


\titlerunning{Unified Reasoning about Robustness Properties of Symbolic-Heaps}

\author{Christina Jansen\inst{1} \and Jens Katelaan\inst{2} \and Christoph Matheja\inst{1} \and \\ Thomas Noll\inst{1} \and Florian Zuleger\inst{2}}
\authorrunning{Jansen, Katelaan, Matheja, Noll, Zuleger}

\institute{Software Modeling and Verification Group, \\ RWTH Aachen University, Germany
\and
TU Wien, Austria
}

\maketitle

\begin{abstract}
  We introduce \emph{heap automata}, a formalism for automatic reasoning
about \emph{robustness properties} of the symbolic heap fragment of
separation logic with user-defined inductive predicates.
Robustness properties, such as satisfiability, reachability, and
acyclicity, are important for a wide range of reasoning tasks in
automated program analysis and verification based on separation logic.
Previously, such properties have appeared in many places in the
separation logic literature, but have not been studied in a systematic
manner.
In this paper, we develop an algorithmic framework based on heap
automata that allows us to derive asymptotically optimal decision
procedures for a wide range of robustness properties in a uniform way.

We implemented a protoype of our framework and obtained promising results
for all of the aforementioned robustness properties.

Further, we demonstrate the applicability of heap automata beyond
robustness properties.  We apply our algorithmic framework to the
\emph{model checking} and the \emph{entailment problem} for
symbolic-heap separation logic.

\end{abstract}





\section{Introduction}
\label{sec:introduction}
        \emph{Separation logic (SL)}~\cite{reynolds2002separation} is a popular formalism
        for Hoare-style verification of imperative, heap-manipulating programs.
        While its symbolic heap fragment originally emerged as an idiomatic
        form of assertions that occur naturally in hand-written
        proofs~\cite{o2001local,berdine2005symbolic,berdine2005smallfoot}, a
        variety of program analyses based on symbolic-heap separation logic
        have been developed~\cite{berdine2005symbolic,berdine2007shape,calcagno2009compositional,le2014shape,ohearn2007resources,brookes2007semantics,gotsman2007thread}.  Consequently, it now
        serves as formal basis for a multitude of automated verification
        tools, such as \cite{berdine2011slayer,calcagno2011infer,chin2012automated,dudka2011predator,jacobs2011verifast,qiu2013natural,magill2008thor,botincan2011corestar},
        capable of proving complex properties of a program's heap, such as
        memory safety, for large code
        bases~\cite{calcagno2009compositional,calcagno2011infer}.
        These tools typically rely on \emph{systems of inductive predicate
          definitions (SID)} to specify the shape of data structures employed
        by a program, such as trees and linked lists.  Originally, separation
        logic tools implemented highly-specialized procedures for such fixed
        SIDs.  As this limits their applicability, there is an ongoing trend
        to support custom SIDs that are either defined manually~\cite{jacobs2011verifast,chin2012automated} or even 
        automatically generated.  The latter may, for example, be
        obtained from the 
        tool
        \textsc{Caber}~\cite{brotherston2014cyclic}.

        \paragraph*{Robustness properties} Allowing for arbitrary SIDs, however, raises various questions about
        their \emph{robustness}. 
        A user-defined or auto-generated SID might, for example, be
        \emph{inconsistent}, introduce \emph{unallocated logical variables},
        specify data structures that contain undesired \emph{cycles}, or
        produce \emph{garbage}, i.e., parts of the heap that are unreachable
        from any program variable. 
        Accidentally introducing such properties into
        specifications can have a negative impact on performance,
        completeness, and even soundness of the employed verification
        algorithms: 
        \begin{itemize}
        \item Brotherston et al.~\cite{brotherston2014decision} point out that
          tools might waste time on inconsistent scenarios due to
          \emph{unsatisfiability} of specifications.
        \item The absence of unallocated logical variables, also known as
          \emph{establishment}, is required by the approach of Iosif et
          al.~\cite{iosif2013tree,iosif2014entailment} to obtain a decidable fragment of symbolic heaps.
        \item Other verification approaches, such as the one by Habermehl et
          al.~\cite{habermehl2011forest,habermehl2012forest}, assume that no
          garbage is introduced by data structure specifications.
        \item During program analysis and verification, questions such as
          \emph{reachability}, \emph{acyclicity} and \emph{garbage-freedom}
          arise depending on the properties of interest.  For example, as
          argued by Zanardini and Genaim~\cite{zanardini2014inference},
          acyclicity of the heap is crucial in automated termination
          proofs.
        \end{itemize}

        Being able to check such \emph{robustness properties} of custom SIDs
        is thus crucial (1) in debugging of separation-logic specifications
        prior to program analysis and (2) in the program analyses themselves.
        So far, however, all of the above properties have either been
        addressed individually or not systematically at all. For example,
        satisfiability is studied in detail by Brotherston et
        al.~\cite{brotherston2014decision}, whereas establishment is
        often addressed with ad-hoc solutions
        \cite{iosif2013tree,habermehl2011forest}.


        Several reasoning tasks arise in the context of robustness
        properties. As a motivation, consider the problem of acyclicity.  If our
        program analysis requires acyclicity, we would like to decide whether
        all interpretations of a symbolic heap are acyclic; if not, to find
        out how cycles can be introduced into the heap (counterexample
        generation); and, finally, to be able to generate a new SID that does
        guarantee acyclicity (called \emph{refinement} below).
        A systematic treatment of robustness properties
        should cover these reasoning tasks in general, not just for the
        problem of acyclicity.

        \paragraph*{Problem statement}
        We would like to develop a framework that enables:
        \begin{enumerate}
        \item \emph{Decision procedures for robustness properties.} In program
          analysis, we generally deal with symbolic heaps that reference
          SIDs specifying unbounded data structures and thus usually have infinitely many
          interpretations. We need to be able to decide whether all, or some, of these
          infinitely many interpretations are guaranteed to satisfy a given
          robustness property.
        %
        \item \emph{Generation of counterexamples} that violate a desired
          property.
        \item \emph{Refinement} of SIDs to automatically generate a new SID
          that respects a given robustness property.
        \item \emph{Automatic combination of decision procedures} to derive decision procedures for
          complex robustness properties from simpler ingredients.
        \end{enumerate}


        \paragraph*{Motivating example: Inductive reasoning about robustness
          properties}
        The key insight underlying our solution to the above problems is that
        many properties of symbolic heaps can be decided iteratively by
        inductive reasoning.  To motivate our approach, we illustrate this
        reasoning process with a concrete
        example.
%
        Consider an SID for acyclic singly-linked list segments with head $x$ and tail $y$:
         \begin{align*}
           \textnormal{\texttt{sll}} ~\SRDARROW~ \EMP : \{ x = y \} \qquad\qquad
           \textnormal{\texttt{sll}} ~\SRDARROW~ \exists u ~.~
           \PTS{x}{u}  \SEP \textnormal{\texttt{sll}}(u\,y) :
                                                    \{ x \neq y \}.
         \end{align*}
         The two \emph{rules} of the SID define a case distinction: A list is
         either empty, but then the first and the last element are the same;
         or, the first element has a successor $u$ (specified by the
         \emph{points-to assertion} $\PTS{x}{u}$), which in turn is at
         the head of a (shorter) singly-linked list segment,
         $\mathtt{sll}(u\,y)$.  The inequality in the second rule
         guarantees that there is no cyclic model.
        Now, consider the following symbolic heap with \emph{predicate calls} to
        \texttt{sll}:
        $
        \sh = \exists x\,y\,z ~.~ \mathtt{sll}(x\,z) \SEP \PTS{z}{y} \SEP \mathtt{sll}(y\,x)
        $, 
        which might appear as an assertion during program analysis. Say our
        program analysis depends on the acyclicity of $\sh$, so we need to
        determine whether $\sh$ is acyclic. We can do so by \emph{inductive
          reasoning} as follows.
        \begin{compactitem}
        \item We 
          analyze the call $\mathtt{sll}(x\,z)$, the first list
          segment in the symbolic heap $\sh$. 
          If it is 
          interpreted by the right-hand side of the first rule of the SID
          from 
          above, then there is no cycle in $\mathtt{sll}(x\,z)$ and $z$
          is reachable from $x$.
        \item If we already know for a call $\mathtt{sll}(u\,z)$ that all of its models are
          acyclic structures and that $z$ is reachable from $u$,
          then $z$ is also reachable from $x$ in the symbolic heap $\exists u
          ~.~ \PTS{x}{u} \SEP \textnormal{\texttt{sll}}(u\,z) : \{ x \neq z \}$
          obtained by the second rule of the SID. 
          Since our SID does not introduce dangling pointers,
          we also know that there is still no cycle.
        \item By induction, $\mathtt{sll}(x\,z)$ is thus acyclic and $z$ is
          reachable from $x$.
        \item Likewise, $\mathtt{sll}(y\,x)$ is acyclic and $x$ is reachable
          from $y$.
        \item Now, based on the information we discovered for
          $\mathtt{sll}(x\,z)$ and $\mathtt{sll}(y\,x)$, we examine $\sh$ and
          conclude that it is \emph{cyclic}, as $z$ is reachable from $x$, $y$
          is reachable from $z$, and $x$ is reachable from $y$. Crucially, we
          reason inductively and thus do \emph{not} re-examine the list
          segments to arrive at our conclusion.
        \end{compactitem}
        In summary, we examine a symbolic heap and corresponding SID
        \emph{bottom-up}, starting from the non-recursive base case. Moreover,
        at each stage of this analysis, we remember a fixed amount of
        information---namely what we discover about reachability between
        parameters and acyclicity of every symbolic heap we examine.
        Similar inductive constructions are defined explicitly for various robustness properties throughout the separation logic literature~\cite{brotherston2014decision,brotherston2016model,iosif2013tree}.
        Our aim is to generalize such manual constructions following an
        \emph{automata-theoretic approach}: We introduce automata that
        operate on symbolic heaps and store the relevant information of each
        symbolic heap they examine in their state space. Whenever such an automaton comes
        across a predicate that it has already analyzed, it can simply replace
        the predicate with the information that is encoded in the corresponding state.
        In other words, our automata \emph{recognize} robustness properties
        \emph{in a compositional way} by exploiting the \emph{inductive
          structure} inherent in the SIDs.

        \paragraph*{Systematic reasoning about robustness properties} 
        Our novel automaton model, \emph{heap automata}, works directly on the
        structure of symbolic heaps as outlined in the example, and can be
        applied to all the problems introduced before.
        In particular, heap automata enable \emph{automatic refinement} of SIDs and
        enjoy a variety of closure properties
        through which we can derive counterexample generation as well as
        decision procedures for various robustness properties---including
        satisfiability, establishment, reachability, garbage-freedom, and
        acyclicity.

        Our approach can thus be seen as an \emph{algorithmic framework}
        for deciding a wide range of robustness properties of symbolic
        heaps. Furthermore, we show asymptotically optimal complexity of our
        automata-based decision procedures in a uniform way.
        %
        By enabling this systematic approach to reasoning about robustness,
        our framework generalizes prior work that studied single robustness
        properties in isolation, such as the work by Brotherston et
        al.~\cite{brotherston2014decision,brotherston2016model}.

        As a natural byproduct of our automata-based approach, we also derive
        decision procedures for the \emph{model-checking
          problem}, which was
        recently studied, and proven to be \CCLASS{ExpTime}--complete in
        general, by Brotherston et al.~\cite{brotherston2016model}.  This makes
        it possible to apply our framework to \emph{run-time verification}---a
        setting in which robustness properties are of particular
        importance~\cite{nguyen2008runtime,jacobs2011verifast,brotherston2016model}.

        \paragraph*{Entailment checking with heap automata}

        Finally, we also address the \emph{entailment problem}.  In
        Hoare-style program analysis, decision procedures for the
        \emph{entailment problem} become essential to discharge implications
        between assertions, as required, for example, by the rule of
        consequence~\cite{hoare1969axiomatic}. Because of this central role in
        verification, there is an extensive body of research on decision
        procedures for entailment; see, for
        example~\cite{berdine2004decidable,brotherston2011automated,brotherston2012generic,iosif2013tree,iosif2014entailment,navarro2013separation,piskac2014automating,enea2014entailment}.
        Antonopoulos et al.~\cite{antonopoulos2014foundations} study the
        complexity of the entailment problem and show that it is undecidable
        in general, and already \CCLASS{ExpTime}--hard for SIDs specifying
        sets of trees.

        We use heap automata to check entailment between \emph{determined} symbolic heaps.
        Intuitively, determinedness is a strong form of the establishment property guaranteeing that two variables are either equal or unequal in every model.
        %
        %
        Unlike other decision
        procedures~\cite{iosif2013tree,iosif2014entailment,berdine2004decidable},
        our approach does not impose syntactic restrictions on the symbolic
        heap under consideration but merely requires that suitable heap
        automata for the predicates on the right-hand side of the entailment
        are provided.
        In particular, we show how to obtain \CCLASS{ExpTime} decision
        procedures from such heap automata---which exist for highly
        non-trivial SIDs.
        If desired, additional syntactic restrictions can be integrated
        seamlessly into our approach to boost our algorithms' performance.

        \paragraph{Contributions}
        Our main contributions can be summarized as follows.
        \begin{itemize}
        \item We introduce \emph{heap automata}, a novel automaton model
          operating directly on symbolic heaps. 
        %
          We prove that heap automata enjoy various useful \emph{closure
            properties}. Besides union, intersection and complement, they are
          closed under the conjunction with pure formulas, allowing the
          construction of complex heap automata from simple ones.
        \item We develop a powerful \emph{algorithmic framework} for automated
          reasoning about and debugging of symbolic heaps with inductive
          predicate definitions based on heap automata.
        \item We show that key robustness properties, such 
          as \emph{satisfiability}, \emph{establishment}, \emph{reachability},
          \emph{garbage freedom} and \emph{acyclicity}, can naturally be
          expressed as heap automata. Moreover, the upper bounds of decision
          procedures obtained from our framework are shown to be
          optimal---i.e., \CCLASS{ExpTime}--complete---in each of these cases.
          Further, they enable automated \emph{refinement} of SIDs to filter
          out (or expose) symbolic heaps with undesired properties.
        \item Additionally, we apply heap automata to tackle the
          \emph{entailment} and the \emph{model checking} problem for symbolic
          heaps.  We show that if each predicate of an
          SID can be represented by a heap automaton, then the entailment
          problem for the corresponding fragment of symbolic heaps is
          decidable in \CCLASS{2-ExpTime} in general and
          \CCLASS{ExpTime}-complete if the maximal arity of predicates and
          points-to assertions is bounded.
          For example, our framework yields an \CCLASS{ExpTime} decision procedure
          for a symbolic heap fragment 
          capable of representing trees with linked leaves---a fragment that
          is out of scope of most \CCLASS{ExpTime} decision procedures known
          so far
          (cf. \cite{berdine2004decidable,enea2014entailment,iosif2014entailment}).
      \item We implemented a prototype of our framework that yields promising results for all robustness properties considered in the paper.
        %
        \end{itemize}

        \paragraph{Organization of the paper}
        The fragment of symbolic heaps with inductive predicate definitions is briefly introduced in Section~\ref{sec:symbolic-heaps}.
        Heap automata and derived decision procedures are studied in Section~\ref{sec:compositional}.
        Section~\ref{sec:zoo} demonstrates that a variety of robustness properties can be checked by heap automata.
        We report on a prototypical implementation of our framework in Section~\ref{sec:implementation}.
        Special attention to the entailment problem is paid in Section~\ref{sec:entailment}.
        %
        Finally, Section~\ref{sec:conclusion} concludes.
        \SEEAPPENDIXPLEASE


%
%
\section{Symbolic Heaps}
\label{sec:symbolic-heaps}
This section briefly introduces the symbolic heap fragment of separation logic equipped with inductive predicate definitions.
\paragraph{Basic Notation}
$\N$ is the set of natural numbers and $2^{S}$ is the powerset of a set $S$.
$(co)\DOM(f)$ is the (co)domain of a (partial) function $f$.
We abbreviate tuples $u_1 \ldots u_n$, $n \geq 0$, by $\T{u}$
and write $\PROJ{\T{u}}{i}$, $1 \leq i \leq \SIZE{\T{u}} = n$, to
denote $u_i$, the $i$-th element of $\T{u}$.
By slight abuse of notation, the same symbol $\T{u}$ is used for the set of all elements occurring in tuple $\T{u}$.
The empty tuple is $\EMPTYSEQ$ and the set of all (non-empty) tuples [of length $n \geq 0$] over a finite set $S$ is  $S^{*}$ ($S^{+}$ [$S^{n}$]).
The concatenation of tuples $\T{u}$ and $\T{v}$
is $\T{u}\,\T{v}$.
\paragraph{Syntax}
%
%
We usually denote \emph{variables} taken from $\VAR$ (including a dedicated variable $\NIL$) by $a,b,c,x,y,z$, etc.
Moreover, let $\PRED$ be a set of \emph{predicate} \emph{symbols} and $\ARITY : \PRED \to \N$ be a function assigning each symbol its \emph{arity}.
\emph{Spatial formulas} $\SPATIAL{}$ and \emph{pure formulas} $\pi$ are given by the following grammar:
\begin{align*}
 \SPATIAL{} ~::=~ \EMP ~|~ \PTS{x}{\T{y}} ~|~ \SPATIAL{} \SEP \SPATIAL{} \qquad \pi ~::=~ x = y ~|~ x \neq y,
\end{align*}
where $\T{y}$ is a non-empty tuple of variables.
Here, $\EMP$ stands for the \emph{empty heap}, $\PTS{x}{\T{y}}$ is a \emph{points-to} \emph{assertion} and $\SEP$ is the \emph{separating} \emph{conjunction}.
Furthermore, for $\PS \in \PRED$ and a tuple of variables $\T{y}$ of length $\ARITY(\PS)$, $\PS\T{y}$ is a \emph{predicate} \emph{call}.
A \emph{symbolic} \emph{heap} $\sh$ with variables $\VAR(\sh)$ and
free variables $\FV{0}{} \subseteq \VAR(\sh)$ is a formula of the
form
~$
 \sh ~=~ \SYMBOLICHEAP{}, ~~ \CALLS{} = \CALLN{1}{} \SEP \ldots \SEP \CALLN{m}{},
$~
where $\SPATIAL{}$ is a spatial formula, $\CALLS{}$ is a sequence of predicate calls and $\PURE{}$ is a finite set of pure formulas, each with variables from $\FV{0}{}$ and $\BV{}$.
This normal form, in which predicate calls and points-to assertions are never mixed, is chosen to simplify formal constructions.
If an element of a symbolic heap is empty, we usually omit it to improve readability.
For the same reason, we fix the notation from above and write $\BV{\sh}$, $\FV{i}{\sh}$, $\SPATIAL{\sh}$ etc.\ to denote the respective component of symbolic heap $\sh$ in formal constructions.
Hence, $\NOFV{\sh}$ and $\NOCALLS{\sh}$ refer to the number of free
variables and the number of predicate calls of $\sh$,
respectively.
We omit the superscript whenever the symbolic heap under consideration is clear from the context.
If a symbolic heap $\rsh$ contains no predicate calls, i.e., $\NOCALLS{\rsh} = 0$, then $\rsh$ is called \emph{reduced}.
%
Moreover, to simplify the technical development, we tacitly assume that $\NIL$ is a free variable that is passed to every predicate call.
Thus, for each $i \in \N$, we write $\PROJ{\FV{i}{}}{0}$ as a shortcut for $\NIL$ and treat
$\PROJ{\FV{i}{}}{0}$ as if $\PROJ{\FV{i}{}}{0} \in
\FV{i}{}$.\footnote{Since $\PROJ{\FV{i}{}}{0}$ is just a shortcut and not a
  proper variable, $\SIZE{\FV{i}{}}$ refers to the number of
variables in $\FV{i}{}$ \emph{apart from} $\PROJ{\FV{i}{}}{0}$.}
%
%

\paragraph{Systems of Inductive Definitions}
Every predicate symbol is associated with one or more symbolic heaps by a \emph{system} \emph{of} \emph{inductive} \emph{definitions} (SID).
Formally, an SID is a finite set of rules of the form $\SRDRULE{\PS}{\sh}$, where $\sh$ is a symbolic heap with $\ARITY(\PS) = \NOFV{\sh}$.
The set of all predicate symbols occurring in SID $\SRD$ and their
maximal arity are denoted by $\PRED(\SRD)$ and $\ARITY(\SRD)$,
respectively.
\begin{example}\label{ex:srd}
 An SID specifying doubly-linked list segments is defined by:
 \begin{align*}
   \textnormal{\texttt{dll}} ~\SRDARROW~ & \EMP : \{ a = c, b = d \} \\
   \textnormal{\texttt{dll}} ~\SRDARROW~ & \exists u ~.~
   \PT{a}{u \, b}
   \SEP \textnormal{\texttt{dll}}(u\,a\,c\,d),
 \end{align*}
 where $a$ corresponds to the \emph{head} of the list, $b$ and $c$ represent the \emph{previous} and the \emph{next} list element and $d$ represents the \emph{tail} of the list.
 For the sake of readability, we often prefer $a$, $b$, $c$, etc. as free variables in examples instead of $\IFV{1},\IFV{2},\IFV{3}$, etc.
 Further, 
 the following rules specify binary trees with
 \emph{root} $a$, \emph{leftmost leaf}
 $b$ and \emph{successor of the rightmost leaf}
 $c$ in which all leaves are connected by a
 singly-linked list from left to right.
 \begin{align*}
   \textnormal{\texttt{tll}} ~\SRDARROW~ & \PT{a}{\NIL\,\NIL\,c} : \{ a = b \} \\
   \textnormal{\texttt{tll}} ~\SRDARROW~ &
   \begin{array}[t]{@{}r@{~}l@{}}
     \exists \ell\,r\,z ~. & \PT{a}{\ell\,r\,\NIL}
     \SEP                    \textnormal{\texttt{tll}}(\ell\,b\,z) \SEP
                             \textnormal{\texttt{tll}}(r\,z\,c).
   \end{array}
 \end{align*}
\end{example}
\begin{definition}
We write $\SL{}{}$ for the set of all symbolic
heaps and $\SL{\SRD}{}$ for the set of symbolic heaps restricted to
predicate symbols taken from SID $\SRD$.
Moreover, given a computable function $\SRDCLASS : \SL{}{} \to \{0, 1\}$,
the set of symbolic heaps $\SL{}{\SRDCLASS}$  is given by
$\SL{}{\SRDCLASS} \DEFEQ \{ \sh \in \SL{}{} \mid \SRDCLASS(\sh) = 1 \}$.
%
We collect all SIDs in which every right-hand side belongs to
$\SL{}{\SRDCLASS}$ in $\SETSRD{\SRDCLASS}$.
To refer to the set of all reduced symbolic heaps (belonging to a
set defined by $\SRDCLASS$), we write $\RSL{}{}$
$(\RSL{}{\SRDCLASS})$.
\end{definition}
\begin{example}\label{ex:shclass}
Let $\alpha \in \mathbb{N}$ and $\SRDCLASSFV{\alpha}(\sh) \DEFEQ
\begin{cases}
  1, & \NOFV{\sh} \leq \alpha \\
  0, & \text{otherwise}
\end{cases}$.

Clearly, $\SRDCLASSFV{\alpha}$ is computable. Moreover, 
$\SL{}{\SRDCLASSFV{\alpha}}$ is the set of all symbolic heaps
having at most $\alpha \geq 0$ free variables.
\end{example}
\paragraph{Semantics}
As in a typical RAM model, we assume heaps to consist of records with a finite number of fields.
Let $\VAL$ denote an infinite set of \emph{values} and $\LOC \subseteq \VAL$ an infinite set of addressable \emph{locations}.
Moreover, we assume the existence of a special non-addressable value $\NIL \in \VAL \setminus \LOC$.
A \emph{heap} is a finite partial function $\heap : \LOC \PTO \VAL^{+}$ mapping locations to non-empty tuples of values.
We write $\heap_1 \HEAPUNION \heap_2$ to denote the union of heaps $\heap_1$ and $\heap_2$ provided that $\DOM(\heap_1) \cap \DOM(\heap_2) = \emptyset$.
Otherwise, $\heap_1 \HEAPUNION \heap_2$ is undefined.
Variables are interpreted by a \emph{stack}, i.e., a partial function $\stack~:~\VAR \PTO \VAL$
with $\stack(\NIL) = \NIL$.
%
Furthermore, stacks are canonically extended to tuples of variables by componentwise application.
We call a stack--heap pair $(\stack,\heap)$ a \emph{state}.
The set of all states is $\STATES$.
The semantics of a symbolic heap with respect to an SID and a state is shown in Figure~\ref{fig:slsemantics}.
Note that the semantics of predicate calls is explained in detail next.
\begin{figure}[t]
 \begin{align*}
     \stack,\heap \,\SAT{\SRD}\,
   & x \sim y
   & \Leftrightarrow~
   & \stack(x) \sim \stack(y), ~\text{where}~\sim \,\in\, \{\,=,\neq\,\} \\
   %
   %
    \stack,\heap \,\SAT{\SRD}\,
   & \EMP
   & \Leftrightarrow~
   & \DOM(\heap) = \emptyset \\
    \stack,\heap \,\SAT{\SRD}\,
   & \PTS{x}{\T{y}}
   & \Leftrightarrow~
   & \DOM(\heap) = \{\stack(x)\} ~\text{and}~ \heap(\stack(x)) = \stack(\T{y}) \\
    \stack,\heap \,\SAT{\SRD}\,
   & \PS\T{y}
   & \Leftrightarrow~
   & \exists \rsh \in \CALLSEM{P\T{y}}{\SRD} \,.\, \stack,\heap \SAT{\emptyset} \rsh \\
    \stack,\heap \,\SAT{\SRD}\,
   & \sh \SEP \sha
   & \Leftrightarrow~
   & \exists \heap_1,\heap_2 \,.\, \heap = \heap_1 \HEAPUNION \heap_2 \\
   & & & \quad \text{and } \stack,\heap_1 \SAT{\SRD} \sh \text{ and } \stack,\heap_2 \SAT{\SRD} \sha \\
    \stack,\heap \,\SAT{\SRD}\,
   & \exists \BV{} . \SPATIAL{} \SEP \CALLS{} \!:\! \PURE{}\!\!\!\!\!\!\!\!
   & ~\Leftrightarrow~
   & \exists \T{v} \in \VAL^{\SIZE{\BV{}}} \,.\,
     \stack\remap{\BV{}}{\T{v}}, \heap \SAT{\SRD} \SPATIAL{} \SEP \CALLS{} \\
   & & & \quad \text{and } \forall \pi \in \PURE{} \,.\, \stack\remap{\BV{}}{\T{v}}, \heap \SAT{\SRD} \pi
 \end{align*}
 \caption{Semantics of the symbolic heap fragment of separation logic with respect to an SID $\SRD$ and a state $(\stack,\heap)$.}
 \label{fig:slsemantics}
\end{figure}
\paragraph{Unfoldings of Predicate Calls}
The semantics of predicate calls is defined in terms of \emph{unfolding} \emph{trees}.
Intuitively, an unfolding tree specifies how predicate calls are \emph{replaced} by symbolic heaps according to a given SID.
The resulting reduced symbolic heap obtained from an unfolding tree is consequently called an \emph{unfolding}.
Formally, let $\sh = \exists \BV{} . \SPATIAL{} \SEP \CALLN{1}{} \SEP \ldots \SEP \CALLN{m}{} : \PURE{}$.
Then a predicate call $\CALLN{i}{}$ may be \emph{replaced} by a reduced
symbolic heap $\rsh$ if $\SIZE{\FV{i}{}} = \NOFV{\rsh}$ and
$\VAR(\sh) \cap \VAR(\rsh) \subseteq \FV{0}{\rsh}$.
The result of such a replacement is
\begin{align*}
 \sh\subst{\PS_i}{\rsh} ~\DEFEQ~ & \exists \BV{}\, \BV{\rsh}
                                               \,.\, \SPATIAL{}
                                               \SEP \SPATIAL{\rsh\subst{\FV{0}{\rsh}}{\FV{i}{}}} \, \SEP \\
                                               & \CALLN{1}{} \SEP \ldots \SEP \CALLN{i-1}{} \SEP \CALLN{i+1}{} \SEP \ldots \SEP \CALLN{m}{}
                                                \,:\, \left(\PURE{} \cup \PURE{\rsh\subst{\FV{0}{\rsh}}{\FV{i}{}}}\right),
\end{align*}
where $\rsh\subst{\FV{0}{\rsh}}{\FV{i}{}}$ denotes the substitution
of each free variable of $\rsh$ by the corresponding parameter of $\PS_i^{}$.
A \emph{tree} over symbolic heaps $\SL{\SRD}{}$ is a finite partial function $t : \N^{*} \PTO \SL{\SRD}{}$ such that $\emptyset \neq \DOM(t) \subseteq \N^{*}$ is prefix-closed and for all $\T{u} \in \DOM(t)$ with $t(\T{u}) = \sh$, we have $\{1,\ldots,\NOCALLS{\sh}\} = \{ i \in \N ~|~ \T{u}\,i \in \DOM(t) \}$.
The element $\EMPTYSEQ \in \DOM(t)$ is called the \emph{root} of tree $t$.
Furthermore, the \emph{subtree} $\SUBTREE{t}{\T{u}}$ of $t$ with root $\T{u}$ is  $\SUBTREE{t}{\T{u}} \,:\,  \{ \T{v} ~|~ \T{u}\,\T{v} \in \DOM(t) \} \to \SL{\SRD}{}$ with $\SUBTREE{t}{\T{u}}(\T{v}) \DEFEQ t(\T{u} \cdot \T{v})$.
%

\begin{definition} \label{def:sl:unfolding-trees}
 Let $\SRD \in \SETSRD{}$ and $\sh \in \SL{\SRD}{}$.
 Then the set of \emph{unfolding} \emph{trees} of $\sh$ w.r.t. $\SRD$,
 written $\UTREES{\SRD}{\sh}$, is the least set that contains all
 trees $t$ that satisfy (1) $t(\EMPTYSEQ) = \sh$ and (2)
 $\SUBTREE{t}{i} \in \UTREES{\SRD}{\sha_i}$ for each
 $1 \leq i \leq \NOCALLS{\sh}$, where
 $\SRDRULE{\PS_i^{\sh}}{\sha_i} \in \SRD$.
\end{definition}
   Note that for every reduced symbolic heap $\rsh$, we have $\NOCALLS{\rsh} = 0$.
   Thus, $\UTREES{\SRD}{\rsh} = \{ t \}$, where $t : \{ \EMPTYSEQ \} \to \{ \rsh \}:
   \EMPTYSEQ \mapsto \rsh$, forms the base case in
   Definition~\ref{def:sl:unfolding-trees}.
%
%
Every unfolding tree $t$ specifies a reduced symbolic heap $\UNFOLD{t}$, which is obtained by recursively replacing predicate calls by reduced symbolic heaps:
\begin{definition} \label{def:symbolic-heaps:unfolding}
The \emph{unfolding} of an unfolding tree $t \in \UTREES{\SRD}{\sh}$ is
%
\begin{align*}
 \UNFOLD{t} ~\DEFEQ~ \begin{cases}
                        t(\EMPTYSEQ) & ,~ \NOCALLS{t(\EMPTYSEQ)} = 0 \\
                        t(\EMPTYSEQ)\left[P_1 / \UNFOLD{\SUBTREE{t}{1}}, \ldots, P_m / \UNFOLD{\SUBTREE{t}{m}}\right] & ,~ \NOCALLS{t(\EMPTYSEQ)} = m > 0~,
                     \end{cases}
\end{align*}
where we tacitly assume that the variables $\BV{t(\EMPTYSEQ)}$, i.e., the existentially quantified variables in $t(\EMPTYSEQ)$, are substituted by fresh variables.
\end{definition}
\begin{example}
  Recall from Example~\ref{ex:srd}
  the two symbolic heaps $\rsh$ (upper) and $\sh$ (lower)
  occurring on the right-hand side of the \texttt{dll} predicate.
  Then $t : \{\EMPTYSEQ,1\} \to \{\sh, \rsh\}: \EMPTYSEQ \mapsto \sh, 1 \mapsto \rsh$
  is an unfolding tree of $\sh$.
  The corresponding unfolding is 
  \begin{align*}
    \UNFOLD{t} = \sh\left[ \PS_{1}^{\sh} / \rsh \right] ~=~ &
        \exists z ~.~ \PT{
           a
        }{
          z ~
          b
        }
        \SEP \EMP
         ~:~ \{ z = c, a = d \}.
  \end{align*}
\end{example}

%
\begin{definition} \label{def:symbolic-heaps:unfoldings}
  The set of all \emph{unfoldings} of a predicate call $\CALLN{i}{}$
  w.r.t. an SID $\SRD$ is denoted by
  $\CALLSEM{\CALLN{i}{}}{\SRD}$.
  Analogously, the \emph{unfoldings} of a symbolic heap $\sh$ are
  $ \CALLSEM{\sh}{\SRD} ~\DEFEQ~ \{ \UNFOLD{t} ~|~ t \in
  \UTREES{\SRD}{\sh} \} $.
\end{definition}

Then, as already depicted in Figure~\ref{fig:slsemantics}, the semantics of predicate calls requires the existence of an unfolding satisfying a given state.
This semantics corresponds to a particular iteration of the frequently used semantics of predicate calls based on least fixed points
(cf.~\cite{brotherston2014decision}).
%
%
%
%
%
Further note that applying the SL semantics to a given symbolic heap coincides with applying them to a suitable unfolding.
\begin{lemma} \label{thm:symbolic-heaps:semantics}
 Let $\sh \in \SL{\SRD}{}$.
 Then, for every $(\stack,\heap) \in \STATES$, we have
 \begin{align*}
   \stack,\heap \SAT{\SRD} \sh ~\text{iff}~ \exists \rsh \in \CALLSEM{\sh}{\SRD} ~.~ \stack,\heap \SAT{\emptyset} \rsh.
 \end{align*}
\end{lemma}
\begin{proof}
 By induction on the height of unfolding trees of $\sh$.
 \qed
\end{proof}

\newcommand{\btllExample}{%
\definecolor{colB}{RGB}{227, 123, 64}
\definecolor{colE}{RGB}{70, 178, 157}

  \begin{tikzpicture}[scale=0.65]
     \tikzset{inner/.style={circle,draw=colB,fill=colB!50,thick,inner
         sep=2pt,minimum size=18pt}}
     \tikzset{leaf/.style={circle,draw=colE,fill=colE!50,thick,inner
         sep=2pt,minimum size=18pt}}
     \tikzset{tedge/.style={->,thick,draw=colB}}
     \tikzset{ledge/.style={->,thick,draw=colE}}
   \node[leaf] (r1) {};
   \node[leaf,right=of r1] (r2) {};

   \node[inner,above left=of r1] (l) {};
   \node[inner,right=of l] (r) {};

   \node[inner,below left=of l] (l1) {};
   \node[leaf,below=of l] (l2) {};
   \node[leaf,below left=of l1] (l11) {f};
   \node[leaf,below right=of l1] (l12) {};

   \node[inner,above right=of l] (root) {r};

    \draw (root) edge[tedge] (l);
    \draw (root) edge[tedge] (r);
    \draw (r) edge[tedge] (r1);
    \draw (r) edge[tedge] (r2);

    \draw (l) edge[tedge] (l1);
    \draw (l) edge[tedge] (l2);
    \draw (l1) edge[tedge] (l11);
    \draw (l1) edge[tedge] (l12);

    \draw (l11) edge[ledge] (l12);
    \draw (l12) edge[ledge] (l2);
    \draw (l2) edge[ledge] (r1);
    \draw (r1) edge[ledge] (r2);

   \node[leaf,dashed,below=of r2,fill=colE!20] (null) {null};
    \draw (r2) edge[ledge] (null);

  \end{tikzpicture}
}


%
\section{Heap Automata}
\label{sec:compositional}
In this section we develop a procedure to reason about robustness properties of symbolic heaps. 
This procedure relies on the notion of \emph{heap automata}; a device that assigns one of finitely many states to any given symbolic heap.
\begin{definition} 
  A \emph{heap} \emph{automaton} over $\SHCLASS$ is a tuple $\HA{A} = (Q,\SHCLASS,\Delta,F)$, where
  $Q$ is a finite set of \emph{states} and $F \subseteq Q$ is a set of \emph{final states}, respectively.
  Moreover, $\Delta \subseteq Q^{*} \times \SHCLASS \times Q$ is a decidable \emph{transition relation} such that $(\T{q},\sh,p) \in \Delta$
  implies that $\SIZE{\T{q}} = \NOCALLS{\sh}$.
  We often write $\MOVE{A}{p}{\sh}{\T{q}}$ instead of $(\T{q},\sh,p) \in \Delta$.
\end{definition}
A transition $\MOVE{A}{p}{\sh}{\T{q}}$ takes a symbolic heap $\sh$ and an input state $q_i$ for every predicate call $\PS_i$ of $\sh$---collected in the tuple $\T{q}$---and assigns an output state $p$ to $\sh$.
Thus, the intuition behind a transition 
is that $\sh$ has a property encoded by state $p$
if every predicate call $\PS_i$ of $\sh$ is replaced by a reduced symbolic heap $\rsh_i$ that has a property encoded by state $\PROJ{\T{q}}{i}$.

Note that every heap automaton $\HA{A}$ assigns a state $p$ to a reduced symbolic heap $\rsh$ within a single transition of the form $\OMEGA{A}{p}{\rsh}$.
Alternatively, $\HA{A}$ may process a corresponding unfolding tree $t$ with $\UNFOLD{t} = \rsh$.
In this case, $\HA{A}$ proceeds similarly to the compositional construction of unfoldings (see Definition~\ref{def:symbolic-heaps:unfolding}).
However, instead of replacing every predicate call $\PS_i$ of the symbolic heap $t(\EMPTYSEQ)$ at the root of $t$ by an unfolding $\UNFOLD{\SUBTREE{t}{i}}$ of a subtree of $t$,
$\HA{A}$ uses states to keep track of the properties of these unfolded subtrees.
Consequently, $\HA{A}$ assigns a state $p$ to the symbolic heap $t(\EMPTYSEQ)$ if $\MOVE{A}{p}{t(\EMPTYSEQ)}{(q_1,\ldots,q_m)}$ holds, where 
for each $1 \leq i \leq m$, $q_i$ is the state assigned to the unfolding of subtree $\SUBTREE{t}{i}$, i.e., there is a transition $\OMEGA{A}{q_i}{\UNFOLD{\SUBTREE{t}{i}}}$.
It is then natural to require that $p$ should coincide with the state assigned directly to the unfolding $\UNFOLD{t}$, i.e., $\OMEGA{A}{p}{\UNFOLD{t}}$.
Hence, we require all heap automata considered in this paper to satisfy a \emph{compositionality property}.

\begin{definition} \label{def:refinement:automaton}
    A heap automaton $\HA{A} = (Q,\SHCLASS,\Delta,F)$ is \emph{compositional} if for every $p \in Q$, every $\sh \in \SHCLASS$ with $m \geq 0$ predicate calls
    $\CALLS{\sh} = \CALLN{1}{} \SEP \ldots \SEP \CALLN{m}{}$, and all reduced symbolic heaps $\rsh_1,\ldots,\rsh_m \in \RSHCLASS$, we have:
  %
  \[
  \begin{array}{c}
     \exists \T{q} \in Q^m ~.~ (\T{q},\sh,p) \in \Delta ~\text{and}~ \bigwedge_{1 \leq i \leq m} (\EMPTYSEQ,\rsh_i,\PROJ{\T{q}}{i}) \in \Delta \\
  \text{if and only if} \\
  (\EMPTYSEQ,~ \sh\left[P_1/\rsh_1,\ldots,P_m/\rsh_m\right],~p) ~\in~ \Delta.
   \end{array}
  \]
\end{definition}

Due to the compositionality property, we can safely define the \emph{language} $L(\HA{A})$ \emph{accepted} by a heap automaton $\HA{A}$ as the set of all reduced symbolic heaps that are assigned a final state, i.e.,
 $L(\HA{A}) \,\DEFEQ\, \{ \rsh \in \RSHCLASS ~|~ \exists q \in F \,.\, \OMEGA{A}{q}{\rsh} \}$.

\begin{example} \label{ex:heap-automaton:toy}
Given a symbolic heap $\sh$, let $|\SPATIAL{\sh}|$ denote the number of points-to assertions in $\sh$.
As a running example, we consider a heap automaton $\HA{A} =
(\{0,1\},\SL{}{},\Delta,\{1\})$,
  where $\Delta$ is given by
\begin{align*}
  \MOVE{A}{p}{\sh}{\T{q}} ~\text{iff}~
  p = \begin{cases}
        1, & \text{if}~ |\SPATIAL{\sh}| + \sum_{i=1}^{\SIZE{\T{q}}} \PROJ{\T{q}}{i}   > 0 \\
        0, & \text{otherwise}.
      \end{cases}
\end{align*}
While $\HA{A}$ is a toy example, it illustrates the compositionality property:
Consider the reduced symbolic heap $\rsh = \exists z . \EMP \SEP \EMP : \{ x = z, z = y \}$.
Since $\rsh$ contains no points-to assertions, $\HA{A}$ rejects $\rsh$ in a single step, i.e., $\OMEGA{A}{0}{\rsh} \notin \{1\}$.
The compositionality property of $\HA{A}$ ensures that $\HA{A}$ yields the same result for every unfolding tree $t$ whose unfolding $\UNFOLD{t}$ is equal to $\rsh$.
For instance, $\rsh$ is a possible unfolding of the symbolic heap $\sh = \exists z . \texttt{sll}(x z) \SEP \texttt{sll}(z y)$, where $\texttt{sll}$ is a predicate specifying singly-linked list segments as in Section~\ref{sec:introduction}.
More precisely, if both predicates are replaced according to the rule $\SRDRULE{\texttt{sll}}{\EMP : \{x = y\}}$, we obtain $\rsh$ again (up to renaming of parameters as per Definition~\ref{def:symbolic-heaps:unfolding}).
In this case, $\HA{A}$ rejects as before: We have $\OMEGA{A}{0}{\EMP : \{x=y\}}$ for both base cases and $\MOVE{A}{0}{\sh}{(0,0)}$ for the symbolic heap $\sh$.
By the compositionality property, this is equivalent to $\OMEGA{A}{0}{\rsh}$.
Analogously, if an $\texttt{sll}$ predicate, say the first, is replaced according to the rule $\SRDRULE{\texttt{sll}}{\sha}$, where $\sha = \exists z . \PTS{x}{z} \SEP \texttt{sll}(z y)$, then
$\MOVE{A}{1}{\sha}{0}$, $\MOVE{A}{1}{\sha}{1}$ and $\MOVE{A}{1}{\sh}{(1,0)}$ holds, i.e., $\HA{A}$ accepts.
In general, $L(\HA{A})$ is the set of all reduced symbolic heaps that contain at least one points-to assertion.
%
%
\end{example}

While heap automata can be applied to check whether a single reduced symbolic heap has a property of interest, i.e., belongs to the language of a heap automaton,
our main application is directed towards reasoning about \emph{infinite} sets of symbolic heaps,
such as all unfoldings of a symbolic heap $\sh$.
Thus, given a heap automaton $\HA{A}$, we would like to answer the following questions:
\begin{compactenum}
    \item[1.] Does there \emph{exist} an unfolding of $\sh$ that is accepted by $\HA{A}$?
    \item[2.] Are \emph{all} unfoldings of $\sh$ accepted by $\HA{A}$?
\end{compactenum}
%
We start with a special case of the first question in which $\sh$ is a single predicate call.
The key idea behind our corresponding decision procedure is to transform the SID $\SRD$ to \emph{filter out} all unfoldings that are \emph{not} accepted by $\HA{A}$.
One of our main results is that such a refinement is always possible.
\begin{theorem}[Refinement Theorem] \label{thm:compositional:refinement}
  Let $\HA{A}$ be a heap automaton over $\SHCLASS$ and $\SRD \in \SETSRD{\SRDCLASS}$.
  Then one can effectively construct a refined $\SRDALT \in \SETSRD{\SRDCLASS}$ such that
  for each $\PS \in \PRED(\SRD)$, we have $\CALLSEM{\PS\FV{0}{}}{\SRDALT} = \CALLSEM{\PS\FV{0}{}}{\SRD} ~\cap~ L(\HA{A})$.
\end{theorem}
\begin{proof}
 We construct $\SRDALT \in \SETSRD{\SRDCLASS}$ over the predicate symbols
 $\PRED(\SRDALT) = (\PRED(\SRD) \times Q_{\HA{A}}) \cup \PRED(\SRD)$
 as follows:
 If $\SRDRULE{\PS}{\sh} \in \SRD$ with
 $\CALLS{\sh} = \CALLN{1}{} \SEP \ldots \SEP \CALLN{m}{}$, $m \geq 0$,
 and $\MOVE{A}{q_0}{\sh}{q_1 \ldots q_m}$, we add a rule to
 $\SRDALT$ in which $\PS$ is substituted by $\langle\PS,q_0\rangle$ and each
 predicate call $\CALLN{i}{}$ is substituted by a call
 $\langle\PS_i,q_i\rangle\FV{i}{}$.
 %
 %
 Furthermore, for each $q \in F_{\HA{A}}$, we add a rule $\PS \SRDARROW \langle\PS,q\rangle\FV{0}{}$ to $\SRDALT$.
 %
 See Appendix~\ref{app:compositional:refinement} for details.
 \qed
\end{proof}
\begin{example} \label{ex:refinement:toy}
Applying the refinement theorem to the heap automaton from Example~\ref{ex:heap-automaton:toy}
and the SID
from Example~\ref{ex:srd} yields a refined SID given by the rules:
%
\begin{align*}
  &\quad~~\, \SRDRULE{
    \textnormal{\texttt{dll}}
  }{
    \langle\textnormal{\texttt{dll}},1\rangle(a\,b\,c\,d)
  }
  && \SRDRULE{
    \langle\textnormal{\texttt{dll}},1\rangle
  }{
   \exists z ~.~ \PT{a}{z \, b}
   \SEP \langle\textnormal{\texttt{dll}},0\rangle(z\,a\,c\,d)
  } \\
  & \SRDRULE{
    \langle\textnormal{\texttt{dll}},0\rangle
  }{
 \EMP : \{ a  = c, b = d \}
  }
  && \SRDRULE{
    \langle\textnormal{\texttt{dll}},1\rangle
  }{
   \exists z ~.~ \PT{a}{z \, b}
   \SEP \langle\textnormal{\texttt{dll}},1\rangle(z\,a\,c\,d)
  }
\end{align*}
Hence, the refined predicate $\textnormal{\texttt{dll}}$ specifies all non-empty doubly-linked lists.
%
\end{example}

To answer question (1) we then check whether the set of unfoldings of a refined SID is non-empty.
This  boils down to a simple reachability analysis.
\begin{lemma} \label{thm:symbolic-heaps:emptiness}
  Given an SID $\SRD$ and a predicate symbol $\PS \in \PRED(\SRD)$,
  it is decidable in linear time
  whether the set of unfoldings of $\PS$ is empty, i.e., $\CALLSEM{\PS\T{x}}{\SRD} = \emptyset$.
\end{lemma}
\begin{proof}[sketch]
    It suffices to check whether the predicate $\PS$ lies in the least set $R$ such that
    (1) $I \in R$ if $\SRDRULE{I}{\rsh} \in \SRD$ for some $\rsh \in \RSL{}{}$, and
    (2) $I \in R$ if $\SRDRULE{I}{\sh} \in \SRD$ and for each $\CALLN{i}{\sh}$, $1 \leq i \leq \NOCALLS{\sh}$, $\PS_i^{\sh} \in R$.
The set $R$ is computable in linear time by a straightforward backward reachability analysis.
\qed
\end{proof}
As outlined before, putting the Refinement Theorem and
Lemma~\ref{thm:symbolic-heaps:emptiness} together immediately yields a
decision procedure for checking whether some unfolding of a predicate
symbol $P$ is accepted by a heap automaton:
%
Construct the refined SID 
and subsequently check whether the set of unfoldings of $P$ is non-empty.

To extend this result from unfoldings of single predicates to
unfoldings of arbitrary symbolic heaps $\sh$, we just add a rule
$\SRDRULE{\PS}{\sh}$, where $\PS$ is a fresh predicate symbol, and
proceed as before.
\begin{corollary} \label{thm:compositional:existence}
 Let $\HA{A}$ be a heap automaton over $\SHCLASS$ and $\SRD \in \SETSRD{\SRDCLASS}$.
 Then, for each $\sh \in \SL{\SRD}{\SRDCLASS}$, it is decidable whether there exists $\rsh \in \CALLSEM{\sh}{\SRD}$ such that $\rsh \in L(\HA{A})$.
\end{corollary}
\begin{algorithm}[t]
  \SetKwInOut{Input}{Input}\SetKwInOut{Output}{Output}
  \hrule \vspace*{0.5em}
  \Input{SID $\SRD$,~ $I \in \PRED(\SRD)$,~ $\HA{A} = (Q,\SHCLASS,\Delta,F)$ }
  \Output{yes iff $\CALLSEM{I\T{x}}{\SRD} \cap L(\mathfrak{A}) = \emptyset$ }
  \vspace*{0.5em} \hrule
  \BlankLine
  {R$~\leftarrow~\emptyset$};\\
  \Repeat{R reaches a fixed point (w.r.t. all choices of rules)}{%
    \lIf{$R \cap \left(\{I\} \times F\right) \neq \emptyset$}{\Return{no}}
    {pick a state  $q$ in $Q$};
    {pick a rule $\SRDRULE{\PS}{\sh}$ in $\SRD$};\\
    s$~\leftarrow~ \EMPTYSEQ$; ~\text{// list of states of $\HA{A}$}\\
    \For{$i$ in $1$ to $\NOCALLS{\sh}$ }
    {
      pick $(\PS_i^{\sh},p) \in R$;
      append(s,$p$) ~\text{// base case if $\NOCALLS{\sh} = 0$}
    }
    \lIf{$(s,\sh,q) \in \Delta$}{
	    R$~\leftarrow~$ R $~\cup~ \{(P,q)\}$
    }
  }
  \Return{yes}
  \hrule \vspace*{0.5em}
  \caption{
     On-the-fly construction of a refined SID with
     emptiness check.
  }
  \label{alg:on-the-fly-refinement}
\end{algorithm}
The refinement and emptiness check can also be integrated:
Algorithm~\ref{alg:on-the-fly-refinement} displays a simple procedure that constructs the refined SID $\SRDALT$ from Theorem~\ref{thm:compositional:refinement} on-the-fly while checking whether its set of unfoldings is empty for a given predicate symbol.
Regarding complexity, the size of a refined SID\footnote{We assume a reasonable function $\SIZE{.}$ assigning a size to SIDs, symbolic heaps, unfolding trees, etc. For instance, the size $\SIZE{\SRD}$ of an SID $\SRD$ is given by the product of its number of rules and the size of the largest symbolic heap contained in any rule.} obtained from an SID $\SRD$
and a heap automaton $\HA{A}$ is bounded by $\SIZE{\SRD} \cdot \SIZE{Q_{\HA{A}}}^{M+1}$,
where $M$ is the maximal number of predicate calls occurring in any rule of $\SRD$.
Thus, the aforementioned algorithm runs in time
$\BIGO{\SIZE{\SRD} \cdot \SIZE{Q_{\HA{A}}}^{M+1} \cdot \SIZE{\Delta_{\HA{A}}}}$,
where $\SIZE{\Delta_{\HA{A}}}$ denotes the complexity of deciding whether the transition relation $\Delta_{\HA{A}}$ holds for a given  tuple of states and a symbolic heap occurring in a rule of $\SRD$.
\begin{example}
  Resuming our toy example, we check whether some unfolding of the doubly-linked list predicate  \texttt{dll} (see Example~\ref{ex:srd})
  contains points-to assertions. 
  Formally, we decide whether $\CALLSEM{\texttt{dll}\,\FV{0}{}}{\SRD} \cap L(\HA{A}) \neq \emptyset$, where $\HA{A}$ is the heap automaton introduced in Example~\ref{ex:heap-automaton:toy}.
  Algorithm~\ref{alg:on-the-fly-refinement} first picks the rule that
  maps \texttt{dll} to the empty list segment and consequently adds
  $\langle\texttt{dll},0\rangle$ to the set $R$ of reachable predicate--state
  pairs.  In the next iteration, it picks the rule that maps to the
  non-empty list.  Since $\langle\mathtt{dll},0\rangle \in R$, $s$ is set to $0$
  in the \textbf{do}-loop. Abbreviating the body of the rule to $\sh$,
  we have $(0,\sh,1) \in \Delta$, so the algorithm adds
  $\langle\texttt{dll},1\rangle$ to $R$.  After that, \emph{no} is returned,
  because $1$ is a final state of $\HA{A}$.  Hence, some unfolding of
  $\texttt{dll}$ is accepted by $\HA{A}$ and thus contains points-to
  assertions.
\end{example}
%
We now revisit question (2) from above--are all unfoldings accepted by a heap automaton?-- and observe that heap automata
enjoy several closure properties.
\begin{theorem}\label{thm:refinement:boolean}
  Let $\HA{A}$ and $\HA{B}$ be heap automata over $\SHCLASS$.
  Then there exist heap automat $\HA{C}_1,\HA{C}_2,\HA{C}_3$ over $\SHCLASS$ with
  $L(\HA{C}_1) = L(\HA{A}) \cup L(\HA{B})$,
  $L(\HA{C}_2) = L(\HA{A}) \cap L(\HA{B})$, and
  $L(\HA{C}_3) = \RSHCLASS \setminus L(\HA{A})$, respectively.\footnote{Formal constructions are found in Appendix~\ref{app:refinement:boolean}.}
  %
\end{theorem}
%
%
Then, by the equivalence $X \subseteq Y \Leftrightarrow X \cap \overline{Y} = \emptyset$ and Theorem~\ref{thm:refinement:boolean}, it is also decidable whether every unfolding of a symbolic heap is accepted by a heap automaton.
\begin{corollary} \label{thm:compositional:inclusion}
 Let $\HA{A}$ be a heap automaton over $\SHCLASS$ and $\SRD \in \SETSRD{\SRDCLASS}$.
 Then, for each $\sh \in \SHCLASS$, it is decidable whether $\CALLSEM{\sh}{\SRD} \subseteq L(\HA{A})$ holds.
\end{corollary}
%
Note that complementation of heap automata
in general leads to an exponentially larger state space and
exponentially higher complexity of evaluating $\Delta$.
Thus, $\CALLSEM{\sh}{\SRD} \subseteq L(\HA{A})$ is decidable in time
$
  \BIGO{\left(\SIZE{\sh} + \SIZE{\SRD}\right) \cdot
  \SIZE{2^{Q_{\HA{A}}}}^{2(M+1)} \cdot  \SIZE{\Delta_{\HA{A}}}}
$.
In many cases it is, however, possibly to construct smaller automata
for the complement directly to obtain more efficient decision
procedures.
For example, this is the case for most heap automata considered in Section~\ref{sec:zoo}.
Apart from decision procedures, Theorem~\ref{thm:compositional:refinement} enables systematic refinement of SIDs according to heap automata in order to establish desired properties.
%
For instance, as shown in Section~\ref{sec:zoo}, an SID in which every unfolding is satisfiable can be constructed from any given SID.
Another application of Theorem~\ref{thm:compositional:refinement} is counterexample generation for
systematic debugging of SIDs that are manually written as data structure specifications
or even automatically generated.
Such counterexamples are obtained by constructing the refined SID of the complement of a given heap automaton.
%
Further applications are examined in the following. 

\begin{remark}
    While we focus on the well-established symbolic heap fragment of separation logic, we remark that the general reasoning principle underlying heap automata is also applicable to
    check robustness properties of richer fragments.
    For example, permissions~\cite{bornat2005permission} are easily integrated within our framework.
\end{remark}


%
\section{A Zoo of Robustness Properties}
\label{sec:zoo}
This section demonstrates the wide applicability of heap automata to
decide and establish robustness properties of SIDs.
In particular, the sets of symbolic heaps informally presented in the introduction can be accepted by heap automata over the set
$\SHCLASSFV{\alpha}$ of symbolic heaps with at most
$\alpha \geq 0$ free variables (cf. Example~\ref{ex:shclass}). 
%
Furthermore, we analyze the complexity of related decision problems. 
Towards a formal presentation, some terminology is needed.
\begin{definition}\label{def:symbolic-heaps:models}
 %
 The set of \emph{tight} \emph{models} of a symbolic heap $\sh \in \SL{\SRD}{}$ is defined as
 %
 $\MODELS{\sh} \DEFEQ \{ (\stack,\heap) \in \STATES \,|\, \DOM(\stack) = \FV{0}{\sh} ,\, \stack,\heap \SAT{\SRD} \sh \}$.
 %
\end{definition}
We often consider relationships between variables that
hold in every tight model of a reduced symbolic heap.
Formally, let $\rsh \DEFEQ \exists \BV{} . \SPATIAL{} : \PURE{} \in \RSL{}{}$.
Moreover, let $\STRIP{\rsh}$ be defined as $\rsh$ except that each of its variables is free,
 i.e., $\STRIP{\rsh} \DEFEQ \SPATIAL{} : \PURE{}$.
Then two variables $x,y \in \VAR(\rsh)$ are \emph{definitely} \emph{(un)equal} in $\rsh$,
written $x \MEQ{\rsh} y$ ($x \MNEQ{\rsh} y$),
 if $\stack(x) = \stack(y)$ ($\stack(x) \neq \stack(y)$) holds for every
$(s,h) \in \MODELS{\STRIP{\rsh}}$.
 %
 Analogously, a variable is \emph{definitely} \emph{allocated} if it is definitely equal to a variable occurring on the left-hand side of a points-to assertion.
 Thus the set of definitely allocated variables in
 $\rsh$ is given by
 \begin{align*}
  \ALLOC{\rsh} ~=~  \{ x \in \VAR(\rsh) ~|~ \forall (\stack,\heap) \in \MODELS{\STRIP{\rsh}} ~.~ \stack(x) \in \DOM(\heap) \}.
 \end{align*}
 Finally, a variable $x$ \emph{definitely} \emph{points-to} variable $y$ in $\rsh$,
 written $x \MPT{\rsh} y$, if for every $(\stack,\heap) \in \MODELS{\STRIP{\rsh}}$,
 we have $\stack(y) \in \heap(\stack(x))$.
\begin{example}
  Recall the symbolic heap $\rsh$ in the first rule of predicate
  $\texttt{tll}$ from Example~\ref{ex:srd}.
  Then $\ALLOC{\rsh} = \{ a, b \}$
  and neither $a \MEQ{\rsh} c$ nor $a \MNEQ{\rsh} c$ holds. 
  Further,
  \begin{align*}
    & a \MEQ{\rsh} b ~\text{is true,}~
    && a \MEQ{\rsh} c ~\text{is false,}~
    & a \MNEQ{\rsh} \NIL ~\text{is true,}~ \\
    & a \MNEQ{\rsh} c   ~\text{is false,}~
    && a \MPT{\rsh} c  ~\text{is true,}~
    & c \MPT{\rsh} a  ~\text{is false.}~ \quad
  \end{align*}
\end{example}
\begin{remark}\label{rem:closure}
%
All 
definite relationships are decidable in
polynomial time.
In fact, each of these relationships boils down
to first adding inequalities $x \neq \NIL$ and $x \neq y$ for every pair $x$,
$y$ of distinct variables occurring on the left-hand side of points-to assertions to the set of pure formulas
and then computing its (reflexive), symmetric (and transitive) closure with respect to $\neq$ (and $=$).
Furthermore, if the closure contains a contradiction, e.g., $\NIL \neq \NIL$, it is set to all pure formulas over the variables of a given reduced symbolic heap.
After that, it is straightforward to decide in polynomial time
whether variables are definitely allocated, (un)equal or pointing to each other.
%
\end{remark}
\subsection{Tracking Equalities and Allocation}\label{sec:zoo:track}
Consider the symbolic heap
$ \sh \DEFEQ \exists x\,y\,z . \PS_1(x~y) \SEP \PS_2(y~z) : \{ x = z \}$.
Clearly, $\sh$ is unsatisfiable if $x = y$ holds for every unfolding of $\PS_1(x~y)$ and $y \neq z$ holds for every unfolding of $\PS_2(y~z)$.
Analogously, $\sh$ is unsatisfiable if $x$ is allocated
in every unfolding of $\PS_1(x~y)$ and $z$ is allocated in every unfolding of $\PS_2(y~z)$,
because $\PTS{x}{\_} \SEP \PTS{z}{\_}$ implies $x \neq z$.
This illustrates that robustness properties, such as satisfiability, require detailed knowledge about the relationships between parameters of predicate calls.
Consequently, we construct a heap automaton $\HATRACK$ that keeps track of this knowledge.
More precisely, $\HATRACK$ should accept those unfoldings
in which it is guaranteed that
\begin{compactitem}
\item given a 
      set $A \subseteq \FV{0}{}$, exactly the variables in
  $A$ are definitely allocated, and
\item 
    exactly the (in)equalities in a given set of pure formulas $\PURE{}$ hold.
\end{compactitem}
%
%
Towards a formal construction, we formalize the desired set of symbolic heaps.
\begin{definition}\label{def:zoo:track}
  Let $\alpha \in \POSN$ and $\FV{0}{}$ be a tuple of variables with $\NOFV{} = \alpha$.
  Moreover, let $A \subseteq \FV{0}{}$ and $\PURE{}$ be a finite set of pure formulas over $\FV{0}{}$.
  %
  %
  The \emph{tracking} \emph{property} $\TRACK(\alpha,A,\PURE{})$ is the set
  \begin{align*}
    & \{ \rsh \in \RSHCLASSFV{\alpha} ~|~ \forall i,j ~.~ \PROJ{\FV{0}{}}{i} \in A ~\text{iff}~ \PROJ{\FV{0}{}}{i} \in \ALLOC{\rsh} \\
    & \qquad \text{and}~ \PROJ{\FV{0}{}}{i} \sim \PROJ{\FV{0}{}}{j} \in \PURE{} ~~\text{iff}~~ \PROJ{\FV{0}{\rsh}}{i} \MSIM{\rsh} \PROJ{\FV{0}{\rsh}}{j} \}.
  \end{align*}
\end{definition}
%
Intuitively, our heap automaton $\HATRACK$
stores in its state space which free variables are definitely equal, unequal and allocated.
Its transition relation then enforces that these stored information are correct, i.e.,
a transition $\MOVE{\HATRACK}{p}{\sh}{\T{q}}$ is only possible if the information stored in $p$ is consistent with
$\sh$ and with the information stored in the states $\T{q}$ for the predicate calls of $\sh$.
Formally, let $\FV{0}{}$ be a tuple of variables with $\NOFV{} = \alpha$ and
$\textnormal{Pure}(\FV{0}{})
  \DEFEQ 2^{\{ \PROJ{\FV{0}{}}{i} \sim \PROJ{\FV{0}{}}{j} ~|~ 0 \leq
  i,j \leq \alpha, \sim \in \{\,=,\,\neq\,\}  \}}$
be the powerset of all pure formulas over $\FV{0}{}$.
The information stored by our automaton consists of a set of free variables
$B \subseteq \FV{0}{}$
and a set of pure formulas $\Lambda \in \textnormal{Pure}(\FV{0}{})$.
Now, for some unfolding $\rsh$ of a symbolic heap $\sh$,
assume that $B$ is chosen as the set of all definitely allocated free variables of $\rsh$.
Moreover, assume $\Lambda$ is the set of all definite (in)equalities between free variables in $\rsh$.
We can then construct a reduced symbolic heap $\SSIGMA{\sh}{(B,\Lambda)}$ from $B$ and $\Lambda$ that precisely captures these relationships between free variables.

\begin{definition} \label{def:zoo:track:kernel}
  Let $\sh$ be a symbolic heap, $B \subseteq \FV{0}{}$ and $\Lambda \in \textnormal{Pure}(\FV{0}{})$.
  Furthermore, let $\textrm{min}(B,\Lambda) = \{ \FV{0}{i} \in B ~|~ \neg \exists \FV{0}{j} \in B . j < i ~\text{and}~ \FV{0}{i} \MEQ{\Lambda} \FV{0}{j} \}$ be the set of minimal (w.r.t. to occurrence in $\FV{0}{}$) allocated free variables.
  Then
  \begin{align*}
   \SSIGMA{\sh}{(B,\Lambda)}
   ~\DEFEQ~ \bigstar_{\PROJ{\FV{0}{}}{i} \in \textrm{min}(B,\Lambda)} ~ \PTS{\PROJ{\FV{0}{\sh}}{i}}{\NIL} ~:~ \Lambda, 
  \end{align*}
  where we write
$\bigstar_{s \in S}\,\PTS{s}{\NIL}$ for 
          $\PTS{s_1}{\NIL} \SEP \ldots \SEP \PTS{s_k}{\NIL}$, $S = \{s_1,\ldots,s_k\}$.
\end{definition}

Consequently, the relationships between free variables remain unaffected if a predicate call of $\sh$ is replaced by $\SSIGMA{\sh}{(B,\Lambda)}$ instead of $\rsh$.
Thus, $\HATRACK$ has one state per pair $(B, \Lambda)$.
In the transition relation of $\HATRACK$ it suffices to replace each predicate call $\PS\FV{0}{}$
by the corresponding 
symbolic heap
$\SSIGMA{\PS\FV{0}{}}{(B,\Lambda)}$.
and check whether the current state
is consistent with the resulting symbolic heap.
%
%
%
%
Intuitively, a potentially large unfolding of a symbolic heap $\sh$ with $m$ predicate calls is ``compressed'' into a small one that contains all necessary information about parameters of predicate calls.
Here, $\T{q}$ is a sequence of pairs $(B,\Lambda)$ as explained above.
%
Formally,
\begin{definition}\label{def:zoo:track-automaton}
 %
 $\HATRACK = (Q,\SHCLASSFV{\alpha},\Delta,F)$
 is given by:
  \begin{align*}
      Q ~\DEFEQ~ & 2^{{\FV{0}{}}} ~\times~ \textnormal{Pure}(\FV{0}{}), \qquad \quad F ~\DEFEQ~ \{ (A,\PURE{}) \}, \\
      \Delta ~~:~~ & \MOVE{\HATRACK}{(A_0,\PURE{}_0)}{\sh}{\T{q}}
    ~\text{iff}~ \forall x,y \in \FV{0}{} ~.~ \\
                 & \quad y \in A_0 \leftrightarrow y^{\sh} \in \ALLOC{\SHRINK{\sh,\T{q}}} \\
                 & \quad \text{and}~ x \sim y \in \PURE{}_0 ~\leftrightarrow~ x^{\sh} \sim_{\SHRINK{\sh,\T{q}}} y^{\sh}~, \\
     \SHRINK{\sh,\T{q}} ~\DEFEQ~ &
        \sh\left[\PS_1 / \SSIGMA{\CALLN{1}{}}{\T{q}[1]}, \ldots, \PS_m / \SSIGMA{\CALLN{m}{}}{\T{q}[m]}\right]~,
  \end{align*}
  where $m = \NOCALLS{\sh} = \SIZE{\T{q}}$ is the number of predicate calls in $\sh$ and
  $y^{\sh}$ denotes the free variable of $\sh$ corresponding to $y \in \FV{0}{}$, i.e., if $y = \PROJ{\FV{0}{}}{i}$ then $y^{\sh} = \PROJ{\FV{0}{\sh}}{i}$.
%
\end{definition}
Since $\SHRINK{\rsh,\EMPTYSEQ} = \rsh$ holds for every reduced symbolic heap $\rsh$, it is straightforward to show that
$L(\HATRACK) = \TRACK(\alpha,A,\Pi)$.
Furthermore, $\HATRACK$ satisfies the compositionality property.
A formal proof is found in Appendix~\ref{app:zoo:track:property}.
%
Hence,
\begin{lemma}\label{thm:zoo:track:property}
 For all $\alpha \in \POSN$ and all sets
 $A \subseteq \FV{0}{}$, 
 $\PURE{} \in \textnormal{Pure}(\FV{0}{})$,
there is a heap automaton over $\SHCLASSFV{\alpha}$ accepting
 $\TRACK(\alpha,A,\PURE{})$.
\end{lemma}
\subsection{Satisfiability} \label{sec:zoo:sat}
Tracking relationships between free variables of symbolic heaps is a
useful auxiliary construction that serves as a building block in
automata for more natural properties. 
For instance, the heap automaton $\HATRACK$ constructed in Definition~\ref{def:zoo:track-automaton} can be reused to deal with the
%
%

%
\noindent\textbf{Satisfiability problem (\DPROBLEM{SL-SAT})}:
Given $\SRD \in \SETSRD{}$ and $\sh \in \SL{\SRD}{}$,
decide whether $\sh$ is satisfiable,
i.e., 
there exists 
$(\stack,\heap) \in \STATES$
such that $s,h \SAT{\SRD} \sh$.
%
%
\begin{theorem}\label{thm:zoo:sat:property}
 For each $\alpha \in \POSN$, 
 there is a heap automaton over $\SHCLASSFV{\alpha}$ accepting
 the set
 $
   \SATPROP(\alpha) \DEFEQ \{ \rsh \in \RSHCLASSFV{\alpha} ~|~ \rsh~\text{is satisfiable} \}
 $ 
 of all satisfiable reduced symbolic heaps with at most $\alpha$ free variables.
\end{theorem}
\begin{proof}
 A heap automaton $\HASAT$ accepting $\SATPROP(\alpha)$ is constructed as in Definition~\ref{def:zoo:track-automaton} except for the set of final states, which is 
 $
  F \DEFEQ \{ (A,\PURE{}) ~|~ \NIL \neq \NIL\,\notin \PURE{} \}
 $.
 See Appendix~\ref{app:zoo:sat:property} for a correctness proof.
\qed
\end{proof}
A heap automaton accepting the complement of $\SATPROP(\alpha)$ is constructed analogously
by choosing $F \DEFEQ \{ (A,\PURE{}) ~|~ \NIL \neq \NIL\,\in \PURE{} \}$.
Thus, together with Corollary~\ref{thm:compositional:existence}, we obtain a decision procedure for the satisfiability problem similar to the one proposed in \cite{brotherston2014decision}.
Regarding complexity, the heap automaton $\HASAT$ from Definition~\ref{def:zoo:track-automaton} has $2^{2\alpha^2 + \alpha}$ states.
By Remark~\ref{rem:closure}, membership in $\Delta_{\HASAT}$ is decidable in polynomial time.
Thus, by Corollary~\ref{thm:compositional:existence}, our construction yields an exponential-time decision procedure for $\DPROBLEM{SL-SAT}$.
If the number of free variables $\alpha$ is bounded, an algorithm in $\CCLASS{NP}$
is easily obtained by guessing a suitable unfolding tree of height at most $\SIZE{Q_{\HASAT}}$
and running $\HASAT$ on it to check whether its unfolding is decidable  (cf. Appendix~\ref{app:zoo:sat:np}).
This is in line with the results of Brotherston et al.~\cite{brotherston2014decision}, where the satisfiability problem is shown to be $\CCLASS{ExpTime}$--complete in general and \CCLASS{NP}--complete if the number of free variables is bounded.
These complexity bounds even hold for the following
special case~\cite{brotherston2016model}:
%

%
\noindent\textbf{Restricted satisfiability problem (\DPROBLEM{SL-RSAT})}
        Given an SID $\SRD$ that contains no 
        points-to assertions,
        and a predicate symbol $\PS$,
        decide whether $\PS\T{x}$ is satisfiable
        w.r.t. $\SRD$.
%
The complement of this problem is denoted by $\COMPLEMENT{SL-RSAT}$.
\subsection{Establishment} \label{sec:zoo:establishment}
%
A symbolic heap $\sh$
is \emph{established} if every existentially quantified variable of
every unfolding of $\sh$
is definitely equal to a free variable or definitely allocated.\footnote{Sometimes this property is also defined by requiring that each existentially quantified variable is "eventually allocated"~\cite{iosif2013tree}.}
This property is natural for symbolic heaps that specify the
shape of data structures; for example, the SIDs in
Example~\ref{ex:srd} define sets of established symbolic heaps.
Further, establishment is often required 
to ensure decidability of the
entailment problem~\cite{iosif2013tree,iosif2014entailment}.
Establishment can also be checked by heap automata. 
\begin{theorem}\label{thm:zoo:establishment}
 For all $\alpha \in \POSN$, there is a heap automaton over $\SHCLASSFV{\alpha}$ accepting the set
 of all established reduced symbolic heaps with at most $\alpha$ free variables:
 \begin{align*}
   \ESTPROP(\alpha) ~\DEFEQ~
   \{ \rsh \in \RSHCLASSFV{\alpha} ~|~ & \forall y \in \VAR(\rsh) ~.~ y \in \ALLOC{\rsh}
                                               ~\text{or}~ \exists x \in \FV{0}{\rsh} ~.~ x \MEQ{\rsh} y  \}
 \end{align*}
 %
 %
\end{theorem}
\begin{proof}
  The main idea in the construction of a heap automaton $\HAEST$
  for $\ESTPROP(\alpha)$
  is to verify that every variable is definitely allocated or equal to
  a free variable while running $\HATRACK$
  (see Definition~\ref{def:zoo:track-automaton}) in parallel to keep
  track of the relationships between free variables.
 An additional flag $q \in \{0,1\}$ is attached to each state of $\HATRACK$ to store whether the establishment condition is already violated ($q=0$) or holds so far ($q=1$).
 Formally, $\HAEST = (Q,\SHCLASSFV{\alpha},\Delta,F)$, where
\begin{align*}
   & Q ~\DEFEQ~ Q_{\HATRACK} \times \{0,1\}, \qquad F ~\DEFEQ~ Q_{\HATRACK} \times \{1\}, \\
   &  \Delta ~~:~~ \MOVE{\HAEST}{(p_0,q_0)}{\sh}{(p_1,q_1) \ldots (p_m,q_m)} \\
   & ~\text{iff}~ \MOVE{\HATRACK}{p_0}{\sh}{p_1\ldots p_m}
     ~\text{and}~ q_0 = \min \{q_1,\ldots,q_m,\CHECK(\sh,p_1 \ldots p_m)\}.
\end{align*}
Here, $\CHECK : \SHCLASSFV{\alpha} \times Q_{\HATRACK}^{*} \to \{0,1\}$ is a predicate given by
 \begin{align*}
   \CHECK(\sh,\T{p}) ~\DEFEQ~ \begin{cases}
                                   1 &, ~\text{if}~ \forall y \in \VAR(\sh) ~.~  y \in \ALLOC{\SHRINK{\sh,\T{p}}} \\
                                    & \qquad\text{or}~ \exists x \in \FV{0}{\sh} ~.~ x \MEQ{\SHRINK{\sh,\T{p}}} y \\
                                   0 &, ~\text{otherwise}~,
                                 \end{cases}
 \end{align*}
 where $\SHRINK{\sh,\T{p}}$ is the reduced symbolic heap obtained from the tracking property as in Definition~\ref{def:zoo:track-automaton}.
 Moreover, unlike in the construction of $\HATRACK$,
 we are not interested in a specific set of relationships between the
 pure formulas, so any state of $\HATRACK$
 is chosen as a final state provided that predicate $\CHECK$
 could be evaluated to $1$.
 See Appendix~\ref{app:zoo:establishment} for a correctness proof.
\qed
\end{proof}
Again, it suffices to swap the final- and non-final states of $\HAEST$ to obtain a heap automaton
$\HA{A}_{\overline{\ESTPROP}}$
accepting the complement
of $\ESTPROP(\alpha)$.
Thus, by Corollary~\ref{thm:compositional:existence} and Remark~\ref{rem:closure},
we obtain an \CCLASS{ExpTime} decision procedure for the
%

%
\noindent\textbf{Establishment problem (\DPROBLEM{SL-EST})}:
  Given an SID $\SRD$ and $\sh \in \SL{\SRD}{}$, decide whether every $\rsh \in \CALLSEM{\sh}{\SRD}$ is established.
%
\begin{lemma}\label{thm:zoo:establishment:lower}
  \COMPLEMENT{SL-RSAT} is polynomial-time reducible to \DPROBLEM{SL-EST}.
  Hence, the establishment problem $\DPROBLEM{SL-EST}$ is \CCLASS{ExpTime}--hard in general and \CCLASS{coNP}--hard if the maximal number of free variables is bounded.
\end{lemma}
\begin{proof}
 Let $(\SRD,\PS)$ be an instance of $\COMPLEMENT{SL-RSAT}$.
 Moreover, let
 $ 
   \sh ~\DEFEQ~ \exists \T{z} y ~.~ \PS\T{z} : \{ \PROJ{\FV{0}{}}{1} = \NIL, y \neq \NIL \}.
 $ 
 As $y$ is neither allocated nor occurs in $\PS\T{z}$, $\sh$ is
 established iff $\IFV{1} = y$ iff $\NIL \neq \NIL$ iff $\PS\T{x}$ is
 unsatisfiable.
 Hence, $(\SRD,\sh) \in \DPROBLEM{SL-EST}$ iff $(\SRD,\PS) \in \COMPLEMENT{SL-RSAT}$.
 A full proof is found in Appendix~\ref{app:zoo:establishment:lower}.
\qed
\end{proof}
\begin{lemma}\label{thm:zoo:establishment:conp}
  \DPROBLEM{SL-EST} is in \CCLASS{coNP}  for a bounded number of free variables $\alpha$.
\end{lemma}
\begin{proof}
  Let $(\SRD,\sh)$ be an instance of $\DPROBLEM{SL-EST}$,
  $N = \SIZE{\SRD} + \SIZE{\sh}$, and
  %
  $M \leq N$ be the maximal number of predicate calls occurring in $\sh$ and any rule of $\SRD$.
  Moreover, let $\HA{A}_{\overline{\ESTPROP}}$ be a heap automaton
  accepting $\overline{\ESTPROP(\alpha)}$---the complement of $\ESTPROP(\alpha)$
  (cf. Theorem~\ref{thm:zoo:establishment}).
  Since $\alpha$ is bounded by a constant, so is the number of states of
  $\HA{A}_{\overline{\ESTPROP}}$, namely
  $\SIZE{Q_{\HA{A}_{\overline{\ESTPROP}}}} \leq k = 2^{2\alpha^2 + \alpha + 1}$.
  Now, let $\UTREES{\SRD}{\sh}^{\leq k}$ denote the set of all unfolding trees $t \in \UTREES{\SRD}{\sh}$ of height at most $k$.
  Clearly, each of these trees is of size $\SIZE{t} \leq M^{k} \leq N^{k}$, i.e., polynomial in $N$.
  Moreover, let $\omega: \DOM(t) \to Q_{\HA{A}_{\overline{\ESTPROP}}}$ be a function mapping each node of $t$ to a state of $\HA{A}_{\overline{\ESTPROP}}$.
  Again, $\omega$ is of size polynomial in $N$; as such $\SIZE{\omega} \leq k \cdot N^{k}$.
  Let $\Omega_{t}$ denote the set of all of these functions $\omega$ for a given unfolding tree $t$ with $\omega(\EMPTYSEQ) \in F_{\HA{A}_{\overline{\ESTPROP}}}$.
  Given an unfolding tree $t \in \UTREES{\SRD}{\sh}^{\leq k}$ and $\omega \in \Omega_{t}$,
  we can easily decide whether $\OMEGA{\HA{A}_{\overline{\ESTPROP}}}{\omega(\EMPTYSEQ)}{\UNFOLD{t}}$ holds:
  For each $u,u1,\ldots,un \in \DOM(t)$, $u(n+1) \notin \DOM(t)$,
  $n \geq 0$, it suffices to check whether $\MOVE{\HA{A}_{\overline{\ESTPROP}}}{\omega(u)}{t(u)}{\omega(u1) \ldots \omega(un)}$.
  Since, by Remark~\ref{rem:closure},
  each of these checks can be performed in time polynomial in $N$
  the whole procedure is feasible in polynomial time.
  %
  We now show that $(\SRD,\sh) \in \DPROBLEM{SL-EST}$ if and only if
  \begin{align*}
    \forall t \in \UTREES{\SRD}{\sh}^{\leq k} ~.~ \forall \omega \in \Omega_{t} ~.~  \text{not}~ \OMEGA{\HA{A}_{\overline{\ESTPROP}}}{\omega(\EMPTYSEQ)}{\UNFOLD{t}}.
  \end{align*}
  Since each $t \in \UTREES{\SRD}{\sh}$ and each $\omega \in \Omega_{t}$ is of size polynomial in $N$,
  this is equivalent to $\DPROBLEM{SL-EST}$ being in $\CCLASS{coNP}$.
  To complete the proof, note that
  $\CALLSEM{\sh}{\SRD} \subseteq \ESTPROP(\alpha)$ holds iff $\UNFOLD{t} \notin \overline{\ESTPROP(\alpha)}$ for each $t \in \UTREES{\SRD}{\sh}$.
  Furthermore, by a standard pumping argument, it suffices to consider trees in $\UTREES{\SRD}{\sh}^{\leq k}$:
  If there exists a taller tree $t$ with $\UNFOLD{t} \in \overline{\ESTPROP(\alpha)}$ then there is some path of length greater $k$ in $t$ on which two nodes are assigned the same state by
  a function $\omega \in \Omega_{t}$ proving membership of $t$ in $\overline{\ESTPROP(\alpha)}$.
  This path can be shortened to obtain a tree of smaller height.
\qed
\end{proof}
Putting upper and lower bounds together, we conclude:
\begin{theorem}\label{thm:zoo:establishment:complexity}
  $\DPROBLEM{SL-EST}$ is $\CCLASS{ExpTime}$--complete in general and $\CCLASS{coNP}$--complete if the number of free variables $\alpha$ is bounded.
\end{theorem}

\subsection{Reachability} \label{sec:zoo:reach}
Another family of robustness properties is based on reachability questions, e.g., ``is every location of every model of a symbolic heap reachable from the location of a program variable?'' or ``is every model of a symbolic heap acyclic?''.
For established SIDs, heap automata accepting these properties are an extension of the tracking automaton introduced in Definition~\ref{def:zoo:track-automaton}.
More precisely, a variable $y$ is \emph{definitely} \emph{reachable}
from $x$ in $\rsh \in \RSL{}{}$, written $\REACH{x}{y}{\rsh}$,
if and only if $x \MPT{\rsh} y$ or there exists a $z \in \VAR(\rsh)$ such that $x \MPT{\rsh} z$ and
$\REACH{z}{y}{\rsh}$.\footnote{The definite points-to
  relation $\MPT{\rsh}$ was defined at the beginning of
  Section~\ref{sec:zoo}.} Note that we define reachability to be
transitive, but not reflexive.
As for the other definite relationships between variables,
definite reachability is computable in polynomial time for reduced symbolic heaps,
e.g., by performing a depth-first search on the definite points-to
relation $\MPT{\rsh}$.
Note that our notion of reachability does not take variables into account
that are only reachable from one another in \emph{some} models of a reduced symbolic heap.
For example, consider the symbolic heap
$\rsh = \PTS{x}{y} \SEP \PTS{z}{\NIL}$.
Then $\REACH{x}{z}{\rsh}$ does \emph{not} hold, but there exists a model $(\stack,\heap)$ with
$\stack(z) = \stack(y) \in \heap(\stack(x))$.
Thus, reachability introduced by unallocated variables is not detected.
However, the existence (or absence) of such variables can be checked first due to Theorem~\ref{thm:zoo:establishment}.
%
%
%
%

%
%
%
\begin{theorem}\label{thm:zoo:reachability:property}
  Let $\alpha \in \POSN$ and $R \subseteq \FV{0}{} \times \FV{0}{}$ be a binary relation
  over the variables $\FV{0}{}$ with $\SIZE{\FV{0}{}} = \alpha$.
  Then the \emph{reachability} \emph{property} $\RPROP(\alpha,R)$, given by the set
  $
    \{ \rsh \in \RSHCLASSFV{\alpha} ~|~ \forall i,j ~.~ (\PROJ{\FV{0}{}}{i},\PROJ{\FV{0}{}}{j}) \in R  ~\text{iff}~ \REACH{\PROJ{\FV{0}{\rsh}}{i}}{\PROJ{\FV{0}{\rsh}}{j}}{\tau} \},
  $ 
  can be accepted by a heap automaton over $\SHCLASSFV{\alpha}$.
 %
\end{theorem}
\begin{proof}[sketch]
A heap automaton $\HAREACH$ accepting $\RPROP(\alpha,R)$ is
constructed similarly to
the heap automaton $\HATRACK$ introduced in Definition~\ref{def:zoo:track-automaton}.
The main difference is that $\HAREACH$
additionally stores a binary relation $S \subseteq \FV{0}{} \times \FV{0}{}$
in its state space to remember which free variables are reachable from one another.
Correspondingly, we adapt Definition~\ref{def:zoo:track:kernel} as follows:
%
\begin{align*}
   \SSIGMA{\sh}{(B,\Lambda,S)}
   ~\DEFEQ~ \exists z ~.~ \bigstar_{\textrm{min}(B,\Lambda)} ~ \PT{\PROJ{\FV{0}{\sh}}{i}}{\T{v}_i} ~:~ \Lambda~,
\end{align*}
where $z$ is a fresh variable and
$\PROJ{\T{v}_i}{j} \DEFEQ  \PROJ{\FV{0}{\sh}}{j}$ if $(i,j) \in S$ and
$\PROJ{\T{v}_i}{j} \DEFEQ  z$, otherwise.
%
The other parameters $\sh,B,\Lambda$ are the same as in Definition~\ref{def:zoo:track-automaton}.
Note that the additional variable $z$ is needed to deal with allocated free variables that cannot reach
any other free variable, including $\NIL$.
Moreover, the set of final states is $F_{\HAREACH} = Q_{\HATRACK} \times \{R\}$.
Correctness of this encoding is verified in the transition relation.
Hence, the transition relation of $\HAREACH$ extends the transition relation of $\HATRACK$
by the requirement
$ 
  (x,y) \in S ~\text{iff}~ \REACH{x^{\sh}}{y^{\sh}}{\SHRINK{\sh,\T{p}}}
$ 
for every pair of free variables $x,y \in \FV{0}{}$.
Here, $\SHRINK{\sh,\T{p}}$ is defined as in Definition~\ref{def:zoo:track-automaton} except
that the new encoding $\SSIGMA{\PS_i\FV{i}{}}{\T{q}[i]}$ from above is used.
Since $\SHRINK{\rsh,\EMPTYSEQ} = \rsh$ holds for every reduced symbolic heap $\rsh$,
it is straightforward to verify that $L(\HAREACH) = \RPROP(\alpha)$.
Further details are found in
Appendix~\ref{app:zoo:reachability:property}.
\qed
\end{proof}
Furthermore, we consider the related
%

%
\noindent\textbf{Reachability problem ($\DPROBLEM{SL-REACH}$)}:
  Given an SID $\SRD$, 
  $\sh \in \SL{\SRD}{}$ with $\alpha = \NOFV{\sh}$
  and 
  variables $x,y \in \FV{0}{\sh}$,
  decide whether $\REACH{x}{y}{\rsh}$ holds
  for all $\rsh \in \CALLSEM{\sh}{\SRD}$.
%
\begin{theorem}\label{thm:zoo:reachability:complexity}
  The decision problem
  \DPROBLEM{SL-REACH} is \CCLASS{ExpTime}--complete in general and \CCLASS{coNP}--complete if the number of free variables is bounded.
\end{theorem}
\begin{proof}
  Membership in \CCLASS{ExpTime} follows from our upper bound derived for Algorithm~\ref{alg:on-the-fly-refinement},
  the size of the state space of $\HAREACH$,
  which is exponential in $\alpha$, and Remark~\ref{rem:closure}.
  If $\alpha$ is bounded, membership in $\CCLASS{coNP}$ is shown analogously to Lemma~\ref{thm:zoo:establishment:conp}.
 Lower bounds are shown by reducing \COMPLEMENT{SL-RSAT} to \DPROBLEM{SL-REACH}.
Formally, let $(\SRD,\PS)$ be an instance of \COMPLEMENT{SL-RSAT}.
Moreover, let
$ 
\sh ~\DEFEQ~ \exists \BV{} ~.~ \PTS{\PROJ{\FV{0}{}}{1}}{\NIL} \SEP \PS\BV{}
~:~ \{ \PROJ{\FV{0}{}}{2} \neq \NIL \}.
$ 
As $\IFV{2}$ is neither allocated nor $\NIL$, $\IFV{2}$ is \emph{not}
definitely reachable from $\IFV{1}$ in \emph{any} model of
$\sh$. Hence
$(\SRD,\sh,\PROJ{\FV{0}{}}{1},\PROJ{\FV{0}{}}{2}) \in
\DPROBLEM{SL-REACH}$ iff $\PS$ is unsatisfiable.
A detailed proof is found in
Appendix~\ref{app:zoo:reachability:complexity}.
\qed
\end{proof}
\subsection{Garbage-Freedom} \label{sec:zoo:garbage}
Like the tracking automaton $\HATRACK$, the automaton
$\HAREACH$ is a useful ingredient in the construction of more complex heap automata.
%
For instance, such an automaton can easily be modified to check
whether a symbolic heap is \emph{garbage-free}, i.e., whether every existentially quantified variable is reachable from some program variable.
Garbage-freedom is a natural requirement if SIDs represent data structure specifications.
For instance, the SIDs in Example~\ref{ex:srd} are garbage-free.
Furthermore, this property is needed by the approach of Habermehl et al.~\cite{habermehl2012forest}.
\begin{lemma}\label{thm:zoo:garbage:property}
 For each $\alpha \in \POSN$, the set $\GARBAGEPROP(\alpha)$, given by
 \begin{align*}
   &\{ \rsh \in \RSHCLASSFV{\alpha} ~|~ \forall y \in \VAR(\rsh) ~.~ \exists x \in \FV{0}{\rsh} ~.~ 
   x \MEQ{\rsh} y ~\text{or}~ \REACH{x}{y}{\rsh} \},
 \end{align*}
 of garbage-free symbolic heaps can be accepted by a heap automaton over $\SHCLASSFV{\alpha}$.
%
\end{lemma}
\begin{proof}[sketch]
A heap automaton $\HAGARBAGE$ accepting $\GARBAGEPROP(\alpha)$
is constructed similarly to the heap automaton $\HAEST$
introduced in the proof of Theorem~\ref{thm:zoo:establishment}.
The main difference is that heap automaton $\HAREACH$ is used
instead of $\HATRACK$.
Furthermore, the predicate $\CHECK : \SHCLASSFV{\alpha} \times Q_{\HAREACH}^{*} \to \{0,1\}$
is redefined to verify that every variable of a symbolic heap $\sh$ is established in $\SHRINK{\sh,\T{p}}$,
where $\SHRINK{\sh,\T{p}}$ is the same as in the construction of $\HAREACH$ (see Theorem~\ref{thm:zoo:reachability:property}):
%
\begin{align*}
\CHECK(\sh,\T{p}) ~\DEFEQ~
 \begin{cases}
        1 &, ~\text{if}~ \forall y \in \VAR(\sh) \,.\, \exists x \in \FV{0}{\sh} ~.~ \\
          & \qquad        x \MEQ{\SHRINK{\sh,\T{p}}} y ~\text{or}~ \REACH{x}{y}{\SHRINK{\sh,\T{p}}} \\
        0 &, ~\text{otherwise}~,
 \end{cases}
\end{align*}
Since $\SHRINK{\rsh,\EMPTYSEQ} = \rsh$ holds for every reduced symbolic heap $\rsh$,
it is straightforward that $L(\HAGARBAGE) = \GARBAGEPROP(\alpha)$.
A proof is found in Appendix~\ref{app:zoo:garbage:property}.
\qed
\end{proof}
To guarantee that symbolic heaps are garbage-free, we solve the 
%

%
\noindent\textbf{Garbage-freedom problem ($\DPROBLEM{SL-GF}$)}:
  Given an SID $\SRD$ and 
  $\sh \in \SL{\SRD}{}$,
  decide whether every $\rsh \in \CALLSEM{\sh}{\SRD}$ is garbage-free,
  i.e., $\rsh \in \GARBAGEPROP(\alpha)$ for some $\alpha \in \N$.
%
\begin{theorem}\label{thm:zoo:garbage:complexity}
  \DPROBLEM{SL-GF} is \CCLASS{ExpTime}--complete in general and $\CCLASS{coNP}$--complete if the number of free variables $\alpha$ is bounded.
\end{theorem}
\begin{proof}
Similar to Theorem~\ref{thm:zoo:establishment:complexity}; see Appendix~\ref{app:zoo:garbage:complexity}.
\qed
\end{proof}
\subsection{Acyclicity} \label{sec:zoo:acyclicity}
Automatic termination proofs of programs frequently rely on
the acyclicity of the data structures used, i.e., they
assume that no variable is reachable from itself.
%
In fact, Zanardini and Genaim~\cite{zanardini2014inference} claim that ``proving termination needs acyclicity, unless program-specific or non-automated reasoning is performed.''
\begin{lemma}\label{thm:zoo:acyclicity:property}
  For each $\alpha \in \POSN$, the set
  of all weakly acyclic symbolic heaps
  \begin{align*}
          \CYCLEPROP(\alpha) ~\DEFEQ~ & \{ \rsh \in \RSHCLASSFV{\alpha} ~|~
               \NIL \MNEQ{\rsh} \NIL
            ~\text{or}~ \forall x \in \VAR(\rsh) ~.~
              \text{not}~ \REACH{x}{x}{\rsh}  \}
  \end{align*}
  can be accepted by a heap automaton over $\SHCLASSFV{\alpha}$.
\end{lemma}
Here, the condition $\NIL \MNEQ{\rsh} \NIL$ ensures
that an unsatisfiable reduced symbolic heap is considered weakly acyclic.
Further, note that our notion of acyclicity is weak in the sense that dangling pointers may introduce cyclic models that are not considered.
For example, $\exists z . \PTS{x}{z}$ is weakly acyclic, but contains cyclic models if $x$ and $z$ are aliases.
However, weak acyclicity coincides with the absence of cyclic models for established SIDs---a property considered in Section~\ref{sec:zoo:establishment}.
\begin{proof}[sketch]
  A heap automaton $\HACYCLE$ for the set of all weakly acyclic reduced
  symbolic heaps is constructed analogously to the heap automaton
  $\HAGARBAGE$ in the proof of
  Lemma~\ref{thm:zoo:garbage:property}.
  The main difference is the predicate
  $\CHECK : \SHCLASSFV{\alpha} \times Q_{\HAREACH}^{*} \to \{0,1\}$, which now checks
  whether a symbolic heap is weakly acyclic:
\begin{align*}
\CHECK(\sh,\T{p}) ~\DEFEQ~
 \begin{cases}
        1 &, ~\text{if}~ \forall y \in \VAR(\sh) ~.~ \text{not}~\REACH{x}{x}{\SHRINK{\sh,\T{p}}} \\
        0 &, ~\text{otherwise}.
 \end{cases}
\end{align*}
  Moreover, the set of final states $F_{\HACYCLE}$ is chosen such that accepted symbolic heaps
  are unsatisfiable or $\CHECK(\sh,\T{p}) =1$. 
  See Appendix~\ref{app:zoo:acyclicity:property} for details.
  \qed
\end{proof}
For example, 
 the symbolic heap $\texttt{sll}\,\FV{0}{}$ 
is weakly acyclic, but $\texttt{dll}\,\FV{0}{}$ (cf. Example~\ref{ex:srd}) is not.
In general, we are interested in the
%

%
\noindent\textbf{Acyclicity problem ($\DPROBLEM{SL-AC}$)}:
  Given an SID $\SRD$ and 
  $\sh \in \SL{\SRD}{}$,
  decide whether every $\rsh \in \CALLSEM{\sh}{\SRD}$ is weakly acyclic,
  i.e., $\rsh \in \CYCLEPROP(\alpha)$ for some $\alpha \in \N$.
%
\begin{theorem}\label{thm:zoo:acyclicity:complexity}
  \DPROBLEM{SL-AC} is \CCLASS{ExpTime}--complete in general and $\CCLASS{coNP}$--com\-plete
  if the number of free variables $\alpha$ is bounded.
\end{theorem}
\begin{proof}
Similar to the proof of Theorem~\ref{thm:zoo:establishment:complexity}.
For lower bounds, we show that $\COMPLEMENT{SL-RSAT}$ is reducible to
\DPROBLEM{SL-AC}.
Let $(\SRD,\PS)$ be an instance of \COMPLEMENT{SL-RSAT}.
Moreover, let $\sh = \exists \BV{} . \PT{\PROJ{\FV{0}{}}{1}}{\PROJ{\FV{0}{}}{1}} \SEP \PS\BV{}$.
Since $\IFV{1}$ is definitely reachable from itself, $\sh$ is weakly acyclic iff $\PS$ is unsatisfiable.
Thus, $(\SRD,\sh) \in \DPROBLEM{SL-AC}$ iff $(\SRD,\PS) \in \COMPLEMENT{SL-RSAT}$.
See Appendix~\ref{app:zoo:acyclicity:complexity} for details.
\qed
\end{proof}
%
%
%


%
\section{Implementation} \label{sec:implementation}

We developed a prototype of our framework---called \textsc{Harrsh}\footnote{Heap Automata for Reasoning about Robustness of Symbolic Heaps}---that implements Algorithm~\ref{alg:on-the-fly-refinement} as well as all heap automata constructed in the previous sections.
The code, the tool and our experiments are available online.\footnote{\url{https://bitbucket.org/jkatelaan/harrsh/}}

For our experimental results, we first considered common SIDs from the literature, such as singly- and doubly-linked lists, trees, trees with linked-leaves etc.
For each of these SIDs, we checked all robustness properties presented throughout this paper, i.e.,
the existence of points-to assertions (Example~\ref{ex:heap-automaton:toy}),
the tracking property $\TRACK(B,\Lambda)$ (Section~\ref{sec:zoo:track}),
satisfiability (Section~\ref{sec:zoo:sat}),
establishment (Section~\ref{sec:zoo:establishment}),
the reachability property $\RPROP(\alpha,R)$ (Section~\ref{sec:zoo:reach}),
garbage-freedom (Section~\ref{sec:zoo:garbage}), and
weak acyclicity (Section~\ref{sec:zoo:acyclicity}).
All in all, our implementation of Algorithm~\ref{alg:on-the-fly-refinement} takes 300ms to successfully check these properties on all 45 problem instances.
Since the SIDs under consideration are typically carefully handcrafted to be robust, the low runtime is to be expected.
Moreover, we ran heap automata on benchmarks of the tool \textsc{Cyclist}~\cite{brotherston2014decision}.
In particular, our results for the satisfiability problem---the only
robustness property checked by both tools---were within the same order of magnitude.

Further details are found in Appendix~\ref{app:implementation}.

\section{Entailment Checking with Heap Automata}
\label{sec:entailment}
So far, we have constructed heap automata for reasoning about robustness properties,
such as 
satisfiability, establishment and acyclicity.
%
This section demonstrates that our approach can also be applied to discharge \emph{entailments} for certain fragments of separation logic.
Formally, we are concerned with the
%

\noindent\textbf{Entailment problem (\DENTAIL{\SRDCLASS}{\SRD})}:
 Given symbolic heaps $\sh,\sha \in \SL{\SRD}{\SRDCLASS}$, decide whether $\sh \ENTAIL{\SRD} \sha$ holds, i.e.,
$\forall (\stack,\heap) \in \STATES ~.~ \stack,\heap \ENTAIL{\SRD} \sh ~\text{implies}~ \stack,\heap \ENTAIL{\SRD} \sha$.
%
%

%
Note that the symbolic heap fragment of separation logic is \emph{not} closed under conjunction and negation.
Thus, a decision procedure for satisfiability (cf. Theorem~\ref{thm:zoo:sat:property}) does \emph{not} yield a decision procedure for the entailment problem.
It is, however, essential to have a decision procedure for entailment,
because this problem underlies the important rule of consequence in
Hoare logic~\cite{hoare1969axiomatic}.
In the words of Brotherston et al.~\cite{brotherston2011automated}, ``effective procedures for establishing entailments are at the foundation of automatic verification based on separation logic''.
%

%
We show how our approach to decide robustness properties,
is applicable to discharge
entailments for certain fragments of symbolic heaps.
This results in an algorithm deciding entailments between so-called
\emph{determined symbolic heaps} for SIDs whose predicates
can be characterized by heap automata.
%
%
\begin{definition} \label{def:entailment:determined}
 A reduced symbolic heap $\rsh$ is \emph{determined} if all tight
 models of $\rsh$ are isomorphic.\footnote{A formal definition of model isomorphism is found in Appendix~\ref{app:supplement}.}
 %
 If $\rsh$ is also satisfiable then we call $\rsh$ \emph{well-determined}.
 Moreover, for some SID $\SRD$, a symbolic heap $\sh \in \SL{\SRD}{}$ is (well-)determined if all of its unfoldings $\rsh \in \CALLSEM{\sh}{\SRD}$ are (well-)determined.
 Consequently, an SID $\SRD$ is (well-)determined if $\PS\T{x}$ is (well-)determined for each predicate symbol $\PS$ in $\SRD$.
\end{definition}
  We present two sufficient conditions for determinedness of symbolic heaps.
  First, a reduced symbolic heap $\rsh$ is determined if
  all equalities and inequalities between variables are explicit, i.e., 
    $\forall x,y \in \VAR(\rsh) \,.\, x = y \in \PURE{\rsh}$
    or $x \neq y \in \PURE{\rsh}$
  ~\TECHNICALREPORT.
  Furthermore, a reduced symbolic heap $\rsh$ is determined if every variable is definitely allocated
  or definitely equal to $\NIL$, i.e.,
    $\forall x \in \VAR(\rsh) \,.\, x \in \ALLOC{\rsh}$
    or $x \MEQ{\rsh} \NIL$.
  These two notions can also be combined: A symbolic heap is determined if every variable $x$ is definitely allocated or definitely equal to $\NIL$ or there is an explicit pure formula $x \sim y$ between $x$ and each other variable $y$.
  %
%
\begin{example}
  By the previous remark, the SID generating acyclic singly-linked lists from Section~\ref{sec:introduction}
  is well-determined.
  Furthermore, although the predicate $\texttt{dll}$ from Example~\ref{ex:srd} 
  is not determined, the following symbolic heap is well-determined:
  $\PTS{\IFV{4}}{\NIL} \SEP \texttt{dll}\,\FV{0}{} : \{ \IFV{1} \neq \IFV{3} \}$.
\end{example}
\subsection{Entailment between predicate calls}
%
We start by considering entailments between predicate calls of well-determined SIDs.
By definition, an entailment $\sh \ENTAIL{\SRD} \sha$ holds if for
every stack--heap pair $(\stack,\heap)$ that satisfies an unfolding of $\sh$,
there exists an unfolding of $\sha$ that is satisfied by $(\stack,\heap)$ as well.
Our first observation is that, for well-determined unfoldings, two quantifiers can be switched:
It suffices for each unfolding $\rsha$ of $\sh$ to find \emph{one} unfolding $\rsh$ of $\sha$ such that every model of $\rsha$ is also a model
of $\rsh$.
\begin{lemma} \label{thm:entailment:predicates}
 Let $\SRD \in \SETSRD{}$ and $\PS_1,\PS_2$ be predicate symbols with $\ARITY(\PS_1) = \ARITY(\PS_2)$.
 Moreover, let $\CALLSEM{\PS_1\T{x}}{\SRD}$ be well-determined.
 Then
 \begin{align*}
 &  \PS_1\T{x} \ENTAIL{\SRD} \PS_2\T{x} ~~\text{iff}~~ \forall \rsha \in \CALLSEM{\PS_1\T{x}}{\SRD} \,.\, \exists \rsh \in \CALLSEM{\PS_2\T{x}}{\SRD} \,.\, \rsha \ENTAIL{\emptyset} \rsh.
 \end{align*}
\end{lemma}

\begin{proof}
 See Appendix~\ref{app:entailment:predicates} for a detailed proof.
 \qed
\end{proof}
Note that, even if only well-determined predicate calls are taken into
account, it is undecidable in general whether an entailment $\PS_1\FV{0}{} \ENTAIL{\SRD} \PS_2\FV{0}{}$ holds~\cite[Theorem 3]{antonopoulos2014foundations}.
To obtain decidability, we additionally require the set
of reduced symbolic heaps entailing a given predicate call to be accepted by a heap automaton. 
%
%
%
\begin{definition} \label{def:entailment:heapmodels}
 Let $\SRD \in \SETSRD{\SRDCLASS}$ and $\sh \in \SL{\SRD}{\SRDCLASS}$.
 Then
 \begin{align*}
  & \USET{\sh}{\SRD}{\SRDCLASS}
  \DEFEQ \{ \rsha \in \RSL{}{\SRDCLASS} ~|~ \NOFV{\rsha} = \NOFV{\sh} ~\text{and}~ \exists \rsh \in \CALLSEM{\sh}{\SRD} \,.\, \rsha \ENTAIL{\emptyset} \rsh \}
 \end{align*}
 is the set of all reduced symbolic heaps in $\SHCLASS$ over the same free variables as $\sh$ that entail an unfolding of $\sh$.
\end{definition}
\begin{example} \label{ex:uset}
  Let $\sh = \texttt{tll}\,\FV{0}{} : \{ \IFV{1} \neq \IFV{2} \}$, where $\texttt{tll}$ is a
  predicate of SID $\SRD$ introduced in Example~\ref{ex:srd}.
  Then $\USET{\sh}{\SRD}{\SRDCLASSFV{3}}$ consists of all reduced symbolic heaps with three free variables
  representing non-empty trees with linked leaves.
  In particular, note that these symbolic heaps do not have to be derived using the SID $\SRD$. For instance, they might contain additional pure formulas.
\end{example}
  In particular, $\USET{\PS\T{x}}{\SRD}{\SRDCLASS}$ can be accepted by a heap automaton 
for common predicates specifying data structures
such as lists, trees, and trees with linked leaves.
  %
%
We are now in a position to decide entailments between predicate calls.
\begin{lemma} \label{thm:entailment:pred-entail}
 Let $\SRD \in \SETSRD{\SRDCLASS}$ and $\PS_1,\PS_2 \in \PRED(\SRD)$ be predicate symbols having the same arity. 
 Moreover, let $\CALLSEM{\PS_1\T{x}}{\SRD}$ be well-determined
 and $\USET{\PS_2\T{x}}{\SRD}{\SRDCLASS}$ be accepted by a heap automaton over $\SHCLASS$. 
 Then the entailment $\PS_1\T{x} \ENTAIL{\SRD} \PS_2\T{x}$ is decidable.
\end{lemma}
\begin{proof}
Let $\HASH{\PS_2\T{x}}$ be a heap automaton over $\SHCLASS$ accepting $\USET{\PS_2\T{x}}{\SRD}{\SRDCLASS}$.
  Then
  \begin{align*}
                      & \PS_1\T{x} \ENTAIL{\SRD} \PS_2\T{x} \\
    ~\Leftrightarrow~ & \forall \rsha \in \CALLSEM{\PS_1\T{x}}{\SRD} \,.\, \exists \rsh \in \CALLSEM{\PS_2\T{x}}{\SRD} \,.\, \rsha \ENTAIL{\emptyset} \rsh 
                        \tag{Lemma~\ref{thm:entailment:predicates}} \\
   ~\Leftrightarrow~  & \forall \rsha \in \CALLSEM{P_1\T{x}}{\SRD} \,.\, \rsha \in \USET{\PS_2\T{x}}{\SRD}{\SRDCLASS} 
                        \tag{Definition~\ref{def:entailment:heapmodels}}\\
    ~\Leftrightarrow~  & \CALLSEM{P_1\T{x}}{\SRD} \subseteq L(\HASH{\PS_2\T{x}}).
    \tag{$L(\HASH{\PS_2\T{x}}) = \USET{\PS_2\T{x}}{\SRD}{\SRDCLASS}$}
  \end{align*}
  where the last inclusion is decidable by Corollary~\ref{thm:compositional:inclusion}.
\qed
\end{proof}
%
%
%
\subsection{Entailment between symbolic heaps}
Our next step is to generalize 
Lemma~\ref{thm:entailment:pred-entail}
to arbitrary determined symbolic heaps $\sh$ instead of single predicate calls.
This requires the construction of heap automata $\HASH{\sh}$ accepting $\USET{\sh}{\SRD}{\SRDCLASS}$.
%
W.l.o.g. we assume 
SIDs and symbolic heaps to be \emph{well}-determined instead of determined only.
Otherwise, we apply Theorem~\ref{thm:compositional:refinement} with the heap automaton $\HASAT$ (cf. Theorem~\ref{thm:zoo:sat:property}) 
to obtain a \emph{well-}determined SID.
Thus, we restrict our attention to the following set. 
\begin{definition} \label{def:entailment:centail}
  The set $\SHCENTAIL{\alpha}$ is given by $\CENTAIL{\alpha} : \SL{}{} \to \{0,1\}$, where
  $\CENTAIL{\alpha}(\sh) = 1$ iff
  $\sh$ is well-determined and every predicate call of $\sh$ has 
  $\leq \alpha \in \N$ parameters.
\end{definition}
Clearly, $\CENTAIL{\alpha}$ is decidable, 
because satisfiability is decidable (cf. Theorem~\ref{thm:zoo:sat:property})
and verifying that a symbolic heap has at most $\alpha$ parameters amounts to a simple syntactic check.
Note that, although the number of parameters in predicate calls is bounded by $\alpha$,
the number of free variables of a symbolic heap $\sh \in \SHCENTAIL{\alpha}$ is not. 
We then construct heap automata for well-determined symbolic heaps.
%
%
\begin{theorem} \label{thm:entailment:top-level}
 Let $\alpha \in \N$ and $\SRD \in \SETSRD{\SRDCLASSFV{\alpha}}$ be established.
 Moreover, for each predicate symbol $\PS \in \PRED(\SRD)$, 
 let there be a heap automaton over $\SHCENTAIL{\alpha}$ accepting 
 $\USET{\PS\T{x}}{\SRD}{\CENTAIL{\alpha}}$.
 %
 Then, for every well-determined symbolic heap $\sh \in \SL{\SRD}{}$, 
 there is a heap automaton over $\SHCENTAIL{\alpha}$ accepting
 $\USET{\sh}{\SRD}{\CENTAIL{\alpha}}$.
 %
\end{theorem}
\begin{proof}
 By structural induction on the syntax of symbolic heaps.
 For each case a suitable heap automaton has to be constructed.
 See Appendix~\ref{app:entailment:top-level} fo details.
\qed
\end{proof}
%
%
%
\begin{remark}\label{rem:model-checking}
  Brotherston et al.~\cite{brotherston2016model}
  studied the \emph{model-checking} \emph{problem} for symbolic heaps,
  i.e., the question whether $\stack,\heap \SAT{\SRD} \sh$ holds for a
  given stack--heap pair $(\stack,\heap)$, an SID $\SRD$, and a
  symbolic heap $\sh \in \SL{}{\SRD}$. They showed that this problem
  is \CCLASS{ExpTime}--complete in general and \CCLASS{NP}--complete
  if the number of free variables is
  bounded.  
  We obtain these results for \emph{determined} symbolic heaps %
  in a natural
  way: Observe that every stack--heap pair $(\stack,\heap)$ is  
  characterized by an established, well-determined, reduced symbolic
  heap, say $\rsh$, that has exactly $(\stack,\heap)$ as a tight model up
  to isomorphism.
  Then Theorem~\ref{thm:entailment:top-level} yields a 
  heap automaton $\HASH{\rsh}$ accepting $\USET{\rsh}{\SRD}{\CENTAIL{\alpha}}$, where $\alpha$ is 
  the maximal arity of any predicate in $\SRD$.
  Thus, $\stack,\heap \SAT{\SRD} \sh$ iff $L(\HASH{\rsh}) \cap \CALLSEM{\sh}{\SRD} \neq \emptyset$,
  which is decidable by
  Corollary~\ref{thm:compositional:existence}.
  Further, note that the general model-checking problem is within the scope of heap automata.
  A suitable statespace is the set of all subformulas of the symbolic heap $\rsh$.
\end{remark}
Coming back to the entailment problem,  it remains to put  
our results together.
Algorithm~\ref{alg:entailment:decision-procedure} depicts a decision
procedure for the entailment problem that, given an entailment $\sh \ENTAIL{\SRD} \sha$,
first removes all unsatisfiable unfoldings of $\sh$, i.e. $\sh$ becomes well-determined.
After that, our previous reasoning techniques for heap automata and SIDs from
Section~\ref{sec:compositional} are applied
to decide whether $\sh \ENTAIL{\SRD} \sha$ holds.
Correctness of Algorithm~\ref{alg:entailment:decision-procedure}
is formalized in 
\begin{theorem}\label{thm:entailment:main}
 Let $\alpha \in \N$ and $\SRD \in \SETSRD{\SRDCLASSFV{\alpha}}$ be established.
 Moreover, for every $\PS \in \PRED(\SRD)$, 
 let $\USET{\PS\T{x}}{\SRD}{\CENTAIL{\alpha}}$ be accepted by a heap automaton over $\SHCENTAIL{\alpha}$.
%
 %
 Then $\sh \ENTAIL{\SRD} \sha$ is decidable for determined $\sh,\sha \in \SL{\SRD}{}$ 
 with $\FV{0}{\sh} = \FV{0}{\sha}$.
\end{theorem}
\begin{proof}
We define a new SID $\Omega \DEFEQ \SRD \cup \{ \SRDRULE{\PS}{\sh} \}$,
where $\PS$ is a fresh predicate symbol of arity $\NOFV{\sh}$.
Clearly, $\sh \ENTAIL{\SRD} \sha$ iff $\PS\FV{0}{\sh} \ENTAIL{\Omega} \sha$.
Furthermore, since $\sh$ and $\SRD$ are established, so is $\Omega$.
Then applying the Refinement Theorem (Theorem~\ref{thm:compositional:refinement})
to $\Omega$ and $\HASAT$ (cf. Theorem~\ref{thm:zoo:sat:property}),
we obtain a well-determined SID $\SRDALT \in \SETSRD{\CENTAIL{\alpha}}$
where none of the remaining unfoldings of $\Omega$ is changed, i.e., for each $\PS \in \PRED(\Omega)$, we have $\CALLSEM{\PS\T{x}}{\SRDALT} \subseteq \CALLSEM{\PS\T{x}}{\Omega}$.
%
By Theorem~\ref{thm:entailment:top-level}, the set
$\USET{\sha}{\SRD}{\CENTAIL{\alpha}} = \USET{\sha}{\SRDALT}{\CENTAIL{\alpha}}$
can be accepted by a heap automaton over $\SHCENTAIL{\alpha}$.
%
Then, analogously to the proof of Lemma~\ref{thm:entailment:pred-entail},
\begin{align*}
  \sh \ENTAIL{\SRD} \sha
  ~\text{iff}~   \PS\FV{0}{\sh} \ENTAIL{\SRDALT} \sha
  ~~\text{iff}~~
  \CALLSEM{\PS\FV{0}{\sh}}{\SRDALT} \subseteq \USET{\sha}{\SRDALT}{\CENTAIL{\alpha}}~,
\end{align*}
where the last inclusion is decidable by Corollary~\ref{thm:compositional:inclusion}.
%
\qed
\end{proof}
\begin{algorithm}[t]
  \footnotesize
  \SetKwData{LHS}{$\Omega$}
  \SetKwData{LHSP}{$\SRDALT$}
  \SetKwData{RHS}{$\HASH{\sha}$}
  \SetKwData{RHSP}{$\overline{\HASH{\sha}}$}
  \SetKwInOut{Input}{Input}\SetKwInOut{Output}{Output}
  \hrule \vspace*{0.5em}
  \Input{
    established SID $\SRD$,~ $\sh,\sha \in \SL{\SRD}{}$ determined, \\
    heap automaton $\HASH{\PS_i}$ for each $\PS_i \in \PRED(\SRD)$
  }
  \Output{yes iff $\sh \ENTAIL{\SRD} \sha$ holds }
  \vspace*{0.5em} \hrule
  \BlankLine
   {
     \LHS$~\leftarrow~\{ \SRDRULE{\PS}{\sh} \} ~\cup~ \SRD$
     \tcp*{$\PS$ fresh predicate symbol}
   }
   {
     \LHSP$~\leftarrow~ \texttt{removeUnsat}(\LHS)$
     \tcp*{Theorem~\ref{thm:zoo:sat:property}}
   }
   {
     \RHS$~\leftarrow~ \texttt{automaton}(\sha,\HASH{\PS_1},\HASH{\PS_2}, \ldots)$ 
     \tcp*{Theorem~\ref{thm:entailment:top-level}}
   }
   {
     \RHSP $~\leftarrow~ \texttt{complement}(\RHS)$
     \tcp*{Lemma~\ref{thm:refinement:boolean}}
   }
   {
     \Return{
       yes iff
       $\CALLSEM{\PS\T{x}}{\LHSP} ~\cap L(\RHSP) = \emptyset$
     }
     \tcp*{Algorithm~\ref{alg:on-the-fly-refinement}}
   }
  \hrule \vspace*{0.5em}
  \caption{Decision procedure for $\sh \ENTAIL{\SRD} \sha$.}\label{alg:entailment:decision-procedure}
\end{algorithm}
%
%
%

\subsection{Complexity}

Algorithm~\ref{alg:entailment:decision-procedure} 
may be fed with arbitrarily large heap automata.
For a meaningful complexity analysis, we thus consider heap automata of bounded size only.
\begin{definition} \label{def:alpha-bounded}
  An SID $\SRD$ is 
  $\alpha$--\emph{bounded} if for each $\PS \in \PRED(\SRD)$
  there exists a heap automaton $\HA{A}_{P}$  over $\SHCENTAIL{\alpha}$ 
  accepting $\USET{\PS\T{x}}{\SRD}{\CENTAIL{\alpha}}$
  such that
  $\Delta_{\HA{A}_P}$ is decidable in $\BIGO{2^{\text{poly}(\SIZE{\SRD})}}$
  and $\SIZE{Q_{\HA{A}_P}} \leq 2^{\text{poly}(\alpha)}$.
\end{definition}
The bounds from above are natural for a large class of heap automata.
In particular, all heap automata constructed in Section~\ref{sec:zoo}
stay within these bounds.
Then a close analysis of Algorithm~\ref{alg:entailment:decision-procedure} for $\alpha$--bounded SIDs yields the following complexity results.
A detailed analysis is provided in Appendix~\ref{app:complexity}.
\begin{theorem} \label{thm:entailment:complexity}
  \DENTAIL{\CENTAIL{\alpha}}{\SRD} is decidable in \CCLASS{2-ExpTime}
  for every $\alpha$--bounded SID $\SRD$.
  If $\alpha \geq 1$ is a constant then
  \DENTAIL{\CENTAIL{\alpha}}{\SRD} is \CCLASS{ExpTime}-complete.
\end{theorem}
Note that lower complexity bounds depend on the 
SIDs under consideration.
Antonopoulos et al.~\cite[Theorem 6]{antonopoulos2014foundations} showed that
the entailment problem is already $\Pi^{P}_{2}$--complete\footnote{$\Pi^{P}_{2}$ denotes the second level of the polynomial hierarchy.} for the base fragment, i.e., $\SRD = \emptyset$.
Thus, under common complexity assumptions, the exponential time upper bound derived in Theorem~\ref{thm:entailment:complexity} is asymptotically optimal for a deterministic algorithm.
Since the entailment problem is already \CCLASS{ExpTime}--hard for points-to assertions of arity $3$ and SIDs specifying regular sets of trees
(cf. \cite[Theorem 5]{antonopoulos2014foundations} and Appendix~\ref{app:entailment:lower} for details),
exponential time is actually needed for certain SIDs.
\subsection{Expressiveness}

We conclude this section with a brief remark regarding the expressiveness of heap automata.
In particular, SIDs specifying common data structures,
such as lists, trees, trees with linked leaves and combinations thereof can be encoded by heap automata.\footnote{Details on the construction of such automata are provided in Appendix~\ref{app:case-study}.} 
In general, the close relationship between established SIDs and context-free graph languages
studied by Dodds~\cite[Theorem 1]{dodds2008separation} 
and Courcelle's work on recognizable graph languages~\cite[Theorems 4.34 and 5.68]{courcelle2012graph},
suggest that heap automata exist for every set of reduced symbolic heaps that can be specified in monadic second-order logic over graphs~\cite{courcelle2012graph}.


%
%
%
%
%
\section{Conclusion}
\label{sec:conclusion}
We developed an algorithmic framework for automatic reasoning about and debugging of the
symbolic heap fragment of separation logic.
%
Our approach is centered around a new automaton model, \emph{heap
automata}, that is specifically tailored to symbolic heaps.
We show that many common robustness properties as well as certain types of entailments
are naturally covered by our framework---often with optimal asymptotic complexity.
%
%
%
There are several directions for future work including automated
learning of heap automata accepting common data structures
and applying heap automata to the abduction problem~\cite{calcagno2009compositional}.


%

%

\bibliographystyle{splncs03}
\bibliography{biblio}
\appendix

\renewcommand\thesection{A.\arabic{section}}

\newpage
\section{Supplementary Material} \label{app:supplement}

\begin{definition}[Isomorphic states]
    Two states $(\stack_1,\heap_1), (\stack_2,\heap_2)$ are \emph{isomorphic} if and only if 
    there exist bijective functions $f : \DOM(\stack_1) \to \DOM(\stack_2)$
    $g : \DOM(\heap_1) \to \DOM(\heap_2)$ such that for all $\ell \in \DOM(\heap_1)$, we have
    $g(\heap_1(\ell)) = \heap_2(g(\ell))$, 
    where $g$ is lifted to tuples by componentwise application.
\end{definition}

\section{Implemenation and Experimental Results} \label{app:implementation}

Overall, the implementation of our tool \textsc{Harrsh} consists of about 1500 lines of Scala
code, not counting test classes, comments and blank lines.
Since our tool is---to our best knowledge---the first one to
systematically reason about robustness properties, we cannot compare
the results of our tool against other implementations.

A notable exception is \textsc{Cyclist}~\cite{brotherston2014decision}, which is capable of proving satisfiability of SIDs.
We evaluated our tool against the large collection of benchmarks that
is distributed with \textsc{Cyclist}.
In particular, this collection includes the following sets:
\begin{enumerate}
  \item A set of handwritten standard predicates from the separation logic literature.
  \item 45945 problem instances that have been automatically generated by the inference tool \textsc{Caber}~\cite{brotherston2014cyclic}.
  \item A set of particularly hard problem instances that are derived from the SIDs used to prove lower complexity bounds for satisfiability.
      These benchmarks have been used to test the scalability of \textsc{Cyclist}.
\end{enumerate}

Experiments were performed on an Intel Core i5-3317U at 1.70GHz with 4GB of RAM.

For the standard predicates in the first set, our implementation runs in total approximately 300ms to check \emph{all} robustness properties on all standard predicates, i.e., a total of 45 problem instances.
As already reported in the paper, this low analysis time is not a surprise, because the standard data structure predicates are generally very well-behaved.

To evaluate the performance of \textsc{Harrsh} on a realistic set of benchmarks, we ran both \textsc{Harrsh} and \textsc{Cyclist} on all 45945 benchmarks generated by \textsc{Caber}. For \textsc{Cyclist}, we only checked satisfiability---the only of the robustness properties supported by \textsc{Cyclist}; for \textsc{Harrsh}, we checked all robustness properties introduced in Section~\ref{sec:zoo}.

Both tools were capable of proving (un)satisfiability of all of these problem instances within a set timeout of $30$ seconds.
All in all, the \emph{accumulated} analysis time of \textsc{Harrsh} for these instances was 12460ms, while \textsc{Cyclist} required 44856ms.\footnote{For both tools we added up the analysis times of individual tasks, reported with millisecond precision. Consequently, we expect that rounding errors influence the accumulated time to a similar degree for both tools.}
For all other properties, \textsc{Harrsh} also achieved \emph{accumulated} analysis time below 20 seconds; see Table~\ref{tbl:robustness-benchmarks}.  These numbers demonstrate the applicability of our tool to problem instances that occur in practice.

\begin{table}
    \begin{center}
    \setlength{\tabcolsep}{10pt}
    \begin{tabular}{l|r}
        \textbf{Robustness Property} & \textbf{Analysis Time (ms)} \\ 
        \hline
        No points-to assertions (Example~\ref{ex:heap-automaton:toy}) & 7230 \\ 
        Tracking property (Section~\ref{sec:zoo:track}) & 11459 \\ 
        Satisfiability (Section~\ref{sec:zoo:sat}) & 12460 \\ 
        Complement of Satisfiability (Section~\ref{sec:zoo:sat}) & 11980 \\ 
        Establishment (Section~\ref{sec:zoo:establishment}) & 18055 \\ 
        Complement of Establishment (Section~\ref{sec:zoo:establishment}) & 17272 \\ 
        Reachability (Section~\ref{sec:zoo:reach}) & 14897 \\ 
        Garbage-Freedom (Section~\ref{sec:zoo:garbage}) & 18192 \\ 
        Weak Acyclicity (Section~\ref{sec:zoo:acyclicity}) & 18505 \\ 
    \end{tabular}
    \end{center}
    \caption{Total analysis time for robustness properties presented throughout the paper on the second set of benchmarks, i.e., 45945 automatically inferred SIDs. }
    \label{tbl:robustness-benchmarks}
\end{table}

Moreover, we ran \textsc{Harrsh} and \textsc{Cyclist} to check satisfiability of the third set and additional handwritten benchmarks distributed with \textsc{Cyclist}.
For both tools, we chose a timeout of 5 minutes.
The measured analysis times for this set are shown in Table~\ref{tbl:hard-benchmarks}.

Notably, the tools yield different results for the SID contained in the file \texttt{inconsistent-ls-of-ls.def}.
While \textsc{Harrsh} states that this SID is satisfiable, \textsc{Cyclist} states that it is not.
Despite the benchmark's name, however, the underlying SID
\begin{align*}
    & \SRDRULE{P}{x = \NIL} \\
    & \SRDRULE{P}{Q(x x) : \{ x \neq \NIL\}} \\
    & \SRDRULE{Q}{\exists c,d . \PT{x}{d,c} \SEP P(d) : \{ y = \NIL, x \neq \NIL \}} \\
    & \SRDRULE{Q}{\exists c,d . \PT{x}{d,c} : \{y \neq \NIL\}}
\end{align*}
is satisfiable:
Clearly $x = \NIL$ is a satisfiable unfolding of $Px$.
Using this unfolding to replace the predicate call $P(d)$ in the third rule, we also obtain a a satisfiable unfolding of $Q(x,y)$:
\[
    \exists c,d . \PT{x}{d,c} : \{ d = \NIL, y = \NIL, x \neq \NIL \}
\]

\begin{table}[!htbp]
    \begin{center}
    \setlength{\tabcolsep}{10pt}
    \begin{tabular}{l|r|r}
        \textbf{Benchmark} & \textsc{Harrsh} & \textsc{Cyclist} \\
        \hline
        inconsistent-ls-of-ls.defs &  1 & 4 (not correct) \\
        succ-rec01.defs &  3 & 0 \\
        succ-rec02.defs &  10 & 8 \\
        succ-rec03.defs &  24 & 12 \\
        succ-rec04.defs &  106 & 20 \\
        succ-rec05.defs &  496 & 128 \\
        succ-rec06.defs &  2175 & 792 \\
        succ-rec07.defs &  9692 & 4900 \\
        succ-rec08.defs &  39408 & 31144 \\
        succ-rec09.defs &  169129 & 164464 \\
        succ-rec10.defs &  TO & TO \\
        succ-circuit01.defs &  80 & 4\\
        succ-circuit02.defs &  142 & 8\\
        succ-circuit03.defs &  699 & 48 \\
        succ-circuit04.defs &  4059 & 832\\
        succ-circuit05.defs &  75110 & 28800\\
        succ-circuit06.defs &  TO & TO \\
    \end{tabular}
    \end{center}
    \caption{Comparison of \textsc{Harrsh} and \textsc{Cyclist} for hard instances of the satisfiability problem. Provided times are in milliseconds. Timeouts (TO) are set to 5 minutes.}
    \label{tbl:hard-benchmarks}
\end{table}


%
%
\section{Proof of Lemma~\ref{thm:symbolic-heaps:semantics}} \label{app:symbolic-heaps:semantics}
 By induction on the height $k$ of unfolding trees of $\sh$.
 \paragraph{I.B.}
 If $k = 0$ then $\NOCALLS{\sh} = 0$, i.e., $\sh$ contains no predicate calls.
 Thus $\UNFOLD{t} = \sh$.
 Then, for each $(\stack,\heap) \in \STATES$, we have
 \begin{align*}
                     & \stack,\heap \SAT{\SRD} \sh \\
   ~\Leftrightarrow~ & \left[ \NOCALLS{\sh} = 0 \right] \\
                     & \stack,\heap \SAT{\emptyset} \sh \\
   ~\Leftrightarrow~ & \left[ \CALLSEM{\sh}{\SRD} = \{ \sh \} \right] \\
                     & \exists \rsh \in \CALLSEM{\sh}{\SRD} ~.~ \stack,\heap \SAT{\emptyset} \rsh.
 \end{align*}
 \paragraph{I.H.}
 Assume for an arbitrary, but fixed, natural number $k$ that for each $\SRD \in \SETSRD{}$, $\sh \in \SL{\SRD}{}$, where each $t \in \UTREES{\SRD}{\sh}$ is of height at most $k$, it holds for each $(\stack,\heap) \in \STATES$ that
 \begin{align*}
   \stack,\heap \SAT{\SRD} \sh \quad\text{iff}\quad \exists \rsh \in \CALLSEM{\sh}{\SRD} ~.~ \stack,\heap \SAT{\emptyset} \rsh.
 \end{align*}
 \paragraph{I.S.}
 Let $\SRD \in \SETSRD{}$ and $\sh \in \SL{\SRD}{}$ such that each $t \in \UTREES{\SRD}{\sh}$ is of height at most $k+1$.
 We proceed by structural induction on the syntax of $\sh$.
 For $\sh = \EMP$, $\sh = \PT{x}{\T{a}}$, $\sh = (a=b)$, $\sh = (a \neq b)$, the height of all unfolding trees is $0 < k+1$, i.e., there is nothing to show.
 If $\sh = \PS\T{a}$ then $\UNFOLD{t} = \UNFOLD{\SUBTREE{t}{1}}$ holds for each unfolding of $\sh$.
 Since $t$ is of height at most $k+1$,  $\SUBTREE{t}{1}$ is of height at most $k$.
 By I.H. we obtain for each $(\stack,\heap) \in \STATES$ that
 \begin{align*}
   \stack,\heap \SAT{\SRD} \PS\T{a} \quad\text{iff}\quad \exists \rsh \in \CALLSEM{\PS\T{a}}{\SRD} ~.~ \stack,\heap \SAT{\emptyset} \rsh.
 \end{align*}
 If $\sh = \sh_1 \SEP \sh_2$, we have for each $(\stack,\heap) \in \STATES$:
 \begin{align*}
                     & \stack,\heap \SAT{\SRD} \sh_1 \SEP \sh_2 \\
   ~\Leftrightarrow~ & \left[ \text{Semantics of $\SEP$} \right] \\
                     & \exists \heap_1,\heap_2 . \heap = \heap_1 \uplus \heap_2 ~\text{and}~ \stack,\heap_1 \SAT{\SRD} \sh_1 ~\text{and}~ \stack,\heap_2 \SAT{\SRD} \sh_2 \\
   ~\Leftrightarrow~ & \left[ \text{I.H. on}~\sh_1 \right] \\
                     & \exists \heap_1,\heap_2 . \heap = \heap_1 \uplus \heap_2 \\
                     & \quad\text{and}~ \exists \rsh_1 \in \CALLSEM{\sh_1}{\SRD} . \stack,\heap_1 \SAT{\SRD} \rsh_1 \\
                     & \quad\text{and}~ \stack,\heap_2 \SAT{\SRD} \sh_2 \\
   ~\Leftrightarrow~ & \left[ \text{I.H. on}~\sh_2 \right] \\
                     & \exists \heap_1,\heap_2 . \heap = \heap_1 \uplus \heap_2 \\
                     & \quad\text{and}~ \exists \rsh_1 \in \CALLSEM{\sh_1}{\SRD} . \stack,\heap_1 \SAT{\emptyset} \rsh_1 \\
                     & \quad\text{and}~ \exists \rsh_2 \in \CALLSEM{\sh_2}{\SRD} . \stack,\heap_2 \SAT{\emptyset} \rsh_2 \\
   ~\Leftrightarrow~ & \big[ \CALLSEM{\sh_1 \SEP \sh_2}{\SRD} = \{ (\sh_1 \SEP \sh_2)[\CALLS{\sh_1} / \rsh_1, \CALLS{\sh_2} / \rsh_2] ~|~ \rsh_1 \in \CALLSEM{\sh_1}{\SRD}, \\
                     & \qquad \qquad \qquad \qquad \qquad \qquad \qquad \qquad \qquad \rsh_2 \in \CALLSEM{\sh_2}{\SRD} \} \big] \\
                     & \exists \rsh \in \CALLSEM{\sh_1 \SEP \sh_2}{\SRD} ~.~ \stack,\heap \SAT{\emptyset} \rsh
 \end{align*}
 Finally, we consider the case $\sh = \SYMBOLICHEAP{}$.
 The crux of the proof relies on the observation that for each $t \in \UTREES{\SRD}{\sh}$ and each $t' \in \UTREES{\SRD}{\SPATIAL{} \SEP \CALLS{}}$, we have
 $t(\EMPTYSEQ) = \exists \BV{} . t'(\EMPTYSEQ) : \PURE{}$. Thus
 \begin{align*}
  \CALLSEM{\sh}{\SRD} = \{ \exists \BV{} . \rsh' : \PURE{} ~|~ \rsh' \in \CALLSEM{\SPATIAL{} \SEP \CALLS{}}{\SRD} \}. \tag{$\dag$}
 \end{align*}

 Then we have for each $(\stack,\heap) \in \STATES$:
 \begin{align*}
                     & \stack,\heap \SAT{\SRD} \exists \BV{} . \SPATIAL{} \SEP \CALLS{} : \PURE{} \\
   ~\Leftrightarrow~ & \left[ \text{SL semantics} \right] \\
                     & \exists \T{v} \in \VAL^{\SIZE{\BV{}}} \,.\, \stack\remap{\BV{}}{\T{v}}, \heap \SAT{\SRD} \SPATIAL{} \SEP \CALLS{} \\
                     & \quad \text{and}~ \forall \pi \in \PURE{} \,.\, \stack\remap{\BV{}}{\T{v}}, \heap \SAT{\SRD} \pi \\
   ~\Leftrightarrow~ & \left[ \text{I.H. on}~\SPATIAL{} \SEP \CALLS{} \right] \\
                     & \exists \T{v} \in \VAL^{\SIZE{\BV{}}} \,.\, \exists \rsh' \in \CALLSEM{\SPATIAL{} \SEP \CALLS{}}{\SRD} \,.\, \stack\remap{\BV{}}{\T{v}}, \heap \SAT{\emptyset} \rsh' \\
                     & \quad \text{and}~ \forall \pi \in \PURE{} \,.\, \stack\remap{\BV{}}{\T{v}}, \heap \SAT{\SRD} \pi \\
   ~\Leftrightarrow~ & \left[  \exists x \exists y \equiv \exists y \exists x \right] \\
                     & \exists \rsh' \in \CALLSEM{\SPATIAL{} \SEP \CALLS{}}{\SRD} \,.\, \exists \T{v} \in \VAL^{\SIZE{\BV{}}} \,.\, \stack\remap{\BV{}}{\T{v}}, \heap \SAT{\emptyset} \rsh' \\
                     & \quad \text{and}~ \forall \pi \in \PURE{} \,.\, \stack\remap{\BV{}}{\T{v}}, \heap \SAT{\SRD} \pi \\
   ~\Leftrightarrow~ & \left[ \text{apply}~(\dag) \right] \\
                     & \exists \rsh \in \CALLSEM{\sh}{\SRD} \,.\, \stack,\heap \SAT{\emptyset} \rsh. 
 \end{align*}
 \qed
%
%
%
\section{Coincidence Lemma for Symbolic Heaps} \label{app:symbolic-heaps:fv-coincidence}
\begin{lemma} \label{thm:symbolic-heaps:fv-coincidence}
 Let $\SRD \in \SETSRD{}$ and $\sh \in \SL{\SRD}{}$.
 Moreover, let $(\stack,\heap) \in \STATES$.
 Then
 $
    \stack,\heap \SAT{\SRD} \sh ~\text{iff}~ (\restr{\stack}{\FV{0}{\sh}}), \heap \SAT{\SRD} \sh,~
 $
 where $\restr{\stack}{\FV{0}{\sh}}$ denotes the restriction of the domain of $\stack$ to the free variables of $\sh$.
\end{lemma}
%
\begin{proof}
  By structural induction on the syntax of symbolic heaps $\sh$.
  \paragraph{$\sh = \EMP$}
  \begin{align*}
                      & \stack,\heap \SAT{\SRD} \EMP \\
    ~\Leftrightarrow~ & \left[ \text{SL semantics} \right] \\
                      & \DOM(\heap) = \emptyset \\
    ~\Leftrightarrow~ & \left[ \FV{0}{\EMP} = \emptyset \right] \\
                      & (\restr{\stack}{\FV{0}{\EMP}}),\heap \SAT{\SRD} \EMP.
  \end{align*}
  \paragraph{$\sh = \PT{x}{\T{a}}$}
  \begin{align*}
                      & \stack,\heap \SAT{\SRD} \PT{x}{\T{a}} \\
    ~\Leftrightarrow~ & \left[ \text{SL semantics} \right] \\
                      & \DOM(\heap) = \{\stack(x)\} ~\text{and}~ h(\stack(x)) = \stack(\T{a}) \\
    ~\Leftrightarrow~ & \left[ \FV{0}{\PT{x}{\T{a}}} = \{x\} \cup \T{a} \right] \\
                      & (\restr{\stack}{\FV{0}{\PT{x}{\T{a}}}}),\heap \SAT{\SRD} \PT{x}{\T{a}}.
  \end{align*}
  \paragraph{$\sh = (a \sim b),~ \sim \in \{=,\neq\}$}
  \begin{align*}
                      & \stack,\heap \SAT{\SRD} a \sim b \\
    ~\Leftrightarrow~ & \left[ \text{SL semantics} \right] \\
                      & \stack(a) \sim \stack(b) \\
    ~\Leftrightarrow~ & \left[ \FV{0}{a \sim b} = \{a,b\} \right] \\
                      & (\restr{\stack}{\FV{0}{a \sim b}}),\heap \SAT{\SRD} a \sim b.
  \end{align*}
  \paragraph{$\sh = \rsha \SEP \rsh$}
    \begin{align*}
                      & \stack,\heap \SAT{\SRD} \rsha \SEP \rsh \\
    ~\Leftrightarrow~ & \left[ \text{SL semantics} \right] \\
                      & \exists h_1,\heap_2 \,.\, h = h_1 \uplus h_2 ~\text{and}~ \stack,\heap_1 \SAT{\SRD} \rsha ~\text{and}~ \stack,\heap_2 \SAT{\SRD} \rsh \\
    ~\Leftrightarrow~ & \left[ \text{I.H.} \right] \\
                      & \exists h_1,\heap_2 \,.\, h = h_1 \uplus h_2 \\
                      & \quad \text{and}~ (\restr{\stack}{\FV{0}{\rsha}},\heap_1) \SAT{\SRD} \rsha ~\text{and}~ (\restr{\stack}{\FV{0}{\rsh}},\heap_2) \SAT{\SRD} \rsh \\
    ~\Leftrightarrow~ & \left[ \FV{0}{\rsh},\FV{0}{\rsha} \,\subseteq\, \DOM(s),~\text{I.H.} \right] \\
                      & \exists h_1,\heap_2 \,.\, h = h_1 \uplus h_2 \\
                      & \quad \text{and}~ (\restr{\stack}{\FV{0}{\rsha} \cup \FV{0}{\rsh}},\heap_1) \SAT{\SRD} \rsha ~\text{and}~ (\restr{\stack}{\FV{0}{\rsha} \cup \FV{0}{\rsh}},\heap_2) \SAT{\SRD} \rsh \\
    ~\Leftrightarrow~ & \left[ \FV{0}{\rsha \SEP \rsh} = \FV{0}{\rsha} \cup \FV{0}{\rsh},~ \text{SL semantics} \right] \\
                      & (\restr{\stack}{\FV{0}{\rsha \SEP \rsh}}),\heap \SAT{\SRD} \rsha \SEP \rsh.
  \end{align*}
  \paragraph{$\sh = \PS\T{a}$}
  \begin{align*}
                      & \stack,\heap \SAT{\SRD} \PS\T{a} \\
    ~\Leftrightarrow~ & \left[ \text{SL semantics} \right] \\
                      & \exists \rsh \in \CALLSEM{\PS\T{a}}{\SRD} \,\, \stack,\heap \SAT{\emptyset} \rsh \\
    ~\Leftrightarrow~ & \left[ \text{I.H.} \right] \\
                      & \exists \rsh \in \CALLSEM{\PS\T{a}}{\SRD} \,\, (\restr{\stack}{\FV{0}{\rsh}}),\heap \SAT{\SRD} \rsh \\
    ~\Leftrightarrow~ & \left[ \FV{0}{\PS\T{a}} = \FV{0}{\rsh} \right] \\
                      & (\restr{\stack}{\FV{0}{\PS\T{a}}}),\heap \SAT{\SRD} \PS\T{a}.
  \end{align*}
  \paragraph{$\sh = \SYMBOLICHEAP{}$}
  \begin{align*}
                      & \stack,\heap \SAT{\SRD} \SYMBOLICHEAP{} \\
    ~\Leftrightarrow~ & \left[ \text{SL semantics} \right] \\
                      & \exists \T{v} \in \VAL^{\SIZE{\BV{}}} \,.\, \stack\remap{\BV{}}{\T{v}}, h \SAT{\SRD} \SPATIAL{} \SEP \CALLS{} \\
                      & \quad \text{and}~ \forall \pi \in \PURE{} \,.\, \stack\remap{\BV{}}{\T{v}}, h \SAT{\SRD} \pi \\
    ~\Leftrightarrow~ & \left[ \text{I.H. on}~\SPATIAL{}\SEP\CALLS{} ~\text{and}~ \pi \right] \\
                      & \exists \T{v} \in \VAL^{\SIZE{\BV{}}} \,.\, (\stack\remap{\BV{}}{\T{v}}\upharpoonright_{\FV{0}{\SPATIAL{} \SEP \CALLS{}}}),\heap \SAT{\SRD} \SPATIAL{} \SEP \CALLS{} \\
                      & \quad \text{and}~ \forall \pi \in \PURE{} \,.\, (\stack\remap{\BV{}}{\T{v}}\upharpoonright_{\FV{0}{\pi}}),\heap \SAT{\SRD} \pi \\
    ~\Leftrightarrow~ & \left[ \DOM(\stack\remap{\BV{}}{\T{v}}\upharpoonright_{\FV{0}{\pi}}) \subseteq \DOM((\restr{\stack}{\FV{0}{\pi}})[\BV{} \mapsto \T{v}]) \right] \\
                      &  \exists \T{v} \in \VAL^{\SIZE{\BV{}}} \,.\, (\restr{\stack}{\FV{0}{\SPATIAL{} \SEP \CALLS{}}}[\BV{} \mapsto \T{v}]),\heap \SAT{\SRD} \SPATIAL{} \SEP \CALLS{} \\
                      & \quad \text{and}~ \forall \pi \in \PURE{} \,.\, (\restr{\stack}{\FV{0}{\pi}}[\BV{} \mapsto \T{v}]),\heap \SAT{\SRD} \pi \\
    ~\Leftrightarrow~ & \left[ \FV{0}{\sh} = \left(\FV{0}{\SPATIAL{} \SEP \CALLS{}} \cup \bigcup_{\pi \in \PURE{}} \FV{0}{\pi} \right) \setminus \BV{} \right] \\
                      &  \exists \T{v} \in \VAL^{\SIZE{\BV{}}} \,.\, (\restr{\stack}{\FV{0}{\sh}}[\BV{} \mapsto \T{v}]),\heap \SAT{\SRD} \SPATIAL{} \SEP \CALLS{} \\
                      & \quad \text{and}~ \forall \pi \in \PURE{} \,.\, (\restr{\stack}{\FV{0}{\sh}}[\BV{} \mapsto \T{v}]),\heap \SAT{\SRD} \pi \\
    ~\Leftrightarrow~ & \left[ \text{SL semantics} \right] \\
                      & (\restr{\stack}{\FV{0}{\sh}}), h \SAT{\SRD} \SYMBOLICHEAP{}. 
  \end{align*}
  \qed
\end{proof}
%
%
%
\section{The emptiness problem for sets of unfolding trees} \label{app:symbolic-heaps:emptiness}
The set of unfolding trees of a given symbolic heap $\sh$ with predicates taken from an SID $\SRD$ can be accepted by a bottom-up tree automaton $\HA{A}$.
Then, the set of unfolding trees of $\sh$ is empty if and only if $\HA{A}$ accepts the empty language (of trees).
Since the emptiness problem of tree automata is $\CCLASS{PTime}$--complete (cf. \cite[Theorem 1.7.4]{comon2007tree}), so is the question whether the set of unfolding trees of a given symbolic heap is empty.

Formally, let $\SRD \in \SETSRD{}$ and $\sh \in \SL{\SRD}{}$.
 Then the set of unfolding trees of $\sh$ w.r.t. $\SRD$ is the set of all trees accepted by the tree automaton $\mathfrak{A} = (Q,A,\Delta,F)$, where
 \begin{align*}
   Q ~=~ & \{ \PS ~|~ \SRDRULE{\PS}{\sha} \in \SRD \} \HEAPUNION \{ S \}, \quad F = \{ S \}, \\
   A ~=~ & \{ \sha ~|~ \SRDRULE{\PS}{\sha} \in \SRD \} \cup \{ \sh \}, \\
   \TARULE{\PS_1 \ldots \PS_m}{\psi}{\PS_0} ~\text{iff}~ & \SRDRULE{\PS_0}{\sha} \in \SRD ~\text{or}~ \sha = \sh ~\text{and}~ \PS_0 = S,~
 \end{align*}
 where $\CALLS{\sha} = \CALLN{1}{} \SEP \ldots \SEP \CALLN{m}{}$.
%


%
%
%
\section{Proof of Theorem~\ref{thm:refinement:boolean}} \label{app:refinement:boolean}
\begingroup
\def\thetheorem{\ref{thm:refinement:boolean}}
\begin{theorem}
  Let $\HA{A}$ and $\HA{B}$ be heap automata over $\SHCLASS$.
  Then there exists a heap automata $\HA{C}_1,\HA{C}_2,\HA{C}_3$ over $\SHCLASS$ with
  $L(\HA{C}_1) = L(\HA{A}) \cup L(\HA{B})$,
  $L(\HA{C}_2) = L(\HA{A}) \cap L(\HA{B})$, and
  $L(\HA{C}_3) = \RSHCLASS \setminus L(\HA{A})$, respectively.
\end{theorem}
\addtocounter{theorem}{-1}
\endgroup
\begin{proof}
 We split the proof across three lemmas that are proven subsequently.
 Closure under union and intersection are shown in Lemma~\ref{thm:refinement:union} and Lemma~\ref{thm:refinement:intersection}, respectively.
 Finally, closure under complement with respect to $\RSL{}{\SRDCLASS}$ is proven in Lemma~\ref{thm:refinement:complement}.
 \qed
\end{proof}

\begin{lemma} \label{thm:refinement:union}
  Let $\HA{A},\HA{B}$ be heap automata over $\SHCLASS$. 
  Then there exists a heap automaton $\HA{C}$ over $\SHCLASS$ such that $L(\HA{C}) = L(\HA{A}) \cup L(\HA{B})$.
\end{lemma}
\begin{proof}
  We construct a heap automaton $\HA{C} = (Q,\SHCLASS,\Delta,F)$ as follows:
 \begin{align*}
   & Q ~\DEFEQ~ \left(Q_{\HA{A}} \cup \{\bot\}\right) \times \left(Q_{\HA{B}} \cup \{\bot\}\right) \\
   & F ~\DEFEQ~ F_{\HA{A}} \times \{\bot\} ~\cup~ \{\bot\} \times F_{\HA{B}} \\
   & \Delta ~:~ \MOVE{C}{(p_0,q_0)}{\sh}{(p_1,q_1)\ldots(p_m,q_m)} ~\Leftrightarrow~ \\
   & \qquad \qquad  \MOVE{A}{p_0}{\sh}{p_1 \ldots p_m} ~\text{and}~ p_2 = \bot ~\text{or} \\
   & \qquad \qquad  \MOVE{B}{q_0}{\sh}{q_1 \ldots q_m} ~\text{and}~ p_1 = \bot~,
 \end{align*}
 where $\sh \in \SL{}{\SRDCLASS}$ with $\NOCALLS{\sh} = m$.
 Moreover, $\bot$ is assumed to be a fresh state.
 Clearly, $\Delta$ is decidable, because $\Delta_{\HA{A}}$ and $\Delta_{\HA{B}}$ are.
 Assuming $\HA{C}$ satisfies the compositionality property, we first show that the language of $\HA{C}$ in fact accepts $L(\HA{A}) \cup L(\HA{B})$: 
 \begin{align*}
     & \rsh \in L(\HA{A}) \cup L(\HA{B}) \\
     ~\Leftrightarrow~ & \left[ \text{Definition of}~\cup \right] \\
     & \rsh \in L(\HA{A}) ~\text{or}~ \rsh \in L(\HA{B}) \\
     ~\Leftrightarrow~ & \left[ \text{Definition of}~L(\HA{A}), L(\HA{B}) \right] \\
                     & \exists p \in F_{\HA{A}} \,.\, \OMEGA{A}{p}{\rsh} ~\text{or}~
                       \exists q \in F_{\HA{B}} \,.\, \OMEGA{B}{q}{\rsh} \\
   ~\Leftrightarrow~ & \left[ \text{construction of}~\HA{C} \right] \\
                     & \exists p \in F_{\HA{A}} \,.\, \OMEGA{C}{(p,\bot)}{\rsh} ~\text{or}~
                       \exists q \in F_{\HA{B}} \,.\, \OMEGA{C}{(\bot,q)}{\rsh} \\
   ~\Leftrightarrow~ & \left[ \text{Definition of}~F \right] \\
                     & \exists (p,q) \in F \,.\, \OMEGA{C}{(p,q)}{\rsh}.
 \end{align*}
 It remains to show the compositionality property, i.e., for each $p \in Q$, $\sh \in \SL{}{\SRDCLASS}$ with $\NOCALLS{\sh} = m$ and $\rsh_1,\ldots, \rsh_m \in \RSL{}{\SRDCLASS}$ with
 $\rsha = \sh[P_1/\rsh_1,\ldots,P_m/\rsh_m]$, we have
  \begin{align*}
    & \OMEGA{C}{p}{\rsha}
    ~\Leftrightarrow~
    \exists \T{q} \in Q^m \,.\, \MOVE{C}{p}{\sh}{\T{q}} ~\text{and}~ \forall 1 \leq i \leq m \,.\, \OMEGA{C}{\PROJ{\T{q}}{i}}{\rsh_i}.
  \end{align*}
 In the following proof, we write $p_1,p_2$ to denote the first and second component of $p \in Q$, respectively.
 \begin{align*}
                      & \OMEGA{C}{p}{\rsha} \\
    ~\Leftrightarrow~ & \left[ \text{construction of}~\HA{C} \right] \\
                      & p_2 = \bot ~\text{and}~ \OMEGA{A}{p_1}{\rsha} ~~\text{or}~~ p_1 = \bot ~\text{and}~ \OMEGA{B}{p_2}{\rsha} \\
    ~\Leftrightarrow~ & \left[ \text{compositionality of}~\HA{A},\HA{B} \right] \\
                      & p_2 = \bot ~\text{and}~ \exists \T{q_1} \in Q_{\HA{A}}^{m} \,.\, \MOVE{A}{p_1}{\sh}{\T{q_1}} ~\text{and}~ \\
                      & \forall 1 \leq i \leq m \,.\, \OMEGA{A}{\PROJ{\T{q_1}}{i}}{\rsh_i} \\
                      & ~~\text{or}~~ \\
                      & p_1 = \bot ~\text{and}~ \exists \T{q_2} \in Q_{\HA{B}}^{m} \,.\, \MOVE{B}{p_2}{\sh}{\T{q_2}} ~\text{and}~ \\
                      & \forall 1 \leq i \leq m \,.\, \OMEGA{B}{\PROJ{\T{q_2}}{i}}{\rsh_i} \\
    ~\Leftrightarrow~ & \big[ \text{setting}~\T{q} = (\PROJ{\T{q_1}}{1},\PROJ{\T{q_2}}{1}) \ldots (\PROJ{\T{q_1}}{m},\PROJ{\T{q_2}}{m}),~\text{construction of}~\HA{C} \big] \\
                      & \exists \T{q} \in Q^{m} \,.\, \MOVE{C}{p}{\sh}{\T{q}} ~\text{and}~ \forall 1 \leq i \leq m \,.\, \OMEGA{C}{\PROJ{\T{q}}{i}}{\rsh_i}. 
 \end{align*}
 \qed
\end{proof}
\begin{lemma} \label{thm:refinement:intersection}
  Let $\HA{A},\HA{B}$ be heap automata over $\SHCLASS$. 
  Then there exists a heap automaton $\HA{C}$ over $\SHCLASS$ such that $L(\HA{C}) = L(\HA{A}) \cap L(\HA{B})$.
\end{lemma}
\begin{proof}
  Let $\HA{A},\HA{B}$ be heap automata over $\SHCLASS$ accepting $H$ and $K$.
  We construct a heap automaton $\HA{C} = (Q,\SHCLASS,\Delta,F)$ as follows: 
  \begin{align*}
   & Q ~\DEFEQ~ Q_{\HA{A}} \times Q_{\HA{B}} \qquad F ~\DEFEQ~ F_{\HA{A}} \times F_{\HA{B}} \\
   & \Delta ~:~ \MOVE{C}{(p_0,q_0)}{\sh}{(p_1,q_1)\ldots(p_m,q_m)} ~\Leftrightarrow~ \\
   & \qquad \MOVE{A}{p_0}{\sh}{p_1 \ldots p_m} ~\text{and}~ \MOVE{B}{q_0}{\sh}{q_1 \ldots q_m}~,
 \end{align*}
 where $\sh \in \SL{}{\SRDCLASS}$ with $\NOCALLS{\sh} = m$.
Decidability of $\Delta$ follows immediately from decidability of $\Delta_{\HA{A}}$ and $\Delta_{\HA{B}}$.
Assuming $\HA{C}$ satisfies the compositionality property, we first show that $L(\HA{C}) = L(\HA{A}) \cap L(\HA{B})$: 
\begin{align*}
    & \rsh \in L(\HA{A}) \cap L(\HA{B}) \\
  ~\Leftrightarrow~ & \left[ \text{Definition of}~\cap \right] \\
    & \rsh \in L(\HA{A}) ~\text{and}~ \rsh \in L(\HA{B}) \\
    ~\Leftrightarrow~ & \left[ \text{Definition of}~L(\HA{A}),~L(\HA{B}) \right] \\
                    & \exists p \in F_{\HA{A}} \,.\, \OMEGA{A}{p}{\rsh}
                      ~\text{and}~
                      \exists q \in F_{\HA{B}} \,.\, \OMEGA{B}{q}{\rsh} \\
  ~\Leftrightarrow~ & \left[ \text{construction of}~\HA{C} \right] \\
                    & \exists (p,q) \in F \,.\, \OMEGA{C}{(p,q)}{\rsh}
\end{align*}
 It remains to show the compositionality property, i.e., for each $p \in Q$, $\sh \in \SL{}{\SRDCLASS}$ with $\NOCALLS{\sh} = m$ and $\rsh_1,\ldots, \rsh_m \in \RSL{}{\SRDCLASS}$ with
 $\rsha = \sh[P_1/\rsh_1,\ldots,P_m/\rsh_m]$, we have
  \begin{align*}
    & \OMEGA{C}{p}{\rsha}
    ~\Leftrightarrow~
    \exists \T{q} \in Q^m \,.\, \MOVE{C}{p}{\sh}{\T{q}} ~\text{and}~ \forall 1 \leq i \leq m \,.\, \OMEGA{C}{\PROJ{\T{q}}{i}}{\rsh_i}.
  \end{align*}
 In the following proof, we write $p_1,p_2$ to denote the first and second component of $p \in Q$, respectively.
\begin{align*}
                   & \OMEGA{C}{p}{\rsha} \\
 ~\Leftrightarrow~ & \left[ \text{construction of}~\HA{C} \right] \\
                   & \OMEGA{A}{p_1}{\rsha} ~\text{and}~ \OMEGA{B}{p_2}{\rsha} \\
 ~\Leftrightarrow~ & \left[ \text{compositionality of}~\HA{A},\HA{B} \right] \\
                   & \exists \T{q_1} \in Q_{\HA{A}}^m \,.\, \MOVE{A}{p_1}{\sh}{\T{q_1}} ~\text{and} \\
                   & \forall 1 \leq i \leq m \,.\, \OMEGA{A}{\PROJ{\T{q_1}}{i}}{\rsh_i} \\
                   & ~\text{and}~ \\
                   & \exists \T{q_2} \in Q_{\HA{B}}^m \,.\, \MOVE{B}{p_2}{\sh}{\T{q_2}} ~\text{and} \\
                   & \forall 1 \leq i \leq m \,.\, \OMEGA{B}{\PROJ{\T{q_2}}{i}}{\rsh_i} \\
 ~\Leftrightarrow~ & \big[ \text{setting}~\T{q} = (\PROJ{\T{q_1}}{1},\PROJ{\T{q_2}}{1}) \ldots (\PROJ{\T{q_1}}{m},\PROJ{\T{q_2}}{m}), \\
                   & \quad \text{construction of}~\HA{C} \big] \\
                   & \exists \T{q} \in Q^{m} \,.\, \MOVE{C}{p}{\sh}{\T{q}} ~\text{and}~ \forall 1 \leq i \leq m \,.\, \OMEGA{C}{\PROJ{\T{q}}{i}}{\rsh_i}. \qedhere
\end{align*}
\end{proof}
\begin{lemma} \label{thm:refinement:complement}
  Let $\HA{A}$ be a heap automaton over $\SHCLASS$.
  Then there exists a heap automaton $\HA{C}$ over $\SHCLASS$ such that 
  $L(\HA{C}) = \RSL{}{\SRDCLASS} \setminus L(\HA{A})$.
\end{lemma}
\begin{proof}
%
  %
    We construct a heap automaton $\HA{C} = (Q,\SHCLASS,\Delta,F)$ accepting $\RSL{}{\SRDCLASS} \setminus L(\HA{A})$ as follows:
 \begin{align*}
   & Q ~\DEFEQ~ 2^{Q_{\HA{A}}} \qquad F ~\DEFEQ~ \{ X \subseteq Q_{\HA{A}} ~|~ X \cap F_{\HA{A}} = \emptyset \} \\
   & \Delta ~:~ \MOVE{C}{X_0}{\sh}{X_1 \ldots X_m} ~\Leftrightarrow~ \\
   & \quad \forall q_0 \in X_0 . \forall 1 \leq i \leq m . \exists q_i \in X_i \,.\, \MOVE{A}{q_0}{\sh}{q_1 \ldots q_m} \\
   & \quad \text{and}~ \\
   & \quad \forall q_0 \notin X_0 . \forall 1 \leq i \leq m . \forall q_i \in X_i \,.\, \text{not}~ \MOVE{A}{q_0}{\sh}{q_1 \ldots q_m}~,
 \end{align*}
 where $\sh \in \SL{}{\SRDCLASS}$ with $\NOCALLS{\sh} = m$.
 Decidability of $\Delta$ follows immediately from decidability of $\Delta_{\HA{A}}$.
 Assuming $\HA{C}$ satisfies the compositionality property, we first show that $L(\HA{C}) = \RSL{}{\SRDCLASS} \setminus L(\HA{A})$:
 \begin{align*}
     & \rsh \in \RSL{}{\SRDCLASS} \setminus L(\HA{A}) \\
   ~\Leftrightarrow~ & \left[ \text{Definition of}~\RSL{}{\SRDCLASS} \setminus L(\HA{A}) \right] \\
                     & \forall p \in F_{\HA{A}} \,.\, \text{not}~ \OMEGA{A}{p}{\rsh} \\
   ~\Leftrightarrow~ & \left[ \text{Definition of}~F \right] \\
   ~\Leftrightarrow~ & \exists X \in F \,.\, \forall p \in X \,.\, \OMEGA{A}{p}{\rsh} ~\text{and}~ \forall p \notin X \,.\, \text{and}~ \OMEGA{A}{p}{\rsh} \\
   ~\Leftrightarrow~ & \left[ \text{construction of}~\HA{C} \right] \\
   ~\Leftrightarrow~ & \exists X \in F \,.\, \OMEGA{C}{X}{\rsh}
 \end{align*}
 It remains to show the compositionality property, i.e., for each $X \in Q$, $\sh \in \SL{}{\SRDCLASS}$ with $\NOCALLS{\sh} = m$ and $\rsh_1,\ldots, \rsh_m \in \RSL{}{\SRDCLASS}$ with
 $\rsha = \sh[P_1/\rsh_1,\ldots,P_m/\rsh_m]$, we have
  \begin{align*}
    & \OMEGA{C}{p}{\rsha}
    ~\Leftrightarrow~
    \exists \T{Y} \in Q^m \,.\, \MOVE{C}{X}{\sh}{\T{Y}} ~\text{and}~ \forall 1 \leq i \leq m \,.\, \OMEGA{C}{\PROJ{\T{Y}}{i}}{\rsh_i}
  \end{align*}
  Assume $\OMEGA{C}{X}{\rsha}$.
  We choose $\T{Y}$ such that for each $1 \leq i \leq m$, we have
  \begin{align*}
   \PROJ{\T{Y}}{i} ~\DEFEQ~ \{ q \in Q ~|~ \OMEGA{A}{q}{\rsh_i} \}.
  \end{align*}
  Then $\OMEGA{C}{\PROJ{\T{Y}}{i}}{\rsh_i}$ and $\MOVE{C}{X}{\sh}{\T{Y}}$ hold immediately by construction and our choice of $\T{Y}$.
  For the converse direction assume there exists $\T{Y} \in Q^{m}$ such that $\MOVE{C}{X}{\sh}{\T{Y}}$ and for each $1 \leq i \leq m$, we have $\OMEGA{C}{\PROJ{\T{Y}}{i}}{\rsh_i}$.
  Two cases arise for each $p \in Q_{\HA{A}}$.
  First, assume $p \in X$.
  By construction, there exists a $q_i \in \PROJ{\T{Y}}{i}$ for each $1 \leq i \leq m$ such that
  $\MOVE{A}{p}{\sh}{q_1 \ldots q_m}$ and $\OMEGA{A}{q_i}{\rsh_i}$, $1 \leq i \leq m$.
  Since $\HA{A}$ is a heap automaton this implies $\OMEGA{A}{p}{\rsha}$.

  Second, assume $p \notin X$.
  Then for each choice of $q_i \in q_i \in \PROJ{\T{Y}}{i}$, $1 \leq i \leq m$, we know by construction of $\Delta$ that
  $\MOVE{A}{p}{\sh}{q_1 \ldots q_m}$ as well as $\OMEGA{A}{q_k}{\rsh_i}$ \emph{do} \emph{not} hold.
  Since $\HA{A}$ is a heap automaton this implies that $\OMEGA{A}{p}{\rsha}$ does not hold.
  Putting both cases together yields $\OMEGA{C}{X}{\rsha}$.
  \qed
\end{proof}
%
%
%
\section{Proof of Theorem~\ref{thm:compositional:refinement}} \label{app:compositional:refinement}
The following lemma is essential to show that our construction of $\SRDALT$ is correct.
\begin{lemma} \label{thm:compositional:psi-correctness}
 Let $\SRD$, $\SRDALT$ and $\HA{A}$ be as in Theorem~\ref{thm:compositional:refinement}.
 Then, for each $\PS \in \PRED(\SRD)$ and $q \in Q_{\HA{A}}$, we have
 \begin{align*}
   & \rsh \in \CALLSEM{(\PS,q)\FV{0}{}}{\SRDALT} \quad\text{iff}\quad \rsh \in \CALLSEM{\PS\FV{0}{}}{\SRD} ~\text{and}~ \OMEGA{A}{q}{\rsh}. 
 \end{align*}
\end{lemma}
\begin{proof}
  Let $\PS \in \PRED(\SRD)$ and $q \in Q_{\HA{A}}$.
  By induction on the height $k$ of unfolding trees $t$, we show
  \begin{align*}
    t \in \UTREES{\SRDALT}{(\PS,q)\FV{0}{}} ~\Leftrightarrow~ & \exists t' \in \UTREES{\SRD}{\PS\FV{0}{}} ~.~ \UNFOLD{t'} = \UNFOLD{t}\\
                                                       & \qquad \text{and}~ \OMEGA{A}{q}{\UNFOLD{t}}.
  \end{align*}
  \paragraph{I.B.}
  Note that, by Definition of unfolding trees, $t(\EMPTYSEQ) = (\PS,q)\FV{0}{}$.
  If $k = 0$ there is nothing to show.
  Thus assume $k = 1$.
  Then
  \begin{align*}
         & \UNFOLD{t} \\
     ~=~ & \left[ \text{Definition of}~\UNFOLD{t} \right] \\
         & (\PS,q)\FV{0}{}\left[(\PS,q) / \UNFOLD{\SUBTREE{t}{1}}\right] \\
     ~=~ & \left[ \text{Definition of predicate replacement} \right] \\
     ~=~ & \UNFOLD{\SUBTREE{t}{1}} \\
     ~=~ & \left[ k = 1 ~\text{implies}~ \UNFOLD{\SUBTREE{t}{1}} = t(1) \in \RSL{}{}  \right] \\
     ~=~ & t(1).
  \end{align*}
  Hence
  \begin{align*}
                       & t \in \UTREES{\SRDALT}{(\PS,q)\FV{0}{}} \\
     ~\Leftrightarrow~ & \left[ \text{Definition}~\ref{def:sl:unfolding-trees},~ k = 1 \right] \\
                       & \SRDRULE{(\PS,q)}{t(1)} \in \SRDALT \\
     ~\Leftrightarrow~ & \left[ \text{Construction of}~\SRDALT \right] \\
                       & \SRDRULE{\PS}{t(1)} \in \SRD ~\text{and}~ \OMEGA{A}{t(1)}{q} \\
     ~\Leftrightarrow~ & \left[ \text{Definition}~\ref{def:sl:unfolding-trees} \right] \\
                       & \exists t' \in \UTREES{\SRD}{\PS\FV{0}{}} \,.\, \UNFOLD{t} = t(1) ~\text{and}~ \OMEGA{A}{t(1)}{q} \\
     ~\Leftrightarrow~ & \left[ \UNFOLD{t} = t(1) \right] \\
                       & \exists t' \in \UTREES{\SRD}{\PS\FV{0}{}} \,.\, \UNFOLD{t'} = \UNFOLD{t} ~\text{and}~ \OMEGA{A}{\UNFOLD{t}}{q}.
  \end{align*}
  \paragraph{I.H.}
  Assume for an arbitrary, but fixed natural number $k$ that for each predicate symbol $\PS \in \PRED(\SRD)$, each state $q \in \HA{A}$ and each unfolding tree $t \in \UTREES{\SRDALT}{(\PS,q)\FV{0}{}}$ of height at most $k$, we have
  \begin{align*}
    t \in \UTREES{\SRDALT}{(\PS,q)\FV{0}{}} ~\Leftrightarrow~ & \exists t' \in \UTREES{\SRD}{\PS\FV{0}{}} ~.~ \UNFOLD{t'} = \UNFOLD{t}\\
                                                       & \qquad \text{and}~ \OMEGA{A}{q}{\UNFOLD{t}}.
  \end{align*}
  \paragraph{I.S.}
  Let $t \in \UTREES{\SRDALT}{(\PS,q)\FV{0}{}}$ be an unfolding tree of height $k+1$.
  Then $t(1)$ is not a reduced symbolic heap (otherwise the height of $t$ would be $1$).
  Thus, given some $m \geq 1$, we may assume that
  \begin{align*}
   t(1) = \exists \BV{} . \SPATIAL{} \SEP (\PS_1,q_1)\FV{1}{} \SEP \ldots \SEP (\PS_m,q_m) \FV{m}{} : \PURE{}
  \end{align*}
  Then
  \begin{align*}
         & \UNFOLD{t} \\
     ~=~ & \left[ \text{Definition of}~\UNFOLD{t} \right] \\
         & t(\EMPTYSEQ)[(\PS,q) / \UNFOLD{\SUBTREE{t}{1}}] \\
     ~=~ & \left[ t(\EMPTYSEQ) = (\PS,q)\FV{0}{},~\text{definition of predicate replacement} \right] \\
         & \UNFOLD{\SUBTREE{t}{1}} \\
     ~=~ & \left[ \text{Definition of}~\UNFOLD{t} \right] \\
         & t(1)[(\PS_1,q_1) / \UNFOLD{\SUBTREE{t}{1\cdot 1}}, \ldots (\PS_m,q_m) / \UNFOLD{\SUBTREE{t}{1\cdot m}}] \tag{$\spadesuit$}.
  \end{align*}
  Now, for each $1 \leq i \leq m$, $\SUBTREE{t}{1\cdot i}$ is an unfolding tree of $(\PS_i,q_i)\FV{0}{}$ of height at most $k$.
  Thus, by I.H. we know that
  \begin{align*}
    \SUBTREE{t}{1\cdot i} \in \UTREES{\SRDALT}{(\PS_i,q_i)\FV{0}{}} ~\Leftrightarrow~ & \exists t_i \in \UTREES{\SRD}{\PS_i\FV{0}{}} ~.~ \UNFOLD{t_i} = \UNFOLD{\SUBTREE{t}{1\cdot i}}\\
                                                       & \qquad \text{and}~ \OMEGA{A}{q_i}{\UNFOLD{\SUBTREE{t}{1\cdot i}}}.
  \end{align*}
  Then
  \begin{align*}
                       & t \in \UTREES{\SRDALT}{(\PS,q)\FV{0}{}} \\
     ~\Leftrightarrow~ & \left[ \text{Definition}~\ref{def:sl:unfolding-trees},~ k > 1 \right] \\
                       & \SRDRULE{(\PS,q)}{t(1)} \in \SRDALT \\
                       & \qquad \text{and}~ \forall 1 \leq i \leq m ~.~ \SUBTREE{t}{i} \in \UTREES{\SRDALT}{(\PS_i,q_i)\FV{0}{}} \\
     ~\Leftrightarrow~ & \left[ \text{Construction of}~\SRDALT \right] \\
                       & \SRDRULE{\PS}{\left(\exists \BV{} . \SPATIAL{} \SEP \CALLN{1}{} \SEP \ldots \SEP \CALLN{m}{} : \PURE{}\right)} \in \SRD \\
                       & \qquad \text{and}~ \MOVE{A}{q}{t(1)}{q_1 \ldots q_m} \\
                       & \qquad \text{and}~ \forall 1 \leq i \leq m ~.~ \SUBTREE{t}{i} \in \UTREES{\SRDALT}{(\PS_i,q_i)\FV{0}{}} \\
     ~\Leftrightarrow~ & \left[ \text{I.H.} \right] \\
                       & \SRDRULE{\PS}{\left(\exists \BV{} . \SPATIAL{} \SEP \CALLN{1}{} \SEP \ldots \SEP \CALLN{m}{} : \PURE{}\right)} \in \SRD \\
                       & \qquad \text{and}~ \MOVE{A}{q}{t(1)}{q_1 \ldots q_m} \\
                       & \qquad \text{and}~ \forall 1 \leq i \leq m ~.~ \exists t_i \in \UTREES{\SRD}{\PS_i\FV{0}{}} ~.~ \UNFOLD{t_i} = \UNFOLD{\SUBTREE{t}{1\cdot i}}\\
                       & \qquad \qquad \text{and}~ \OMEGA{A}{q_i}{\UNFOLD{\SUBTREE{t}{1\cdot i}}}. \\
     ~\Leftrightarrow~ & \left[ \text{Definition}~\ref{def:sl:unfolding-trees},~t'(1) \DEFEQ \exists \BV{} . \SPATIAL{} \SEP \CALLN{1}{} \SEP \ldots \SEP \CALLN{m}{} : \PURE{} \right] \\
                       & \exists t' \in \UTREES{\SRD}{\PS\FV{0}{}} ~.~ \UNFOLD{t'} = t'(1)[\PS_1 / \UNFOLD{\SUBTREE{t}{1\cdot 1}}, \ldots, \PS_m / \UNFOLD{\SUBTREE{t}{1\cdot m}}] \\
                       & \qquad \text{and}~ \MOVE{A}{q}{t(1)}{q_1 \ldots q_m} \\
                       & \qquad \text{and}~ \forall 1 \leq i \leq m ~.~ \OMEGA{A}{q_i}{\UNFOLD{\SUBTREE{t}{1\cdot i}}} \\
     ~\Leftrightarrow~ & \left[ \text{Definition}~\ref{def:refinement:automaton} \right] \\
                       & \exists t' \in \UTREES{\SRD}{\PS\FV{0}{}} ~.~ \UNFOLD{t'} = t'(1)[\PS_1 / \UNFOLD{\SUBTREE{t}{1\cdot 1}}, \ldots, \PS_m / \UNFOLD{\SUBTREE{t}{1\cdot m}}] \\
                       & \qquad \text{and}~ \OMEGA{A}{q}{\UNFOLD{\SUBTREE{t}{1}}} \\
     ~\Leftrightarrow~ & \left[ \text{applying}~(\spadesuit) \right] \\
                       & \exists t' \in \UTREES{\SRD}{\PS\FV{0}{}} ~.~ \UNFOLD{t'} = \UNFOLD{t} ~\text{and}~ \OMEGA{A}{q}{\UNFOLD{t}}.
  \end{align*}
  To finish the proof, note that
  \begin{align*}
                      & \rsh \in \CALLSEM{(\PS,q)\FV{0}{}}{\SRDALT} \\
    ~\Leftrightarrow~ & \left[ \text{Definition}~\ref{def:symbolic-heaps:unfoldings} \right] \\
                      & \exists t \in \UTREES{\SRDALT}{(\PS,q)\FV{0}{}} ~.~ \UNFOLD{t} = \rsh \\
    ~\Leftrightarrow~ & \left[ \text{previously shown proposition} \right] \\
                      & \exists t' \in \UTREES{\SRD}{\PS\FV{0}{}} ~.~ \UNFOLD{t'} = \rsh ~\text{and}~ \OMEGA{A}{q}{\rsh} \\
    ~\Leftrightarrow~ & \left[ \text{Definition}~\ref{def:symbolic-heaps:unfoldings} \right] \\
                      & \rsh \in \CALLSEM{\PS\FV{0}{}}{\SRD} ~\text{and}~ \OMEGA{A}{q}{\rsh}. \qedhere
  \end{align*}
\end{proof}
We are now in a position to prove Theorem~\ref{thm:compositional:refinement}.
\begin{proof}[Proof of Theorem~\ref{thm:compositional:refinement}]
 It remains to show that for each $\PS \in \PRED(\SRD)$, we have $\CALLSEM{\PS\FV{0}{}}{\SRDALT} = \CALLSEM{P\FV{0}{}}{\SRD} \cap L(\HA{A})$.
  By Lemma~\ref{thm:compositional:psi-correctness}, we know that for each $\PS \in \PRED(\SRD)$ and $q \in Q_{\HA{A}}$, we have
 \begin{align*}
   \rsh \in \CALLSEM{(\PS,q)\FV{0}{}}{\SRDALT} \quad\text{iff}\quad \rsh \in \CALLSEM{\PS\FV{0}{}}{\SRD} ~\text{and}~ \OMEGA{A}{q}{\rsh}. \tag{$\bigstar$}
 \end{align*}
 Then
 \begin{align*}
                     & \rsh \in \CALLSEM{\PS\FV{0}{}}{\SRDALT} \\
   ~\Leftrightarrow~ & \left[ \text{Definition}~\ref{def:symbolic-heaps:unfoldings} \right] \\
                     & \exists t \in \UTREES{\SRDALT}{\PS\FV{0}{}} ~.~ \UNFOLD{t} = \rsh \\
   ~\Leftrightarrow~ & \left[ \UNFOLD{t} = \PS\FV{0}{}[\PS\FV{0}{} / \UNFOLD{\SUBTREE{t}{1}}] = \UNFOLD{\SUBTREE{t}{1}} \right] \\
                     & \exists t \in \UTREES{\SRDALT}{\PS\FV{0}{}} ~.~ \UNFOLD{\SUBTREE{t}{1}} = \rsh \\
   ~\Leftrightarrow~ & \big[ \PS \SRDARROW (\PS,q) \in \SRDALT \Leftrightarrow q \in F_{\HA{A}} ~\text{and}~ t(\EMPTYSEQ) = \PS\FV{0}{} \\
                     & \quad\text{implies}~ \exists q \in F_{\HA{A}} . t(1) = (\PS,q)\FV{0}{} \big] \\
                     & \exists q \in F_{\HA{A}} . \exists t' \in \UTREES{\SRDALT}{(\PS,q)\FV{0}{}} ~.~ \UNFOLD{t'} = \rsh  \\
   ~\Leftrightarrow~ & \left[ \text{Definition}~\ref{def:symbolic-heaps:unfoldings} \right] \\
                     & \exists q \in F_{\HA{A}} ~.~ \rsh \in \CALLSEM{(\PS,q)\FV{0}{}}{\SRDALT} \\
   ~\Leftrightarrow~ & \left[ \text{applying}~(\bigstar) \right] \\
                     & \exists q \in F_{\HA{A}} ~.~ \rsh \in \CALLSEM{P\FV{0}{}}{\SRD} ~\text{and}~ \OMEGA{A}{q}{\rsh}. \\
   ~\Leftrightarrow~ & \left[ \rsh \in L(\HA{A}) \Leftrightarrow \exists q \in F_{\HA{A}} ~.~ \OMEGA{A}{q}{\rsh} \right] \\
                     & \rsh \in \CALLSEM{P\FV{0}{}}{\SRD} ~\text{and}~ \rsh \in L(\HA{A}) \\
   ~\Leftrightarrow~ & \left[ x \in A \wedge x \in B \Leftrightarrow x \in A \cap B \right] \\
                     & \rsh \in \CALLSEM{P\FV{0}{}}{\SRD} ~\cap~ \rsh \in L(\HA{A}). 
 \end{align*}
 \qed

\end{proof}

%
%
\section{Details on definite relationships} \label{app:zoo:definite}
This section briefly shows how to compute the definite relationships $\MEQ{\rsh}, \MNEQ{\rsh}$,
$\ALLOC{\rsh}$ and $x \MPT{\rsh} y$ introduced at the beginning of Section~\ref{sec:zoo}.
\begin{definition} \label{def:zoo:completion}
  Let $\rsh = \exists \BV{} \,.\, \SPATIAL{}{} : \PURE{} \in \RSL{}{}$.
  Then the \emph{completion} of $\rsh$ is given by the symbolic heap
  $
    \textrm{complete}(\rsh) ~\DEFEQ~ \exists \BV{} \,.\, \SPATIAL{}{} : \textrm{closure}(\Lambda)
  $, 
  where $\textrm{closure}(\Lambda)$ denotes the reflexive, symmetric, and transitive closure w.r.t. $=$ and the
  symmetric closure w.r.t. $\neq$ of the set of pure formulas $\Lambda$ given by
  \begin{align*}
           \PURE{}{}
    ~\cup~ & \{ x \neq \NIL ~|~ \PT{x}{\_} ~\text{occurs in}~ \SPATIAL{}{} \} \\
    ~\cup~ & \{ x \neq y ~|~ \PT{x}{\_},\PT{y}{\_} ~\text{distinct points-to assertions in}~ \SPATIAL{}{} \}.
  \end{align*}
  Moreover, if $\textrm{closure}(\Lambda)$ is inconsistent, i.e.,
  $x \neq x$ or $\NIL \neq \NIL$ is contained in $\textrm{closure}(\Lambda)$,
  we define $\textrm{closure}(\Lambda)$ to be the set of all pure formulas over $\VAR(\rsh)$.
\end{definition}
Clearly, $\Lambda$ is computable in at most quadratic time.
Then computing $\textrm{closure}(\Lambda)$ is possible in at most cubic time.
Finally, checking for inconsistencies can be performed in linear time.
Hence, the completion of a reduced symbolic heap is computable in polynomial time.
The following lemma provides a characterization of each of the definite relationships
between variables introduced in Section~\ref{sec:zoo}.
\begin{lemma} \label{thm:zoo:relationships:characterizations}
  Let $\rsh \in \RSL{}{}$ and $x,y \in \VAR(\rsh)$.
  Moreover, let $\rsha = \textrm{complete}(\rsh)$.
  Then
  \begin{itemize}
    \item $x \MEQ{\rsh} y$ iff $x = y \in \PURE{\rsha}$,
    \item $x \MNEQ{\rsh} y$ iff $x \neq y \in \PURE{\rsha}$,
    \item $x \in \ALLOC{\rsh}$ iff
          $x \neq x \in \PURE{\rsha}$
          or there exists $\PT{z}{\_}$ in $\SPATIAL{\rsha}$ such that $x = z \in \PURE{\rsha}$,
    \item $x \MPT{\rsh} y$ iff
          $x \neq x \in \PURE{\rsha}$
          or there exists $\PT{z_1}{(\_,z_2,\_)}$  in $\SPATIAL{\rsha}$
          such that $x = z_1 \in \PURE{\rsha}$ and $y = z_2 \in \PURE{\rsha}$. \qedhere
  \end{itemize}
\end{lemma}
In particular, the right-hand side of each of these characterizations is easily computable in
polynomial time once a completion of a symbolic heap has been computed.
\begin{proof}
Let $\rsh = \exists \BV{} ~.~ \SPATIAL{} : \PURE{}$.
Then, for each $(\stack,\heap) \in \MODELS{\rsh}$, we know by the SL semantics (cf. Figure~\ref{fig:slsemantics})
that there exists $\T{v} \in \VAL^{\SIZE{\BV{}}}$ such that $\stack' = \stack[\BV{} \mapsto \T{v}]$,
$\DOM(\heap) = \{ \stack'(x) ~|~ \PT{x}{\T{y}} ~\text{in}~ \SPATIAL{} \}$,
$\heap(\stack'(x)) = \stack'(\T{y})$ for each $\PT{x}{\T{y}}$ in $\SPATIAL{}$, and for each $x \sim y \in \PURE{}$,
$\stack'(x) \sim \stack'(y)$.
The same holds for $\rsha$ except that $\PURE{}$ is substituted by $\textrm{closure}(\Lambda)$
 (although each additional pure formula is already implied by the conditions from above).
Thus, $\MODELS{\rsh} = \MODELS{\rsha}$.
Furthermore, by construction of $\PURE{\rsha} = \textrm{closure}(\Lambda)$, we have $x \MEQ{\rsha} y$ iff $x = y \in \PURE{\rsha}$
and $x \MNEQ{\rsha}$ iff $x \neq y \in \PURE{\rsha}$.
Then, each of the properties in Lemma~\ref{thm:zoo:relationships:characterizations} is easy to verify.
\qed
\end{proof}
\section{Appendix to Section~\ref{sec:zoo}}
%
\subsection{Proof of Lemma~\ref{thm:zoo:track:property}}
\label{app:zoo:track:property}
Note that some statements are broken down into separate lemmas which are
shown immediately afterwards.
\begingroup
\def\thelemma{\ref{thm:zoo:track:property}}
\begin{lemma}
 For all $\alpha \in \POSN$ and all sets
 $A \subseteq \FV{0}{}$, 
 $\PURE{} \in \textnormal{Pure}(\FV{0}{})$,
there is a heap automaton over $\SHCLASSFV{\alpha}$ accepting
 $\TRACK(\alpha,A,\PURE{})$.
\end{lemma}
\addtocounter{theorem}{-1}
\endgroup
\begin{proof}
  It suffices to show that the heap automaton $\HATRACK$ constructed in
  Definition~\ref{def:zoo:track-automaton} satisfies the compositionality
  property and accepts $\TRACK(\alpha,A,\PURE{})$.
  We first show the latter, i.e., $L(\HATRACK) = \TRACK(\alpha,A,\PURE{})$.
  Let $\rsh \in \RSL{}{\SRDCLASSFV{\alpha}}$. Then
  \begin{align*}
                       & \rsh \in L(\HATRACK) \\
     ~\Leftrightarrow~ & \left[ \text{Definition}~L(\HATRACK) \right] \\
                       & \exists (A',\Pi') \in F_{\HATRACK} ~.~ \OMEGA{\HATRACK}{(A',\Pi')}{\rsh} \\
     ~\Leftrightarrow~ & \left[ F_{\HATRACK} = \{(A,\Pi)\} \right] \\
                       & \OMEGA{\HATRACK}{(A,\Pi)}{\rsh} \\
     ~\Leftrightarrow~ & \left[ \text{Definition}~\Delta \right] \\
                       & \forall x,y \in \FV{0}{\rsh} ~.~
                           x \in A \leftrightarrow x \in \ALLOC{\SHRINK{\rsh,\EMPTYSEQ}} \\
                       & \qquad \text{and}~ (x \sim y) \in \Pi \leftrightarrow x \sim_{\SHRINK{\rsh,\EMPTYSEQ}} y \\
     ~\Leftrightarrow~ & \left[ \text{Definition}~\ref{def:zoo:track-automaton}:~\SHRINK{\rsh,\EMPTYSEQ} = \rsh  \right] \\
                       & \forall x,y \in \FV{0}{\rsh} ~.~
                           x \in A \leftrightarrow x \in \ALLOC{\rsh} \\
                       & \qquad \text{and}~ (x \sim y) \in \Pi \leftrightarrow x \sim_{\rsh} y \\
     ~\Leftrightarrow~ & \left[ \text{Definition}~\TRACK(\alpha,A,\Pi) \right] \\
                       & \rsh \in \TRACK(\alpha,A,\Pi).
  \end{align*}
  The compositionality property of $\HATRACK$ is verified separately in
  Lemma~\ref{thm:zoo:track:compositional},
  which is proven in the remainder of this section.
  \qed
\end{proof}
\begin{lemma} \label{thm:zoo:track:compositional}
 Let $\sh \in \SL{}{\SRDCLASSFV{\alpha}}$ with $\NOCALLS{\sh} = m \geq 0$.
 Moreover, for each $1 \leq i \leq m$, let $\rsh_i \in \RSL{}{\SRDCLASSFV{\alpha}}$ with $\SIZE{\FV{0}{\rsh_i}} = \SIZE{\FV{i}{\sh}}$.
 Then, for $\rsh = \sh[\PS_1 / \rsh_1,\ldots \PS_m / \rsh_m]$, we have $\OMEGA{\HATRACK}{(A_0,\Pi_0)}{\rsh}$ if and only if there exist
 $(A_1,\Pi_1), \ldots, (A_m,\Pi_m) \in Q_{\HATRACK}$ such that
  \begin{align*}
  \MOVE{\HATRACK}{(A_0,\Pi_0)}{\sh}{(A_1,\Pi_1) \ldots (A_m,\Pi_m)}
  \end{align*}
 and, for each $1 \leq i \leq m$, we have $\OMEGA{\HATRACK}{(A_i,\Pi_i)}{\rsh_i}$.
\end{lemma}
\begin{proof}
  Recall from Definition~\ref{def:zoo:track-automaton} that
  \begin{align*}
     \SHRINK{\sh,\T{q}} ~\DEFEQ~ &
        \sh\left[\PS_1 / \SSIGMA{\CALLN{1}{}}{\T{q}[1]}, \ldots, \PS_m / \SSIGMA{\CALLN{m}{}}{\T{q}[m]}\right].
  \end{align*}
  Moreover, for each $1 \leq i \leq m$, let
  \begin{align*}
    A_i ~\DEFEQ~ & \{ \PROJ{\FV{0}{}}{\ell} ~|~ \PROJ{\FV{0}{\rsh_i}}{\ell} \in \ALLOC{\rsh_i} \},~\text{and} \\
    \Pi_i ~\DEFEQ~ & \{ \PROJ{\FV{0}{}}{\ell} \sim \PROJ{\FV{0}{}}{k} ~|~ \PROJ{\FV{0}{\rsh_i}}{\ell} \MSIM{\rsh_i} \PROJ{\FV{0}{\rsh_i}}{k} \}.
  \end{align*}
  Then
  \begin{align*}
                      & \OMEGA{\HATRACK}{(A_0,\Pi_0)}{\rsh} \\
    ~\Leftrightarrow~ & \left[ \text{Definition of}~\Delta~\text{and}~ \SIZE{\CALLS{\rsh}} = 0  \right] \\
                      & \forall x,y \in \FV{0} ~.~ x \in A_0 ~\text{iff}~ x^{\sh} \in \ALLOC{\rsh} \\
                      & \qquad \text{and}~ x \sim y \in \Pi_0 ~\text{iff}~ x^{\sh} \MSIM{\rsh} y^{\sh} \\
    ~\Leftrightarrow~ & \left[ \text{choice of}~A_i~\text{and}~\Pi_i~\text{for each}~1 \leq i \leq m \right] \\
                      & \exists \T{q} = (A_1,\Pi_1) \ldots (A_m,\Pi_m) \in Q_{\HATRACK}^{m} ~.~ \forall x,y \in \FV{0} ~.~ \\
                      & \qquad x \in A_0 ~\text{iff}~ x^{\sh} \in \ALLOC{\rsh} \\
                      & \qquad \text{and}~ x \sim y \in \Pi_0 ~\text{iff}~ x^{\sh} \MSIM{\rsh} y^{\sh} \\
                      & \qquad \text{and}~ \forall 1 \leq i \leq m ~.~ \OMEGA{\HATRACK}{(A_i,\Pi_i)}{\rsh_i} \\
    ~\Leftrightarrow~ & \big[ \text{Lemma}~\ref{thm:zoo:track:congruence} \big] \\
                      & \exists \T{q} = (A_1,\Pi_1) \ldots (A_m,\Pi_m) \in Q_{\HATRACK}^{m} ~.~ \forall x,y \in \FV{0} ~.~ \\
                      & \qquad x \in A_0 ~\text{iff}~ x^{\sh} \in \ALLOC{\SHRINK{\sh,\T{q}}} \\
                      & \qquad \text{and}~  x \sim y \in \Pi_0 ~\text{iff}~ x^{\sh} \MSIM{\SHRINK{\sh,\T{q}}} y^{\sh} \\
                      & \qquad \text{and}~ \forall 1 \leq i \leq m ~.~ \OMEGA{\HATRACK}{(A_i,\Pi_i)}{\rsh_i} \\
    ~\Leftrightarrow~ & \left[ \text{Definition of}~\Delta  \right] \\
                      & \exists \T{q} = (A_1,\Pi_1) \ldots (A_m,\Pi_m) \in Q_{\HATRACK}^{m} ~.~ \\
                      & \qquad \MOVE{A}{(A_0,\Pi_0)}{\sh}{(A_1,\Pi_1) \ldots (A_m,\Pi_m)} \\
                      & \qquad \text{and}~ \forall 1 \leq i \leq m ~.~ \OMEGA{\HATRACK}{(A_i,\Pi_i)}{\rsh_i}. 
  \end{align*}
  \qed
\end{proof}
It remains to show that Lemma~\ref{thm:zoo:track:congruence}.
The crux of the proof of Lemma~\ref{thm:zoo:track:congruence} is the following observation.
\begin{lemma} \label{thm:zoo:track:congruence:auxiliary}
 Let $\sh = \exists \BV{} . \SPATIAL{} \SEP \CALLN{1}{} : \PURE{} \in \SL{}{\SRDCLASSFV{\alpha}}$.
 Moreover, let $\rsh \in \RSL{}{}$ with $\SIZE{\FV{0}{\rsh}} = \SIZE{\FV{1}{\sh}}$, $A = \{ y \in \FV{0}{} ~|~ y^{\rsh} \in \ALLOC{\rsh} \}$ and $\Lambda = \{ x \sim y ~|~ x^{\rsh} \MSIM{\rsh} y^{\rsh} \}$.
 Then, for each $x,y \in \VAR(\sh)$, we have
 \begin{align*}
    x \MSIM{\sh[\PS_1 / \rsh]} y \quad\text{iff}&\quad x \MSIM{\sh\left[\PS_1 / \SSIGMA{\CALLN{1}{}}{(A,\Lambda)}\right]} y, \quad \text{and}~  \\
    x \in \ALLOC{\sh[\PS_1 / \rsh]} \quad\text{iff}&\quad x \in \ALLOC{\sh\left[\PS_1 / \SSIGMA{\CALLN{1}{}}{(A,\Lambda)}\right]}~,
 \end{align*}
 where $\SSIGMA{\CALLN{1}{}}{(A,\Lambda)}$ is defined in Section~\ref{sec:zoo} (above Definition~\ref{def:zoo:track-automaton}).
\end{lemma}
Proving this lemma in turn relies on some auxiliary results, in particular, we need 
Lemma~\ref{thm:zoo:track-congruence-right} and Lemma~\ref{thm:zoo:track-congruence-left}, which are proven at the end of this section.
\begin{proof}
 \begin{align*}
                     & x \MSIM{\sh[\PS_1 / \rsh]} y \\
   ~\Leftrightarrow~ & \left[ \text{Definition of}~\MSIM{\sh[\PS_1 / \rsh]} \right] \\
                     & \forall (\stack,\heap) ~.~ \stack,\heap \SAT{\emptyset} \sh[\PS_1 / \rsh] ~\text{implies}~ \stack(x) \sim \stack(y) \\
   ~\Leftrightarrow~ & \left[ \text{logic} \right] \\
                     & \forall (\stack,\heap) ~.~ (\exists \heap' ~.~ \stack,\heap' \SAT{\emptyset} \sh[\PS_1 / \rsh]) ~\text{implies}~ \stack(x) \sim \stack(y) \\
   ~\Rightarrow~ & \left[ \text{Lemma}~\ref{thm:zoo:track:heap-existence} \right] \\
                     & \forall (\stack,\heap) ~.~ (\exists \heap' ~.~ \stack,\heap' \SAT{\emptyset} \sh[\PS_1 / \rsh]) ~\text{implies}~ \stack(x) \sim \stack(y) \\
                     & ~\text{and}~ \stack,\heap \SAT{\emptyset} \sh\left[\PS_1 / \SSIGMA{\CALLN{1}{}}{(A,\Lambda)}\right] \\
                     & \qquad ~\text{implies}~ (\exists \heap' ~.~ \stack,\heap' \SAT{\emptyset} \sh[\PS_1 / \rsh]) \\
       ~\Rightarrow~ & \left[ B \rightarrow C ~\wedge~ A \rightarrow B ~\text{implies}~ A \rightarrow C \right] \\
                     & \forall (\stack,\heap) ~.~ \stack,\heap \SAT{\emptyset} \sh\left[\PS_1 / \SSIGMA{\CALLN{1}{}}{(A,\Lambda)}\right] ~\text{implies}~ \stack(x) \sim \stack(y) \\
       ~\Leftrightarrow~ & \left[ \text{Definition of}~\MSIM{\sh\left[\PS_1 / \SSIGMA{\CALLN{1}{}}{(A,\Lambda)}\right]} \right] \\
                     & x \MSIM{\sh\left[\PS_1 / \SSIGMA{\CALLN{1}{}}{(A,\Lambda)}\right]} y.
 \end{align*}
 Furthermore, we obtain
 \begin{align*}
                     & x \in \ALLOC{\sh[\PS_1 / \rsh]} \\
   ~\Leftrightarrow~ & \left[ \text{Definition of}~\ALLOC{\sh[\PS_1 / \rsh]} \right] \\
                     & \forall (\stack,\heap) ~.~ \stack,\heap \SAT{\emptyset} \sh[\PS_1 / \rsh] ~\text{implies}~ \stack(x) \in \DOM(\heap) \\
   ~\Leftrightarrow~ & \left[ \text{logic} \right] \\
                     & \forall (\stack,\heap) ~.~ (\exists \heap' ~.~ \stack,\heap' \SAT{\emptyset} \sh[\PS_1 / \rsh] \\
                     & \qquad \text{and}~ \forall y \in \DOM(\stack) ~.~ \stack(y) \in \DOM(\heap) ~\text{iff}~ \stack(y) \in \DOM(\heap')) \\
                     & \quad \text{implies}~ \stack(x) \in \DOM(\heap) \\
   ~\Rightarrow~ & \left[ \text{Lemma}~\ref{thm:zoo:track:heap-existence} \right] \\
                     & \forall (\stack,\heap) ~.~ (\exists \heap' ~.~ \stack,\heap' \SAT{\emptyset} \sh[\PS_1 / \rsh] \\
                     & \qquad \text{and}~ \forall y \in \DOM(\stack) ~.~ \stack(y) \in \DOM(\heap) ~\text{iff}~ \stack(y) \in \DOM(\heap')) \\
                     & \qquad \quad \text{implies}~ \stack(x) \in \DOM(\heap) \\
                     & ~\text{and}~ \stack,\heap \SAT{\emptyset} \sh\left[\PS_1 / \SSIGMA{\CALLN{1}{}}{(A,\Lambda)}\right] \\
                     & \quad \text{implies}~ (\exists \heap' ~.~ \stack,\heap' \SAT{\emptyset} \sh[\PS_1 / \rsh] \\
                     & \qquad \text{and}~ \forall y \in \DOM(\stack) ~.~ \stack(y) \in \DOM(\heap) ~\text{iff}~ \stack(y) \in \DOM(\heap')) \\
       ~\Rightarrow~ & \left[ Y \rightarrow Z ~\wedge~ X \rightarrow Y ~\text{implies}~ X \rightarrow Z \right] \\
                     & \forall (\stack,\heap) ~.~ \stack,\heap \SAT{\emptyset} \sh\left[\PS_1 / \SSIGMA{\CALLN{1}{}}{(A,\Lambda)}\right] \\
                     & \qquad \text{implies}~ \stack(x) \in \DOM(\heap) \\
       ~\Leftrightarrow~ & \left[ \text{Definition of}~\ALLOC{\sh\left[\PS_1 / \SSIGMA{\CALLN{1}{}}{(A,\Lambda)}\right]} \right] \\
                     & x \in \ALLOC{\sh\left[\PS_1 / \SSIGMA{\CALLN{1}{}}{(A,\Lambda)}\right]}.
 \end{align*}
 The proof of the converse direction is analogous, where the main difference is consists of using the property
 \begin{align*}
         & \stack,\heap ~\SAT{\emptyset}~ \sh[\PS_1 / \rsh] ~\text{implies}~\exists \heap' ~.~ \stack,\heap' \SAT{\emptyset}
           \sh\left[\PS_1 / \SSIGMA{\CALLN{1}{}}{(A,\Lambda)}\right] \\
         & \qquad ~\text{and}~ \forall y \in \DOM(\stack) ~.~
           \stack(y) \in \DOM(h) ~\text{iff}~ \stack(y) \in \DOM(\heap')
 \end{align*}
 instead of Lemma~\ref{thm:zoo:track:heap-existence}.
 This property is shown analogous to Lemma~\ref{thm:zoo:track:heap-existence} except that we apply Lemma~\ref{thm:zoo:track-congruence-left}
 instead of Lemma~\ref{thm:zoo:track-congruence-right}.
 \qed
\end{proof}
\begin{lemma} \label{thm:zoo:track:congruence}
 Let $\sh \in \SHCLASSFV{\alpha}$ with $\NOCALLS{\sh} = m \geq 0$.
 For each $1 \leq i \leq m$, let $\rsh_i \in \RSHCLASSFV{\alpha}$ with $\NOFV{\rsh_i} = \SIZE{\FV{i}{\sh}}$, $A_i = \{ y \in \FV{0} ~|~ y^{\rsh_i} \in \ALLOC{\rsh_i} \}$ and $\PURE{}_i = \{ x \sim y ~|~ x^{\rsh_i} \MSIM{\rsh_i} y^{\rsh_i} \}$.
 Moreover, let
 \begin{align*}
   \T{q} ~\DEFEQ~ & (A_1,\PURE{}_1) \ldots (A_m,\PURE{}_m), \\
   \rsh ~\DEFEQ~ & \sh\left[\PS_1 / \rsh_1, \ldots, \PS_m / \rsh_m\right],~\text{and} \\
   \SHRINK{\sh,\T{q}} ~\DEFEQ~ & \sh\left[\PS_1 / \SSIGMA{\CALLN{1}{}}{\T{q}[1]}, \ldots, \PS_m / \SSIGMA{\CALLN{m}{}}{\T{q}[m]}\right].
 \end{align*}
 Then, for each $x,y \in \VAR(\sh)$, we have
 \begin{align*}
   & x \in \ALLOC{\rsh} ~\text{iff}~ x \in \ALLOC{\SHRINK{\sh,\T{q}}}, ~\text{and} \\
   & x \MSIM{\rsh} y ~\text{iff}~ x \MSIM{\SHRINK{\sh,\T{q}}} y. 
 \end{align*}
\end{lemma}
\begin{proof}
  We apply Lemma~\ref{thm:zoo:track:congruence:auxiliary} iteratively to
  \begin{align*}
   \psi_0 ~=~ & \sh[\mathbf{P_1 / \rsh_1}, \PS_2 / \rsh_2, \ldots, \PS_m / \rsh_m] ~=~ \rsh \\
   \psi_1 ~=~ & \sh[\PS_1 / \SSIGMA{\CALLN{1}{}}{\T{q}[1]}, \mathbf{\PS_2 / \rsh_2}, \ldots, \PS_m / \rsh_m] \\
   \psi_2 ~=~ & \sh[\PS_1 / \SSIGMA{\CALLN{1}{}}{\T{q}[1]}, \PS_2 / \SSIGMA{\CALLN{2}{}}{\T{q}[2]}, \mathbf{\PS_3 / \rsh_3}, \ldots, \PS_m / \rsh_m] \\
              & \ldots \\
   \psi_i ~=~ & \sh[\PS_1 / \SSIGMA{\CALLN{1}{}}{\T{q}[1]}, \ldots, \PS_i / \SSIGMA{\CALLN{i}{}}{\T{q}[i]}, \\
              & \qquad  \mathbf{\PS_{i+1} / \rsh_{i+1}},\ldots \PS_m / \rsh_m] \big] \\
              & \ldots \\
  \psi_m ~=~ & \sh[\PS_1 / \SSIGMA{\CALLN{1}{}}{\T{q}[1]}, \ldots, \PS_m / \SSIGMA{\CALLN{m}{}}{\T{q}[m]}] = \SHRINK{\sh,\T{q}}.
  \end{align*}
  Here, the replacement written in bold is the single predicate call to which Lemma~\ref{thm:zoo:track:congruence:auxiliary} is applied.
  As a result, we obtain for each $0 \leq i,j \leq m$ and each pair of variables $x,y \in \VAR(\sh)$ that
  \begin{align*}
    x \MSIM{\psi_i} y ~\text{iff}~ x \MSIM{\psi_j} y ~\text{and}~ x \in \ALLOC{\psi_i} ~\text{iff}~ x \in \ALLOC{\psi_j}.
  \end{align*}
  In particular, for $i = 0$ and $j = m$, this yields
 \begin{align*}
   & x \in \ALLOC{\rsh} ~\text{iff}~ x \in \ALLOC{\SHRINK{\sh,\T{q}}}, ~\text{and} \\
   & x \MSIM{\rsh} y ~\text{iff}~ x \MSIM{\SHRINK{\sh,\T{q}}} y. 
 \end{align*}
 \qed
\end{proof}
It remains to prove the missing Lemmas.
\begin{lemma} \label{thm:zoo:track:heap-existence}
  Given the setting of Lemma~\ref{thm:zoo:track:congruence:auxiliary}, it holds that
  \begin{align*}
      & \stack,\heap \SAT{\emptyset} \sh\left[\PS_1 / \SSIGMA{\CALLN{1}{}}{(A,\Lambda)}\right] ~\text{implies} \\
    & \exists \heap' ~.~ \stack,\heap' \SAT{\emptyset} \sh[ \PS_1 / \rsh] ~\text{and}~ \forall y \in \DOM(\stack) ~.~ \\
    & \qquad \qquad \qquad \stack(y) \in \DOM(\heap) ~\text{iff} ~\stack(y) \in \DOM(\heap'). 
  \end{align*}
\end{lemma}
\begin{proof}
   Assume $\stack,\heap \SAT{\emptyset} \sh\left[\PS_1 / \SSIGMA{\CALLN{1}{}}{(A,\Lambda)}\right]$.
   Then
 \begin{align*}
                    & \stack,\heap \SAT{\emptyset} \sh\left[\PS_1 / \SSIGMA{\CALLN{1}{}}{(A,\Lambda)}\right] \\
  ~\Leftrightarrow~ & \left[ \text{SL semantics} \right] \\
                    & \exists \T{v} \in \VAL^{\SIZE{\BV{}}} ~.~ \exists \heap_1, \heap_2 ~.~ \heap = \heap_1 \uplus \heap_2 \\
                    & ~\text{and}~ \stack[\BV{} \mapsto \T{v}],\heap_1 \SAT{\emptyset} \SPATIAL{} \\
                    & ~\text{and}~ \stack[\BV{} \mapsto \T{v}],\heap_2 \SAT{\emptyset} \SSIGMA{\CALLN{1}{}}{(A,\Lambda)} \\
                    & ~\text{and}~ \forall a \sim b \in \Pi ~.~ \stack[\BV{} \mapsto \T{v}](a) \sim \stack[\BV{} \mapsto \T{v}](b) \\
  ~\Leftrightarrow~ & \left[ \text{Lemma}~\ref{thm:symbolic-heaps:fv-coincidence} \right] \\
                    &\exists \T{v} \in \VAL^{\SIZE{\BV{}}} ~.~ \exists \heap_1, \heap_2 ~.~ \heap = \heap_1 \uplus \heap_2 \\
                    & ~\text{and}~ \stack[\BV{} \mapsto \T{v}],\heap_1 \SAT{\emptyset} \SPATIAL{} \\
                    & ~\text{and}~ (\stack[\BV{} \mapsto \T{v}]\upharpoonright\FV{0}{\SSIGMA{\CALLN{1}{}}{(A,\Lambda)}}),\heap_2 \SAT{\emptyset} \SSIGMA{\CALLN{1}{}}{(A,\Lambda)} \\
                    & ~\text{and}~ \forall a \sim b \in \Pi ~.~ \stack[\BV{} \mapsto \T{v}](a) \sim \stack[\BV{} \mapsto \T{v}](b) \\
      ~\Rightarrow~ & \left[ \text{Lemma}~\ref{thm:zoo:track-congruence-right},~\FV{0}{\SSIGMA{\CALLN{1}{}}{(A,\Lambda)}} = \FV{0}{\rsh} \right] \\
                    & \exists \T{v} \in \VAL^{\SIZE{\BV{}}} ~.~ \exists \heap_1, \heap_2 ~.~ \heap = \heap_1 \uplus \heap_2 \\
                    & ~\text{and}~ \stack[\BV{} \mapsto \T{v}],\heap_1 \SAT{\emptyset} \SPATIAL{} \\
                    & ~\text{and}~ \exists \heap_2' . (\stack[\BV{} \mapsto \T{v}]\upharpoonright\FV{0}{\rsh}),\heap_2' \SAT{\emptyset} \rsh \\
                    & ~\text{and}~ \forall y \in \DOM(\stack[\BV{} \mapsto \T{v}]\upharpoonright\FV{0}{\rsh}) ~.~ \\
                    & \qquad \quad \stack(y) \in \DOM(\heap_2) ~\text{iff} ~\stack(y) \in \DOM(\heap_2') \\
                    & ~\text{and}~ \forall a \sim b \in \Pi ~.~ \stack[\BV{} \mapsto \T{v}](a) \sim \stack[\BV{} \mapsto \T{v}](b) \\
      ~\Rightarrow~ & \big[ \text{SL semantics}~,~\DOM(h_2) \subseteq (\stack[\BV{} \mapsto \T{v}]\upharpoonright\FV{0}{\rsh})~,\\
                    & \qquad \qquad \text{set}~\heap' = \heap_1 \uplus \heap_2' \big] \\
                    & \exists \heap' ~.~ \stack,\heap' \SAT{\emptyset} \sh[\PS_1 / \rsh] ~\text{and}~ \forall y \in \DOM(\stack) ~.~ \\
                    & \qquad \qquad \qquad \stack(y) \in \DOM(\heap) ~\text{iff} ~\stack(y) \in \DOM(\heap'). 
 \end{align*}
 \qed
\end{proof}
\begin{lemma} \label{thm:zoo:track-congruence-semantics}
 Let $\rsh \in \RSL{}{}$, $A = \{ x \in \FV{0} ~|~ x^{\rsh} \in \ALLOC{\rsh} \}$ and $\Lambda = \{ x \sim y ~|~ x^{\rsh} \MSIM{\rsh} y^{\rsh} \}$.
 Moreover, let $\SSIGMA{\CALLN{1}{}}{(A,\Lambda)}$ be defined as in Section~\ref{sec:zoo}.
 Then, for each $(\stack,\heap) \in \STATES$ with $\DOM(s) = \FV{0}{\rsh}$, $\stack,\heap \SAT{\emptyset} \SSIGMA{\CALLN{1}{}}{(A,\Lambda)}$ holds if and only if~\footnote{to be strict: $\heap = \{\ldots\}$ is a shortcut for $\DOM(\heap) = \{ \stack(x) ~|~ x \in A\}$, $\heap(\stack(x)) = \NIL$ for each $x \in A$}
 \begin{align*}
   & \heap = \{ \stack(x) \mapsto \NIL ~|~ x \in A \}
   ~\text{and}~ \bigwedge_{x \sim y \in \Lambda} \stack(x) \sim \stack(y)  \\
   & \text{and}~ \bigwedge_{x \in A, y \in (A \setminus \{x\}) \cup \{\NIL\}} \stack(x) \neq \stack(y). 
 \end{align*}
\end{lemma}
\begin{proof}
 Follows immediately from the SL semantics (cf. Figure~\ref{fig:slsemantics}).
\end{proof}
\begin{lemma} \label{thm:zoo:track-congruence-right}
 Let $\rsh \in \RSL{}{}$, $A = \{ x \in \FV{0} ~|~ x^{\rsh} \in \ALLOC{\rsh} \}$ and $\Lambda = \{ x \sim y ~|~ x^{\rsh} \MSIM{\rsh} y^{\rsh} \}$.
 Moreover, let $\SSIGMA{\CALLN{1}{}}{(A,\Lambda)}$ be defined as in Section~\ref{sec:zoo}.
 Then, for each $(\stack,\heap) \in \STATES$ with $\DOM(s) = \FV{0}{\rsh}$, we have
 \begin{align*}
   \stack,\heap \SAT{\emptyset} \SSIGMA{\CALLN{1}{}}{(A,\Lambda)} \quad\text{implies}\quad \exists h' ~.~ \stack,\heap' \SAT{\emptyset} \rsh.
 \end{align*}
 Furthermore, for each $x \in \DOM(s)$, we have $\stack(x) \in \DOM(\heap)$ iff $\stack(x) \in \DOM(\heap')$.
\end{lemma}
\begin{proof}
 By structural induction on $\rsh$.
 For readability, we omit equalities of the form $x = x$ in pure formulas
 although they are contained in pure formulas of $\SSIGMA{\CALLN{1}{}}{(A,\Lambda)}$ by definition.
 \paragraph{The case $\rsh = a \sim b$}
 is trivial, because $\SSIGMA{\CALLN{1}{}}{(A,\Lambda)}$ contains $a \sim b$.
 \paragraph{The case $\rsh = \EMP$}
 is trivial, because $A = \Lambda = \emptyset$ and $\SSIGMA{\CALLN{1}{}}{(A,\Lambda)} = \EMP$.
 \paragraph{The case $\rsh = \PT{x}{\T{y}}$}
 Then $A = \{x\}$, $\Lambda = \{ x \neq \NIL \}$. Moreover, we know that $\SSIGMA{\CALLN{1}{}}{(A,\Lambda)} = \PT{x}{\NIL} : \{ x \neq \NIL \}$.
 Hence,
 \begin{align*}
                & \stack,\heap \SAT{\emptyset} \SSIGMA{\CALLN{1}{}}{(A,\Lambda)} \\
  ~\Rightarrow~ &  \left[ \text{Definition of}~\SSIGMA{\CALLN{1}{}}{(A,\Lambda)} \right] \\
                & \stack,\heap \SAT{\emptyset} \PT{x}{\NIL} : \{ x \neq \NIL \} \\
  ~\Rightarrow~ &  \left[ \text{SL semantics} \right] \\
                & \stack(x) \neq \NIL ~\text{and}~ h = \{ \stack(x) \mapsto \NIL \} \\
  ~\Rightarrow~ & \left[ \text{choose}~h' = \{ \stack(x) \mapsto \stack(\T{y}) \} \right] \\
                & \stack(x) \neq \NIL ~\text{and}~ \DOM(\heap') = \{\stack(x)\} ~\text{and}~ h(\stack(x)) = \stack(\T{y}) \\
  ~\Rightarrow~ &  \left[ \text{SL semantics} \right] \\
                & \stack,\heap' \SAT{\emptyset} \PT{x}{\T{y}}.
 \end{align*}
 Moreover, by our choice of $h'$, we have $\DOM(\heap) = \DOM(\heap')$.
 \paragraph{The case $\rsh = \Sigma_1 \SEP \Sigma_2$}
 Let $A_1,A_2$ and $\Lambda_1$, $\Lambda_2$ denote the sets corresponding to $\Sigma_1$ and $\Sigma_2$, respectively.
 \begin{align*}
                & \stack,\heap \SAT{\emptyset} \SSIGMA{\CALLN{1}{}}{(A,\Lambda)} \\
  ~\Rightarrow~ &  \left[ \text{Definition}~\SSIGMA{\CALLN{1}{}}{(A,\Lambda)} \right] \\
                & \stack,\heap \SAT{\emptyset} \SSIGMA{\CALLN{1}{}}{(A,\Lambda)} ~\text{and}~h = \{\stack(x) \mapsto \NIL ~|~ x \in A \} \\
  ~\Leftrightarrow~ &  \left[ \text{Lemma}~\ref{thm:zoo:track-congruence-semantics} \right] \\
                &  \bigwedge_{x \sim y \in \Lambda} \stack(x) \sim \stack(y) ~\wedge~ \bigwedge_{x \in A, y \in (A \setminus \{x\}) \cup \{\NIL\}} \stack(x) \neq \stack(y) \\
                & \qquad\text{and}~h = \{\stack(x) \mapsto \NIL ~|~ x \in A \} \\
  ~\Rightarrow~ &  \left[ A = A_1 \uplus A_2 \right] \\
                & \bigwedge_{x \sim y \in \Lambda} \stack(x) \sim \stack(y) ~\wedge~ \bigwedge_{x \in A, y \in (A \setminus \{x\}) \cup \{\NIL\}} \stack(x) \neq \stack(y) \\
                & \wedge~h = \{\stack(x) \mapsto \NIL ~|~ x \in A_1 \} \uplus \{\stack(x) \mapsto \NIL ~|~ x \in A_1 \} \\
  ~\Rightarrow~ &  \left[ \text{for}~i=1,2~\text{choose}~h_i = \{ \stack(x) \mapsto \NIL ~|~ \PT{x}{\T{y}} \in \Sigma_i \} \right] \\
                & \bigwedge_{x \sim y \in \Lambda} \stack(x) \sim \stack(y) ~\wedge~ \bigwedge_{x \in A, y \in (A \setminus \{x\}) \cup \{\NIL\}} \stack(x) \neq \stack(y) \\
                & \wedge~h = h_1 \uplus h_2 \\
  ~\Rightarrow~ &  \left[ \Lambda_1,\Lambda_2 \subseteq \Lambda \right] \\
                & \bigwedge_{x \sim y \in \Lambda_1} \stack(x) \sim \stack(y) ~\wedge~ \bigwedge_{x \sim y \in \Lambda_2} \stack(x) \sim \stack(y)\\
                & \wedge~ \bigwedge_{x \in A, y \in (A \setminus \{x\}) \cup \{\NIL\}} \stack(x) \neq \stack(y) ~\wedge~h = h_1 \uplus h_2 \\
  ~\Rightarrow~ & \left[ \text{Lemma}~\ref{thm:zoo:track-congruence-semantics} \right] \\
                & \stack,\heap_1 \SAT{\emptyset} \SSIGMA{\CALLN{1}{}}{(A_1,\Lambda_1)} ~\text{and}~ \stack,\heap_2 \SAT{\emptyset} \SSIGMA{\CALLN{1}{}}{(A_2,\Lambda_2)} \\
                & \wedge~ \bigwedge_{x \in A, y \in (A \setminus \{x\}) \cup \{\NIL\}} \stack(x) \neq \stack(y) ~\wedge~h = h_1 \uplus h_2 \\
  ~\Rightarrow~ & \left[ \text{I.H.} \right] \\
                & \exists h_1',h_2' ~.~ \stack,\heap_1' \SAT{\emptyset} \Sigma_1 ~\text{and}~ \stack,\heap_2' \SAT{\emptyset} \Sigma_2 \\
                & \wedge~\forall x \in \DOM(s) ~.~ \stack(x) \in \DOM(\heap_1) ~\text{iff}~ \stack(x) \in \DOM(\heap_1') \\
                & \wedge~\forall x \in \DOM(s) ~.~ \stack(x) \in \DOM(\heap_2) ~\text{iff}~ \stack(x) \in \DOM(\heap_2') \\
                & \wedge~ \bigwedge_{x \in A, y \in (A \setminus \{x\}) \cup \{\NIL\}} \stack(x) \neq \stack(y) ~\wedge~h = h_1 \uplus h_2 \\
  ~\Rightarrow~ & \left[ \DOM(\heap_1) \cap \DOM(\heap_2) = \emptyset,~h' = h_1' \uplus h_2' \right] \\
                & \exists h' = h_1' \uplus h_2' ~.~ \stack,\heap_1' \SAT{\emptyset} \Sigma_1 ~\text{and}~ \stack,\heap_2' \SAT{\emptyset} \Sigma_2 \\
                & ~\text{and}~ \forall x \in \DOM(s) ~.~ \stack(x) \in \DOM(\heap) ~\text{iff}~ \stack(x) \in \DOM(\heap') \\
  ~\Rightarrow~ & \left[ \text{SL semantics} \right] \\
                & \stack,\heap' \SAT{\emptyset} \Sigma_1 \SEP \Sigma_2 \\
                & ~\text{and}~ \forall x \in \DOM(s) ~.~ \stack(x) \in \DOM(\heap) ~\text{iff}~ \stack(x) \in \DOM(\heap')
 \end{align*}
 \paragraph{The case $\rsh = \exists z \,.\, \SPATIAL{} \,:\, \Pi$}
 Let $\rsh' = \STRIP{\rsh} = \SPATIAL{} \,:\, \Pi$ be as $\rsh$ except that $z$ is a free variable.
 Moreover, let $A'$ and $\Lambda'$ be the corresponding sets of allocated variables and pure formulas between free variables of $\rsh'$.
 Then
 \begin{align*}
                 & \stack,\heap \SAT{\emptyset} \SSIGMA{\CALLN{1}{}}{(A,\Lambda)} \\
   ~\Rightarrow~ & \left[ \text{Lemma}~\ref{thm:zoo:track-congruence-semantics} \right] \\
                 & h = \{ \stack(x) \mapsto \NIL ~|~ x \in A\} ~\text{and}~ \\
                 & \bigwedge_{x \sim y \in \Lambda} \stack(x) \sim \stack(y) ~\wedge~ \bigwedge_{x \in A, y \in (A \setminus \{x\}) \cup \{\NIL\}} \stack(x) \neq \stack(y) \\
   ~\Rightarrow~ & \left[ \ALLOC{\rsh} = \ALLOC{\rsh'},~x\MSIM{\rsh}y \Leftrightarrow x\MSIM{\rsh'}y,~\DOM(s) = \FV{0}{\rsh} \right] \\
                 & h = \{ \stack(x) \mapsto \NIL ~|~ x \in A\} ~\text{and}~ \\
                 & \exists a \in \VAL ~.~ \bigwedge_{x \sim y \in \Lambda'} s[z \mapsto a](x) \sim s[z \mapsto a](y) \\
                 & \wedge~ \bigwedge_{x \in A', y \in (A' \setminus \{x\}) \cup \{\NIL\}} s[z \mapsto a](x) \neq s[z \mapsto a](y) \\
   ~\Rightarrow~ & \big[ \text{Lemma}~\ref{thm:zoo:track-congruence-semantics},\\
                 & \qquad \qquad ~\text{set}~ h' = h \uplus \{ \stack(z) \mapsto \NIL \} ~\text{if}~ z \in \ALLOC{\rsh} \} \big] \\
                 & \exists a \in \VAL ~.~ s[z \mapsto a],h' \SAT{\emptyset} \SSIGMA{\CALLN{1}{}}{(A',\Lambda')} \\
   ~\Rightarrow~ & \left[ \text{I.H.} \right] \\
                 & \exists a \in \VAL ~.~ \exists h'' ~.~ s[z \mapsto a],h'' \SAT{\emptyset} \rsh' \\
                 & ~\text{and}~ \forall x \in \DOM(s[z \mapsto a]) ~.~ s[z \mapsto a](x) \in \DOM(\heap') \\
                 & \qquad \text{iff}~ s[z \mapsto a](x) \in \DOM(\heap'') \\
   ~\Rightarrow~ & \left[ \exists x \exists y \cong \exists y \exists x \right] \\
                 & \exists h'' ~.~ \exists a \in \VAL ~.~ s[z \mapsto a],h'' \SAT{\emptyset} \rsh' \\
                 & ~\text{and}~ \forall x \in \DOM(s[z \mapsto a]) ~.~ s[z \mapsto a](x) \in \DOM(\heap') \\
                 & \qquad \text{iff}~ s[z \mapsto a](x) \in \DOM(\heap'') \\
   ~\Rightarrow~ & \left[ \text{SL semantics} \right] \\
                 & \exists h'' ~.~ \stack,\heap'' \SAT{\emptyset} \rsh \\
                 & ~\text{and}~ \forall x \in \DOM(s) ~.~ \stack(x) \in \DOM(\heap) \\
                 & \qquad \text{iff}~\stack(x) \in \DOM(\heap''). 
 \end{align*}
 \qed
\end{proof}
We also need the (simpler) converse direction.
\begin{lemma} \label{thm:zoo:track-congruence-left}
 Let $\rsh \in \RSL{}{}$, $A = \{ x \in \FV{0} ~|~ x^{\rsh} \in \ALLOC{\rsh} \}$ and $\Lambda = \{ x \sim y ~|~ x^{\rsh} \MSIM{\rsh} y^{\rsh} \}$.
 Moreover, let $\SSIGMA{\CALLN{1}{}}{(A,\Lambda)}$ be defined as in Section~\ref{sec:zoo}.
 Then, for each $(\stack,\heap) \in \STATES$ with $\DOM(s) = \FV{0}{\rsh}$, we have
 \begin{align*}
   \stack,\heap \SAT{\emptyset} \rsh \quad\text{implies}\quad \exists h' ~.~ \stack,\heap' \SAT{\emptyset} \SSIGMA{\CALLN{1}{}}{(A,\Lambda)}.
 \end{align*}
 Furthermore, for each $x \in \DOM(s)$, we have $\stack(x) \in \DOM(\heap)$ iff $\stack(x) \in \DOM(\heap')$.
\end{lemma}
\begin{proof}
  Assume $\stack,\heap \SAT{\emptyset} \rsh$ and $\DOM(s) = \FV{0}{\rsh}$.
  By definition of $\MSIM{\rsh}$ we know that for each $x,y \in \DOM(s)$ with $x \MSIM{\rsh} y$, we have $\stack(x) \sim \stack(y)$.
  Moreover, by definition of $\ALLOC{\rsh}$, $\stack(x) \in \DOM(\heap)$ iff $x \in \DOM(s) \cap \ALLOC{\rsh}$.
  By definition of $A$ and $\Lambda$ this means
  \begin{align*}
    \bigwedge_{x \sim y \in \Lambda} \stack(x) \sim \stack(y) ~\wedge~ \bigwedge_{x \in A, y \in (A \setminus \{x\}) \cup \{\NIL\}} \stack(x) \neq \stack(y).
  \end{align*}
  Thus, by Lemma~\ref{thm:zoo:track-congruence-semantics}, we obtain
  \begin{align*}
   & s,\{ \stack(x) \mapsto \NIL ~|~ x \in A\} \SAT{\emptyset} \SSIGMA{\CALLN{1}{}}{(A,\Lambda)}.  \qedhere
  \end{align*}
  %
  %
  \qed
\end{proof}
%
%
\section{Proof of Theorem~\ref{thm:zoo:sat:property}} \label{app:zoo:sat:property}
 A heap automaton $\HASAT$ accepting $\SATPROP(\alpha)$ is constructed as presented in Definition~\ref{def:zoo:track-automaton} except for the set of final states being set to
 $
  F_{\HASAT} ~\DEFEQ~ \{ (A,\Pi) ~|~ \NIL \neq \NIL \,\notin\, \Pi \}
 $.
 Since we already now that $\HATRACK$ and thus also $\HASAT$ satisfies the compositionality property (cf. Lemma~\ref{thm:zoo:track:compositional}), it suffices to show
 $L(\HASAT) = \SATPROP(\alpha)$.
 Let $\rsh \in \RSL{}{\SRDCLASSFV{\alpha}}$. Then
 \begin{align*}
                      & \rsh~\text{satisfiable} \\
    ~\Leftrightarrow~ & \left[ \text{Definition of satisfiability} \right] \\
                      & \exists (\stack,\heap) ~.~ \stack,\heap \SAT{\emptyset} \rsh \\
    ~\Leftrightarrow~ & \left[ \text{Lemma}~\ref{thm:symbolic-heaps:fv-coincidence} \right] \\
                      & \MODELS{\rsh} \neq \emptyset \\
    ~\Leftrightarrow~ & \big[ \big(\forall (\stack,\heap) \in \MODELS{\rsh} ~.~ \stack(\NIL) \neq \stack(\NIL)\big) \\
                      & \qquad ~\text{iff}~ \MODELS{\rsh} = \emptyset ] \\
                      & \exists \Pi ~.~ \forall x,y \in \FV{0}{\rsh} ~.~ \NIL \neq \NIL \notin \Pi \\
                      & \qquad \text{and}~ (x \sim y) \in \Pi \\
                      & \qquad \qquad \leftrightarrow \left( \forall (\stack,\heap) \in \MODELS{\rsh} ~.~ \stack(x^{\rsh}) \sim \stack(y^{\rsh}) \right) \\
    ~\Leftrightarrow~ & \left[ \text{Definition}~ x \MSIM{\rsh} y  \right] \\
                      & \exists \Pi ~.~ \forall x,y \in \FV{0}{\rsh} ~.~ \NIL \neq \NIL \notin \Pi \\
                      & \qquad \text{and}~ (x \sim y) \in \Pi \leftrightarrow x^{\rsh} \MSIM{\rsh} y^{\rsh} \\
    ~\Leftrightarrow~ & \left[ \text{Definition of}~F_{\HASAT} \right] \\
                      & \exists (A,\Pi) \in F_{\HASAT} ~.~ \forall x,y \in \FV{0}{\rsh} ~.~ x \in A \leftrightarrow x^{\rsh} \in \ALLOC{\rsh} \\
                      & \qquad \text{and}~ (x \sim y) \in \Pi \leftrightarrow x^{\rsh} \MSIM{\rsh} y^{\rsh} \\
    ~\Leftrightarrow~ & \left[ \text{Definition}~\ref{def:zoo:track-automaton}:~\SHRINK{\rsh,\EMPTYSEQ} = \rsh \right] \\
                      & \exists (A,\Pi) \in F_{\HASAT} ~.~ \forall x,y \in \FV{0}{\SHRINK{\rsh,\EMPTYSEQ}} ~.~ \\
                      & \qquad \qquad x \in A \leftrightarrow x^{\rsh} \in \ALLOC{\SHRINK{\rsh,\EMPTYSEQ}} \\
                      & \qquad \text{and}~ (x \sim y) \in \Pi \leftrightarrow x^{\rsh} \MSIM{\SHRINK{\rsh,\EMPTYSEQ}} y^{\rsh} \\
    ~\Leftrightarrow~ & \left[ \text{Definition}~\Delta_{\HASAT} \right] \\
                      & \exists (A,\Pi) \in F ~.~ \OMEGA{\HASAT}{(A,\Pi)}{\rsh} \\
    ~\Leftrightarrow~ & \left[ \text{Definition}~L(\HASAT) \right] \\
                      & \rsh \in L(\HASAT).
 \end{align*}
 Thus, $\rsh$ is satisfiable if and only if $\rsh \in L(\HASAT)$.
 \qed
 %
%
\section{Satisfiability is in \CCLASS{NP}} \label{app:zoo:sat:np}
Brotherston et al.~\cite{brotherston2014decision} already showed that the satisfiability problem
is decidable in \CCLASS{NP} if the maximal number of free variables $\alpha$ is bounded.
In this section we briefly show that such an \CCLASS{NP}--decision procedure naturally emerged from our heap
automaton $\HASAT$ accepting $\SATPROP(\alpha)$ (see Theorem~\ref{thm:zoo:sat:property}).
\begin{lemma}
  \DPROBLEM{SL-SAT} is in \CCLASS{NP} if the maximal number $\alpha$ of free variables is bounded.
\end{lemma}
\begin{proof}
  Let $(\SRD,\sh) \in \DPROBLEM{SL-SAT}$ and $N = \SIZE{\SRD} + \SIZE{\sh}$.
  %
  Moreover, let $n \leq N$ be the maximal number of predicate calls occurring in $\sh$ and any rule of $\SRD$.
  Since $\alpha$ is a constant, the number of states of
  heap automaton $\HASAT$ (cf. the proof of Theorem~\ref{thm:zoo:sat:property})
  is a constant,
  namely $k = 2^{2\alpha^2 + \alpha}$.
  Now, let $\UTREES{\SRD}{\sh}^{\leq k}$ denote the set of all unfolding trees $t \in \UTREES{\SRD}{\sh}$ of height at most $k$.
  Clearly, each of these trees is of size $\SIZE{t} \leq n^{k} \leq N^{k}$, i.e., polynomial in $N$.
  Moreover, let $\omega: \DOM(t) \to Q_{\HASAT}$ be a function mapping each node of $t$ to a state of
  $\HASAT$. 
  Again, $\omega$ is of size polynomial in $N$; as such $\SIZE{\omega} \leq k \cdot N^{k}$.
  Let $\Omega_{t}$ denote the set of all of these functions $\omega$ for a given
  unfolding tree $t$ with $\omega(\EMPTYSEQ) \in F_{\HASAT}$.
  Now, given an unfolding tree $t \in \UTREES{\SRD}{\sh}^{\leq k}$ and $\omega \in \Omega_{t}$, we can easily decide whether $\OMEGA{A}{\omega(\EMPTYSEQ)}{\UNFOLD{t}}$ holds:
  For each $u,u1,\ldots,un \in \DOM(t)$, $u(n+1) \notin \DOM(t)$, $n \geq 0$ it suffices to check whether $\MOVE{A}{\omega(u)}{t(u)}{\omega(u1) \ldots \omega(un)}$.
  Since, by Remark~\ref{rem:closure}, 
  each of these checks can be performed in time polynomial in $N$,
  the whole procedure is feasible in polynomial time.

  %
  Then, our decision procedure for $\DPROBLEM{SL-SAT}$ answers yes on input $(\SRD,\sh)$ if and only if
  \begin{align*}
    \exists t \in \UTREES{\SRD}{\sh}^{\leq k} ~.~ \forall \omega \in \Omega_{t} ~.~  \text{not}~ \OMEGA{A}{\omega(\EMPTYSEQ)}{\UNFOLD{t}}.
  \end{align*}
  Since $t$ and $\omega$ are both of size polynomial in $N$, this procedure is in $\CCLASS{NP}$.
  Regarding correctness, we first note that
  $\CALLSEM{\sh}{\SRD} \cap \SATPROP(\alpha) \neq \emptyset$ holds iff $\UNFOLD{t} \in \SATPROP(\alpha)$
  for some $t \in \UTREES{\SRD}{\sh}$.
  Furthermore, by a standard pumping argument, it suffices to consider trees in $\UTREES{\SRD}{\sh}^{\leq k}$:
  If there exists a taller tree $t$ with $\UNFOLD{t} \in \SATPROP(\alpha)$ then there is some path of length greater $k$ in $t$ on which two nodes are assigned the same state by
  a function $\omega \in \Omega_{t}$ proving membership of $t$ in $\SATPROP(\alpha)$.
  Thus, this path can be shortened to obtain a tree of smaller height whose unfolding is satisfiable.
  \qed
  \end{proof}
%
\section{
Compositionality of $\HAEST$, $\HAGARBAGE$, $\HACYCLE$
}
\label{app:zoo:scheme:compositionality}
Some of our heap automata presented in Section~\ref{sec:zoo}
follow a common construction scheme.
In particular this holds for
\begin{itemize}
  \item $\HAEST$ (cf. Lemma~\ref{thm:zoo:establishment} and Appendix~\ref{app:zoo:establishment}),
  \item $\HAGARBAGE$ (cf. Lemma~\ref{thm:zoo:garbage:property} and Appendix~\ref{app:zoo:garbage:property}), and
  \item $\HACYCLE$ (cf. Lemma~\ref{thm:zoo:acyclicity:property} and Appendix~\ref{app:zoo:acyclicity:property}).
\end{itemize}
Intuitively, each of these automata evaluates a predicate
$\CHECK : \SL{}{\SRDCLASSFV{\alpha}} \times Q_{\HA{A}}^{*} \to \{0,1\}$
while running a heap automaton $\HA{A}$
-- which is either $\HATRACK$ (cf. Definition~\ref{def:zoo:track-automaton})
or $\HAREACH$ (cf. Lemma~\ref{thm:zoo:reachability:property})  --
in parallel to collect required knowledge about the relationships (equalities, allocation, reachability)
between free variables.
Due to these similarities, we prove the compositionality property for a general construction
scheme as described above.
More formally,
\begin{definition} \label{def:zoo:scheme-automaton}
  Let $\HA{A} \in \{ \HASAT, \HAREACH \}$
  be either the tracking automaton
  $\HATRACK$ (cf. Definition~\ref{def:zoo:track-automaton})
  or the reachability automaton
  $\HAREACH$ (cf. Lemma~\ref{thm:zoo:reachability:property}).
  Moreover, let
  $F \subseteq Q_{\HA{A}} \times \{0,1\}$.
  Further, let $\CHECK : \SL{}{\SRDCLASSFV{\alpha}} \times Q_{\HA{A}}^{\star} \to \{0,1\}$
  be a Boolean predicate such that for each
  $\sh \in \SL{}{\SRDCLASSFV{\alpha}}$ with $\NOCALLS{\sh} = m$
  $\rsh_1,\ldots,\rsh_m \in \RSL{}{\SRDCLASSFV{\alpha}}$,
  $\MOVE{A}{p_0}{\sh}{p_1 \ldots p_m}$, and $\OMEGA{A}{p_i}{\rsh_i}$,
  we have
  \begin{align*}
   \CHECK(\sh[\PS_1^{\sh} / \rsh_1, \ldots, \PS_m^{\sh} / \rsh_m],\EMPTYSEQ) = 1
  \end{align*}
  if and only if such that
  \begin{align*}
    \CHECK(\sh, p_1 \ldots p_m) = \CHECK(\rsh_1,\EMPTYSEQ) = \ldots = \CHECK(\rsh_m,\EMPTYSEQ) = 1.
  \end{align*}
  Then the heap automaton $\HASCHEME{\HA{A}}{\CHECK}{F}$ is given by
  \begin{align*}
   & \HASCHEME{\HA{A}}{\CHECK}{F} = (Q,\SRDCLASSFV{\alpha},\Delta,F)~,\text{where} \\
   & Q ~\DEFEQ~ Q_{\HA{A}} \times \{0,1\} \\
   &  \Delta ~~:~~ \MOVE{B}{(p_0,q_0)}{\sh}{(p_1,q_1) \ldots (p_m,q_m)} \\
   & ~\text{iff}~ \MOVE{A}{p_0}{\sh}{p_1\ldots p_m} \\
   & \qquad ~\text{and}~ q_0 = \min \{q_1,\ldots,q_m,\CHECK(\sh,p_1 \ldots p_m)\}~,
  \end{align*}
  where $\NOCALLS{\sh} = m \geq 0$.
\end{definition}
\begin{lemma} \label{thm:zoo:scheme:compositionality}
  The heap automaton $\HA{B} = \HASCHEME{\HA{A}}{\CHECK}{F}$ satisfies the
  compositionality property.
\end{lemma}
\begin{proof}
  Let  $\sh \in \SL{}{\SRDCLASSFV{\alpha}}$ with $\NOCALLS{\sh} = m$
  and $\rsh_1,\ldots,\rsh_m \in \RSL{}{\SRDCLASSFV{\alpha}}$.
  Moreover, let $\rsh = \sh[\PS_1^{\sh} / \rsh_1, \ldots, \PS_m^{\sh} / \rsh_m]$.
  We have to show that for each $(p_0,q_0) \in Q_{\HA{B}}$ it holds that
  \begin{align*}
    \OMEGA{B}{(p_0,q_0)}{\rsh} \quad\text{iff}\quad & \exists (p_1,q_1),\ldots,(p_m,q_m) \in Q_{\HA{B}} ~.~ \\
    & \quad\MOVE{B}{(p_0,q_0)}{\sh}{(p_1,q_1)\ldots(p_m,q_m)} \\
    & \quad\text{and}~ \forall 1 \leq i \leq m ~.~ \OMEGA{B}{(p_i,q_i)}{\rsh_i}.
  \end{align*}
  Assume $\OMEGA{B}{(p_0,q_0)}{\rsh}$.
  By definition of $\Delta_{\HA{B}}$
  this is the case if and only if
  $ 
     \OMEGA{A}{p_0}{\rsh} ~\text{and}~ q_0 = \CHECK(\rsh,\EMPTYSEQ).
  $ 
  For each $1 \leq i \leq m$, we set $q_i = \CHECK(\tau_i,\EMPTYSEQ) \in \{0,1\}$.
  Moreover, since $\HA{A}$ is known to satisfy the compositionality property by Lemma~\ref{thm:zoo:track:property} (or Lemma~\ref{thm:zoo:reachability:property}), this is equivalent to
  \begin{align*}
                  & q_0 = \CHECK(\rsh,\EMPTYSEQ)
                    ~\text{and}~ \forall 1 \leq i \leq m ~.~ q_i = \CHECK(\rsh_i,\EMPTYSEQ) \\
     ~\text{and}~ &\exists p_1,\ldots,p_m \in Q_{\HA{A}} ~.~ \MOVE{A}{p_0}{\sh}{p_1 \ldots p_m} \\
     ~\text{and}~ &\forall 1 \leq i \leq m ~.~ \OMEGA{A}{p_0}{\rsh_i}.
  \end{align*}
  By construction of $\HA{B}$ and since each $\rsh_i$ contains no predicate calls this is equivalent to
  \begin{align*}
                  & q_0 = \CHECK(\rsh,\EMPTYSEQ) \tag{$\dag$} \\
     ~\text{and}~ &\exists (p_1,q_1),\ldots,(p_m,q_m) \in Q_{\HA{B}} ~.~ \MOVE{A}{p_0}{\sh}{p_1 \ldots p_m} \\
     ~\text{and}~ &\forall 1 \leq i \leq m ~.~ \OMEGA{B}{(p_0,q_0)}{\rsh_i}.
  \end{align*}
  Now, by Definition~\ref{def:zoo:scheme-automaton},
  $\CHECK(\rsh,\EMPTYSEQ) = 1$ if and only if
  $\CHECK(\sh,p_1 \ldots p_m) = 1$ and
  $\CHECK(\rsh_1,\EMPTYSEQ) = \ldots = \CHECK(\rsh_m,\EMPTYSEQ) = 1$.
  Thus
  \begin{align*}
   q_0 ~=~ \CHECK(\rsh,\EMPTYSEQ) ~=~ & \min \{ \CHECK(\sh,p_1\ldots p_m), \\
                                      & \qquad \quad
                                         \CHECK(\rsh_1,\EMPTYSEQ), \ldots \CHECK(\rsh_m,\EMPTYSEQ) \} \\
                         ~=~ & \min \{ \CHECK(\sh,p_1\ldots p_m), q_1, \ldots, q_m \}.
  \end{align*}
  Putting this equation into $(\dag)$, we obtain the equivalent statement
  \begin{align*}
                  & \exists (p_1,q_1),\ldots,(p_m,q_m) \in Q_{\HA{B}} ~.~ \MOVE{A}{p_0}{\sh}{p_1 \ldots p_m} \\
     ~\text{and}~ & q_0 = \min \{ \CHECK(\sh,p_1 \ldots p_m), q_1, \ldots, q_m \} \\
     ~\text{and}~ & \forall 1 \leq i \leq m ~.~ \OMEGA{B}{(p_0,q_0)}{\rsh_i}.
  \end{align*}
  By definition of $\Delta_{\HA{B}}$, this is equivalent to
  \begin{align*}
                  & \exists (p_1,q_1),\ldots,(p_m,q_m) \in Q_{\HA{B}} ~.~ \\
                  & \qquad \MOVE{B}{(p_0,q_0)}{\sh}{(p_1,q_1) \ldots (p_m,q_m)} \\
                  & \qquad \text{and}~ \forall 1 \leq i \leq m ~.~ \OMEGA{B}{(p_0,q_0)}{\rsh_i}.
  \end{align*}
  Hence, $\HA{B} = \HASCHEME{\HA{A}}{\CHECK}{F}$ satisfies the compositionality property.
  \qed
\end{proof}
%
%
\section{Proof of Theorem~\ref{thm:zoo:establishment}} \label{app:zoo:establishment}
We have to construct a heap automaton
$\HAEST$ over $\SHCLASSFV{\alpha}$ that satisfies the compositionality property
and accepts $\ESTPROP(\alpha)$.
In order to highlight the necessary proof obligations, the actual construction of $\HAEST$
and its correctness proof are splitted into several definitions and lemmas that are
provided afterwards.
The construction of $\HAEST$ was already sketched in the paper.
A fully formal construction of $\HAEST$ is found in Definition~\ref{def:zoo:establishment:automaton}.
It then remains to show the correctness of our construction of $\HAEST$:
\begin{itemize}
  \item Lemma~\ref{thm:zoo:establishment:language} establishes that $\HAEST$ indeed accepts $\ESTPROP(\alpha)$, i.e., $L(\HAEST) = \ESTPROP(\alpha)$.
  \item To prove the compositionality property, we show that $\HAEST$ is an instance of a more general
        construction scheme whose compositionality property is shown in
        Lemma~\ref{thm:zoo:scheme:compositionality}.
        In order to apply Lemma~\ref{thm:zoo:scheme:compositionality}, we have to show that
        \begin{align*}
                          & \CHECK(\sh[\PS_1^{\sh} / \rsh_1, \ldots, \PS_m^{\sh} / \rsh_m], \EMPTYSEQ) = 1 \\
             ~\text{iff}~ & \CHECK(\sh,p_1 \ldots p_m) = 1 \\
                          & \text{and}~ \CHECK(\rsh_1,\EMPTYSEQ) = \ldots = \CHECK(\rsh_m,\EMPTYSEQ) = 1.
         \end{align*}
         This is verified in Lemma~\ref{thm:zoo:establishment:compositionality}.
         Then, by Lemma~\ref{thm:zoo:scheme:compositionality}, we know that 
         $\HAEST = \HASCHEME{\HATRACK}{\CHECK}{F}$ satisfies the compositionality property.
\end{itemize}
Putting both together, we obtain a heap automaton $\HAEST$ over $\SHCLASSFV{\alpha}$
that satisfies the compositionality property and accepts $\ESTPROP(\alpha)$.
\qed
The remainder of this section fills the gaps in the proof from above.
\begin{definition} \label{def:zoo:establishment:automaton}
$\HAEST = (Q,\SL{}{\SRDCLASSFV{\alpha}},\Delta,F)$ is given by
  \begin{align*}
   & Q ~\DEFEQ~ Q_{\HATRACK} \times \{0,1\}, \qquad F ~\DEFEQ~ Q_{\HATRACK} \times \{1\}, \\
   &  \Delta ~~:~~ \MOVE{\HAEST}{(p_0,q_0)}{\sh}{(p_1,q_1) \ldots (p_m,q_m)} \\
   & ~\text{iff}~ \MOVE{\HATRACK}{p_0}{\sh}{p_1\ldots p_m} \\
   & \qquad \text{and}~ q_0 = \min \{q_1,\ldots,q_m,\CHECK(\sh,p_1 \ldots p_m)\}.
  \end{align*}
  Here, $\CHECK : \SHCLASSFV{\alpha} \times Q_{\HATRACK}^{*} \to \{0,1\}$ is a predicate given by
  \begin{align*}
   \CHECK(\sh,\T{p}) ~\DEFEQ~ \begin{cases}
                                   1 &, ~\text{if}~ \forall y \in \VAR(\sh) ~.~  y \in \ALLOC{\SHRINK{\sh,\T{p}}} \\
                                    & \qquad\text{or}~ \exists x \in \FV{0}{\sh} ~.~ x \MEQ{\SHRINK{\sh,\T{p}}} y \\
                                   0 &, ~\text{otherwise}~,
                                 \end{cases}
  \end{align*}
  where $\SHRINK{\sh,\T{p}}$ is the reduced symbolic heap obtained from the tracking property as in Definition~\ref{def:zoo:track-automaton}.
\end{definition}
\begin{lemma} \label{thm:zoo:establishment:language}
  $L(\HAEST) = \ESTPROP(\alpha)$.
\end{lemma}
\begin{proof}
Let $\rsh \in \RSL{}{\SRDCLASSFV{\alpha}}$. Then:
\begin{align*}
                   & \rsh \in L(\HAEST) \\
 ~\Leftrightarrow~ & \left[ \text{Definition of}~L(\HAEST) \right] \\
                   & \exists q \in F_{\HAEST} ~.~ \OMEGA{\HAEST}{q}{\rsh} \\
 ~\Leftrightarrow~ & \left[ \text{Definition of}~F_{\HAEST},~q = (p,1) \right] \\
                   & \exists p \in Q_{\HATRACK} ~.~ \OMEGA{\HAEST}{(p,1)}{\rsh} \\
 ~\Leftrightarrow~ & \left[ \text{Definition of}~\Delta_{\HAEST} \right] \\
                   & \exists p \in Q_{\HATRACK} ~.~ \OMEGA{\HATRACK}{p}{\rsh} ~\text{and}~ \CHECK(\rsh,\EMPTYSEQ) = 1 \\
 ~\Leftrightarrow~ & \left[ \text{Definition of}~\CHECK(\rsh,\EMPTYSEQ) \right] \\
                   & \exists p \in Q_{\HATRACK} ~.~ \OMEGA{\HATRACK}{p}{\rsh} \\
                   & \qquad \text{and}~ \forall y \in \VAR(\rsh) ~.~ y \in \ALLOC{\SHRINK{\rsh,\EMPTYSEQ}} \\
                   & \qquad \qquad \text{or}~ \exists x \in \FV{0}{\rsh} ~.~ x \MEQ{\SHRINK{\rsh,\EMPTYSEQ}} y \\
 ~\Leftrightarrow~ & \left[ \NOCALLS{\rsh} = 0 ~\text{implies}~ \rsh = \SHRINK{\rsh,\EMPTYSEQ} \right] \\
                   & \forall y \in \VAR(\rsh) ~.~ y \in \ALLOC{\rsh} ~\text{or}~ \exists x \in \FV{0}{\rsh} ~.~ x \MEQ{\rsh} y \\
 ~\Leftrightarrow~ & \left[ \text{Definition of}~\ESTPROP(\alpha) \right] \\
                   & \rsh \in \ESTPROP(\alpha).
\end{align*}
\qed
\end{proof}
\begin{lemma} \label{thm:zoo:establishment:compositionality}
  Let  $\sh \in \SL{}{\SRDCLASSFV{\alpha}}$ with $\NOCALLS{\sh} = m$
  and $\rsh_1,\ldots,\rsh_m \in \RSL{}{\SRDCLASSFV{\alpha}}$. Then
  \begin{align*}
                   & \CHECK(\sh[\PS_1^{\sh} / \rsh_1, \ldots, \PS_m^{\sh} / \rsh_m], \EMPTYSEQ) = 1 \\
      ~\text{iff}~ & \CHECK(\sh,\T{p}) \\
                   & \text{and}~ \CHECK(\rsh_1,\EMPTYSEQ) = \ldots = \CHECK(\rsh_m,\EMPTYSEQ) = 1. 
  \end{align*}
\end{lemma}
The proof of Lemma~\ref{thm:zoo:establishment:compositionality} relies
on a technical observation, which intuitively states that
equalities between variables belonging to different nodes of an unfolding tree have to be propagated through parameters.
Formally,
%
\begin{lemma} \label{obs:zoo:establishment:equality-propagation}
 Let  $\sh \in \SL{}{\SRDCLASSFV{\alpha}}$ with $\NOCALLS{\sh} = m$,  $\rsh_1,\ldots,\rsh_m \in \RSL{}{\SRDCLASSFV{\alpha}}$ and
 $\rsh = \sh[\PS_1^{\sh} / \rsh_1, \ldots, \PS_m^{\sh} / \rsh_m]$.
 Moreover, for some $1 \leq i \leq m$, let $x \in \VAR(\rsh_i[\FV{0}{\rsh_i} / \FV{i}{\sh}])$ and $y \in \VAR(\rsh) \setminus \VAR(\rsh_i[\FV{0}{\rsh_i} / \FV{i}{\sh}])$. Then
 $
    x \MEQ{\rsh} y \quad\text{iff}\quad \exists z \in \FV{i}{\sh} ~.~ x \MEQ{\rsh_i[\FV{0}{\rsh_i} / \FV{i}{\sh}]} z ~\text{and}~ z \MEQ{\rsh} y.
 $
\end{lemma}
\begin{proof}[sketch]
  The direction from right to left is straightforward.
  For the converse direction, observe that, by Lemma~\ref{thm:zoo:relationships:characterizations},
  a pure formula $x = y$ is an element of the closure of pure formulas of $\rsh$.
  However, such an equality cannot be an element of the closures of $\PURE{\sh},\PURE{\rsh_1},\ldots,\PURE{\rsh_m}$, because, by assumption,
  $x$ and $y$ are not both contained in the set of variables of these symbolic heaps.
  Thus, since the only variables shared by $\sh$ and $\rsh_1,\ldots,\rsh_m$ are the parameters of the predicate calls of $\sh$,
  there exists a parameter $z$ of a suitable predicate call such that $x = z$
  is contained in one of the aforementioned closures of pure formulas.
  \qed
\end{proof}
Note that $\NIL$ is always assumed to be a free variable and no other constants occur in our fragment of symbolic heaps.
Otherwise, the observation from above would be wrong.
%

%
\begin{proof}[of Lemma~\ref{thm:zoo:establishment:compositionality}]
  Let $\rsh = \sh[\PS_1^{\sh} / \rsh_1, \ldots, \PS_m^{\sh} / \rsh_m]$.
  Recall that $\SHRINK{\rsh,\EMPTYSEQ}$ denotes the reduced symbolic heap
  introduced in Definition~\ref{def:zoo:track-automaton}.
  Then:
  \begin{align*}
                       & \CHECK(\rsh,\EMPTYSEQ) = 1 \\
     ~\Leftrightarrow~ & \left[ \text{Definition of}~\CHECK,~\SHRINK{\rsh,\EMPTYSEQ} = \rsh \right] \\
                       & \forall y \in \VAR(\rsh) ~.~  y \in \ALLOC{\rsh} ~\text{or}~ \exists x \in \FV{0}{\rsh} ~.~ x \MEQ{\rsh} y \\
     ~\Leftrightarrow~ & \left[ \VAR(\rsh) = \VAR(\sh) \cup \bigcup_{1 \leq i \leq m} \VAR(\rsh_i\left[ \FV{0}{\rsh_i} / \FV{i}{\sh} \right]) \right] \\
                       & \forall y \in \VAR(\sh) ~.~  y \in \ALLOC{\rsh} ~\text{or}~ \exists x \in \FV{0}{\rsh} ~.~ x \MEQ{\rsh} y \\
                       & \quad \text{and}~ \forall 1 \leq i \leq m ~.~ \forall y \in \VAR(\rsh_i\left[ \FV{0}{\rsh_i} / \FV{i}{\sh} \right])  ~.~ \\
                       & \qquad y \in \ALLOC{\rsh} ~\text{or}~ \exists x \in \FV{0}{\rsh} ~.~ x \MEQ{\rsh} y \\
     ~\Leftrightarrow~ & \left[ \text{Lemma}~\ref{thm:zoo:track:congruence}~\text{applied to}~y \in \VAR(\sh) \text{and}~ x \in \VAR(\FV{0}{\rsh}) \right] \\
                       & \forall y \in \VAR(\sh) ~.~  y \in \ALLOC{\SHRINK{\sh,\T{p}}} \\
                       & \qquad \text{or}~ \exists x \in \FV{0}{\sh} ~.~ x \MEQ{\SHRINK{\sh,\T{p}}} y \\
                       & \quad \text{and}~ \forall 1 \leq i \leq m ~.~ \forall y \in \VAR(\rsh_i\left[ \FV{0}{\rsh_i} / \FV{i}{\sh} \right]) ~.~ \\
                       & \qquad y \in \ALLOC{\rsh} ~\text{or}~ \exists x \in \FV{0}{\rsh} ~.~ x \MEQ{\rsh} y \\
     ~\Leftrightarrow~ & \left[ \text{Definition of}~\CHECK \right] \\
                       & \CHECK(\SHRINK{\sh,\T{p}}) = 1 \\
                       & ~\text{and}~ \forall 1 \leq i \leq m ~.~ \forall y \in \VAR(\rsh_i\left[ \FV{0}{\rsh_i} / \FV{i}{\sh} \right]) ~.~ \\
                       & \qquad y \in \ALLOC{\rsh} ~\text{or}~ \exists x \in \FV{0}{\rsh} ~.~ x \MEQ{\rsh} y \\
     ~\Leftrightarrow~ & \big[ y \in \ALLOC{\rsh} ~\text{iff}~ y \in \ALLOC{\rsh_i\left[ \FV{0}{\rsh_i} / \FV{i}{\sh} \right]} \\
                       & \qquad \text{or}~ \exists x \in \ALLOC{\rsh} ~.~ y \MEQ{\rsh} x \big] \\
                       & \CHECK(\SHRINK{\sh,\T{p}}) = 1 \\
                       & ~\text{and}~ \forall 1 \leq i \leq m ~.~ \forall y \in \VAR(\rsh_i\left[ \FV{0}{\rsh_i} / \FV{i}{\sh} \right]) ~.~ \\
                       & \qquad \left( y \in \ALLOC{\rsh_i\left[ \FV{0}{\rsh_i} / \FV{i}{\sh} \right]} ~\text{or}~ \exists x \in \ALLOC{\rsh} ~.~ y \MEQ{\rsh} x\right)  \\
                       & \qquad \text{or}~ \exists x \in \FV{0}{\rsh} ~.~ x \MEQ{\rsh} y \\
     ~\Leftrightarrow~ & \left[ \text{Lemma}~\ref{obs:zoo:establishment:equality-propagation} \right] \\
                       & \CHECK(\SHRINK{\sh,\T{p}}) = 1 \\
                       & ~\text{and}~ \forall 1 \leq i \leq m ~.~ \forall y \in \VAR(\rsh_i\left[ \FV{0}{\rsh_i} / \FV{i}{\sh} \right]) ~.~ \\
                       & \qquad y \in \ALLOC{\rsh_i} \\
                       & \qquad \quad \text{or}~ \exists z \in \FV{i}{\sh} . \exists x \in \ALLOC{\rsh} ~.~ y \MEQ{\rsh_i} z ~\text{and}~ z \MEQ{\rsh} x  \\
                       & \qquad \text{or}~ \exists x \in \FV{i}{\sh} . \exists z \in \FV{0}{\rsh} ~.~ y \MEQ{\rsh_i} x ~\text{and}~ x \MEQ{\rsh} z \\
     ~\Leftrightarrow~ & \big[ \FV{0}{\rsh_i}~\text{is substituted by}~\FV{i}{\sh} \in \VAR(\sh)~\text{in}~\rsh \\
                       & \qquad \text{which are all established due to $\CHECK(\SHRINK{\sh,\T{p}}) = 1$} \big] \\
                       & \CHECK(\SHRINK{\sh,\T{p}}) = 1 \\
                       & ~\text{and}~ \forall 1 \leq i \leq m ~.~ \forall y \in \VAR(\rsh_i) ~.~ \\
                       & \qquad y \in \ALLOC{\rsh_i} ~\text{or}~ \exists z \in \FV{0}{\rsh_i} ~.~ y \MEQ{\rsh_i} z \\
     ~\Leftrightarrow~ & \left[ \text{Definition of}~\CHECK \right] \\
                       & \CHECK(\SHRINK{\sh,\T{p}}) = 1 \\
                       & ~\text{and}~ \CHECK(\rsh_1,\EMPTYSEQ) = \ldots = \CHECK(\rsh_m,\EMPTYSEQ) = 1. 
  \end{align*}
  \qed
\end{proof}
%
%
%
\section{Proof of Lemma~\ref{thm:zoo:establishment:lower}} \label{app:zoo:establishment:lower}
 Let $(\SRD,\PS) \in \COMPLEMENT{SL-RSAT}$ be an instance of the complement of the reduced satisfiability problem.
 Moreover, consider the instance $(\SRD,\sh)$ of the establishment problem, where
 \begin{align*}
   \sh ~\DEFEQ~ \exists \T{z} z' ~.~ \PS\T{z} : \{ x = \NIL, z' \neq \NIL \}
 \end{align*}
 and $x$ is the single free variable (other than $\NIL$) of $\sh$.
 Then
 \begin{align*}
                      & \CALLSEM{\sh}{\SRD} \subseteq \ESTPROP(\alpha) \\
    ~\Leftrightarrow~ & \left[ A \subseteq B ~\text{iff}~ \forall x \in A . x \in B \right] \\
                      & \forall \rsh \in \CALLSEM{\sh}{\SRD} ~.~ \rsh \in \ESTPROP(\alpha) \\
    ~\Leftrightarrow~ & \left[ \text{Definition of}~\ESTPROP(\alpha) \right] \\
                      & \forall \rsh \in \CALLSEM{\sh}{\SRD} ~.~ \forall y \in \VAR(\rsh) ~.~ \\
                      & \qquad y \in \ALLOC{\rsh} ~\text{or}~ \exists x \in \FV{0}{\rsh} ~.~ y \MEQ{\rsh} x \\
    ~\Leftrightarrow~ & \left[ \SRD~\text{contains no points-to assertions} \right] \\
                      & \forall \rsh \in \CALLSEM{\sh}{\SRD} ~.~ \forall y \in \VAR(\rsh) ~.~ \exists x \in \FV{0}{\rsh} ~.~ y \MEQ{\rsh} x \\
    ~\Leftrightarrow~ & \left[ \PROJ{\FV{0}{\rsh}}{1} \MEQ{\rsh} \NIL \right] \\
                      & \forall \rsh \in \CALLSEM{\sh}{\SRD} ~.~ \forall y \in \VAR(\rsh) ~.~ y \MEQ{\rsh} \NIL \\
    ~\Leftrightarrow~ & \left[ z' \MNEQ{\rsh} \NIL \right] \\
                      & \forall \rsh \in \CALLSEM{\sh}{\SRD} ~.~ \NIL \MNEQ{\rsh} \NIL \\
    ~\Leftrightarrow~ & \left[ \NIL \MNEQ{\rsh} \NIL ~\text{iff}~ \rsh~\text{unsatisfiable} \right] \\
                      & \forall \rsh \in \CALLSEM{\sh}{\SRD} ~.~ \rsh~\text{is unsatisfiable} \\
    ~\Leftrightarrow~ & \left[ \forall x . \neg A \equiv \neg \exists x . A \right] \\
                      & \text{not}~\exists \rsh \in \CALLSEM{\sh}{\SRD} ~.~ \rsh~\text{if satisfiable} \\
    ~\Leftrightarrow~ & \left[ \text{Definition of satisfiability} \right] \\
                      & \sh~\text{is unsatisfiable}.
 \end{align*}
 Then it remains to show that $\sh$ is unsatisfiable if and only if $\PS$ is unsatisfiable:
 \begin{align*}
                     & \exists (\stack,\heap) \in \STATES ~.~ (\stack,\heap) \SAT{\SRD} \sh \\
   ~\Leftrightarrow~ & \left[ \text{Applying}~(\spadesuit)~\text{see below} \right] \\
                     & \exists (\stack,\heap) \in \STATES ~.~ \exists \T{u} \in \VAL^{\SIZE{\T{z}}} ~.~ \\
                     & \qquad (\restr{\stack[\T{z} \mapsto \T{u}]}{\T{z}},\heap) \SAT{\SRD} \PS\T{z} \\
                     & \qquad \text{and}~ \stack(x) = \NIL ~\text{and}~ \stack(z') \neq \NIL \\
   ~\Leftrightarrow~ & \left[ \DOM(\restr{\stack[\T{z} \mapsto \T{z}]}{\T{z}}) = \T{z} \right] \\
                     & \exists (\stack,\heap) \in \STATES ~.~ (\stack,\heap) \SAT{\SRD} \PS\T{z}.
 \end{align*}
 Here, the missing step marked with ($\spadesuit$) corresponds to the following property:
 \begin{align*}
                      & \stack,\heap \SAT{\SRD} \sh \\
   ~\Leftrightarrow~  & \exists \T{u} \in \VAL^{\SIZE{\T{z}}} ~.~ \tag{$\spadesuit$} \\
                      & \qquad (\restr{\stack[\T{z} \mapsto \T{u}]}{\T{z}}),\heap \SAT{\SRD} \PS\T{z} \\
                      & \qquad \text{and}~ \stack(x) = \NIL ~\text{and}~ \stack(z') \neq \NIL.
 \end{align*}
 To complete the proof, let $(\stack,\heap) \in \STATES$. Then:
 \begin{align*}
                      & \stack,\heap \SAT{\SRD} \sh \\
    ~\Leftrightarrow~ & \left[ \text{Construction of}~\sh \right] \\
                      & \stack,\heap \SAT{\SRD}  \exists \T{z} z' ~.~ \PS\T{z} : \{ x = \NIL, z' \neq \NIL \} \\
    ~\Leftrightarrow~ & \left[ \text{SL semantics} \right] \\
                      & \exists \T{u} \in \VAL^{\SIZE{\T{z}}} . \exists v \in \VAL ~.~ \\
                      & \qquad \stack[\T{z} \mapsto \T{u}, z' \mapsto v],\heap \SAT{\SRD} \PS\T{z} \\
                      & \qquad \text{and}~ \stack[\T{z} \mapsto \T{u}, z' \mapsto v],\heap \SAT{\SRD} x = \NIL \\
                      & \qquad \text{and}~ \stack[\T{z} \mapsto \T{u}, z' \mapsto v],\heap \SAT{\SRD} z' \neq \NIL \\
    ~\Leftrightarrow~ & \left[ \text{Lemma}~\ref{thm:symbolic-heaps:fv-coincidence} \right] \\
                      & \exists \T{u} \in \VAL^{\SIZE{\T{z}}} . \exists v \in \VAL ~.~ \\
                      & \qquad (\restr{\stack[\T{z} \mapsto \T{u}, z' \mapsto v]}{\T{z}}),\heap \SAT{\SRD} \PS\T{z} \\
                      & \qquad \text{and}~ (\restr{\stack[\T{z} \mapsto \T{u}, z' \mapsto v]}{x}),\heap \SAT{\SRD} x = \NIL \\
                      & \qquad \text{and}~ (\restr{\stack[\T{z} \mapsto \T{u}, z' \mapsto v]}{z'}),\heap \SAT{\SRD} z' \neq \NIL \\
    ~\Leftrightarrow~ & \left[ \text{SL semantics} \right] \\
                      & \exists \T{u} \in \VAL^{\SIZE{\T{z}}} ~.~ \\
                      & \qquad (\restr{\stack[\T{z} \mapsto \T{u}]}{\T{z}}),\heap \SAT{\SRD} \PS\T{z} \\
                      & \qquad \text{and}~ \stack(x) = \NIL ~\text{and}~ \stack(z') \neq \NIL,
 \end{align*}
 which coincides with $(\spadesuit)$.
 \qed
%
%
\section{Proof of Theorem~\ref{thm:zoo:reachability:property}}
\label{app:zoo:reachability:property}
We have to construct a heap automaton
$\HAREACH$ over $\SHCLASSFV{\alpha}$ that satisfies the compositionality property
and accepts $\RPROP(\alpha,R)$.
In order to highlight the necessary proof obligations, the actual construction of $\HAREACH$
and its correctness proof are splitted into several definitions and lemmas that are
provided afterwards.
The construction of $\HAREACH$ was already sketched in the paper.
%
A formal construction is provided in Definition~\ref{def:zoo:reachability-automaton}.
It then remains to show the correctness of our construction of $\HAREACH$:
\begin{itemize}
  \item Lemma~\ref{thm:zoo:reachability:language} shows that $\HAREACH$ indeed
        accepts $\RPROP(\alpha,R)$, i.e., $L(\HAREACH) = \RPROP(\alpha,R)$.
  \item In order to prove that $\HAREACH$ satisfies the compositionality property,
        we lift Lemma~\ref{thm:zoo:track:congruence} to cover reachability as well.
        After that the compositionality property of $\HAREACH$ is verified analogously
        to the compositionality property of $\HATRACK$ (cf. Lemma~\ref{thm:zoo:track:compositional}).

        The lifting of Lemma~\ref{thm:zoo:track:congruence} is presented in Lemma~\ref{thm:zoo:reachability:congruence}.
        Similar to the proof of Lemma~\ref{thm:zoo:track:congruence:auxiliary},
        the proof of Lemma~\ref{thm:zoo:reachability:congruence} relies on an
        auxiliary property showing that
        \[ \REACH{x}{y}{\sh[\PS_1 / \rsh]} \quad\text{iff} \quad  \REACH{x}{y}{\sh\left[ \PS_1 / \SSIGMA{\CALLN{1}{}}{(B,\Lambda,S)} \right]} \]
        holds for all $x,y \in \VAR(\sh)$ and symbolic heaps $\sh$ containing a single predicate call $\PS_1$.
        This is formalized in Lemma~\ref{thm:zoo:reachability:congruence:auxiliary}.
\end{itemize}
Putting both together, we obtain a heap automaton $\HAREACH$ over $\SHCLASSFV{\alpha}$ accepting
$\RPROP(\alpha,R)$. 
\qed
%
%

%
%
\begin{definition} \label{def:zoo:reachability-automaton}
 Let $\FV{0}{}$ be a tuple of variables with $\NOFV{} = \alpha \in \POSN$.
 Then $\HAREACH = (Q,\SRDCLASSFV{\alpha},\Delta,F)$ is given by
%
  \begin{align*}
      Q ~\DEFEQ~ & Q_{\HAREACH} ~\times~ 2^{\FV{0}{} \times \FV{0}{}}, \qquad \quad F ~\DEFEQ~ Q_{\HATRACK} \times \{ R \} \\
      \Delta ~~:~~ & \MOVE{\HAREACH}{(q_0,S_0)}{\sh}{\T{p}},~ \T{p} = (q_1,S_1) \ldots (q_m,S_m) \\
    ~\text{iff}~ & \MOVE{\HATRACK}{q_0}{\varphi}{q_1 \ldots q_m} \\
                 & \text{and}~ \forall u,v \in \FV{0}{} ~.~
                     (u,v) \in S_0 \leftrightarrow \REACH{u^{\sh}}{v^{\sh}}{\SHRINK{\sh,\T{p}}}~,
  \end{align*}
  where
  \begin{align*}
     \SHRINK{\sh,\T{p}} ~\DEFEQ~ & \sh\left[\PS_1 / \SSIGMA{\CALLN{1}{}}{\T{p}[1]}, \ldots, \PS_m / \SSIGMA{\CALLN{m}{}}{\T{p}[m]} \right].
  \end{align*}
Here, $m = \NOCALLS{\sh}$ stands for the number of predicate calls occurring in $\sh$ and $u^{\sh}$ denotes the free variable of $\sh$ corresponding to $u \in \FV{0}{}$.\footnote{formally if $u = \PROJ{\FV{0}{}}{i}$ then $u^{\sh} = \PROJ{\FV{0}{\sh}}{i}$}
Moreover, $q_i = (A_i,\Pi_i) \in Q_{\HATRACK}$ for each $0 \leq i \leq m$.
\end{definition}
\begin{lemma} \label{thm:zoo:reachability:language}
  $L(\HAREACH) = \RPROP(\alpha,R)$.
\end{lemma}
\begin{proof}
  Let $\rsh \in \RSL{}{\SRDCLASSFV{\alpha}}$. Then
  \begin{align*}
                        & \rsh \in L(\HAREACH) \\
      ~\Leftrightarrow~ & \left[ \text{Definition of}~L(\HAREACH) \right] \\
                        & \exists q \in F_{\HAREACH} ~.~  \OMEGA{\HAREACH}{q}{\rsh} \\
      ~\Leftrightarrow~ & \left[ \text{Definition of}~F_{\HAREACH} \right] \\
                        & \exists p \in Q_{\HATRACK} ~.~ \OMEGA{\HAREACH}{(p,R)}{\rsh} \\
      ~\Leftrightarrow~ & \left[ \text{Definition of}~\Delta_{\HAREACH} \right] \\
                        & \exists p \in Q_{\HATRACK} ~.~ \OMEGA{\HATRACK}{p}{\rsh} \\
                        & \qquad \text{and}~ \forall u,v \in \FV{0}{} ~.~ (u,v) \in R \leftrightarrow \REACH{u^{\rsh}}{v^{\rsh}}{\SHRINK{\rsh,\EMPTYSEQ}} \\
      ~\Leftrightarrow~ & \left[ \text{Definition}~\ref{def:zoo:reachability-automaton}:~\SHRINK{\rsh,\EMPTYSEQ} = \rsh \right] \\
                        & \forall u,v \in \FV{0}{} ~.~ (u,v) \in R \leftrightarrow \REACH{u^{\rsh}}{v^{\rsh}}{\rsh} \\
      ~\Leftrightarrow~ & \left[ \text{Definition of}~\RPROP(\alpha,R) \right] \\
                        & \rsh \in \RPROP(\alpha,R). 
  \end{align*}
  \qed
\end{proof}
\begin{lemma} \label{thm:zoo:reachability:congruence:auxiliary}
 Let $\sh = \exists \BV{} . \SPATIAL{} \SEP \CALLN{1}{} : \PURE{} \in \SL{}{\SRDCLASSFV{\alpha}}$.
 Moreover, let $\rsh \in \RSL{}{}$ with $\SIZE{\FV{0}{\rsh}} = \SIZE{\FV{1}{\sh}}$, $B = \{ y \in \FV{0}{} ~|~ y^{\rsh} \in \ALLOC{\rsh} \}$, $\Lambda = \{ x \sim y ~|~ x^{\rsh} \MSIM{\rsh} y^{\rsh} \}$, and
 $S = \{ (x,y) \in \FV{0}{} \times \FV{0}{} ~|~ \REACH{x^{\rsh}}{y^{\rsh}}{\rsh} \}$.
 Then, for each $x,y \in \VAR(\sh)$, we have
 \begin{align*}
    x \MSIM{\sh[\PS_1 / \rsh]} y \quad\text{iff}&\quad x \MSIM{\sh\left[ \PS_1 / \SSIGMA{\CALLN{1}{}}{(B,\Lambda,S)} \right]} y, \quad \text{and}~  \\
    x \in \ALLOC{\sh[\PS_1 / \rsh]} \quad\text{iff}&\quad x \in \ALLOC{\sh\left[ \PS_1 / \SSIGMA{\CALLN{1}{}}{(B,\Lambda,S)} \right]}, \quad \text{and}~  \\
    \REACH{x}{y}{\sh[\PS_1 / \rsh]} \quad\text{iff} &\quad  \REACH{x}{y}{\sh\left[ \PS_1 / \SSIGMA{\CALLN{1}{}}{(B,\Lambda,S)} \right]}~,
 \end{align*}
 where $\rsha_{\PS_1\FV{1}{},(B,\Lambda,S)}$ has been defined at the beginning of this section.
\end{lemma}
\begin{proof}
 The proof of the first two equivalences is completely analogous to the proof of Lemma~\ref{thm:zoo:track:congruence:auxiliary}.
 Hence, it remains to prove that for each pair of variables $x,y \in \VAR(\sh)$, we have
 \begin{align*}
  \REACH{x}{y}{\sh[\PS_1 / \rsh]} ~\text{iff}~  \REACH{x}{y}{\sh\left[ \PS_1 / \SSIGMA{\CALLN{1}{}}{(B,\Lambda,S)} \right]}.
 \end{align*}
 Since definite reachability always holds for unsatisfiable reduced symbolic heaps
 and, by the first equivalence, either both symbolic heaps are satisfiable or both are unsatisfiable,
 assume without loss of generality that $\sh[\PS_1 / \rsh]$ is satisfiable.
 Both directions of the proposition from above are shown by induction on the number $n$ of definite points-to assertions $\MPT{\rsh}$
 needed to reach $y$ from $x$ in $\REACH{x}{y}{\sh[\PS_1 / \rsh]}$
 and $\REACH{x}{y}{\sh\left[ \PS_1 / \SSIGMA{\CALLN{1}{}}{(B,\Lambda,S)} \right]}$, respectively.
 We first show that $\REACH{x}{y}{\sh[\PS_1 / \rsh]}$ implies $\REACH{x}{y}{\sh\left[ \PS_1 / \SSIGMA{\CALLN{1}{}}{(B,\Lambda,S)} \right]}$.
 \emph{I.B.}\qquad
 For the base case assume that $\REACH{x}{y}{\sh[\PS_1 / \rsh]}$ holds and $y$ is reachable by taking $n = 1$ definite points-to assertions, i.e., $x \MPT{\sh[\PS_1 / \rsh]} y$ holds.
 By Lemma~\ref{thm:zoo:relationships:characterizations}, this is equivalent to
 \begin{align*}
   & \exists z_1,z_2 ~.~ x \MEQ{\sh[\PS_1 / \rsh]} z_1 ~\text{and}~ y \MEQ{\sh[\PS_1 / \rsh]} z_2 \\
   & \qquad \qquad ~\text{and}~ \PT{z_1}{(\_,z_2,\_)} \in \SPATIAL{\sh[\PS_1 / \rsh]}.
 \end{align*}
 Then, by definition of predicate replacement, two cases arise:
 \begin{enumerate}
    \item Case: $\PT{z_1}{(\_,z_2,\_)} \in \SPATIAL{\sh}$.
          Then $z_1,z_2 \in \VAR(\sh)$ and thus, by the already known first property of Lemma~\ref{thm:zoo:reachability:congruence:auxiliary}, we have
          $x \MEQ{\sh\left[ \PS_1 / \SSIGMA{\CALLN{1}{}}{(B,\Lambda,S)} \right]} z_1$ and $y \MEQ{\sh\left[ \PS_1 / \SSIGMA{\CALLN{1}{}}{(B,\Lambda,S)} \right]} z_2$.
          Since $\SPATIAL{\sh}$ is contained in the symbolic heap $\SPATIAL{\sh\left[ \PS_1 / \SSIGMA{\CALLN{1}{}}{(B,\Lambda,S)} \right]}$,
          this is means that
          \[ x \MPT{\sh\left[ \PS_1 / \SSIGMA{\CALLN{1}{}}{(B,\Lambda,S)} \right]} y. \]
    %
    \item Case: $\PT{z_1}{(\_,z_2,\_)} \in \SPATIAL{\rsh}$.
          Then $z_1,z_2 \in \VAR(\rsh)$.
          By Lemma~\ref{obs:zoo:establishment:equality-propagation}, there exist two variables $u,v \in \FV{1}{\sh}$ such that
          $x \MEQ{\sh[\PS_1 / \rsh]} u$, $u \MEQ{\rsh} z$, $y \MEQ{\sh[\PS_1 / \rsh]} v$, and $v \MEQ{\rsh} z'$.
          Thus, $u \MPT{\rsh} v$ holds.
          In particular, this means that $(u,v) \in S$ and thus $\PT{u}{(\_,v,\_)} \in \SPATIAL{\sh\left[ \PS_1 / \SSIGMA{\CALLN{1}{}}{(B,\Lambda,S)} \right]}$.
          Now, by the already known first property of Lemma~\ref{thm:zoo:reachability:congruence:auxiliary}, we also have
          $x \MEQ{\sh\left[ \PS_1 / \SSIGMA{\CALLN{1}{}}{(B,\Lambda,S)} \right]} u$ and $y \MEQ{\sh\left[ \PS_1 / \SSIGMA{\CALLN{1}{}}{(B,\Lambda,S)} \right]} v$.
          Hence, $x \MPT{\sh\left[ \PS_1 / \SSIGMA{\CALLN{1}{}}{(B,\Lambda,S)} \right]} y$.
 \end{enumerate}
 Thus, $\REACH{x}{y}{\sh\left[ \PS_1 / \SSIGMA{\CALLN{1}{}}{(B,\Lambda,S)} \right]}$ holds in both cases.
 \emph{I.H.}\qquad
 Assume for an arbitrary, but fixed, natural number $n$ that $\REACH{x}{y}{\sh[\PS_1 / \rsh]}$, where at most $n$ points-to assertions are used to reach $y$ from $x$, implies $\REACH{x}{y}{\sh\left[ \PS_1 / \SSIGMA{\CALLN{1}{}}{(B,\Lambda,S)} \right]}$.
 \\
 \emph{I.S.}\qquad
 Assume $n+1$ definite points-to assertions are needed to reach $y$ from $x$ in $\sh[\PS_1 / \rsh]$.
 By Lemma~\ref{thm:zoo:relationships:characterizations},
 there exists $z \in \VAR(\sh[\PS_1 / \rsh])$ such that
 $x \MPT{\sh[\PS_1 / \rsh]} z$ and $\REACH{z}{y}{\sh[\PS_1 / \rsh}$,
 where at most $n$ definite points-to assertions are needed to reach $y$ from $z$.
 Thus, by I.H. we know that $\REACH{z}{y}{\sh\left[ \PS_1 / \SSIGMA{\CALLN{1}{}}{(B,\Lambda,S)} \right]}$.
 Furthermore, by the same argument as in the base case, we obtain
 $x \MPT{\sh\left[ \PS_1 / \SSIGMA{\CALLN{1}{}}{(B,\Lambda,S)} \right]} z$.
 %
 Putting both together yields $\REACH{x}{y}{\sh\left[ \PS_1 / \SSIGMA{\CALLN{1}{}}{(B,\Lambda,S)} \right]}$.
The proof of the converse direction is analogous.
The only difference is that we obtain $\REACH{x}{y}{\sh[\PS_1 / \rsh]}$ instead of
$x \MPT{\sh[\PS_1 / \rsh]} y$
in the case analysis of the base case, because $\PT{u}{(\_,v\_,)} \in \SPATIAL{\SSIGMA{\CALLN{1}{}}{(B,\Lambda,S)}}$
implies $\REACH{u}{v}{\rsh}$ only.
\qed
\end{proof}
\begin{lemma} \label{thm:zoo:reachability:congruence}
 Let $\sh \in \SL{}{\SRDCLASSFV{\alpha}}$ with $\NOCALLS{\sh} = m \geq 0$.
 For each $1 \leq i \leq m$,
 let $\rsh_i \in \RSL{}{\SRDCLASSFV{\alpha}}$ with $\NOFV{\rsh_i} = \SIZE{\FV{i}{\sh}}$,
 $A_i = \{ y \in \FV{0} ~|~ y^{\rsh_i} \in \ALLOC{\rsh_i} \}$,
 $\PURE{}_i = \{ x \sim y ~|~ x^{\rsh_i} \MSIM{\rsh_i} y^{\rsh_i} \}$
 and $S_i  = \{ (x,y) \in \FV{0}{} \times \FV{0}{} ~|~ \REACH{x^{\rsh_i}}{y^{\rsh_i}}{\rsh_i} \}$.
 Moreover, let
 \begin{align*}
  \rsh ~\DEFEQ~ & \sh\left[\PS_1 / \rsh_1, \ldots, \PS_m / \rsh_m\right], \\
  \T{p} ~\DEFEQ~ & (A_1,\Pi_1,S_1) \ldots (A_m,\Pi_m,S_m), ~\text{and} \\
  \SHRINK{\sh,\T{p}} ~\DEFEQ~ & \sh\left[\PS_1 / \SSIGMA{\CALLN{1}{}}{\T{p}[1]},
                   \ldots,
                   \PS_m / \SSIGMA{\CALLN{m}{}}{\T{p}[m]}
                   \right].
 \end{align*}
 Then, for each $x,y \in \VAR(\sh)$, we have
 \begin{align*}
   & x ~\MSIM{\rsh}~ y ~\text{iff}~  x ~\MSIM{\SHRINK{\sh,\T{p}}}~ y, \\ 
   & x \in \ALLOC{\rsh} ~\text{iff}~ x \in \ALLOC{\SHRINK{\sh,\T{p}}}, \\
   & \REACH{x}{y}{\rsh} ~\text{iff}~ \REACH{x}{y}{\SHRINK{\sh,\T{p}}}. 
 \end{align*}
\end{lemma}
\begin{proof}
  Analogous to the proof of Lemma~\ref{thm:zoo:track:congruence}
  except for the use of
  Lemma~\ref{thm:zoo:reachability:congruence:auxiliary} instead of
  Lemma~\ref{thm:zoo:track:congruence:auxiliary}.
  \qed
\end{proof}
%
%
\section{Proof of Theorem~\ref{thm:zoo:reachability:complexity}} \label{app:zoo:reachability:complexity}
   Upper bounds are obtained as in the proof of Theorem~\ref{thm:zoo:establishment:complexity}.
   Thus,  we only show that the reduction provided in the paper proving the lower bounds is correct.
   We first briefly recall how \COMPLEMENT{SL-RSAT} is reduced to \DPROBLEM{SL-REACH}.
   Let $(\SRD,\PS)$ be an instance of \COMPLEMENT{SL-RSAT}.
   Then an instance of \DPROBLEM{SL-REACH} is given by $(\SRD,\sh,\PROJ{\FV{0}{}}{1},\PROJ{\FV{0}{}}{2})$,  where

   $ 
    \sh ~\DEFEQ~ \exists \BV{} ~.~
        \PT{\PROJ{\FV{0}{}}{1}}{\NIL} \SEP \PS\BV{} ~:~ \{ \PROJ{\FV{0}{}}{2} \neq \NIL \}.
   $ 
   Now,
   \begin{align*}
                      & \forall \rsh \in \CALLSEM{\sh}{\SRD} ~.~
                            \REACH{\PROJ{\FV{0}{}}{1}}{\PROJ{\FV{0}{}}{2}}{\rsh} \\
    ~\Leftrightarrow~ & \left[ \text{Construction of}~\sh,
                       (\rsh~\text{contains one points-to assertion}) \right] \\
                      & \forall \rsh \in \CALLSEM{\sh}{\SRD} ~.~
                          \PROJ{\FV{0}{}}{1} \MPT{\rsh} \PROJ{\FV{0}{}}{2} \\
    ~\Leftrightarrow~ & \left[ \text{Definition of}~\MPT{\rsh} \right] \\
                      & \forall \rsh \in \CALLSEM{\sh}{\SRD} ~.~
                          \forall (\stack,\heap) \in \MODELS{\rsh} ~.~
                                \stack(\PROJ{\FV{0}{}}{2}) \in \heap(\stack(\PROJ{\FV{0}{}}{1})) \\
    ~\Leftrightarrow~ & \left[
                          \forall (\stack,\heap) \in \MODELS{\rsh} ~.~
                            \heap(\stack(\PROJ{\FV{0}{}}{1})) = \NIL
                        \right] \\
                      & \forall \rsh \in \CALLSEM{\sh}{\SRD} ~.~
                          \forall (\stack,\heap) \in \MODELS{\rsh} ~.~ \\
                      & \qquad \stack(\PROJ{\FV{0}{}}{2}) \in \heap(\stack(\PROJ{\FV{0}{}}{1}))
                            ~\text{and}~ \stack(\PROJ{\FV{0}{}}{2}) = \NIL \\
    ~\Leftrightarrow~ & \left[ \PROJ{\FV{0}{}}{2} \neq \NIL \in \PURE{\sh} \right] \\
                      & \forall \rsh \in \CALLSEM{\sh}{\SRD} ~.~
                          \forall (\stack,\heap) \in \MODELS{\rsh} ~.~ \\
                      & \qquad \stack(\PROJ{\FV{0}{}}{2}) \in \heap(\stack(\PROJ{\FV{0}{}}{1}))
                            ~\text{and}~ \stack(\PROJ{\FV{0}{}}{2}) = \NIL \\
                      & \qquad\text{and}~ \stack(\PROJ{\FV{0}{}}{2}) \neq \NIL \\
    ~\Leftrightarrow~ & \left[ \NIL = \stack(\PROJ{\FV{0}{}}{2}) \neq \NIL
                         ~\text{iff}~ \MODELS{\rsh} = \emptyset \right] \\
                      & \forall \rsh \in \CALLSEM{\sh}{\SRD} ~.~
                            \MODELS{\rsh} = \emptyset \\
    ~\Leftrightarrow~ & \left[
                         \text{Lemma}~\ref{thm:symbolic-heaps:fv-coincidence}
                        \right] \\
                      & \MODELS{\sh} = \emptyset \\
    ~\Leftrightarrow~ & \left[ \text{Definition of satisfiability} \right] \\
                      & \sh~\text{unsatisfiable}.
   \end{align*}
   Since $\sh$ is satisfiable if and only if $\PS$ is satisfiable by an analogous argument as in the proof of Lemma~\ref{thm:zoo:establishment:lower},
   it follows that
   $\REACH{\PROJ{\FV{0}{}}{1}}{\PROJ{\FV{0}{}}{2}}{\sh}$ holds for each $\rsh \in \CALLSEM{\sh}{\SRD}$
   if and only if $\PS$ is unsatisfiable.
   \qed
%
%
\section{Proof of Lemma~\ref{thm:zoo:garbage:property}} \label{app:zoo:garbage:property}
We have to construct a heap automaton
$\HAGARBAGE$ over $\SHCLASSFV{\alpha}$ that satisfies the compositionality property
and accepts $\GARBAGEPROP(\alpha)$.
In order to highlight the necessary proof obligations, the actual construction of $\HAGARBAGE$
and its correctness proof are splitted into several definitions and lemmas that are
provided afterwards.
The construction of $\HAGARBAGE$ is similar to the construction of $\HAEST$
as presented in Theorem~\ref{thm:zoo:establishment}.
The main difference is that the reachability automaton $\HAREACH$,
formally introduced in Definition~\ref{def:zoo:reachability-automaton},
is used instead of the tracking automaton $\HATRACK$
and that the predicate $\CHECK$ is adapted.
A fully formal construction of $\HAGARBAGE$ is found in Definition~\ref{def:zoo:garbage:automaton}.
It then remains to show the correctness of our construction of $\HAGARBAGE$:
\begin{itemize}
  \item Lemma~\ref{thm:zoo:garbage:language} establishes that $\HAGARBAGE$ indeed accepts $\GARBAGEPROP(\alpha)$, i.e., $L(\HAGARBAGE) = \GARBAGEPROP(\alpha)$.
  \item To prove the compositionality property, we show that $\HAGARBAGE$ is an instance of a more general
        construction scheme whose compositionality property is shown in
        Lemma~\ref{thm:zoo:scheme:compositionality}.
        In order to apply Lemma~\ref{thm:zoo:scheme:compositionality}, we have to show that
        \begin{align*}
                          & \CHECK(\sh[\PS_1^{\sh} / \rsh_1, \ldots, \PS_m^{\sh} / \rsh_m], \EMPTYSEQ) = 1 \\
             ~\text{iff}~ & \CHECK(\sh,p_1 \ldots p_m) = 1 \\
                          & \text{and}~ \CHECK(\rsh_1,\EMPTYSEQ) = \ldots = \CHECK(\rsh_m,\EMPTYSEQ) = 1.
         \end{align*}
         This is verified in Lemma~\ref{thm:zoo:garbage:compositionality}.
         Then, by Lemma~\ref{thm:zoo:scheme:compositionality}, we know that 
         $\HAGARBAGE = \HASCHEME{\HAREACH}{\CHECK}{F}$ satisfies the compositionality property.
\end{itemize}
Putting both together, we obtain a heap automaton $\HAGARBAGE$ over $\SHCLASSFV{\alpha}$
that satisfies the compositionality property and accepts $\GARBAGEPROP(\alpha)$.
\qed
\begin{definition} \label{def:zoo:garbage:automaton}
$\HAGARBAGE = (Q,\SHCLASSFV{\alpha},\Delta,F)$ is given by
\begin{align*}
   & Q ~\DEFEQ~ Q_{\HAREACH} \times \{0,1\} \qquad F ~\DEFEQ~ Q_{\HAREACH} \times \{1\} \\
   &  \Delta ~~:~~ \MOVE{\HAGARBAGE}{(p_0,q_0)}{\sh}{(p_1,q_1) \ldots (p_m,q_m)} \\
   & ~\text{iff}~ \MOVE{\HAREACH}{p_0}{\sh}{p_1\ldots p_m} \\
   & \qquad \text{and}~ q_0 = \min \{q_1,\ldots,q_m,\CHECK(\sh,p_1 \ldots p_m)\}~,
\end{align*}
where $\HAREACH$ is as in Definition~\ref{def:zoo:reachability-automaton}.
Furthermore, the predicate $\CHECK : \SHCLASSFV{\alpha} \times Q_{\HAREACH}^{*} \to \{0,1\}$
verifies that each variable of a symbolic heap $\sh$ is established in $\SHRINK{\sh,\T{p}}$,
where $\SHRINK{\sh,\T{p}}$ is the same as in the construction of $\HAREACH$.
Hence,
\begin{align*}
\CHECK(\sh,\T{p}) ~\DEFEQ~
 \begin{cases}
        1 &, ~\text{if}~ \forall y \in \VAR(\sh) \,.\, \exists x \in \FV{0}{\sh} ~.~ \\
          & \qquad        x \MEQ{\SHRINK{\sh,\T{p}}} y ~\text{or}~ \REACH{x}{y}{\SHRINK{\sh,\T{p}}} \\
        0 &, ~\text{otherwise}~,
 \end{cases}
\end{align*}
\end{definition}
\begin{lemma} \label{thm:zoo:garbage:language}
  $L(\HAGARBAGE) = \GARBAGEPROP(\alpha)$.
\end{lemma}
\begin{proof}
Let $\rsh \in \RSL{}{\SRDCLASSFV{\alpha}}$. Then:
\begin{align*}
                   & \rsh \in L(\HAGARBAGE) \\
 ~\Leftrightarrow~ & \left[ \text{Definition of}~L(\HAGARBAGE) \right] \\
                   & \exists q \in F_{\HAGARBAGE} ~.~ \OMEGA{\HAGARBAGE}{q}{\rsh} \\
 ~\Leftrightarrow~ & \left[ \text{Definition of}~F_{\HAGARBAGE},~q = (p,1) \right] \\
                   & \exists p \in Q_{\HAREACH} ~.~ \OMEGA{\HAGARBAGE}{(p,1)}{\rsh} \\
 ~\Leftrightarrow~ & \left[ \text{Definition of}~\Delta_{\HAGARBAGE} \right] \\
                   & \exists p \in Q_{\HAREACH} ~.~ \OMEGA{\HAREACH}{p}{\rsh}
                     ~\text{and}~ \CHECK(\SHRINK{\rsh,\EMPTYSEQ}) = 1 \\
 ~\Leftrightarrow~ & \left[ \text{Definition of}~\CHECK(\SHRINK{\rsh,\EMPTYSEQ}) \right] \\
                   & \exists p \in Q_{\HAREACH} ~.~ \OMEGA{\HAREACH}{p}{\rsh} \\
                   & \qquad \text{and}~ \forall y \in \VAR(\rsh) ~.~ \exists x \in \FV{0}{\rsh} ~.~ x \MEQ{\SHRINK{\rsh,\EMPTYSEQ}} y \\
                   & \qquad \quad \text{or}~ \REACH{x}{y}{\SHRINK{\rsh,\EMPTYSEQ}} \\
 ~\Leftrightarrow~ & \left[ \NOCALLS{\rsh} = 0 ~\text{implies}~ \rsh = \SHRINK{\rsh,\EMPTYSEQ} \right] \\
                   & \forall y \in \VAR(\rsh) ~.~ \exists x \in \FV{0}{\rsh} ~.~ x \MEQ{\rsh} y ~\text{or}~ \REACH{x}{y}{\rsh} \\
 ~\Leftrightarrow~ & \left[ \text{Definition of}~\GARBAGEPROP(\alpha) \right] \\
                   & \rsh \in \GARBAGEPROP(\alpha).
\end{align*}
\qed
\end{proof}
\begin{lemma} \label{thm:zoo:garbage:compositionality}
  Let  $\sh \in \SL{}{\SRDCLASSFV{\alpha}}$ with $\NOCALLS{\sh} = m$
  and $\rsh_1,\ldots,\rsh_m \in \RSL{}{\SRDCLASSFV{\alpha}}$. Then
  \begin{align*}
                   & \CHECK(\sh[\PS_1^{\sh} / \rsh_1, \ldots, \PS_m^{\sh} / \rsh_m], \EMPTYSEQ) = 1 \\
      ~\text{iff}~ & \CHECK(\sh,\T{p}) \\
                   & \text{and}~ \CHECK(\rsh_1,\EMPTYSEQ) = \ldots = \CHECK(\rsh_m,\EMPTYSEQ) = 1. 
  \end{align*}
\end{lemma}
The following technical observation, similar to Lemma~\ref{obs:zoo:establishment:equality-propagation}, is essential.
Intuitively, it states that a variable in one part of an unfolding is definitely reachable by a variable in a different part of an unfolding
only through a parameter (including $\NIL$) of one or more suitable predicate calls.
This is illustrated in Figure~\ref{fig:reach:propagation}.
\begin{figure}
\begin{center}
\begin{tikzpicture}[shorten >=1pt,->]
  \node (top) at (-0.3,0) {$\sh = \exists \BV{} ~.~ \SPATIAL{} ~ \SEP$};
  \node (call1) at (1,0) {$\CALLN{1}{}$};
  \node (sep)   at (2,0) {$\SEP~\ldots~\SEP$};
  \node (callm) at (3,0) {$\CALLN{m}{}$};

  \node (rsh1)  at (1, -2.3) {$\rsh_1$};
  \draw [fill=gray!25]  (1,-0.5) --  (0.3,-2) --  (1.7,-2) -- cycle;

  \node (rsh2)  at (3, -2.3) {$\rsh_m$};
  \draw [fill=gray!25]  (3,-0.5) --  (2.3,-2) --  (3.7,-2) -- cycle;

  \node (u) at (1,-1.5) {$u$};
  \node (v) at (3,-1.5) {$v$};

  \draw[->,decorate,decoration=snake] (u) -- (call1);
  \draw[->,decorate,decoration=snake] (callm) -- (v);
  \path (call1) edge[->,decorate,decoration=snake,bend left=30] (callm);
  \draw[->,decorate,decoration=snake,thick] (u) -- (v);
\end{tikzpicture}
\end{center}
\caption{Propagation of reachability through parameters in a symbolic heap $\rsh = \sh[\PS_1 / \rsh_1, \ldots, \PS_m / \rsh_m]$.
         Here, $\REACH{u}{v}{\rsh}$ holds (thick arrow). 
         Since $u \in \VAR(\rsh_1)$ and $v \in \VAR(\rsh_m)$, there exist parameters $x,y$ of $\PS_1$ and $\PS_m$ that
         $\REACH{u}{x}{\rsh_1}$, $\REACH{x}{y}{\rsh}$, and $\REACH{y}{v}{\rsh_m}$ (thin arrows).
        }
 \label{fig:reach:propagation}
\end{figure}

\begin{lemma} \label{obs:zoo:reachability:propagation}
 Let  $\sh \in \SL{}{\SRDCLASSFV{\alpha}}$ with $\NOCALLS{\sh} = m$,  $\rsh_1,\ldots,\rsh_m \in \RSL{}{\SRDCLASSFV{\alpha}}$ and
 $\rsh = \sh[\PS_1^{\sh} / \rsh_1, \ldots, \PS_m^{\sh} / \rsh_m]$.
 Moreover, for some $1 \leq i \leq m$, let $x \in \VAR(\rsh_i[\FV{0}{\rsh_i} / \FV{i}{\sh}])$ and $y \in \VAR(\rsh) \setminus \VAR(\rsh_i[\FV{0}{\rsh_i} / \FV{i}{\sh}])$.
 Then $\REACH{x}{y}{\rsh}$ holds if and only if there exists $z \in \FV{i}{\sh}$ such that
 \begin{itemize}
   \item $x \MEQ{\rsh_i[\FV{0}{\rsh_i} / \FV{i}{\sh}]} z$ and and $\REACH{z}{y}{\rsh}$, or
   \item $\REACH{x}{z}{\rsh_i[\FV{0}{\rsh_i} / \FV{i}{\sh}]}$ and $z \MEQ{\rsh} y$, or
   \item $\REACH{x}{z}{\rsh_i[\FV{0}{\rsh_i} / \FV{i}{\sh}]}$ and $\REACH{z}{y}{\rsh}$.
 \end{itemize}
 The same holds for the converse direction, i.e., if
 $y \in \VAR(\rsh_i[\FV{0}{\rsh_i} / \FV{i}{\sh}])$ and $x \in \VAR(\rsh) \setminus \VAR(\rsh_i[\FV{0}{\rsh_i} / \FV{i}{\sh}])$.
\end{lemma}
\begin{proof}[sketch]
  The implication from right to left is straightforward.
  The implication from left to right is shown by a lengthy complete induction on the number of
  definitive points-to assertions to reach $y$ from $x$.
  \qed
\end{proof}
Note that this observation relies on the fact that the reachability relation $\REACH{}{}{\sh}$ is defined with respect to \emph{all} models of $\sh$, not just a single one. Otherwise, the observation is wrong for symbolic heaps that are not established.
\begin{proof}[of Lemma~\ref{thm:zoo:garbage:compositionality}]
Let $\rsh = \CHECK(\sh[\PS^{\sh}_1 / \rsh_1, \ldots, \PS^{\sh}_m / \rsh_m])$.
Then:
\begin{align*}
                    & \CHECK(\rsh,\EMPTYSEQ) = 1 \\
  ~\Leftrightarrow~ & \left[ \text{Definition of}~\CHECK,\SHRINK{\rsh,\EMPTYSEQ} = \rsh \right] \\
                    & \forall y \in \VAR(\rsh) ~.~ \exists x \in \FV{0}{\rsh} ~.~ x \MEQ{\rsh} y ~\text{or}~ \REACH{x}{y}{\rsh} \\
  ~\Leftrightarrow~ & \left[
                        \VAR(\rsh) = \VAR(\sh) \cup
                        \bigcup_{1 \leq i \leq m} \VAR(\rsh_i\left[ \FV{0}{\rsh_i} / \FV{i}{\sh}  \right])
                      \right] \\
                    & \forall y \in \VAR(\sh) . \exists x \in \FV{0}{\rsh} ~.~ x \MEQ{\rsh} y ~\text{or}~ \REACH{x}{y}{\rsh} \\
                    & \quad \text{and}~ \forall 1 \leq i \leq m ~.~
                      \forall y \in \VAR\left(\rsh_i\left[ \FV{0}{\rsh_i} / \FV{i}{\sh}  \right]\right) ~. \\
                    & \qquad \qquad \qquad \exists x \in \FV{0}{\rsh} ~.~ x \MEQ{\rsh} y ~\text{or}~ \REACH{x}{y}{\rsh} \\
  ~\Leftrightarrow~ & \left[ y \in \VAR(\sh),
                             ~\text{Lemma}~\ref{thm:zoo:reachability:congruence} \right] \\
                    & \forall y \in \VAR(\sh) ~.~ \exists x \in \FV{0}{\sh} ~.~ x \MEQ{\SHRINK{\sh,\T{p}}} y \\
                    & \qquad \quad ~\text{or}~ \REACH{x}{y}{\SHRINK{\sh,\T{p}}} \\
                    & \quad \text{and}~ \forall 1 \leq i \leq m ~.~
                      \forall y \in \VAR\left(\rsh_i\left[ \FV{0}{\rsh_i} / \FV{i}{\sh}  \right]\right) ~. \\
                    & \qquad \qquad \qquad \exists x \in \FV{0}{\rsh} ~.~ x \MEQ{\rsh} y ~\text{or}~ \REACH{x}{y}{\rsh} \\
  ~\Leftrightarrow~ & \left[ \text{Lemma}~\ref{obs:zoo:establishment:equality-propagation},~\text{Lemma}~\ref{obs:zoo:reachability:propagation} \right] \\
                    & \forall y \in \VAR(\sh) ~.~ \exists x \in \FV{0}{\sh} ~.~ x \MEQ{\SHRINK{\sh,\T{p}}} y ~\text{or}~ \REACH{x}{y}{\SHRINK{\sh,\T{p}}} \\
                    & \quad \text{and}~ \forall 1 \leq i \leq m ~.~
                      \forall y \in \VAR\left(\rsh_i\left[ \FV{0}{\rsh_i} / \FV{i}{\sh}  \right]\right) ~. \\
                    & \qquad  \exists x \in \FV{0}{\rsh} ~.~ \exists z \in \FV{i}{\sh}  ~.~
                      \left( x \MEQ{\rsh} z ~\text{or}~ \REACH{x}{z}{\rsh} \right)  \\
                    & \qquad \qquad ~\text{and}~
                      \left( \REACH{z}{y}{\rsh_i\left[ \FV{0}{\rsh_i} / \FV{i}{\sh} \right]} ~\text{or}~
                      z \MEQ{\rsh_i\left[ \FV{0}{\rsh_i} / \FV{i}{\sh} \right]} y \right)  \\
  ~\Leftrightarrow~ & \left[ \exists a \exists b \equiv \exists b \exists a,
                            ~x,y \in \VAR(\sh) \right] \\
                    & \forall y \in \VAR(\sh) . \exists x \in \FV{0}{\sh} ~.~ \REACH{x}{y}{\SHRINK{\sh,\T{p}}} \\
                    & \quad \text{and}~ \forall 1 \leq i \leq m ~.~
                      \forall y \in \VAR\left(\rsh_i\left[ \FV{0}{\rsh_i} / \FV{i}{\sh}  \right]\right) ~. \\
                    & \qquad \exists z \in \FV{i}{\sh}  ~.~
                      \left( \REACH{z}{y}{\rsh_i\left[ \FV{0}{\rsh_i} / \FV{i}{\sh} \right]}
                      ~\text{or}~ z \MEQ{\rsh_i\left[ \FV{0}{\rsh_i} / \FV{i}{\sh} \right]} y \right)  \\
  ~\Leftrightarrow~ & \left[ \text{Definition of}~\CHECK \right] \\
                    & \CHECK(\sh,\T{p}) = 1 ~\text{and}~ \forall 1 \leq i \leq m ~.~
                      \CHECK(\rsh_i[\FV{0}{\rsh_i} / \FV{i}{\sh}],\EMPTYSEQ) = 1 \\
  ~\Leftrightarrow~ & \left[ \CHECK(\rsh_i[\FV{0}{\rsh_i} / \FV{i}{\sh}],\EMPTYSEQ) = 1
                             ~\text{iff}~ \CHECK(\rsh_i,\EMPTYSEQ) = 1 \right] \\
                    & \CHECK(\sh,\T{p}) = 1 ~\text{and}~ \forall 1 \leq i \leq m ~.~
                      \CHECK(\rsh_i,\EMPTYSEQ) = 1. \qedhere
\end{align*}
\qed
\end{proof}
%

%
%
\section{Proof of Theorem~\ref{thm:zoo:garbage:complexity}} \label{app:zoo:garbage:complexity}
   Since upper bounds are obtained analogously to Theorem~\ref{thm:zoo:establishment:complexity},
   we only show that the reduction provided in the paper proving the lower bounds is correct.
   We first briefly recall how \COMPLEMENT{SL-RSAT} is reduced to \DPROBLEM{SL-GF}.
   Let $(\SRD,\PS)$ be an instance of \COMPLEMENT{SL-RSAT}.
   Then a corresponding instance of \DPROBLEM{SL-GF} is given by
   $(\SRD,\sh)$, where $\FV{0}{\sh} = x$ and
    $\sh ~\DEFEQ~ \exists \BV{} z' ~.~ \SEP \PS\BV{} ~:~ \{ x = \NIL, z' \neq \NIL \}$.
   Now,
   \begin{align*}
                        & \CALLSEM{\sh}{\SRD} \subseteq \GARBAGEPROP(\alpha) \\
      ~\Leftrightarrow~ & \left[ A \subseteq B ~\text{iff}~ \forall x \in A ~.~ x \in B \right] \\
                        & \forall \rsh \in \CALLSEM{\sh}{\SRD} ~.~ \rsh \in \GARBAGEPROP(\alpha) \\
      ~\Leftrightarrow~ & \left[ \text{Definition of}~\GARBAGEPROP(\alpha,R) \right] \\
                        & \forall \rsh \in \CALLSEM{\sh}{\SRD} ~.~
                          \forall u \in \VAR(\rsh) ~.~ \exists v \in \FV{0}{\rsh} ~.~ u \MEQ{\rsh} v ~\text{or}~ \REACH{u}{v}{\rsh} \\
      ~\Leftrightarrow~ & \left[ \sh~\text{contains no points-to assertions} \right] \\
                        & \forall \rsh \in \CALLSEM{\sh}{\SRD} ~.~
                          \forall u \in \VAR(\rsh) ~.~ \exists v \in \FV{0}{\rsh} ~.~ u \MEQ{\rsh} v \\
      ~\Leftrightarrow~ & \left[ \FV{0}{\rsh} = x \right] \\
                        & \forall \rsh \in \CALLSEM{\sh}{\SRD} ~.~
                          \forall u \in \VAR(\rsh) ~.~ u \MEQ{\rsh} \NIL \\
      ~\Leftrightarrow~ & \left[ z' \neq \NIL \in \PURE{\rsh} \right] \\
                        & \sh~\text{unsatisfiable}.
   \end{align*}
   Since $\sh$ is satisfiable if and only if $\PS$ is satisfiable by an analogous argument as in the proof of Lemma~\ref{thm:zoo:establishment:lower},
   it follows that $\CALLSEM{\sh}{\SRD} \subseteq \GARBAGEPROP(\alpha)$ holds if and only if $\PS$ is unsatisfiable.
   \qed
%
%
\section{Proof of Lemma~\ref{thm:zoo:acyclicity:property}} \label{app:zoo:acyclicity:property}
We have to construct a heap automaton
$\HACYCLE$ over $\SHCLASSFV{\alpha}$ that satisfies the compositionality property
and accepts $\CYCLEPROP(\alpha)$.
In order to highlight the necessary proof obligations, the actual construction of $\HACYCLE$
and its correctness proof are splitted into several definitions and lemmas that are
provided afterwards.
The construction of $\HACYCLE$ is similar to the construction of $\HAEST$
as presented in Theorem~\ref{thm:zoo:establishment}.
The main difference is that the reachability automaton $\HAREACH$,
formally introduced in Definition~\ref{def:zoo:reachability-automaton},
is used instead of the tracking automaton $\HATRACK$
and that the predicate $\CHECK$ is adapted.
A fully formal construction of $\HACYCLE$ is found in Definition~\ref{def:zoo:acyclicity:automaton}.
It then remains to show the correctness of our construction of $\HACYCLE$:
\begin{itemize}
  \item Lemma~\ref{thm:zoo:acyclicity:language} establishes that $\HACYCLE$ indeed accepts $\CYCLEPROP(\alpha)$, i.e., we have $L(\HACYCLE) = \CYCLEPROP(\alpha)$.
  \item To prove the compositionality property, we show that $\HACYCLE$ is an instance of a more general
        construction scheme whose compositionality property is shown in
        Lemma~\ref{thm:zoo:scheme:compositionality}.
        In order to apply Lemma~\ref{thm:zoo:scheme:compositionality}, we have to show that
        \begin{align*}
                          & \CHECK(\sh[\PS_1^{\sh} / \rsh_1, \ldots, \PS_m^{\sh} / \rsh_m], \EMPTYSEQ) = 1 \\
             ~\text{iff}~ & \CHECK(\sh,p_1 \ldots p_m) = 1 \\
                          & \text{and}~ \CHECK(\rsh_1,\EMPTYSEQ) = \ldots = \CHECK(\rsh_m,\EMPTYSEQ) = 1.
         \end{align*}
         This is verified in Lemma~\ref{thm:zoo:acyclicity:compositionality}.
         Then, by Lemma~\ref{thm:zoo:scheme:compositionality}, we know that
         $\HACYCLE = \HASCHEME{\HAREACH}{\CHECK}{F}$ satisfies the compositionality property.
\end{itemize}
Putting both together, we obtain a heap automaton $\HACYCLE$ over $\SHCLASSFV{\alpha}$
that satisfies the compositionality property and accepts $\CYCLEPROP(\alpha)$.
\qed
\begin{definition} \label{def:zoo:acyclicity:automaton}
  $\HACYCLE = (Q,\SHCLASSFV{\alpha},\Delta,F)$ is given by
\begin{align*}
   & Q ~\DEFEQ~ Q_{\HAREACH} \times \{0,1\} \\
   & F ~\DEFEQ~ Q_{\HAREACH} \times \{1\} \\
   & \qquad \quad ~\cup~ \{ (A,\Pi,S) \in Q_{\HAREACH} ~|~ \NIL \neq \NIL \in \Pi \} \times \{0\}  \\
   &  \Delta ~~:~~ \MOVE{\HACYCLE}{(p_0,q_0)}{\sh}{(p_1,q_1) \ldots (p_m,q_m)} \\
   & ~\text{iff}~ \MOVE{\HAREACH}{p_0}{\sh}{p_1\ldots p_m} \\
   & \qquad \text{and}~ q_0 = \min \{q_1,\ldots,q_m,\CHECK(\sh,p_1 \ldots p_m)\}~,
\end{align*}

where $\HAREACH$ is is as in Definition~\ref{def:zoo:reachability-automaton}.
Furthermore, the predicate $\CHECK : \SHCLASSFV{\alpha} \times Q_{\HAREACH}^{*} \to \{0,1\}$
verifies that each variable of a symbolic heap $\sh$ is not reachable from itself in $\SHRINK{\sh,\T{p}}$,
where $\SHRINK{\sh,\T{p}}$ is the same as in the construction of $\HAREACH$.
Hence,
\begin{align*}
\CHECK(\sh,\T{p}) ~\DEFEQ~
 \begin{cases}
        1 &, ~\text{if}~ \forall y \in \VAR(\sh) ~.~ \text{not}~\REACH{x}{x}{\SHRINK{\sh,\T{p}}} \\
        0 &, ~\text{otherwise}.
 \end{cases}
\end{align*}
\end{definition}
\begin{lemma} \label{thm:zoo:acyclicity:language}
  $L(\HACYCLE) = \CYCLEPROP(\alpha)$.
\end{lemma}
\begin{proof}
Let $\rsh \in \RSL{}{\SRDCLASSFV{\alpha}}$. Then:
\begin{align*}
                   & \rsh \in L(\HACYCLE) \\
 ~\Leftrightarrow~ & \left[ \text{Definition of}~L(\HACYCLE) \right] \\
                   & \exists q \in F_{\HACYCLE} ~.~ \OMEGA{\HACYCLE}{q}{\rsh} \\
 ~\Leftrightarrow~ & \left[ \text{Definition of}~F_{\HACYCLE}, \right] \\
                   & \exists p \in Q_{\HAREACH} ~.~ \OMEGA{\HACYCLE}{(p,1)}{\rsh} \\
                   & \text{or}~ \exists p=(A,\Pi,S) \in Q_{\HAREACH} ~.~ \exists r \in \{0,1\} ~.~  \\
                   & \qquad \NIL \neq \NIL \in \Pi ~\text{and}~ \OMEGA{\HACYCLE}{(p,r)}{\rsh} \\
 ~\Leftrightarrow~ & \left[ \text{Definition of}~\Delta_{\HACYCLE} \right] \\
                   & \exists p \in Q_{\HAREACH} ~.~ \OMEGA{\HAREACH}{p}{\rsh}
                     ~\text{and}~ \CHECK(\rsh, \EMPTYSEQ) = 1 \\
                   & \text{or}~ \exists p=(A,\Pi,S) \in Q_{\HAREACH} ~.~ \exists r \in \{0,1\} ~.~  \\
                   & \qquad \NIL \neq \NIL \in \Pi ~\text{and}~ \OMEGA{\HAREACH}{p}{\rsh}
                     ~\text{and}~ \CHECK(\rsh, \EMPTYSEQ) = r \\
 ~\Leftrightarrow~ & \left[ \CHECK(\rsh, \EMPTYSEQ) \in \{0,1\} \right] \\
                   & \exists p \in Q_{\HAREACH} ~.~ \OMEGA{\HAREACH}{p}{\rsh}
                     ~\text{and}~ \CHECK(\rsh, \EMPTYSEQ) = 1 \\
                   & \text{or}~ \exists p=(A,\Pi,S) \in Q_{\HAREACH} ~.~
                        \NIL \neq \NIL \in \Pi ~\text{and}~ \OMEGA{\HAREACH}{p}{\rsh} \\
 ~\Leftrightarrow~ & \left[ \text{Definition of}~\CHECK \right] \\
                   & \exists p \in Q_{\HAREACH} ~.~ \OMEGA{\HAREACH}{p}{\rsh} \\
                   & \qquad \text{and}~ \forall x \in \VAR(\rsh) ~.~ \text{not}~\REACH{x}{x}{\SHRINK{\rsh,\EMPTYSEQ}} \\
                   & \text{or}~ \exists p=(A,\Pi,S) \in Q_{\HAREACH} ~.~
                        \NIL \neq \NIL \in \Pi ~\text{and}~ \OMEGA{\HAREACH}{p}{\rsh} \\
 ~\Leftrightarrow~ & \left[ \NOCALLS{\rsh} = 0 ~\text{implies}~ \rsh = \SHRINK{\rsh,\EMPTYSEQ} \right] \\
                   & \forall x \in \VAR(\rsh) ~.~ \text{not}~\REACH{x}{x}{\rsh} \\
                   & \text{or}~ \exists p=(A,\Pi,S) \in Q_{\HAREACH} ~.~
                        \NIL \neq \NIL \in \Pi ~\text{and}~ \OMEGA{\HAREACH}{p}{\rsh} \\
 ~\Leftrightarrow~ & \left[ \text{Definition of}~\Delta_{\HAREACH} \right] \\
                   & \forall x \in \VAR(\rsh) ~.~ \text{not}~\REACH{x}{x}{\rsh}
                        ~\text{or}~ \NIL \MNEQ{\rsh} \NIL \\
 ~\Leftrightarrow~ & \left[ \text{Definition of}~\CYCLEPROP(\alpha) \right] \\
                   & \rsh \in \CYCLEPROP(\alpha).
\end{align*}
\qed
\end{proof}
\begin{lemma} \label{thm:zoo:acyclicity:compositionality}
  Let  $\sh \in \SL{}{\SRDCLASSFV{\alpha}}$ with $\NOCALLS{\sh} = m$
  and $\rsh_1,\ldots,\rsh_m \in \RSL{}{\SRDCLASSFV{\alpha}}$. Then
  \begin{align*}
                   & \CHECK(\sh[\PS_1^{\sh} / \rsh_1, \ldots, \PS_m^{\sh} / \rsh_m], \EMPTYSEQ) = 1 \\
      ~\text{iff}~ & \CHECK(\sh,\T{p}) \\
                   & \text{and}~ \CHECK(\rsh_1,\EMPTYSEQ) = \ldots = \CHECK(\rsh_m,\EMPTYSEQ) = 1. 
  \end{align*}
\end{lemma}
\begin{proof}
Let $\rsh = \CHECK(\sh[\PS^{\sh}_1 / \rsh_1, \ldots, \PS^{\sh}_m / \rsh_m])$.
Then
\begin{align*}
                   & \CHECK(\rsh,\EMPTYSEQ) = 1 \\
 ~\Leftrightarrow~ & \left[ \text{Definition of}~\CHECK, \rsh = \SHRINK{\rsh,\EMPTYSEQ} \right] \\
                   & \forall x \in \VAR(\rsh) ~.~ \text{not}~ \REACH{x}{x}{\rsh} \\
 ~\Leftrightarrow~ & \left[
                        \VAR(\rsh) = \VAR(\sh) \cup
                        \bigcup_{1 \leq i \leq m} \VAR(\rsh_i\left[ \FV{0}{\rsh_i} / \FV{i}{\sh}  \right])
                      \right] \\
                   & \forall x \in \VAR(\sh) ~.~ \text{not}~ \REACH{x}{x}{\rsh} \\
                   & \text{and}~ \forall 1 \leq i \leq m ~.~
                         \forall x \in \VAR(\rsh_i\left[ \FV{0}{\rsh_i} / \FV{i}{\sh}  \right])
                             ~.~ \text{not}~ \REACH{x}{x}{\rsh} \\
 ~\Leftrightarrow~ & \left[ \text{Lemma}~\ref{thm:zoo:reachability:congruence:auxiliary}  \right] \\
                   & \forall x \in \VAR(\sh) ~.~ \text{not}~ \REACH{x}{x}{\SHRINK{\sh,\T{p}}} \\
                   & \text{and}~ \forall 1 \leq i \leq m ~.~
                        \forall x \in \VAR(\rsh_i\left[ \FV{0}{\rsh_i} / \FV{i}{\sh}  \right])
                             ~.~ \text{not}~ \REACH{x}{x}{\rsh}.
\end{align*}
Assume towards a contradiction that $\CHECK(\rsh_i,\EMPTYSEQ) = 0$ for some $1 \leq i \leq m$.
By definition and $\SHRINK{\rsh_i,\EMPTYSEQ} = \rsh_i$, this means that $\REACH{x}{x}{\rsh_i}$.
However, since $\rsh_i$ is contained in $\rsh$, this means that $\REACH{x}{x}{\rsh}$
and thus $\CHECK(\rsh,\EMPTYSEQ) = 0$.
Conversely, assume $\CHECK(\rsh,\EMPTYSEQ) = 0$, but $\CHECK(\sh,\T{p}) = 1$
and $\CHECK(\rsh_i,\EMPTYSEQ) = 1$ for each $1 \leq i \leq m$.
Thus,
\begin{align*}
                   & \forall x \in \VAR(\sh) ~.~ \text{not}~ \REACH{x}{x}{\SHRINK{\sh,\T{p}}} \tag{$\clubsuit$} \\
                   & \text{and}~ \forall 1 \leq i \leq m ~.~  \tag{$\spadesuit$}
                     \forall x \in \VAR(\rsh_i\left[ \FV{0}{\rsh_i} / \FV{i}{\sh}  \right])
                         ~.~ \text{not}~ \REACH{x}{x}{\rsh_i}
\end{align*}
Then, there exists $x \in \VAR(\rsh)$ such that $\REACH{x}{x}{\rsh}$.
We proceed by case distinction:
\begin{enumerate}
  \item Case: $x \in \VAR(\sh)$.
        Then we immediately obtain a contradiction due to $(\clubsuit)$.
  \item Case:
        $\REACH{x}{x}{\rsh_i\left[ \FV{0}{\rsh_i} / \FV{i}{\sh} \right]}$ holds.
        Since
        $\REACH{x}{x}{\rsh_i\left[ \FV{0}{\rsh_i} / \FV{i}{\sh} \right]}$
        holds if and only if
        $\REACH{x}{x}{\rsh_i}$ holds, we immediately obtain a contradiction due to $(\spadesuit)$.
  \item Case:
        For some $1 \leq i \leq m$,
        $x \in \VAR(\rsh_i\left[ \FV{0}{\rsh_i} / \FV{i}{\sh} \right])$
        and
        $\REACH{x}{x}{\rsh_i\left[ \FV{0}{\rsh_i} / \FV{i}{\sh} \right]}$
        does not hold.
        Then there exists
        $y \in \VAR(\rsh) \setminus \VAR(\rsh_i\left[ \FV{0}{\rsh_i} / \FV{i}{\sh} \right])$
        such that one of the following two cases holds:
        \begin{enumerate}
            \item $x \MEQ{\rsh} y$ and $\REACH{y}{y}{\rsh}$.
                  Then, by Lemma~\ref{obs:zoo:establishment:equality-propagation},
                  there exists $z \in \FV{i}{\sh}$ such that
                  $x \MEQ{\rsh_i\left[ \FV{0}{\rsh_i} / \FV{i}{\sh} \right]} z$
                  and
                  $z \MEQ{\rsh} y$.
                  However, by construction of $\SHRINK{\sh,\T{p}}$, this means that $\REACH{z}{z}{\SHRINK{\sh,\T{p}}}$ holds, which
                  contradicts $(\clubsuit)$.
            \item $\REACH{x}{y}{\rsh}$ and $\REACH{y}{x}{\rsh}$.
                  Then, by Lemma~\ref{obs:zoo:reachability:propagation},
                  there exist $u,v \in \FV{i}{\sh}$ such that
                  $\REACH{u}{v}{\rsh}$ and $\REACH{v}{u}{\rsh}$.
                  However, by construction of $\SHRINK{\sh,\T{p}}$, this means that $\REACH{u}{u}{\SHRINK{\sh,\T{p}}}$ holds.
                  Thus $\CHECK(\sh,\T{p}) = 0$, which contradicts $(\clubsuit)$.
        \end{enumerate}
\end{enumerate}
Since each case leads to a contradiction, we conclude $\CHECK(\rsh,\T{p}) = 1$.
\qed
\end{proof}
%
%
\section{Proof of Theorem~\ref{thm:zoo:acyclicity:complexity}} \label{app:zoo:acyclicity:complexity}
   Upper bounds are obtained analogously to Theorem~\ref{thm:zoo:establishment:complexity}.
   Thus, we only show that the reduction provided in the paper proving the lower bounds is correct.
   We first briefly recall how \COMPLEMENT{SL-RSAT} is reduced to \DPROBLEM{SL-AC}.
   Let $(\SRD,\PS)$ be an instance of \COMPLEMENT{SL-RSAT}.
   Then an instance of \DPROBLEM{SL-AC} is given by $(\SRD,\sh)$,
   where $\FV{0}{\sh} = x$ and
   $\sh = \exists \BV{} . \PT{x}{x} \SEP \PS\BV{}$.
   Then
   \begin{align*}
                        & \CALLSEM{\sh}{\SRD} \subseteq \CYCLEPROP(\alpha) \\
      ~\Leftrightarrow~ & \left[ A \subseteq B ~\text{iff}~ \forall x \in A ~.~ x \in B \right] \\
                        & \forall \rsh \in \CALLSEM{\sh}{\SRD} ~.~  \rsh \in \CYCLEPROP(\alpha) \\
      ~\Leftrightarrow~ & \left[ \text{Definition of}~\CYCLEPROP(\alpha) \right] \\
                        & \forall \rsh \in \CALLSEM{\sh}{\SRD} ~.~
                            \NIL \MNEQ{\rsh} \NIL ~\text{or}~
                            \forall y \in \VAR(\rsh) ~.~ \text{not}~ \REACH{y}{y}{\rsh} \\
      ~\Leftrightarrow~ & \left[ \sh~\text{contains exactly one points-to assertion} \right] \\
                        & \forall \rsh \in \CALLSEM{\sh}{\SRD} ~.~
                            \NIL \MNEQ{\rsh} \NIL ~\text{or}~
                            \text{not}~ \REACH{x}{x}{\rsh} \\
      ~\Leftrightarrow~ & \left[ \REACH{x}{x}{\rsh} ~\text{always hold by construction of}~ \sh \right] \\
                        & \forall \rsh \in \CALLSEM{\sh}{\SRD} ~.~
                            \NIL \MNEQ{\rsh} \NIL \\
      ~\Leftrightarrow~ & \left[ \NIL \MNEQ{\rsh} \NIL ~\text{iff}~ \MODELS{\rsh} = \emptyset \right] \\
                        & \forall \rsh \in \CALLSEM{\sh}{\SRD} ~.~ \MODELS{\rsh} = \emptyset \\
      ~\Leftrightarrow~ & \left[ \text{Definition}~\ref{def:symbolic-heaps:models} \right] \\
                        & \sh~\text{unsatisfiable}.
   \end{align*}
   Now, by an analogous argument as in the proof of Theorem~\ref{thm:zoo:establishment:complexity},
   we obtain that $\sh$ is satisfiable iff $\PS$ is satisfiable.
   Then $\CALLSEM{\sh}{\SRD} \subseteq \CYCLEPROP(\alpha)$ holds iff $\PS$ is unsatisfiable.
   Hence, $\DPROBLEM{SL-AC}$ holds for $(\SRD,\sh)$ iff $\COMPLEMENT{SL-RSAT}$ holds for $(\SRD,\PS)$.
   \qed
%
%


%
%
\section{Proof of Lemma~\ref{thm:entailment:predicates}} \label{app:entailment:predicates}
The crux of the proof relies on the fact that each unfolding of $\PS_1\T{x}$ has exactly one canonical model up to isomorphism, i.e.,
\begin{align*}
 \forall \rsha \in \CALLSEM{\PS_1\T{x}}{\SRD} \,.\, \SIZE{\MODELS{\rsha}} = 1 \tag{$\bigstar$}.
\end{align*}
This property is a direct consequence of two properties:
By Definition~\ref{def:entailment:determined}, $\rsha$ has at most one canonical model up to isomorphism, because $\SRD$ is determined.
Furthermore, each symbolic heap $\rsha \in \CALLSEM{\PS_1\T{x}}{\SRD}$ has at least one canonical model, because $\SRD$ is well--determined.
Then
  \begin{align*}
                       & \PS_1\T{x} \SAT{\SRD} \PS_2\T{x} \\
     ~\Leftrightarrow~ & \left[ \text{Definition of entailments} \right] \\
                       & \forall (\stack,\heap) \in \STATES \,.\, \stack,\heap \SAT{\SRD} \PS_1\T{x} ~\text{implies}~ \stack,\heap \SAT{\SRD} \PS_2\T{x} \\
     ~\Leftrightarrow~ & \left[ \text{Lemma}~\ref{thm:symbolic-heaps:fv-coincidence},~\STATES_{\T{x}} \DEFEQ \{ (\stack,\heap) \in \STATES ~|~ \DOM(\stack) = \T{x} \} \right] \\
                       & \forall (\stack,\heap) \in \STATES_{\T{x}} \,.\, \stack,\heap \SAT{\SRD} \PS_1\T{x} ~\text{implies}~ \stack,\heap \SAT{\SRD} \PS_2\T{x} \\
     ~\Leftrightarrow~ & \left[ \text{Lemma}~\ref{thm:symbolic-heaps:semantics} \right] \\
                       & \forall (\stack,\heap) \in \STATES_{\T{x}} \,.\, \\
                       & \qquad \exists \rsha \in \CALLSEM{\PS_1\T{x}}{\SRD} \,.\, \stack,\heap \SAT{\emptyset} \rsha \\
                       & \qquad \text{implies}~ \exists \rsh \in \CALLSEM{\PS_2\T{x}}{\SRD} \,.\, \stack,\heap \SAT{\emptyset} \rsh \\
     ~\Leftrightarrow~ & \left[ A~\rightarrow~B \equiv \neg A \vee B \right] \\
                       & \forall (\stack,\heap) \in \STATES_{\T{x}} \,.\, \\
                       & \qquad \neg \left( \exists \rsha \in \CALLSEM{\PS_1\T{x}}{\SRD} \,.\, \stack,\heap \SAT{\emptyset} \rsha \right) \\
                       & \qquad \text{or}~ \exists \rsh \in \CALLSEM{\PS_2\T{x}}{\SRD} \,.\, \stack,\heap \SAT{\emptyset} \rsh \\
     ~\Leftrightarrow~ & \left[ \neg \exists x . A \equiv \forall x . \neg A \right] \\
                       & \forall (\stack,\heap) \in \STATES_{\T{x}} \,.\, \\
                       & \qquad \forall \rsha \in \CALLSEM{\PS_1\T{x}}{\SRD} \,.\, \stack,\heap \not \SAT{\emptyset} \rsha \\
                       & \qquad \text{or}~ \exists \rsh \in \CALLSEM{\PS_2\T{x}}{\SRD} \,.\, \stack,\heap \SAT{\emptyset} \rsh \\
     ~\Leftrightarrow~ & \left[ \forall x \forall y . A \equiv \forall y \forall x . A \right] \\
                       & \forall \rsha \in \CALLSEM{\PS_1\T{x}}{\SRD} \,.\, \forall (\stack,\heap) \in \STATES_{\T{x}} \,.\, \\
                       & \qquad \stack,\heap \not \SAT{\emptyset} \rsha ~\text{or}~ \exists \rsh \in \CALLSEM{\PS_2\T{x}}{\SRD} \,.\, \stack,\heap \SAT{\emptyset} \rsh \\
     ~\Leftrightarrow~ & \left[ A~\rightarrow~B \equiv \neg A \vee B \right] \\
                       & \forall \rsha \in \CALLSEM{\PS_1\T{x}}{\SRD} \,.\, \forall (\stack,\heap) \in \STATES_{\T{x}} \,.\, \\
                       & \qquad \stack,\heap \SAT{\emptyset} \rsha ~\text{implies}~ \exists \rsh \in \CALLSEM{\PS_2\T{x}}{\SRD} \,.\, \stack,\heap \SAT{\emptyset} \rsh \\
     ~\Leftrightarrow~ & \left[ \forall x . A \rightarrow B \equiv \forall x \in A . B,~\text{Def.}~\ref{def:symbolic-heaps:models} \right] \\
                       & \forall \rsha \in \CALLSEM{\PS_1\T{x}}{\SRD} \,.\, \forall (\stack,\heap) \in \MODELS{\rsha} \,.\,
                         \exists \rsh \in \CALLSEM{\PS_2\T{x}}{\SRD} \,.\, \stack,\heap \SAT{\emptyset} \rsh \\
     ~\Leftrightarrow~ & \left[ \SIZE{\MODELS{\rsha}} = 1~\text{by}~(\bigstar)~\text{and}~ \forall x \in \{y\} . A \equiv \exists x \in \{y\} . A  \right] \\
                       & \forall \rsha \in \CALLSEM{\PS_1\T{x}}{\SRD} \,.\, \exists (\stack,\heap) \in \MODELS{\rsha} \,.\,
                         \exists \rsh \in \CALLSEM{\PS_2\T{x}}{\SRD} \,.\, \stack,\heap \SAT{\emptyset} \rsh \\
     ~\Leftrightarrow~ & \left[ \exists x \exists y . A \equiv \exists y \exists x . A  \right] \\
                       & \forall \rsha \in \CALLSEM{\PS_1\T{x}}{\SRD} \,.\, \exists \rsh \in \CALLSEM{\PS_2\T{x}}{\SRD} \,.\, \exists (\stack,\heap) \in \MODELS{\rsha} \,.\, \stack,\heap \SAT{\emptyset} \rsh \\
     ~\Leftrightarrow~ & \left[ \SIZE{\MODELS{\rsha}} = 1~\text{by}~(\bigstar)~\text{and}~ \forall x \in \{y\} . A \equiv \exists x \in \{y\} . A  \right] \\
                       & \forall \rsha \in \CALLSEM{\PS_1\T{x}}{\SRD} \,.\, \exists \rsh \in \CALLSEM{\PS_2\T{x}}{\SRD} \,.\, \forall (\stack,\heap) \in \MODELS{\rsha} \,.\, \stack,\heap \SAT{\emptyset} \rsh \\
     ~\Leftrightarrow~ & \left[ \forall x . A \rightarrow B \equiv \forall x \in A . B \right] \\
                       & \forall \rsha \in \CALLSEM{\PS_1\T{x}}{\SRD} \,.\, \exists \rsh \in \CALLSEM{\PS_2\T{x}}{\SRD} \,.\, \forall (\stack,\heap) \in \STATES \,.\, \\
                       & \qquad (\stack,\heap) \in \MODELS{\rsha} ~\text{implies}~ \stack,\heap \SAT{\emptyset} \rsh \\
     ~\Leftrightarrow~ & \left[ \text{Def.}~\ref{def:symbolic-heaps:models}:~(\stack,\heap) \in \MODELS{\rsha} ~\Leftrightarrow~ \DOM(\stack) = \T{x} ~\text{and}~ \stack,\heap \SAT{\emptyset} \rsha \right] \\
                       & \forall \rsha \in \CALLSEM{\PS_1\T{x}}{\SRD} \,.\, \exists \rsh \in \CALLSEM{\PS_2\T{x}}{\SRD} \,.\, \forall (\stack,\heap) \in \STATES \,.\, \\
                       & \qquad \DOM(\stack) = \T{x} ~\text{and}~ \stack,\heap \SAT{\emptyset} \rsha ~\text{implies}~ \stack,\heap \SAT{\emptyset} \rsh \\
     ~\Leftrightarrow~ & \left[ \forall x . (A \wedge B) \rightarrow C \equiv \forall x \in A . B \rightarrow C \right] \\
                       & \forall \rsha \in \CALLSEM{\PS_1\T{x}}{\SRD} \,.\, \exists \rsh \in \CALLSEM{\PS_2\T{x}}{\SRD} \,.\, \forall (\stack,\heap) \in \STATES_{\T{x}} \,.\, \\
                       & \qquad \stack,\heap \SAT{\emptyset} \rsha ~\text{implies}~ \stack,\heap \SAT{\emptyset} \rsh \\
     ~\Leftrightarrow~ & \left[ \text{Lemma}~\ref{thm:symbolic-heaps:fv-coincidence} \right] \\
                       & \forall \rsha \in \CALLSEM{\PS_1\T{x}}{\SRD} \,.\, \exists \rsh \in \CALLSEM{\PS_2\T{x}}{\SRD} \,.\, \forall (\stack,\heap) \in \STATES \,.\, \\
                       & \qquad \stack,\heap \SAT{\emptyset} \rsha ~\text{implies}~ \stack,\heap \SAT{\emptyset} \rsh \\
     ~\Leftrightarrow~ & \left[ \text{Definition of entailments} \right] \\
                       & \forall \rsha \in \CALLSEM{\PS_1\T{x}}{\SRD} \,.\, \exists \rsh \in \CALLSEM{\PS_2\T{x}}{\SRD} \,.\, \rsha \SAT{\emptyset} \rsh. 
  \end{align*}
  \qed
\subsection{Proof of Lemma~\ref{thm:entailment:general-upper-bound}}\label{app:entailment:general-upper-bound}
  Our previous complexity analysis of Algorithm~\ref{alg:on-the-fly-refinement} reveals that
  $\CALLSEM{\PS\T{x}}{\SRDALT} \cap L(\overline{\HASH{\sha}}) = \emptyset$ is decidable in
  \begin{align*}
    \BIGO{\SIZE{\SRDALT} \cdot \SIZE{Q_{\overline{\HASH{\sha}}}}^{M+1}
    \cdot \SIZE{\Delta_{\overline{\HASH{\sha}}}}}. \tag{$\clubsuit$}
  \end{align*}
  Regarding $\SIZE{\SRDALT}$, applying 
  the Refinement Theorem (Theorem~\ref{thm:compositional:refinement})
  to $\SRD \cup \{ \SRDRULE{\PS}{\sh} \}$ and $\HASAT$ (cf. Theorem~\ref{thm:zoo:sat:property})
  yields an SID $\SRDALT$ of size
  \begin{align*}
    \SIZE{\SRDALT} \leq c \cdot \SIZE{\SRD} \cdot 2^{\SIZE{\sh}^2} \cdot 2^{2\alpha^2 + \alpha} \leq 2^{\text{poly}(k)}~,
  \end{align*}
  for some positive constant $c$.
  Then $\SRDALT$ is computable in $\BIGO{2^{\text{poly}(k)}}$.
  Furthermore, $\overline{\HASH{\sha}}$ is obtained from complementation of $\HASH{\sha}$.
  Thus, by the construction to prove Lemma~\ref{thm:refinement:boolean}, we obtain that
  $\SIZE{Q_{\overline{\HASH{\sha}}}} \leq 2^{\SIZE{Q_{\HASH{\sha}}}}$ and
  that $\Delta_{\overline{\HASH{\sha}}}$ is decidable in
  $\left(2^{\SIZE{Q_{\HASH{\sha}}}}\right)^{M+1} \cdot \SIZE{\Delta_{\HASH{\sha}}}$.
  Putting both into $(\clubsuit)$ yields the result.
  \qed
\subsection{Proof of Lemma~\ref{thm:entailment:top-level-complexity}}
\label{app:entailment:top-level-complexity}
  By induction on the structure of symbolic heaps.
  The base cases are:
  \emph{The empty heap $\EMP$}:
  In Section~\ref{sec:entailment:emp}, we construct a heap automaton $\HASH{\EMP}$ accepting $\USET{\EMP}{\SRD}{\CENTAIL{\alpha}}$ with
  \begin{align*}
    \SIZE{Q_{\HA{\EMP}}} \leq 1 + (1+\alpha)^0 = 2 \leq 2^{\text{poly}(\alpha)}.
  \end{align*}
  Moreover, $\Delta_{\HASH{\EMP}}$ is decidable in polynomial time by Remark~\ref{rem:closure}.
  \emph{The points-to assertion $\PT{x}{\T{y}}$}:
  In Section~\ref{sec:entailment:points-to}, we construct a heap automaton $\HASH{\PT{x}{\T{y}}}$ accepting $\USET{\PT{x}{\T{y}}}{\SRD}{\CENTAIL{\alpha}}$ with
  \begin{align*}
    \SIZE{Q_{\HASH{\PT{x}{\T{y}}}}} \leq 1 + (1+\alpha)^{\SIZE{\T{y}}} \leq 1 + (1+\alpha)^{\gamma} \leq 2^{\text{poly}(\alpha)}~,
  \end{align*}
  because $\SIZE{\T{y}} \leq \gamma$ is considered to be a constant.
  Moreover, $\Delta_{\HASH{\PT{x}{\T{y}}}}$ is decidable in polynomial time by Remark~\ref{rem:closure}.

  \emph{The predicate call $\PS\T{x}$}:
  By assumption, there is nothing to show.
  Thus, it remains to consider the composite cases:
  \emph{The separating conjunction $\sha_1 \SEP \sha_2$}:
  By I.H. there exist heap automata $\HASH{\sha_1}$, $\HASH{\sha_2}$
  $\USET{\sha_1}{\SRD}{\CENTAIL{\alpha}}$ and $\USET{\sha_2}{\SRD}{\CENTAIL{\alpha}}$
  such that $\Delta_{\HASH{\sha_1}}$ and $\Delta_{\HASH{\sha_2}}$ are decidable in
  $\BIGO{2^{\text{poly}(k_1)}}$ and $\BIGO{2^{\text{poly}(k_2)}}$, where
  $k_1 = \SIZE{\SRD} + \SIZE{\sha_1}$ and $k_2 = \SIZE{\SRD} + \SIZE{\sha_2}$.
  Moreover, $\SIZE{Q_{\HASH{\sha_1}}} \leq 2^{\text{poly}(\alpha)}$
  and $\SIZE{Q_{\HASH{\sha_2}}} \leq 2^{\text{poly}(\alpha)}$, respectively.

  In Section~\ref{sec:entailment:sepcon}, we construct a heap automaton $\HASH{\sha_1 \SEP \sha_2}$
  accepting $\USET{\sha_1 \SEP \sha_2}{\SRD}{\CENTAIL{\alpha}}$ with
  \begin{align*}
    \SIZE{Q_{\HA{\sha_1 \SEP \sha_2}}} & \leq 2^{\alpha} \cdot \SIZE{Q_{\HA{\sha_1}}} \cdot \SIZE{Q_{\HA{\sha_2}}} \\
                      & \leq 2^{\alpha} \cdot 2^{\text{poly}(\alpha)} \cdot 2^{\text{poly}(\alpha)} \\
                      & =    2^{\alpha + \text{poly}(\alpha) + \text{poly}(\alpha)} \\
                      & \leq 2^{\text{poly}(\alpha)}.
  \end{align*}
  Moreover, the time to decide $\Delta_{\HASH{\sha_1 \SEP \sha_2}}$ depends on two factors:
  First, the number of subformulas occurring in $\sha_1$, $\sha_2$ and $\SRD$, which is bounded by
  $2^{\SIZE{\SRD}+\SIZE{\sha_1}+\SIZE{\sha_2}} \leq 2^{\text{poly}(k)}$.
  Second, each of these subformulas is applied to $\Delta_{\HASH{\sha_1}}$ and $\Delta_{\HASH{\sha_2}}$.
  Thus, $\Delta_{\HASH{\sha_1 \SEP \sha_2}}$ is decidable in
  $\BIGO{2^{\text{poly}(k)} \cdot \left(\SIZE{\Delta_{\HASH{\sha_1}}} + \SIZE{\Delta_{\HASH{\sha_2}}}\right)}$.
  By I.H. this means that $\Delta_{\HASH{\sha_1 \SEP \sha_2}}$ is decidable in
  \begin{align*}
    \BIGO{2^{\text{poly}(k)} \cdot (2^{\text{poly}(k_1)} + 2^{\text{poly}(k_2)})} = \BIGO{2^{\text{poly}(k)}}.
  \end{align*}

  \emph{The pure formula $\sh : \Pi$}:
  By I.H. there exists a heap automaton $\HASH{\sh}$ accepting
  $\USET{\sh}{\SRD}{\CENTAIL{\alpha}}$
   such that $\Delta_{\HA{\sh}}$ is decidable in $\BIGO{2^{\text{poly}(k_1)}}$, where
  $k_1 = \SIZE{\SRD} + \SIZE{\sh}$.
  Moreover, $\SIZE{Q_{\HASH{\sh}}} \leq 2^{\text{poly}(\alpha)}$.
  In Section~\ref{sec:entailment:pure}, we construct a heap automaton $\HASH{\sh : \Pi}$ accepting $\USET{\sh : \Pi}{\SRD}{\CENTAIL{\alpha}}$ with
  \begin{align*}
    \SIZE{Q_{\HASH{\sh : \Pi}}} \leq 2 \cdot \SIZE{Q_{\HASH{\sh}}} \leq 2 \cdot 2^{\text{poly}(\alpha)} \leq 2^{\text{poly}(\alpha)}.
  \end{align*}
  Moreover, $\Delta_{\HASH{\sh : \Pi}}$ boils down to deciding $\Delta_{\HASH{\sh}}$ plus some at most polynomial overhead.
  Thus, by I.H., $\Delta_{\HASH{\sh : \Pi}}$ is decidable in $\BIGO{2^{\text{poly}(k)}}$.
  \emph{The existential quantification $\exists z . \sh$}:
  By I.H. there exists a heap automaton $\HASH{\sh}$ accepting
  $\USET{\exists x . \sh}{\SRD}{\CENTAIL{\alpha}}$
   such that $\Delta_{\HASH{\sh}}$ is decidable in $\BIGO{2^{\text{poly}(k_1)}}$, where
  $k_1 = \SIZE{\SRD} + \SIZE{\sh}$.
  Moreover, $\SIZE{Q_{\HASH{\sh}}} \leq 2^{\text{poly}(\alpha)}$.

  In Section~\ref{sec:entailment:points-to}, we construct a heap automaton $\HASH{\exists z . \sh}$ accepting $\USET{\exists x . \sh}{\SRD}{\CENTAIL{\alpha}}$ with
  \begin{align*}
    \SIZE{Q_{\HASH{\exists z . \sh}}} & \leq (\alpha+1) \cdot \SIZE{Q_{\HASH{\sh}}} \cdot 2^{2\alpha^2+\alpha} \\
                      & \leq (\alpha+1) \cdot 2^{\text{poly}(\alpha)} \cdot 2^{2\alpha^2+\alpha} \\
                      & \leq 2^{\text{poly}(\alpha)}.
  \end{align*}
  Moreover, $\Delta_{\HASH{\exists z . \sh}}$ boils down to deciding $\Delta_{\HASH{\sh}}$ plus some overhead which is at most polynomial due to Remark~\ref{rem:closure}.
  Thus, by I.H., $\Delta_{\HASH{\exists z . \sh}}$ is decidable in $\BIGO{2^{\text{poly}(k)}}$.
  \qed
\section{Omitted calculations in proof of Theorem~\ref{thm:entailment:double-exptime}}
\label{app:entailment:double-exptime}
  \begin{align*}
        & \BIGO{\SIZE{\SRDALT} \cdot \SIZE{Q_{\overline{\HASH{\sha}}}}^{M+1} \cdot \SIZE{\Delta_{\overline{\HASH{\sha}}}}} \\
    ~=~ &  \left[ \SIZE{\SRDALT} \leq 2^{\text{poly}(k)}, \SIZE{Q_{\overline{\HASH{\sha}}}} \leq 2^{\SIZE{Q_{\HASH{\sha}}}}, M \leq 2k \right] \\
        & \BIGO{ 2^{\text{poly}(k)} \cdot \left(2^{\SIZE{Q_{\HASH{\sha}}}}\right)^{2k} \cdot \SIZE{\Delta_{\overline{\HASH{\sha}}}}} \\
    ~=~ & \left[ \SIZE{\Delta_{\overline{\HASH{\sha}}}} \leq \left(2^{\SIZE{Q_{\HASH{\sha}}}}\right)^{M+1} \cdot \SIZE{\Delta_{\HASH{\sha}}} \right] \\
        & \BIGO{ 2^{\text{poly}(k)} \cdot \left(2^{\SIZE{Q_{\HASH{\sha}}}}\right)^{4k} \cdot \SIZE{\Delta_{\HASH{\sha}}}} \\
    ~=~ &  \left[ \SIZE{Q_{\HASH{\sha}}} \leq 2^{\text{poly}(k)} \right] \\
        & \BIGO{ 2^{\text{poly}(k)} \cdot \left(2^{2^{\text{poly}(k)}}\right)^{4k} \cdot \SIZE{\Delta_{\HASH{\sha}}}} \\
    ~=~ &  \left[ (a^b)^c = a^{bc},~ \SIZE{\Delta_{\HASH{\sha}}} \in \BIGO{2^{\text{poly}(k)}}  \right] \\
        & \BIGO{ 2^{\text{poly}(k)} \cdot 2^{4k \cdot 2^{\text{poly}(k)}} \cdot 2^{\text{poly}(k)}} \\
    ~=~ & \BIGO{ 2^{2^{\text{poly}(k)}} }
  \end{align*}
\section{Proof of Lemma~\ref{thm:entailment:exptime}}\label{app:entailment:exptime}
By Lemma~\ref{thm:entailment:top-level-complexity},
$\SIZE{Q_{\HASH{\sha}}} \leq 2^{\text{poly}(\alpha)}$ and $\Delta_{\HASH{\sha}}$ is decidable in
$\BIGO{2^{\text{poly}(k)}}$.
Since $\alpha$ is bounded by a constant, so is $\SIZE{Q_{\HASH{\sha}}}$.
Then, by Lemma~\ref{thm:entailment:general-upper-bound},
$\DENTAIL{\CENTAIL{\alpha}}{\SRD}$ is decidable in
  \begin{align*}
   \BIGO{ 2^{\text{poly}(k)} \cdot \left(2^{\SIZE{Q_{\HASH{\sha}}}}\right)^{2(M+1)}
         \cdot 2^{\text{poly}(k)} }
   ~=~
   \BIGO{ 2^{\text{poly}(k)} }, ~
  \end{align*}
which clearly is in \CCLASS{ExpTime}.
\qed
\section{Lower bound for entailments}
\label{app:entailment:lower}
The proof of the \CCLASS{ExpTime}--lower bound in~\cite{antonopoulos2014foundations}
is by reducing the inclusion problem for nondeterministic finite tree automata (NFTA, cf. \cite{comon2007tree}) to
the entailment problem.
Their proof requires a constant (or free variable) for each symbol in the tree automatons alphabet.
In contrast, we prove their result by encoding the alphabet in a null-terminated singly-linked list.
Thus, a tree $a(b,a(b,b)$ is encoded by a reduced symbolic heap
\begin{align*}
  & \exists z_1 z_2 z_3 z_4 z_5 z_6 z_7 ~.~  \\
  & \quad  \PT{x}{z_1~z_2~\NIL}  \\
  & \quad  \SEP \PT{z_1}{\NIL~\NIL~z_3} \SEP \PT{z_3}{\NIL~\NIL~\NIL} \\
  & \quad  \SEP \PT{z_2}{z_4~z_5~\NIL}  \\
  & \quad  \SEP \PT{z_4}{\NIL~\NIL~z_6} \SEP \PT{z_6}{\NIL~\NIL~\NIL} \\
  & \quad  \SEP \PT{z_5}{\NIL~\NIL~ z_7} \SEP \PT{z_7}{\NIL~\NIL~\NIL},
\end{align*}
where the symbol $a$ is encoded by having $\NIL$ as third component in a points-to assertion and
symbol $b$ by a $\NIL$ terminated list of length one.
Now, given some NFTA $\HA{T} = (Q,\Sigma,\Delta,F)$ with $\Sigma = \{a_1,\ldots,a_n\}$,
we construct a corresponding $\SRD$.
Without less of generality, we assume that $\HA{T}$ contains no unreachable or unproductive states.
We set $\PRED(\SRD) \DEFEQ Q \cup \Sigma \cup \{I\}$, where each predicate symbol is of arity one.
Then, for each symbol $a_i \in \Sigma$ one rule of the form $\SRDRULE{a_{1}}{\PROJ{\FV{0}{}}{1}=\NIL}$
or, for $1 < i \leq n$,
\begin{align*}
  a_{i} ~\SRDARROW~ & \exists z_1~z_2~\ldots~z_{i-1} ~.~
                        \PT{\PROJ{\FV{0}{}{1}}}{\NIL~\NIL~z_1} \\
                      & \qquad \SEP \bigstar_{1 \leq j < i} \PT{z_j}{\NIL~\NIL~z_{j+1}} : \{ z_{i-1} = \NIL \}
\end{align*}
is added to $\SRD$.
Furthermore, for each $(p_1 \ldots p_m, a_i, p_0) \in \Delta$, $1 \leq i \leq n$, we add a rule
\begin{align*}
  p_0 ~\SRDARROW~ & \exists z_1 \ldots z_{m+1} ~.~
                        \PT{\PROJ{\FV{0}{}}{1}}{z_1~\ldots~z_{m+1}}  \\
                    & \qquad    ~\SEP~ a_i(z_{m+1})
                        ~\SEP~ \bigstar_{1 \leq i \leq m} p_i(z_i).
\end{align*}
Finally, we add rules $\SRDRULE{I}{p\PROJ{\FV{0}}{1} : \{\PROJ{\FV{0}{}}{1} \neq \NIL \}}$ for each $p \in F$.
Clearly $\SRD$ is established.
Moreover, it is easy to verify that, given two NFTAs $\HA{T}_1$ and $\HA{T}_2$ with distinct sets of states,
we have
\begin{align*}
  I_1 x \ENTAIL{\SRD_1 \cup \SRD_2} I_2x ~\text{iff}~ L(\HA{T}_1) \subseteq L(\HA{T}_2).
\end{align*}
Thus, following~\cite{antonopoulos2014foundations},
if $\USET{Ix}{\SRD}{\CENTAIL{1}}$ can be accepted by a heap automaton,
the entailment problem $\DENTAIL{\SRD}{\CENTAIL{\alpha}}$ is \CCLASS{ExpTime}--hard for certain SIDs $\SRD$
fixed $\alpha = 1$, and a fixed arity of points-to assertions $\gamma = 3$.
Such a heap automaton can easily be constructed.
Formally, let $\HA{T} = (Q,\Sigma,\Delta,F)$ be an NFTA as above and $Q = \{p_1,\ldots,p_k\}$ for some $k > 0$.
Furthermore, for each state $p_i$, let $t_i$ be some fixed finite tree that is accepted
by the tree automaton $\HA{T}_{i} = (Q,\Sigma,\Delta,\{p_i\})$ and $\rsh_i$ be the corresponding
encoding as a reduced symbolic heap.
One possible (not necessarily efficient)
 heap automaton $\HA{A} = (Q_{\HA{A}},\SHCENTAIL{1},\Delta_{\HA{A}},F_{\HA{A}})$
is given by:
\begin{align*}
  & Q_{\HA{A}} ~\DEFEQ~  \{ \rsh_i ~|~ 1 \leq i \leq k\} \cup \{ a_i ~|~ 1 \leq i \leq n \} \\
  & F_{\HA{A}} ~\DEFEQ~  F \\
  & \MOVE{A}{q_1 \ldots q_m}{\sh}{q_0} ~\text{iff}~ \sh\left[ \PS_1 / q_1, \ldots, \PS_m / q_m\right] \ENTAIL{\SRD} \PS\PROJ{\FV{0}{}}{1}~,
\end{align*}
where each $a_i$ corresponds to the reduced symbolic heap encoding symbol $a_i$ and $\PS$ is the predicate $p_i$ corresponding to reduced symbolic heap $\rsh_i$ as previously described.


%
\section{Constructing Well-Determined SIDs} \label{app:entailment:free-vars}


%
%
In this section, we show that an established SID  -- and analogously an established symbolic heap --
can be automatically transformed into a well-determined one.
Intuitively, the main difference between an established symbolic heap $\sh$ and a determined one $\sha$
lies in their free variables:
Since each existentially quantified variable of $\sh$ is eventually allocated or equal to a free variable,
all variables apart from free variables are definitely unequal in $\sh$.
The same holds for $\sha$, but each free variable is additionally either definitely equal or definitely unequal
to any other variable of $\sha$ 
(cf. Section~\ref{sec:entailment}).
%
Thus, the main idea to transform $\sh$ into a determined symbolic heap is to add explicit pure formulas
between each free variable and each other variable.
More formally, let
$\textrm{Cmp}(\T{y},\sh)$ be the set of all sets $\PURE{}$ consisting of exactly one pure formula
$x \sim y$ for each $x \in \VAR(\sh)$ and each $y \in \T{y}$,
where $\sh \in \SL{}{}$ and $\T{y}$ is some tuple of variables.
Then, the following construction is essential to transform symbolic heaps into determined ones.
\begin{definition} \label{def:entailment:det}
  Let $\sh = \SYMBOLICHEAP{}$ be an established symbolic heap with $\NOCALLS{\sh} = m$ .
  Moreover, let $\T{y}$ be a tuple of variables of length $k \geq 0$
  and $\Lambda \in \textrm{Cmp}(\T{y},\sh)$.
  Then,
  \begin{align*}
    \textrm{det}(\sh,\T{y},\Lambda) ~\DEFEQ~ &
        \exists \BV{} ~.~ \SPATIAL{} \SEP \textrm{propagate}(\sh,\T{y}) : \PURE{} \cup \Lambda, \\
    \textrm{propagate}(\sh,\T{y}) ~\DEFEQ~ &
         (\PS_1,k)\FV{1}{}\T{y} \SEP \ldots (\PS_m,k)\FV{m}{}\T{y},
  \end{align*}
  where $(\PS_i^{\sh},\beta)$ is a predicate symbol of arity $\ARITY(\PS_i^{\sh}) + \beta$.
\end{definition}
\begin{example}
  Let $\sh$ denote the symbolic heap in the lower rule of predicate $\texttt{sll}$
  in Example~\ref{ex:srd}. Then $\textrm{det}(\sh,\FV{0}{\sh}, \emptyset)$ is given by
  \begin{align*}
    &\exists z ~.~
      \PT{\IFV{1}}{z} ~\SEP~ (\texttt{sll},2)\,z\,\IFV{2}\,\FV{0}{\sh} ~:~ \{ \IFV{1} \neq \IFV{2}\}.
  \end{align*}
\end{example}
\begin{lemma} \label{thm:entailment:free-vars}
 Let $\SRD \in \SETSRD{}$ be established. Then one can construct an SID $\SRDALT$ with $\PRED(\SRD) \subseteq \PRED(\SRDALT)$ such that for each $\PS \in \PRED(\SRD)$, we have:
 \begin{itemize}
  \item $\rsh$ is well-determined for each $\rsh \in \CALLSEM{\PS\FV{0}{}}{\SRDALT}$,
  \item $\rsh \ENTAIL{\SRD} \PS\FV{0}{}$ holds for each $\rsh \in \CALLSEM{\PS\FV{0}{}}{\SRDALT}$, and
  \item $\rsh \ENTAIL{\SRDALT} \PS\FV{0}{}$ holds for each $\rsh \in \CALLSEM{\PS\FV{0}{}}{\SRD}$. 
 \end{itemize}
\end{lemma}
Here, the last two points ensure that models of predicate calls remain unchanged,
although their actual unfoldings might differ.
\begin{proof}
 We first construct a determined SID $\Omega$.
 For each $Y,Z \in \PRED(\SRD)$, we introduce a fresh predicate symbol
 $(Y,\ARITY(Z))$, i.e.
 \begin{align*}
   & \PRED(\Omega) ~\DEFEQ~ \PRED(\SRD) ~\cup~ \{ (Y,\ARITY(Z))  ~|~ Y,Z \in \PRED(\SRD) \}.
 \end{align*}
 Furthermore, for each arity $k$ of some predicate symbol in $\SRD$,
 $\Omega$ contains two kinds of rules:
 First, there is a rule that propagates the free variables of a symbolic heap
 through all of its predicate calls while comparing it to every variable.
 Formally, we add a rule
 $\SRDRULE{\PS}{\textrm{det}(\sh,\FV{0}{\sh},\Lambda)}$
 for each $\SRDRULE{\PS}{\sh} \in \SRD$ and each $\Lambda \in \textrm{Cmp}(\FV{0}{\sh},\sh)$.
 The second kind of rule continues the propagation and comparison of a selected tuple of variables initiated
 by the first rule.
 Formally, we add a rule
 $\SRDRULE{(\PS,k)}{\textrm{det}(\sh,\T{y},\Lambda)}$, 
 where $\T{y}$ is a tuple of $k$ fresh variables,
 for each $\SRDRULE{\PS}{\sh} \in \SRD$ and each $\Lambda \in \textrm{Cmp}(\T{y},\sh)$.
 Then $\Omega$ is determined.
 To obtain a \emph{well}-determined SID,
 we apply the Refinement Theorem (Theorem~\ref{thm:compositional:refinement}) to $\Omega$
 and the heap automaton $\HASAT$ (cf. Theorem~\ref{thm:zoo:sat:property}) accepting all satisfiable reduced symbolic heaps.
 This yields the desired SID $\SRDALT$ in which every unfolding of each predicate call is satisfiable and determined.
 It remains to prove that our construction of $\SRDALT$ is correct.
 \paragraph{$\rsh$ is well--determined for each $\rsh \in \CALLSEM{\PS\FV{0}{}}{\SRDALT}$.}
 Let $\rsh \in \CALLSEM{\PS\FV{0}{}}{\SRDALT}$.
 Since each rule in $\SRDALT$ is obtained from $\SRD$ by adding pure formulas and free variables only and $\SRD$ is established, $\SRDALT$ is established as well.
 Then models of $\rsh$ differ in the interpretation of free variables only.
 By construction of $\SRDALT$ we know that for each free variable $x \in \FV{0}{}$ and each variable $y \in \VAR(\rsh)$ there exists either an equality $x = y$ or an inequality $x \neq y$ in $\PURE{\rsh}$.
 Hence, each variable is either allocated or definitely equal or unequal to each other variable by a pure formula.
 Then $\rsh$ is determined 
 (cf. Section~\ref{sec:entailment}).
 %
 Furthermore, by Theorem~\ref{thm:zoo:sat:property}, the Refinement Theorem ensures that every unfolding of each predicate symbol is satisfiable without changing the unfoldings themselves.
 Thus $\SRDALT$ is well-determined.
 \paragraph{$\rsh \ENTAIL{\SRDALT} \PS\FV{0}{}$ for each $\rsh \in \CALLSEM{\PS\FV{0}{}}{\SRD}$.}
 Let $\rsh \in \CALLSEM{\PS\FV{0}{}}{\SRD}$.
 By construction of $\Omega$, we have 
 $\rsha = \textrm{det}(\rsh,\FV{0}{\rsh},\Lambda) \in \CALLSEM{\PS\FV{0}{}}{\SRDALT}$
 for each $\Lambda \in \textrm{Cmp}(\FV{0}{\rsh},\rsh)$ such that $\rsha$ is satisfiable.
 %
 Then, for each stack-heap pair $\stack,\heap$, we have
 \begin{align*}
                     & \stack,\heap \SAT{\emptyset} \rsh \\
       ~\Rightarrow~ & \left[ \text{set}~\Lambda \DEFEQ \{ y ~\sim~ \FV{0}{\rsh} ~|~ s(y) \sim s(\FV{0}{\rsh}) \} \right] \\
                     & \stack,\heap \SAT{\emptyset} \rsh ~\wedge~ \stack,\heap \SAT{\emptyset} \Lambda \\
       ~\Rightarrow~ & \left[ \text{Definition of}~\rsha = \textrm{det}(\rsh,\FV{0}{\rsh},\Lambda) \in \CALLSEM{\PS\FV{0}{}}{\SRDALT} \right] \\
                     & \exists \rsha \in \CALLSEM{\PS\FV{0}{}}{\SRDALT} ~.~ \stack,\heap \SAT{\emptyset} \rsha \\
       ~\Rightarrow~ & \left[ \text{SL semantics} \right] \\
                     & \stack,\heap \SAT{\SRDALT} \PS\FV{0}{}.
 \end{align*}
 Hence, $\rsh \ENTAIL{\SRDALT} \PS\FV{0}{}$.
 \paragraph{$\rsh \ENTAIL{\SRD} \PS\FV{0}{}$ for each $\rsh \in \CALLSEM{\PS\FV{0}{}}{\SRDALT}$.}
 Let $\rsh \in \CALLSEM{\PS\FV{0}{}}{\SRDALT}$.
 By construction of $\SRDALT$, there exists $\rsha \in \CALLSEM{\PS\FV{0}{}}{\SRD}$ and 
 $\Lambda \in \textrm{Cmp}(\FV{0}{\rsha},\rsha)$ such that 
 $\rsh = \textrm{det}(\rsha,\FV{0}{\rsha},\Lambda)$.
 Then, for each stack-heap pair $\stack,\heap$, we have
 \begin{align*}
                  & \stack,\heap \SAT{\emptyset} \rsh \\
    ~\Rightarrow~ & \left[ \text{Definition of}~\rsh = \textrm{det}(\rsha,\FV{0}{\rsha},\Lambda) \right] \\ 
                  & \stack,\heap \SAT{\emptyset} \rsha ~\text{and}~ \stack,\heap \SAT{\emptyset} \Lambda \\
    ~\Rightarrow~ & \left[ A \wedge B \rightarrow A \right] \\
                  & \stack,\heap \SAT{\emptyset} \rsha \\
    ~\Rightarrow~ & \left[ \rsha \in \CALLSEM{\PS\FV{0}{}}{\SRD} \right] \\
                  & \exists \rsha \in \CALLSEM{\PS\FV{0}{}}{\SRD} ~.~ \stack,\heap \SAT{\emptyset} \rsha \\
    ~\Rightarrow~ & \left[ \text{SL semantics} \right] \\
                  & \stack,\heap \SAT{\SRD} \PS\FV{0}{}.
 \end{align*}
 Hence, $\rsh \ENTAIL{\SRD} \PS\FV{0}{}$
 \qed
\end{proof}
%
%
\section{Proof of Theorem~\ref{thm:entailment:top-level}} \label{app:entailment:top-level}
Recall from Definition~\ref{def:entailment:heapmodels} the set of all reduced symbolic heaps in $\SRDCLASS$ (over the same free variables as $\sh$) entailing an unfolding of $\sh$:
 \begin{align*}
  & \USET{\sh}{\SRD}{\SRDCLASS} 
  \DEFEQ \{ \rsha \in \RSL{}{\SRDCLASS} ~|~ \NOFV{\rsha} = \NOFV{\sh} \,\text{and}\, \exists \rsh \in \CALLSEM{\sh}{\SRD} \,.\, \rsha \ENTAIL{\emptyset} \rsh \}
 \end{align*}
\begingroup
\def\thetheorem{\ref{thm:entailment:top-level}}
\begin{theorem}
 Let $\alpha \in \N$ and $\SRD \in \SETSRD{\SRDCLASSFV{\alpha}}$ be established.
 Moreover, for each predicate symbol $\PS \in \PRED(\SRD)$, 
 let there be a heap automaton over $\SHCENTAIL{\alpha}$ accepting 
 $\USET{\PS\T{x}}{\SRD}{\CENTAIL{\alpha}}$.
 %
 Then, for every well-determined symbolic heap $\sh \in \SL{\SRD}{}$, 
 there is a heap automaton over $\SHCENTAIL{\alpha}$ accepting
 $\USET{\sh}{\SRD}{\CENTAIL{\alpha}}$.
 %
 %
 %
 %
 %
\end{theorem}
\endgroup
  The proof is by induction on the syntax of symbolic heaps.
  Since each case requires the construction of a suitable heap automaton and a corresponding correctness proof, we split the proof across multiple lemmas, which are shown subsequently in the remainder of this section.
  Furthermore, we make some preliminary remarks to reduce the technical effort.
  \begin{remark}
  \label{remark:entailment:top-level:track}
  We assume without loss of generality
  that every predicate call occurring in a symbolic heap $\sh \in \SHCENTAIL{\alpha}$
  is annotated with a state of the tracking automaton $\HATRACK$ introduced in
  Definition~\ref{def:zoo:track-automaton}.
  Thus, for every predicate call $\CALLN{i}{\sh}$, we assume the availability of 
  \begin{itemize}
    \item a set $A_i \subseteq \FV{i}{\sh}$ capturing all allocated variables in a (fixed)
          unfolding of $\CALLN{i}{\sh}$, and
    \item a finite set of pure formulas $\PURE{}_i$ over $\FV{i}{\sh}$ capturing all definite
          equalities and inequalities between variables in $\FV{i}{\sh}$. 
  \end{itemize}
  This assumption is an optimization that prevents us from applying 
  Lemma~\ref{thm:zoo:track:property} again and again to keep track
  of the aforementioned relationships between parameters of predicate calls.
  \end{remark}  
  \begin{remark}
  \label{remark:entailment:top-level:alpha}
  Furthermore, by Definition~\ref{def:entailment:centail}, 
  every predicate call has at most $\alpha$ parameters.
  %
  %
  However, the number of free variables of the symbolic heap itself, 
  i.e. the formula at the root of an unfolding tree, 
  is not necessarily bounded by $\alpha$.
  To avoid additional case distinctions in correctness proofs, we present
  constructions where the number of free variables is bounded by $\alpha$ as well.
  In each of our constructions,
  this corner case can be dealt with by adding one dedicated final state $q_{>\alpha}$
  such that $(\EMPTYSEQ,\rsh,q_{> \alpha}) \in \Delta$ if and only if
  the number of free variables $\NOFV{\rsh}$ is greater than $\alpha$
  and the property in question is satisfied by $\rsh$.
  General symbolic heaps with more than $\alpha$ free variables are treated analogously.
  Once the final state $q_{> \alpha}$ has been reached no further transition is possible.
  Thus, only symbolic heaps at the root of an unfolding tree may enter $q_{>\alpha}$.
  \end{remark}
\begin{proof}[Proof of Theorem~\ref{thm:entailment:top-level}]
  The remainder of the proof is summarized as follows:
  \begin{itemize}
   \item The base cases $\PT{x}{\T{y}}$ and $\EMP$ are shown in Section~\ref{sec:entailment:points-to} and Section~\ref{sec:entailment:emp}, respectively.
   \item For predicate calls $\PS\FV{0}{}$,
         we already know by assumption that $\USET{\PS\T{x}}{\SRD}{\CENTAIL{\alpha}}$
         is accepted by a heap automaton over $\SHCENTAIL{\alpha}$.
   %
   \item Section~\ref{sec:entailment:sepcon} shows that $\USET{\sh \SEP \sha}{\SRD}{\CENTAIL{\alpha}}$ is 
         accepted by a heap automaton over $\SHCENTAIL{\alpha}$  
         if $\sh,\sha$ are defined over the same set of free variables.
         This is sufficient, because we do not require that these variables actually occur in both symbolic heaps.
   \item Section~\ref{sec:entailment:pure} shows that $\USET{\sh : \Pi}{\SRD}{\CENTAIL{\alpha}}$ is 
         accepted by a heap automaton over $\SHCENTAIL{\alpha}$,
         where $\Pi$ is a finite set of pure formulas over the free variables of $\sh$.
   \item Section~\ref{sec:entailment:existential} shows that $\USET{\exists z ~.~ \sh}{\SRD}{\CENTAIL{\alpha}}$ is 
         accepted by a heap automaton over $\SHCENTAIL{\alpha}$,
         where $z$ is a single variable.
         Repeated application then completes the proof.
  \end{itemize}
  Hence, for each symbolic heap $\sh \in \SL{\SRD}{\CENTAIL{\alpha}}$, a corresponding heap automaton over $\CENTAIL{\alpha}$ accepting $\USET{\sh}{\SRD}{\CENTAIL{\alpha}}$ can be constructed.
  %
  %
  \qed
\end{proof}
%
%
\subsubsection{The Points-to Assertion} \label{sec:entailment:points-to}
Let $\T{u} \in U_n = \{1, \ldots,\alpha\} \times \{0,1, \ldots,\alpha\}^{n-1}$ for some fixed natural number $n \geq 1$.
Moreover, we associate the points-to assertion
\begin{align*}
 \rsha_{\T{u}} ~\DEFEQ~ \PT{\PROJ{\FV{0}{}}{\PROJ{\T{u}}{1}}}{
    \PROJ{\FV{0}{}}{\PROJ{\T{u}}{2}} \ldots \PROJ{\FV{0}{}}{\PROJ{\T{u}}{n}}
 }
\end{align*}
with each tuple $\T{u} \in U_n$.\footnote{Recall that $\PROJ{\FV{0}{}}{0} = \NIL$.}
We construct a heap automaton $\HASH{\rsha_{\T{v}}} = (Q,\SHCENTAIL{\alpha},\Delta,F)$
accepting $\USET{\rsha_{\T{v}}}{\SRD}{\CENTAIL{\alpha}}$ for some fixed tuple $\T{v} \in U_n$ as follows:
\begin{align*}
  Q ~\DEFEQ~ \{ \rsha_{\T{u}} ~|~ \T{u} \in U_n \} \cup \{\EMP\}
  \qquad
  F ~\DEFEQ~ \{ \rsha_{\T{v}} \}
\end{align*}
Moreover, for $\sh \in \SL{}{\CENTAIL{\alpha}}$ with $\NOCALLS{\sh} = m$,
 the transition relation $\Delta$ is given by:
\begin{align*}
              & \MOVE{\HASH{\rsha_{\T{v}}}}{q_0}{\sh}{q_1 \ldots q_m}
 \quad\text{iff}\quad 
                    \REDUCE{\sh, q_1 \ldots q_m} \ENTAIL{\emptyset} q_0
\end{align*}
where $\REDUCE{\sh, q_1 \ldots q_m} \DEFEQ \sh[\PS_1 / q_1 : \Pi_i, \ldots, \PS_m / q_m : \Pi_m]$ and each set $\Pi_i$ denotes the set of pure formulas obtained from the tracking automaton (see Definition~\ref{def:zoo:track}),
which is assumed to be readily available by 
 Remark~\ref{remark:entailment:top-level:track}. 
%
\begin{lemma} \label{thm:entailment:points-to:compositionality}
 $\HASH{\rsha_{\T{v}}}$ satisfies the compositionality property.
\end{lemma}
\begin{proof}
 Let $\sh \in \SL{}{\CENTAIL{\alpha}}$ with $\NOCALLS{\sh} = m$.
 Moreover, for each $1 \leq i \leq m$, let $\rsh_i \in \RSL{}{\CENTAIL{\alpha}}$ and $\rsh \DEFEQ \sh[\PS_1^{\sh} / \rsh_1, \ldots, \PS_m^{\sh} / \rsh_m]$.
 Assume $\OMEGA{\HASH{\rsha_{\T{v}}}}{q_0}{\rsh}$.
 Three cases arise depending on the value of $q_0 \in Q$.
 \\
 \emph{Case 1:} $q_0 = \EMP$\qquad
 \begin{align*}
                         & \OMEGA{\HASH{\rsha_{\T{v}}}}{\EMP}{\rsh} \\
       ~\Leftrightarrow~ & \left[ \text{Definition of}~\Delta,~\NOCALLS{\rsh} = 0 \right] \\
                         & \rsh \ENTAIL{\SRD} \EMP \\
       ~\Leftrightarrow~ & \left[ \text{choose}~q_1 = \ldots = q_m = \EMP \right] \\
                         & \bigwedge_{1 \leq i \leq m} q_i = \EMP ~\text{and}~ \rsh \ENTAIL{\SRD} \EMP \\
       ~\Leftrightarrow~ & \left[ \rsh = \sh[\PS_1^{\sh} / \rsh_1, \ldots, \PS_m^{\sh} / \rsh_m] ~\text{contains no points-to assertion} \right] \\
                         & \bigwedge_{1 \leq i \leq m} q_i = \EMP ~\text{and}~ \sh \ENTAIL{\SRD} \EMP \\
                         & \qquad \text{and}~ \forall 1 \leq i \leq m ~.~ \rsh_i \ENTAIL{\SRD} \EMP \\
       ~\Leftrightarrow~ & \left[ \text{Definition of}~\Delta \right] \\
                         & \exists q_1,\ldots,q_m \in Q ~.~ \MOVE{\HASH{\rsha_{\T{v}}}}{q_0}{\sh}{q_1 \ldots q_m} \\
                         & \qquad \text{and}~ \forall 1 \leq i \leq m ~.~ \OMEGA{\HASH{\rsha_{\T{v}}}}{q_i}{\rsh_i}.
 \end{align*}
 \emph{Case 2:} $q_0 = \rsha_{\T{u}}$\qquad
  \begin{align*}
                         & \OMEGA{\HASH{\rsha_{\T{v}}}}{\rsha_{\T{u}}}{\rsh} \\
       ~\Leftrightarrow~ & \left[ \text{Definition of}~\Delta,~\NOCALLS{\rsh} = 0 \right] \\
                         & \rsh \ENTAIL{\SRD} \rsha_{\T{u}} \tag{$\dag$}
  \end{align*}
  Now, $\rsh \ENTAIL{\SRD} \rsha_{\T{u}}$ holds if and only if $\rsh$ contains at most one points-to assertion $\rsha_{\T{w}}$, $\T{w} \in U_n$, where each variable is free.
  If this points-to assertion is contained in $\sh$ then choose $q_1 = \ldots = q_m = \EMP$.
  Otherwise, if exactly $\rsh_k$, $1 \leq k \leq m$ contains a points-to assertion, choose $q_k = \rsha_{\T{w}}$ and $q_i = \EMP$ for each $i \neq k$.
  With these choices, our computation continues as follows:
  \begin{align*}
     (\dag) ~\Leftrightarrow~ & \left[ \text{choice of}~q_1,\ldots,q_m \right] \\
                              & \exists q_1,\ldots,q_m \in Q ~.~ \rsh \ENTAIL{\SRD} \rsha_{\T{u}} \\
                              & \qquad \text{and}~ \forall 1 \leq i \leq m ~.~ \rsh_i \ENTAIL{\SRD} q_i \\
            ~\Leftrightarrow~ & \left[ \text{Definition of}~\Delta \right] \\
                              & \exists q_1,\ldots,q_m \in Q ~.~ \rsh \ENTAIL{\SRD} \rsha_{\T{u}} \\
                              & \qquad \text{and}~ \forall 1 \leq i \leq m ~.~ \OMEGA{\HASH{\rsha_{\T{v}}}}{q_i}{\rsh_i}
  \end{align*}
  Since each variable of points-to assertion $\rsha_{\T{w}}$ is equal to a free variable and, by Lemma~\ref{thm:zoo:track:property}, two free variables are equal in $\rsh$ if and only if they are equal in $\REDUCE{\sh, q_1 \ldots q_m}$, we have $\REDUCE{\sh, q_1 \ldots q_m} \ENTAIL{\SRD} \rsha_{\T{u}}$ if and only if $\rsh \ENTAIL{\sh} \rsha_{\T{u}}$. Thus:
  \begin{align*}
            ~\Leftrightarrow~ & \left[ \REDUCE{\sh, q_1 \ldots q_m} \ENTAIL{\SRD} \rsha_{\T{u}} ~\text{iff}~ \rsh \ENTAIL{\sh} \rsha_{\T{u}} \right] \\
                              & \exists q_1,\ldots,q_m \in Q ~.~ \REDUCE{\sh, q_1 \ldots q_m} \ENTAIL{\SRD} \rsha_{\T{u}} \\
                              & \qquad \text{and}~ \forall 1 \leq i \leq m ~.~ \OMEGA{\HASH{\rsha_{\T{v}}}}{q_i}{\rsh_i} \\
            ~\Leftrightarrow~ & \left[ \text{Definition of}~\Delta \right] \\
                              & \exists q_1,\ldots,q_m \in Q ~.~ \MOVE{\HASH{\rsha_{\T{v}}}}{q_0}{\sh}{q_1 \ldots q_m} \\
                              & \qquad \text{and}~ \forall 1 \leq i \leq m ~.~ \OMEGA{\HASH{\rsha_{\T{v}}}}{q_i}{\rsh_i}. 
  \end{align*}
  \qed
\end{proof}
\begin{lemma}
 $L(\HASH{\rsha_{\T{v}}}) = \USET{\rsha_{\T{v}}}{\SRD}{\CENTAIL{\alpha}}$.
\end{lemma}
\begin{proof}
  Let $\rsh \in \RSL{}{\CENTAIL{\alpha}}$. Then
  \begin{align*}
                        & \rsh \in L(\HASH{\rsha_{\T{v}}}) \\
      ~\Leftrightarrow~ & \left[ \text{Definition of}~L(\HASH{\rsha_{\T{v}}}) \right] \\
                        & \exists q \in F ~.~ \OMEGA{\HASH{\rsha_{\T{v}}}}{q}{\rsh} \\
      ~\Leftrightarrow~ & \left[ \text{Definition of}~F \right] \\
                        & \OMEGA{\HASH{\rsha_{\T{v}}}}{\rsha_{\T{v}}}{\rsh} \\
      ~\Leftrightarrow~ & \left[ \text{Definition of}~\Delta \right] \\
                        & \REDUCE{\sh, \EMPTYSEQ} \ENTAIL{\emptyset} \rsha_{\T{v}} \\
      ~\Leftrightarrow~ & \left[ \NOCALLS{\rsh} = 0 ~\text{implies}~ \rsh = \REDUCE{\sh, \EMPTYSEQ} \right] \\
                        & \rsh \ENTAIL{\emptyset} \rsha_{\T{v}} \\
      ~\Leftrightarrow~ & \left[ \text{Definition of}~\USET{\rsha_{\T{v}}}{\SRD}{\CENTAIL{\alpha}} \right] \\
                        & \rsh \in \USET{\rsha_{\T{v}}}{\SRD}{\CENTAIL{\alpha}}. 
  \end{align*}
  \qed
\end{proof}
%
\subsubsection{The Empty Heap} \label{sec:entailment:emp}
A heap automaton $\HASH{\EMP}$ accepting $\USET{\EMP}{\SRD}{\CENTAIL{\alpha}}$ is constructed analogously to
$\HASH{\rsha_{\T{v}}}$, the heap automaton constructed for the points-to assertion in Section~\ref{sec:entailment:points-to}.
The only exception is that $F_{\HASH{\EMP}} \DEFEQ \{ \EMP \}$ is chosen as set of final states.\footnote{The set of points-to assertions $U_n$ in this construction can even be chosen to be the empty set.} Since the compositionality property for this automaton already has been shown in Lemma~\ref{thm:entailment:points-to:compositionality}, it suffices to show
that $\HASH{\EMP}$ accepts the correct language.
\begin{lemma}
  $L(\HASH{\EMP}) = \USET{\EMP}{\SRD}{\CENTAIL{\alpha}}$.
\end{lemma}
\begin{proof}
  Let $\rsh \in \RSL{}{\CENTAIL{\alpha}}$. Then
  \begin{align*}
                      & \rsh \in L(\HASH{\EMP}) \\
    ~\Leftrightarrow~ & \left[ \text{Definition of}~L(\HASH{\EMP}) \right] \\
                      & \exists q_0 \in F ~.~ \OMEGA{\HASH{\EMP}}{q_0}{\rsh} \\
    ~\Leftrightarrow~ & \left[ \text{Definition of}~F \right] \\
                      & \OMEGA{\HASH{\EMP}}{\EMP}{\rsh} \\
    ~\Leftrightarrow~ & \left[ \text{Definition of}~\Delta,~\REDUCE{\sh, \EMPTYSEQ} = \rsh \right] \\
                      & \rsh \ENTAIL{\emptyset} \EMP \\
    ~\Leftrightarrow~ & \left[ \rsh \in \RSL{}{\CENTAIL{\alpha}} \right] \\
                      & \rsh \ENTAIL{\SRD} \EMP \\
    ~\Leftrightarrow~ & \left[ \text{Definition of}~\USET{\EMP}{\SRD}{\CENTAIL{\alpha}} \right] \\
                      & \rsh \in \USET{\EMP}{\SRD}{\CENTAIL{\alpha}}. 
  \end{align*}
  \qed
\end{proof}
\subsubsection{The Separating Conjunction}  \label{sec:entailment:sepcon}
Some additional notation is needed.
First, we write $\sh \equiv \sha$ if two symbolic heaps $\sh,\sha$ are identical up to the ordering of spatial formulas.
Given symbolic heaps $\sh$,$\sha$ and some variable $x$, we define
\begin{align*}
  [x]_{\sh,\sha} ~\DEFEQ~
    \begin{cases}
      x & ~\text{if}~x \in \VAR(\sh) \cup \VAR(\sha) \\
      \EMPTYSEQ & ~\text{otherwise}.
    \end{cases}
\end{align*}
Moreover, given two tuples $\T{x}$,$\T{y}$ taken from a finite set of variables $\{z_1,\ldots,z_n\}$, let $\T{x} \uplus \T{y} \DEFEQ [z_1]_{\sh,\sha} \,.\, \ldots \,.\, [z_n]_{\sh,\sha}$.
Similar to splitting a heap $\heap = \heap_1 \uplus \heap_2$ into two heaps, we define an operation to split \emph{symbolic} \emph{heaps}.
\begin{definition} \label{def:entailment:uplus}
 Let $\sh,\sha \in \SL{}{}$ with free variables $\FV{0}{\sh}$, $\FV{0}{\sha}$ and
 \begin{align*}
   \Gamma^{\sh} ~=~ \PS_1\FV{1}{\sh} \SEP \ldots \SEP \PS_m\FV{m}{\sh}
   ~\text{and}~
   \Gamma^{\sha} ~=~ \PS_1\FV{1}{\sha} \SEP \ldots \SEP \PS_m\FV{m}{\sha}.
 \end{align*}
 Moreover, let $\BV{\sh},\BV{\sha} \subseteq \{z_1,\ldots,z_n\}$ for some $n \in \N$ and for each $0 \leq i \leq m$, let $\FV{i}{\sh},\FV{i}{\sh} \subseteq \FV{i}{}$, where $\FV{i}{}$ is some finite set of variables.
 Then $\sh \uplus \sha$ is defined as
 \begin{align*}
   \sh \uplus \sha ~\DEFEQ~ & \exists (\BV{\sh} \uplus \BV{\sha}) ~.~ \SPATIAL{\sh} \SEP \SPATIAL{\sha} \SEP \Gamma : \PURE{\sh} \cup \PURE{\sha}, \text{where} \\
   \Gamma ~\DEFEQ~ & \PS_1(\FV{1}{\sh} \uplus\FV{1}{\sha}) \SEP \ldots \SEP \PS_m(\FV{m}{\sh} \uplus\FV{m}{\sha})
 \end{align*}
 with free variables $\FV{0}{\sh} \uplus \FV{0}{\sha}$.%
 \footnote{Note that technically, the predicate symbols $\PS_i$ may have different arities for each of the aforementioned symbolic heaps due to different numbers of parameters. This can easily be fixed by renaming predicate symbols. However, we chose to ignore the arity here to improve readability.}
 Otherwise, $\sh \uplus \sha$ is undefined.
\end{definition}
Since inequalities between allocated variables of a symbolic heap $\sh$ are redundant, let $\KERNEL(\sh)$ denote the corresponding symbolic heap in which such inequalities are removed.
Formally, $\KERNEL(\sh)$ is given by:
\begin{align*}
 \SYMBOLICHEAP{\sh} \setminus \{ a \neq b ~|~ a,b \in \ALLOC{\sh} \cup \{\NIL\} \}.
\end{align*}
Note that $\KERNEL(\sh)$ is computable in polynomial time if $\sh$ is reduced or each predicate is attached with variable tracking information (see Definition~\ref{def:zoo:track-automaton}).

We are now in a position to deal with the separating conjunction of two symbolic heaps.
Let $\sha_1,\sha_2 \in \SL{}{\CENTAIL{\alpha}}$ with $\VAR(\sha_1) = \VAR(\sha_2) = \FV{0}{\sha_1} = \FV{0}{\sha_2}$.
Moreover, let $\HASH{\sha_1}$, $\HASH{\sha_2}$ be heap automata over $\SHCENTAIL{\alpha}$ accepting $\USET{\sha_1}{\SRD}{\CENTAIL{\alpha}}$ and $\USET{\sha_2}{\SRD}{\CENTAIL{\alpha}}$, respectively.
We construct a heap automaton $\HASH{\sha_1 \SEP \sha_2} = (Q,\SHCENTAIL{\alpha},\Delta,F)$ accepting $\USET{\sha_1 \SEP \sha_2}{\SRD}{\CENTAIL{\alpha}}$ as follows:
\begin{align*}
   Q ~\DEFEQ~ & 2^{\T{u}} \times Q_{\HASH{\sha_1}} \times 2^{\T{u}} \times Q_{\HASH{\sha_2}} \times 2^{\T{u}}
   \\
   F ~\DEFEQ~ & 2^{\T{u}} \times F_{\HASH{\sha_1}} \times 2^{\T{u}} \times F_{\HASH{\sha_2}} \times 2^{\T{u}},~
\end{align*}
where $\T{u}$ denotes a tuple of $\alpha$ variables.
The main idea of choosing this state space is that each state $(U,p,V,q,H)$ stores the current states $p,q$ of heap automata $\HASH{\sha_1}$ and $\HASH{\sha_2}$ running on two separated symbolic heaps that are obtained from splitting the originally given symbolic heap into two parts.
These splitted symbolic heaps require some synchronization information encoded in sets $U,V,K$.
More precisely, $U,V$ store which free variables of the original symbolic heap belong to the symbolic heap fed into $\HASH{\sha_1}$ and $\HASH{\sha_2}$, respectively.
Finally, $H$ stores which variables are \emph{allocated} in the symbolic heap fed into $\HASH{\sha_1}$.
Correspondingly, all other allocated variables must be allocated in the symbolic heap fed into $\HASH{\sha_2}$.
These sets are checked in the transition relation to ensure compositionality.
Formally, for $\sh \in \SL{}{\CENTAIL{\alpha}}$ with $\NOCALLS{\sh} = m$, the transition relation $\Delta$ is given by:
\begin{align*}
                & (U_1,p_1,V_1,q_1,H_1) \ldots (U_m,p_m,V_m,q_m,H_m) \\
                & \qquad \xrightarrow{\sh}_{\HASH{\sha_1 \SEP \sha_2}} (U_0,p_0,V_0,q_0,H_0) \\
   ~\text{iff}~ & \exists \sha_p,\sha_q ~.~ \KERNEL(\sh) \equiv \sha_p \uplus \sha_q \\
                & \quad \text{and}~ \MOVE{\HASH{\sha_1}}{p_0}{\sha_p}{p_1 \ldots p_m} ~\text{and}~ \MOVE{\HASH{\sha_2}}{q_0}{\sha_q}{q_1 \ldots q_m} \\
                & \quad \text{and}~ \forall 0 \leq i \leq m \,.\, \forall 1 \leq j \leq \alpha ~.~ \\
                & \qquad \quad~~ \left( \PROJ{\T{u}}{j} \in U_i ~\text{iff}~ \PROJ{\FV{i}{\sh}}{j} \in \FV{i}{\sha_p}  \right) \\
                & \qquad \text{and}~ \left( \PROJ{\T{u}}{j} \in V_i ~\text{iff}~ \PROJ{\FV{i}{\sh}}{j} \in \FV{i}{\sha_q}  \right) \\
                & \qquad \text{and}~ \left( \PROJ{\T{u}}{j} \in H_i ~\text{iff}~ \PROJ{\FV{i}{\sh}}{j} \in \ALLOC{\sha_p}  \right).
\end{align*}
\begin{lemma}
 $\HASH{\sha_1 \SEP \sha_2}$ satisfies the compositionality property.
\end{lemma}
\begin{proof}
Let $\sh \in \SL{}{\CENTAIL{\alpha}}$ with $\NOCALLS{\sh} = m$.
Moreover, for each $1 \leq i \leq m$, let $\rsh_i \in \RSL{}{\CENTAIL{\alpha}}$ and $\rsh \DEFEQ \sh[\PS_1^{\sh} / \rsh_1, \ldots, \PS_m^{\sh} / \rsh_m]$.
To improve readability, we write $U_i,p_i,V_i,q_i,H_i$ to denote the respective component of an automaton state
$S_i = (U_i,p_i,V_i,q_i,H_i) \in Q$.
Let $S_0 \in Q$. Then:
\begin{align*}
                    & \exists S_1,\ldots,S_m \in Q ~.~ \\
                    & \quad \MOVE{\HASH{\sha_1 \SEP \sha_2}}{S_0}{\sh}{S_1 \ldots S_m} ~ \text{and}~ \forall 1 \leq k \leq m ~.~ \OMEGA{\HASH{\sha_1 \SEP \sha_2}}{S_k}{\rsh_i} \\
  ~\Leftrightarrow~ & \left[ \text{Definition of}~\Delta \right] \\
                    & \exists S_1,\ldots,S_m \in Q ~.~ \\
                    & \quad \exists \sha_p,\sha_q ~.~ \KERNEL(\sh) \equiv \sha_p \uplus \sha_q \\
                    & \qquad \text{and}~ \MOVE{\HASH{\sha_1}}{p_0}{\sha_p}{p_1 \ldots p_m} ~\text{and}~ \MOVE{\HASH{\sha_2}}{q_0}{\sha_q}{q_1 \ldots q_m} \\
		    & \quad \text{and}~ \forall 0 \leq i \leq m \,.\, \forall 1 \leq j \leq \alpha ~.~ \\
		    & \qquad \quad~~ \left( \PROJ{\T{u}}{j} \in U_i ~\text{iff}~ \PROJ{\FV{i}{\sh}}{j} \in \FV{i}{\sha_p}  \right) \\
		    & \qquad \text{and}~ \left( \PROJ{\T{u}}{j} \in V_i ~\text{iff}~ \PROJ{\FV{i}{\sh}}{j} \in \FV{i}{\sha_q}  \right) \\
		    & \qquad \text{and}~ \left( \PROJ{\T{u}}{j} \in H_i ~\text{iff}~ \PROJ{\FV{i}{\sh}}{j} \in \ALLOC{\sha_p}  \right) \\
                    & \quad \text{and}~ \forall 1 \leq k \leq m \,.\, \exists \rsh_{k,p}, \rsh_{k,q} ~.~ \\
                    & \qquad \KERNEL(\rsh_k) \equiv \rsh_{k,p} \uplus \rsh_{k,q} \\
                    & \qquad \text{and}~ \OMEGA{\HASH{\sha_1}}{p_0}{\rsh_{k,p}} ~\text{and}~ \OMEGA{\HASH{\sha_2}}{q_0}{\rsh_{k,q}} \\
		    & \qquad \text{and}~ \forall 1 \leq j \leq \alpha ~.~ \\
		    & \qquad \quad~~ \left( \PROJ{\T{u}}{j} \in U_k ~\text{iff}~ \PROJ{\FV{0}{\rsh_k}}{j} \in \FV{0}{\rsh_{k,p}}  \right) \\
		    & \qquad \text{and}~ \left( \PROJ{\T{u}}{j} \in V_k ~\text{iff}~ \PROJ{\FV{0}{\rsh_k}}{j} \in \FV{0}{\rsh_{k,q}}  \right) \\
		    & \qquad \text{and}~ \left( \PROJ{\T{u}}{j} \in H_k ~\text{iff}~ \PROJ{\FV{0}{\rsh_k}}{j} \in \ALLOC{\rsh_{k,p}}  \right)
\end{align*}
To avoid introducing further notation, 
we do not distinguish between $\KERNEL(\rsh)$ and $\rsh$ in the following.\footnote{recall that $\KERNEL$ just removes inequalities that are contained implicitly in a symbolic heap due to points-to assertions}
Now, for each $1 \leq k \leq m$, the conditions on $U_k$ and $V_k$ guarantee that the free variables of each $\rsh_{k,p}$ and $\rsh_{k,q}$, $1 \leq k \leq m$ match with the parameters of each predicate call $\CALLN{k}{\sha_p}$ and $\CALLN{k}{\sha_q}$, respectively.
Hence, the symbolic heaps
\begin{align*}
  \rsh_p ~\DEFEQ~ \sha_{p}[\PS_1 / \rsh_{1,p}, \ldots, \PS_m / \rsh_{m,p}], \\
  \rsh_q ~\DEFEQ~ \sha_{q}[\PS_1 / \rsh_{1,q}, \ldots, \PS_m / \rsh_{m,q}]
\end{align*}
exist.
Assume, for the moment, that $\rsh_p \uplus \rsh_q \equiv \rsh$ holds under the conditions on sets $U_k,V_k,H_k$.
Then, the calculations from above continue as follows:
\begin{align*}
  ~\Leftrightarrow~ & \left[ \text{Compositionality of}~\HASH{\sha_1},\HASH{\sha_2} \right] \\
                    & \exists S_1,\ldots,S_m \in Q ~.~ \\
                    & \quad \exists \sha_p,\sha_q ~.~ \KERNEL(\sh) \equiv \sha_p \uplus \sha_q \\
                    & \qquad \text{and}~ \OMEGA{\HASH{\sha_1}}{p_0}{\rsh_p} ~\text{and}~ \OMEGA{\HASH{\sha_2}}{q_0}{\rsh_q} \\
		    & \quad \text{and}~ \forall 0 \leq i \leq m \,.\, \forall 1 \leq j \leq \alpha ~.~ \\
		    & \qquad \quad~~ \left( \PROJ{\T{u}}{j} \in U_i ~\text{iff}~ \PROJ{\FV{i}{\sh}}{j} \in \FV{i}{\sha_p}  \right) \\
		    & \qquad \text{and}~ \left( \PROJ{\T{u}}{j} \in V_i ~\text{iff}~ \PROJ{\FV{i}{\sh}}{j} \in \FV{i}{\sha_q}  \right) \\
		    & \qquad \text{and}~ \left( \PROJ{\T{u}}{j} \in H_i ~\text{iff}~ \PROJ{\FV{i}{\sh}}{j} \in \ALLOC{\sha_p}  \right) \\
                    & \quad \text{and}~ \forall 1 \leq k \leq m \,.\, \exists \rsh_{k,p}, \rsh_{k,q} ~.~ \\
                    & \qquad \KERNEL(\rsh_k) \equiv \rsh_{k,p} \uplus \rsh_{k,q} \\
		    & \qquad \text{and}~ \forall 1 \leq j \leq \alpha ~.~ \\
		    & \qquad \quad~~ \left( \PROJ{\T{u}}{j} \in U_k ~\text{iff}~ \PROJ{\FV{0}{\rsh_k}}{j} \in \FV{0}{\rsh_{k,p}}  \right) \\
		    & \qquad \text{and}~ \left( \PROJ{\T{u}}{j} \in V_k ~\text{iff}~ \PROJ{\FV{0}{\rsh_k}}{j} \in \FV{0}{\rsh_{k,q}}  \right) \\
		    & \qquad \text{and}~ \left( \PROJ{\T{u}}{j} \in H_k ~\text{iff}~ \PROJ{\FV{0}{\rsh_k}}{j} \in \ALLOC{\rsh_{k,p}}  \right) \\
  ~\Leftrightarrow~ & \left[ \text{Assumption:}~\rsh_p \uplus \rsh_q \equiv \rsh \right] \\
                    & \exists \rsh_p,\rsh_q ~.~ \KERNEL(\sh) \equiv \rsh_p \uplus \rsh_q \\
                    & \qquad \text{and}~ \OMEGA{\HASH{\sha_1}}{p_0}{\rsh_p} ~\text{and}~ \OMEGA{\HASH{\sha_2}}{q_0}{\rsh_q} \\
		    & \quad \text{and}~ \forall 1 \leq j \leq \alpha ~.~ \\
		    & \qquad \quad~~ \left( \PROJ{\T{u}}{j} \in U_0 ~\text{iff}~ \PROJ{\FV{0}{\sh}}{j} \in \FV{0}{\rsh_p}  \right) \\
		    & \qquad \text{and}~ \left( \PROJ{\T{u}}{j} \in V_0 ~\text{iff}~ \PROJ{\FV{0}{\sh}}{j} \in \FV{0}{\rsh_q}  \right) \\
		    & \qquad \text{and}~ \left( \PROJ{\T{u}}{j} \in H_0 ~\text{iff}~ \PROJ{\FV{0}{\sh}}{j} \in \ALLOC{\rsh_p}  \right) \\
  ~\Leftrightarrow~ & \left[ \text{Definition of}~\Delta \right] \\
                    & \OMEGA{\HASH{\sha_1 \SEP \sha_2}}{S_0}{\rsh}.
\end{align*}
It remains to prove the assumption $\rsh_p \uplus \rsh_q \equiv \rsh$.
Since both symbolic heaps are subformulas of $\rsh$ and, by the condition on sets $H_k$, no variable is allocated twice, we have:
\begin{align*}
           & \rsh_p \uplus \rsh_q \\
       ~=~ & \left[ ~\text{Definition}~\rsh_p,\rsh_q \right] \\
           & \sha_{p}[\PS_1 / \rsh_{1,p}, \ldots, \PS_m / \rsh_{m,p}] \uplus \sha_{q}[\PS_1 / \rsh_{1,q}, \ldots, \PS_m / \rsh_{m,q}] \\
      ~=~  & \left[ \text{Definition of predicate replacement} \right] \\
           & \exists \BV{\sha_p} \cdot \BV{\rsh_{1,p}} \cdot \ldots \cdot \BV{\rsh_{m,p}} ~.~ \\
           & \qquad \SPATIAL{\sha_p} \SEP \SPATIAL{\rsh_{1,p}}[\FV{0}{\rsh_{1,p}} / \FV{1}{\sha_p}] \SEP \SPATIAL{\rsh_{m,p}}[\FV{0}{\rsh_{m,p}} / \FV{m}{\sha_p}] \\
           & \qquad : \PURE{\sha_p} \cup \PURE{\rsh_{1,p}}[\FV{0}{\rsh_{1,p}} / \FV{1}{\sha_p}] \cup {\rsh_{m,p}}[\FV{0}{\rsh_{m,p}} / \FV{m}{\sha_p}] \\
           & \uplus \\
           & \exists \BV{\sha_q} \cdot \BV{\rsh_{1,q}} \cdot \ldots \cdot \BV{\rsh_{m,q}} ~.~ \\
           & \qquad \SPATIAL{\sha_q} \SEP \SPATIAL{\rsh_{1,q}}[\FV{0}{\rsh_{1,q}} / \FV{1}{\sha_q}] \SEP \SPATIAL{\rsh_{m,q}}[\FV{0}{\rsh_{m,q}} / \FV{m}{\sha_q}] \\
           & \qquad : \PURE{\sha_q} \cup \PURE{\rsh_{1,q}}[\FV{0}{\rsh_{1,q}} / \FV{1}{\sha_q}] \cup {\rsh_{m,q}}[\FV{0}{\rsh_{m,q}} / \FV{m}{\sha_q}] \\
       ~=~ & \left[ \text{Definition}~\ref{def:entailment:uplus} \right] \\
           & \exists \left( \BV{\sha_p} \cdot \BV{\rsh_{1,p}} \cdot \ldots \cdot \BV{\rsh_{m,p}} \uplus \BV{\sha_q} \cdot \BV{\rsh_{1,q}} \cdot \ldots \cdot \BV{\rsh_{m,q}} \right) ~.~ \\
           & \qquad \SPATIAL{\sha_p} \SEP \SPATIAL{\rsh_{1,p}}[\FV{0}{\rsh_{1,p}} / \FV{1}{\sha_p}] \SEP \SPATIAL{\rsh_{m,p}}[\FV{0}{\rsh_{m,p}} / \FV{m}{\sha_p}] \\
           & \qquad \SEP \SPATIAL{\sha_q} \SEP \SPATIAL{\rsh_{1,q}}[\FV{0}{\rsh_{1,q}} / \FV{1}{\sha_q}] \SEP \SPATIAL{\rsh_{m,q}}[\FV{0}{\rsh_{m,q}} / \FV{m}{\sha_q}] \\
           & \qquad : \PURE{\sha_p} \cup \PURE{\rsh_{1,p}}[\FV{0}{\rsh_{1,p}} / \FV{1}{\sha_p}] \cup {\rsh_{m,p}}[\FV{0}{\rsh_{m,p}} / \FV{m}{\sha_p}] \\
           & \qquad \cup \PURE{\sha_q} \cup \PURE{\rsh_{1,q}}[\FV{0}{\rsh_{1,q}} / \FV{1}{\sha_q}] \cup {\rsh_{m,q}}[\FV{0}{\rsh_{m,q}} / \FV{m}{\sha_q}]
\end{align*}
Now, by definition, we know that $\BV{\rsh_{k,p} \uplus \rsh_{k,q}} = \BV{\rsh_{k}}$ and $\BV{\sha_p \uplus \sha_q} = \BV{\sh}$, i.e., the existentially quantified variables reduce to $\BV{\sh} \cdot \BV{\rsh_1} \cdot \ldots \cdot \BV{\rsh_m}$.
Analogously, for each $1 \leq k \leq m$, we have
\begin{align*}
 & \SPATIAL{\rsh_{k}}[\FV{0}{\rsh_{k}} / \FV{k}{\sh}] \equiv \SPATIAL{\rsh_{k,p}}[\FV{0}{\rsh_{k,p}} / \FV{k}{\sha_p}] \SEP \SPATIAL{\rsh_{k,q}}[\FV{0}{\rsh_{k,q}} / \FV{k}{\sha_q}], ~\text{and} \\
 & \PURE{\rsh_k}[\FV{0}{\rsh_{k}} / \FV{k}{\sh}] = \PURE{\rsh_{k,q}}[\FV{0}{\rsh_{k,q}} / \FV{k}{\sha_q}].
\end{align*}
%
Thus
\begin{align*}
  ~\equiv~ & \left[ \text{comment from above} \right] \\
           & \exists \BV{\sh} \cdot \BV{\rsh_1} \cdot \ldots \cdot \BV{\rsh_m} ~.~ \\
           & \qquad \SPATIAL{\sha_p} \SEP \SPATIAL{\sha_q} \\
           & \qquad \SEP \SPATIAL{\rsh_{1}}[\FV{0}{\rsh_{1}} / \FV{1}{\sh}]
                    \SEP \ldots
                    \SEP \SPATIAL{\rsh_{m}}[\FV{0}{\rsh_{m}} / \FV{m}{\sh}] \\
           & \qquad : \PURE{\sha_p} \cup \PURE{\sha_q} \\
           & \qquad \cup \PURE{\rsh_1}[\FV{0}{\rsh_{1}} / \FV{1}{\sh}]
                    \cup \ldots
                    \cup \PURE{\rsh_m}[\FV{0}{\rsh_{m}} / \FV{m}{\sh}] \\
  ~\equiv~ & \left[ \SPATIAL{\sh} \equiv \SPATIAL{\sha_p} \SEP \SPATIAL{\sha_q}, \PURE{\sh} = \PURE{\sha_p} \cup \PURE{\sha_q} \right] \\
           & \exists \BV{\sh} \cdot \BV{\rsh_1} \cdot \ldots \cdot \BV{\rsh_m} ~.~ \\
           & \qquad \SPATIAL{\sh} \SEP \SPATIAL{\rsh_{1}}[\FV{0}{\rsh_{1}} / \FV{1}{\sh}]
                    \SEP \ldots
                    \SEP \SPATIAL{\rsh_{m}}[\FV{0}{\rsh_{m}} / \FV{m}{\sh}] \\
           & \qquad : \PURE{\sh} \cup \PURE{\rsh_1}[\FV{0}{\rsh_{1}} / \FV{1}{\sh}]
                    \cup \ldots
                    \cup \PURE{\rsh_m}[\FV{0}{\rsh_{m}} / \FV{m}{\sh}] \\
  ~\equiv~ & \left[ \text{Definition of predicate replacement} \right] \\
           & \sh\left[\CALLN{1}{\sh} / \rsh_1, \ldots \CALLN{m}{\sh} / \rsh_m\right] \\
  ~\equiv~ & \left[ \text{Definition of}~\rsh \right] \\
           & \rsh.
\end{align*}
Hence, $\HASH{\sha_1 \SEP \sha_2}$ satisfies the compositionality property.
\qed
\end{proof}
\begin{lemma}
 $L(\HASH{\sha_1 \SEP \sha_2}) = \USET{\sha_1 \SEP \sha_2}{\SRD}{\CENTAIL{\alpha}}$.
\end{lemma}
\begin{proof}
 Let $\rsh \in \RSL{}{\CENTAIL{\alpha}}$. Then
 \begin{align*}
                      & \rsh \in L(\HASH{\sha_1 \SEP \sha_2}) \\
    ~\Leftrightarrow~ & \left[ \text{Definition of}~L(\HASH{\sha_1 \SEP \sha_2}) \right] \\
                      & \exists S_0 = (U_0,p_0,V_0,q_0,H_0) \in F ~.~ \OMEGA{\HASH{\sha_1 \SEP \sha_2}}{S_0}{\rsh} \\
    ~\Leftrightarrow~ & \left[ \text{Definition of}~\Delta \right] \\
                      & \exists S_0 = (U_0,p_0,V_0,q_0,H_0) \in F ~.~ \\
                      & \quad \exists \sha_p,\sha_q ~.~ \KERNEL(\rsh) \equiv \sha_p \uplus \sha_q \\
                      & \qquad \text{and}~ \OMEGA{\HASH{\sha_1}}{p_0}{\sha_p} ~\text{and}~ \OMEGA{\HASH{\sha_2}}{q_0}{\sha_q} \\
		      & \quad \text{and}~ \forall 1 \leq j \leq \alpha ~.~ \\
		      & \qquad \quad~~ \left( \PROJ{\T{u}}{j} \in U_0 ~\text{iff}~ \PROJ{\FV{0}{\rsh}}{j} \in \FV{0}{\sha_p}  \right) \\
		      & \qquad \text{and}~ \left( \PROJ{\T{u}}{j} \in V_0 ~\text{iff}~ \PROJ{\FV{0}{\rsh}}{j} \in \FV{0}{\sha_q}  \right) \\
		      & \qquad \text{and}~ \left( \PROJ{\T{u}}{j} \in H_0 ~\text{iff}~ \PROJ{\FV{0}{\rsh}}{j} \in \ALLOC{\sha_p}  \right) \\
    ~\Leftrightarrow~ & \left[ p_0 \in F_{\HASH{\sha_1}}, q_0 \in F_{\HASH{\sha_2}}, \text{Definition of}~L(\HASH{\sha_1}),L(\HASH{\sha_2}) \right] \\
                      & \exists \sha_p,\sha_q ~.~ \KERNEL(\rsh) \equiv \sha_p \uplus \sha_q \\
                      & \quad \text{and}~ \sha_p \ENTAIL{\SRD} \sha_1 ~\text{and}~ \sha_q \ENTAIL{\SRD} \sha_2 \\
    ~\Leftrightarrow~ & \left[ \KERNEL(\rsh) \ENTAIL{\emptyset} \rsh ~\text{and}~ \rsh \ENTAIL{\emptyset} \KERNEL(\rsh) \right] \\
                      & \exists \sha_p,\sha_q ~.~ \sha_p \uplus \sha_q \ENTAIL{\emptyset} \rsh ~\text{and}~ \rsh \ENTAIL{\emptyset} \sha_p \uplus \sha_q \\
                      & \quad \text{and}~ \sha_p \ENTAIL{\SRD} \sha_1 ~\text{and}~ \sha_q \ENTAIL{\SRD} \sha_2.
 \end{align*}
 To complete the proof we show
 \begin{align*}
                      &\exists \sha_p,\sha_q ~.~ \sha_p \uplus \sha_q \ENTAIL{\emptyset} \rsh ~\text{and}~ \rsh \ENTAIL{\emptyset} \sha_p \uplus \sha_q \\
                      & \quad \text{and}~ \sha_p \ENTAIL{\SRD} \sha_1 ~\text{and}~ \sha_q \ENTAIL{\SRD} \sha_2 \\
         ~\text{iff}~ & \rsh \ENTAIL{\SRD} \sha_1 \SEP \sha_2.
 \end{align*}
 This, together with the previously shown equivalences, immediately yields
 \begin{align*}
   \rsh \in L(\HASH{\sha_1 \SEP \sha_2}) ~\text{iff}~ \rsh \ENTAIL{\SRD} \sha_1 \SEP \sha_2 ~\text{iff}~ \rsh \in \HEAPMODELS{\sha_1 \SEP \sha_2}{\SRD}.
 \end{align*}

 For the first direction, we have
 \begin{align*}
                     & \stack,\heap \SAT{\emptyset} \rsh \\
       ~\Rightarrow~ & \left[ \rsh \ENTAIL{\emptyset} \sha_p \uplus \sha_q \right] \\
                     & \stack,\heap \SAT{\emptyset} \sha_p \uplus \sha_q \\
       ~\Rightarrow~ & \left[ \text{Definition}~\ref{def:entailment:uplus},~\BV{\rsh} = \BV{\sha_p} \uplus \BV{\sha_q} \right] \\
                     & \stack,\heap \SAT{\emptyset} \exists \BV{\rsh} ~.~ \SPATIAL{\sha_p} \SEP \SPATIAL{\sha_q} ~:~ \PURE{\sha_p} \cup \PURE{\sha_q} \\
                     & \quad \text{and}~ \stack,\heap \SAT{\SRD} \sha_1 \SEP \sha_2 \\
       ~\Rightarrow~ & \left[ \text{SL semantics}~ \right] \\
                     & \exists \T{v} \in \VAL^{\SIZE{\BV{\rsh}}} \,.\, \exists \heap_1,\heap_2 ~.~ \heap = \heap_1 \uplus \heap_2 \\
                     & \quad \text{and}~ \stack[\BV{\rsh} \mapsto \T{v}],\heap_1 \SAT{\emptyset} \SPATIAL{\sha_p}
                            ~\text{and}~ \stack[\BV{\rsh} \mapsto \T{v}],\heap_2 \SAT{\emptyset} \SPATIAL{\sha_q} \\
                     & \quad \text{and}~ \stack[\BV{\rsh} \mapsto \T{v}],\heap \SAT{\emptyset} \PURE{\sha_p} \cup \PURE{\sha_q} \\
       ~\Rightarrow~ & \left[ \text{elementary logic},~\text{SL semantics}~ \right] \\
                    & \exists \heap_1,\heap_2 ~.~ \heap = \heap_1 \uplus \heap_2 \\
                     & \quad \text{and}~ \stack,\heap_1 \SAT{\emptyset} \sha_p ~\text{and}~ \stack,\heap_2 \SAT{\emptyset} \sha_q \\
       ~\Rightarrow~ & \left[ \sha_p \ENTAIL{\SRD} \sha_1, \sha_q \ENTAIL{\SRD} \sha_2 \right] \\
                   & \exists \heap_1,\heap_2 ~.~ \heap = \heap_1 \uplus \heap_2 \\
                     & \quad \text{and}~ \stack,\heap_1 \SAT{\SRD} \sha_1 ~\text{and}~ \stack,\heap_2 \SAT{\SRD} \sha_2 \\
       ~\Rightarrow~ & \left[ \text{SL semantics}~ \right] \\
                     & \stack,\heap \SAT{\SRD} \sha_1 \SEP \sha_2.
 \end{align*}
 Hence, $\rsh \ENTAIL{\SRD} \sha_1 \SEP \sha_2$.
 For the converse direction, we have
 \begin{align*}
                     & \rsh \ENTAIL{\SRD} \sha_1 \SEP \sha_2 \\
   ~\Rightarrow~ & \left[ \text{Definition entailment} \right] \\
                     & \forall (\stack,\heap) ~.~ \stack,\heap \SAT{\SRD} \rsh ~\text{implies}~ \stack,\heap \SAT{\SRD} \sha_1 \SEP \sha_2 \\
   ~\Rightarrow~ & \left[ \rsh~\text{reduced and well-determined} \right] \\
                     & \exists (\stack,\heap) ~.~ \stack,\heap \SAT{\emptyset} \rsh ~\text{and}~ \stack,\heap \SAT{\SRD} \sha_1 \SEP \sha_2 \\
   ~\Rightarrow~ & \left[ \rsh = \exists \BV{} . \SPATIAL{} : \PURE{},~\text{SL semantics} \right] \\
                     & \exists (\stack,\heap) . \exists \heap_1,\heap_2 ~.~ \heap = \heap_1 \uplus \heap_2 \\
                     & \quad \text{and}~ \exists \BV{} \in \VAL^{\SIZE{\BV{}}} \,.\, \stack[\BV{} \mapsto \T{v}], \heap \SAT{\emptyset} \SPATIAL{} \\
                     & \qquad \text{and}~ \stack[\BV{} \mapsto \T{v}], \heap \SAT{\emptyset} \PURE{} \\
                     & \quad \text{and}~ \stack,\heap_1 \SAT{\SRD} \sha_1
                            ~\text{and}~ \stack,\heap_2 \SAT{\SRD} \sha_2 \\
   ~\Rightarrow~ & \left[ \text{SL semantics} \right] \\
                     & \exists (\stack,\heap) . \exists \heap_1,\heap_2 ~.~ \heap = \heap_1 \uplus \heap_2 ~\text{and}~ \exists \BV{} \in \VAL^{\SIZE{\BV{}}} \,.\, \\
                     & \quad \heap = \{ \stack[\BV{} \mapsto \T{v}](x) \mapsto \stack[\BV{} \mapsto \T{v}](\T{y}) ~|~ \PT{x}{\T{y}} ~\text{in}~ \SPATIAL{} \} \\
                     & \qquad \text{and}~ \stack[\BV{} \mapsto \T{v}], \heap \SAT{\emptyset} \PURE{} \\
                     & \quad \text{and}~ \stack,\heap_1 \SAT{\SRD} \sha_1
                            ~\text{and}~ \stack,\heap_2 \SAT{\SRD} \sha_2.
 \end{align*}
 Now, we choose $\sha_p, \sha_q$ such that $\KERNEL(\rsh) = \sha_p \uplus \sha_q$ by adding exactly those points-to assertions $\PT{x}{\T{y}}$ to $\SPATIAL{\sha_p}$, where $\stack[\BV{} \mapsto \T{v}](x) \in \DOM(\heap_1)$. All other points-to assertions are added to $\sha_q$.
 The existentially quantified variables and pure formulas are chosen such that both symbolic heaps are well-determined. This is always possible, because $\sha_1$ and $\sha_2$ are well-determined on their own.
 Then $\stack,\heap_1 \SAT{\emptyset} \sha_p$ and $\stack,\heap_2 \SAT{\emptyset} \sha_q$.
 Since both $\sha_p$ and $\sha_q$ are well-determined, this implies $\sha_p \ENTAIL{\SRD} \sha_1$ and $\sha_q \ENTAIL{\SRD} \sha_q$.
 Moreover, since each points-to assertion and each pure formula in the core occurs in $\sha_p$ or $\sha_q$, we have $\KERNEL(\rsh) = \sha_p \uplus \sha_q$.
 Hence, $\rsh \ENTAIL{\SRD} \sha_1 \SEP \sha_2$ implies
 \begin{align*}
                      &\exists \sha_p,\sha_q ~.~ \sha_p \uplus \sha_q \ENTAIL{\emptyset} \rsh ~\text{and}~ \rsh \ENTAIL{\emptyset} \sha_p \uplus \sha_q \\
                      & \quad \text{and}~ \sha_p \ENTAIL{\SRD} \sha_1 ~\text{and}~ \sha_q \ENTAIL{\SRD} \sha_2. 
 \end{align*}
 \qed
\end{proof}
\subsubsection{The Pure Formula}  \label{sec:entailment:pure}
Dealing with pure formulas is rather straightforward, because we already know how to track equalities and inequalities using the construction from Definition~\ref{def:zoo:track-automaton}.
Let $\shb \in \SL{}{\CENTAIL{\alpha}}$ such that $\USET{\shb}{\SRD}{\CENTAIL{\alpha}}$ 
there exists a heap automaton, say $\HASH{\shb}$, accepting $\USET{\shb}{\SRD}{\CENTAIL{\alpha}}$.
Moreover, let $\PURE{}$ be some finite set of pure formulas over the set of free variables of $\shb$.
We construct a heap automaton $\HASH{\shb : \PURE{}} = (Q,\SHCENTAIL{\alpha},\Delta,F)$ accepting $\USET{\sha}{\SRD}{\CENTAIL{\alpha}}$, where $\sha \DEFEQ \SYMBOLICHEAP{\shb} \cup \PURE{}$, as follows:
\begin{align*}
 Q ~\DEFEQ~ Q_{\HASH{\shb}} \times \{0,1\} \qquad F ~\DEFEQ~ F_{\HASH{\shb}} \times \{1\}.
\end{align*}
Moreover, for $\sh \in \SL{}{\CENTAIL{\alpha}}$ with $\NOCALLS{\sh} = m$, the transition relation $\Delta$ is given by:
\begin{align*}
              & \MOVE{\HASH{\shb : \PURE{}}}{(p_0,q_0)}{\sh}{(p_1,q_1) \ldots (p_m,q_m)} \\
 ~\text{iff}~ & \MOVE{\HASH{\shb}}{p_0}{\sh}{p_1 \ldots p_m} ~\text{and}~
                \bigwedge_{0 \leq i \leq m} q_i = 1 ~\text{iff}~ \PURE{} \subseteq \PURE{}_i,
\end{align*}
where $\PURE{}_i$ is the set of pure formulas obtained from the tracking automaton (see Definition~\ref{def:zoo:track-automaton}) that is available by Remark~\ref{remark:entailment:top-level:track}. 
%
\begin{lemma}
 $\HASH{\shb : \PURE{}}$ satisfies the compositionality property.
\end{lemma}
\begin{proof}
 Let $\sh \in \SL{}{\CENTAIL{\alpha}}$ with $\NOCALLS{\sh} = m$.
 Moreover, for each $1 \leq i \leq m$, let $\rsh_i \in \RSL{}{\CENTAIL{\alpha}}$ and $\rsh \DEFEQ \sh[\PS_1^{\sh} / \rsh_1, \ldots, \PS_m^{\sh} / \rsh_m]$.
 For each $1 \leq i \leq m$, let $\PURE{}_i$ be the set of pure formulas obtained by the tracking automaton (cf. Definition~\ref{def:zoo:track-automaton}) for $\rsh_i$.
 Note that this information is available without running the full automaton again 
due to Remark~\ref{remark:entailment:top-level:track}. 
 %
 Then, for each $(p_0,q_0) \in Q$, we have
 \begin{align*}
                      & \OMEGA{\HASH{\shb : \PURE{}}}{(p_0,q_0)}{\rsh} \\
    ~\Leftrightarrow~ & \left[ \text{Definition of}~\Delta \right] \\
                      & \OMEGA{\HASH{\shb}}{p_0}{\rsh} ~\text{and}~ q_0 = 1 ~\text{iff}~ \PURE{} \subseteq \PURE{}_0 \\
                      & \left[ \text{for each $1 \leq i \leq m$ set}~q_i = 1~\text{iff}~ \PURE{} \subseteq \Pi_i \right] \\
                      & \exists q_1,\ldots,q_m \in \{0,1\} ~.~ \\
                      & \qquad \OMEGA{\HASH{\shb}}{p_0}{\rsh} ~\text{and}~ \bigwedge_{0 \leq i \leq m} q_i = 1 ~\text{iff}~ \PURE{} \subseteq \PURE{}_i \\
    ~\Leftrightarrow~ & \left[ \text{Compositionality of}~\HASH{\shb} \right] \\
                      & \exists q_1,\ldots,q_m \in \{0,1\} ~.~ \exists p_1,\ldots,p_m \in Q_{\HASH{\shb}} ~.~ \\
                      & \qquad \MOVE{\HASH{\shb}}{p_0}{\sh}{p_1 \ldots p_m} ~\text{and}~ \bigwedge_{0 \leq i \leq m} q_i = 1 ~\text{iff}~ \PURE{} \subseteq \PURE{}_i \\
                      & \qquad \text{and}~ \forall 1 \leq i \leq m ~.~ \OMEGA{\HASH{\shb}}{p_i}{\rsh_i} \\
    ~\Leftrightarrow~ & \left[ \text{Definition of}~Q,~\text{regrouping} \right] \\
                      & \exists (p_1,q_1),\ldots,(p_m,q_m) \in Q ~.~ \\
                      & \qquad \MOVE{\HASH{\shb}}{p_0}{\sh}{p_1 \ldots p_m} ~\text{and}~ \bigwedge_{0 \leq i \leq m} q_i = 1 ~\text{iff}~ \PURE{} \subseteq \PURE{}_i \\
                      & \qquad \text{and}~ \forall 1 \leq i \leq m ~.~ \OMEGA{\HASH{\shb}}{p_i}{\rsh_i} ~\text{and}~ q_i = 1 ~\text{iff}~ \PURE{} \subseteq \PURE{}_i  \\
    ~\Leftrightarrow~ & \left[ \text{Definition of}~\Delta \right] \\
                      & \exists (p_1,q_1),\ldots,(p_m,q_m) \in Q ~.~ \\
                      & \qquad \MOVE{\HASH{\shb : \PURE{}}}{(p_0,q_0)}{\sh}{(p_1,q_1) \ldots (p_m,q_m)} \\
                      & \qquad \text{and}~ \forall 1 \leq i \leq m ~.~ \OMEGA{\HASH{\shb : \PURE{}}}{(p_i,q_i)}{\rsh_i}. 
 \end{align*}
 \qed
\end{proof}
\begin{lemma}
 $L(\HASH{\shb : \PURE{}}) = \USET{\sha}{\SRD}{\CENTAIL{\alpha}}$, where $\sha \DEFEQ \SYMBOLICHEAP{\shb} \cup \PURE{}$ and $\shb \in \SL{\SRD}{\CENTAIL{\alpha}}$.
\end{lemma}
\begin{proof}
 Let $\rsh \in \RSL{}{\CENTAIL{\alpha}}$. Then
 \begin{align*}
                       & \rsh \in L(\HASH{\shb : \PURE{}}) \\
     ~\Leftrightarrow~ & \left[ ~\text{Definition of}~L(\HASH{\shb : \PURE{}}) \right] \\
                       & \exists (p_0,q_0) \in F ~.~ \OMEGA{\HASH{\shb : \PURE{}}}{(p_0,q_0)}{\rsh} \\
     ~\Leftrightarrow~ & \left[ ~\text{Definition of}~F \right] \\
                       & \exists p_0 \in F_{\HASH{\shb}} ~.~ \OMEGA{\HASH{\shb : \PURE{}}}{(p_0,1)}{\rsh} \\
     ~\Leftrightarrow~ & \left[ ~\text{Definition of}~\Delta, q_0 = 1 \right] \\
                       & \exists p_0 \in F_{\HASH{\shb}} ~.~ \OMEGA{\HASH{\shb}}{p_0}{\rsh} \\
                       & \qquad \text{and}~ \bigwedge_{0 \leq i \leq \NOCALLS{\rsh}} q_i = 1 ~\text{iff}~ \PURE{} \subseteq \PURE{}_i \\
~\Leftrightarrow~ & \left[ ~\NOCALLS{\rsh} = 0, q_0 = 1 \right] \\
                       & \exists p_0 \in F_{\HASH{\shb}} ~.~ \OMEGA{\HASH{\shb}}{p_0}{\rsh} ~\text{and}~ \PURE{} \subseteq \PURE{}_0 \\
     ~\Leftrightarrow~ & \left[ ~\text{Definition of}~L(\HASH{\shb}) \right] \\
                       & \rsh \in L(\HASH{\shb}) ~\text{and}~ \PURE{} \subseteq \PURE{}_0 \\
     ~\Leftrightarrow~ & \left[ L(\HASH{\shb}) = \USET{\shb}{\SRD}{\CENTAIL{\alpha}} \right] \\
                       & \rsh \in \USET{\shb}{\SRD}{\CENTAIL{\alpha}} ~\text{and}~ \PURE{} \subseteq \PURE{}_0 \\
     ~\Leftrightarrow~ & \left[ ~\text{Definition of}~\Pi_0~\text{(see Definition~\ref{def:zoo:track})} \right] \\
                       & \rsh \in \USET{\shb}{\SRD}{\CENTAIL{\alpha}} ~\text{and}~ \PURE{} \subseteq \{ x \sim y ~|~ x \MSIM{\rsh} y \} \\
     ~\Leftrightarrow~ & \left[ \text{Definition}~\USET{\shb}{\SRD}{\CENTAIL{\alpha}},~\MSIM{\rsh} \right] \\
                       & \rsh \ENTAIL{\SRD} \shb ~\text{and}~ \rsh \ENTAIL{\emptyset} \PURE{} \\
     ~\Leftrightarrow~ & \left[ \sha \DEFEQ \SYMBOLICHEAP{\shb} \cup \PURE{} \right] \\
                       & \rsh \ENTAIL{\SRD} \sha \\
     ~\Leftrightarrow~ & \left[ \text{Definition}~\USET{\sha}{\SRD}{\CENTAIL{\alpha}} \right] \\
                       & \rsh \in \USET{\sha}{\SRD}{\CENTAIL{\alpha}}. 
 \end{align*}
 \qed
\end{proof}
\subsubsection{The Existential Quantification}  \label{sec:entailment:existential}
In order to deal with existential quantifiers, we exploit the following structural property.
\begin{lemma} \label{thm:entailment:existential:structural-auxiliary}
  For $\SRD \in \SETSRD{}$, let $\rsh \in \RSL{}{\alpha}$ be a well--determined symbolic heap and $\sha \in \SL{\SRD}{\alpha}$.
  Moreover, let $x \notin \VAR(\rsh)$ and $\FV{0}{\sha} = \FV{0}{\rsh} \cdot x$.
  Then $\rsh \ENTAIL{\SRD} \exists x \,.\, \sha$ holds if and only if there exists $y \in \VAR(\rsh)$ such that
  %
    $\exists \BV{\rsh} . \SPATIAL{\rsh} : \PURE{\rsh} \cup \{ x = y \} \ENTAIL{\SRD} \sha$.
  %
\end{lemma}
\begin{proof}
 To improve readability, we write $\rsh_y$ as a shortcut for $\exists \BV{\rsh} . \SPATIAL{\rsh} : \PURE{\rsh} \cup \{ x = y \}$.
 \\
 \emph{``only-if''}\qquad
 Towards a contradiction, assume that $\rsh \ENTAIL{\SRD} \exists x . \sha$ holds, but for each $y \in \VAR(\rsh)$, we have $\rsh_y \not \ENTAIL{\SRD} \sha$.
  Since $\rsh$ is established and well--determined and $x \notin \VAR(\rsh)$, we know that
  $\rsh_y$ is established and well--determined as well.
  Thus, let $(\stack,\heap) \in \MODELS{\rsh_y}$ be the unique model of $\rsh_y$ up to isomorphism.
  Clearly, $\stack(x) = \stack(y)$ due to the pure formula $x = y$.
  Then, by assumption, $\stack,\heap \not \SAT{\SRD} \sha$.

  Now, let $(\stack',\heap') \in \MODELS{\rsh}$ be the unique model of $\rsh$ up to isomorphism.
  Since $\rsh \ENTAIL{\SRD} \exists x . \sha$, applying the SL semantics yields that there exists $v \in \VAL$ such that $\stack'[x \mapsto v],\heap' \SAT{\SRD} \sha$.
  Moreover, $v$ is drawn from the image (or codomain) of $\stack'$ and $\DOM(\heap')$, because $\sha$ is established.
  By Lemma~\ref{thm:symbolic-heaps:fv-coincidence} and $x \notin \VAR(\rsh)$, we know that $\stack'[x \mapsto v],\heap' \SAT{\emptyset} \rsh$.
  Observe that for each of these values $v$ there exists at least one variable $y \in \VAR(\tau)$ such that $y$ is evaluated to $v$.%
  \footnote{For $v \in co\DOM(\stack')$ this is clear, because there exists $y \in \DOM(\stack')$ with $\stack'(y) = v$. Otherwise, if $v \in \DOM(\heap')$, applying the SL semantics yields that there exists $\T{u} \in co\DOM(\stack') \cup \DOM(\heap')$ such that $\stack'[\BV{\rsh} \mapsto \T{u}],\heap' \SAT{\emptyset} \SPATIAL{\rsh} : \PURE{\rsh}$. Then, since $\DOM(\heap') = \{\stack'[\BV{\rsh} \mapsto \T{u}](z) ~|~ \PT{z}{\_} ~\text{occurs in}~ \SPATIAL{\rsh}\}$, $v$ is contained in $co\DOM(\stack'[\BV{\rsh} \mapsto \T{u}])$. Thus there exists some variable $y \in \DOM(\stack'[\BV{\rsh} \mapsto \T{u}]) = \VAR(\rsh)$ with $\stack'[\BV{\rsh} \mapsto \T{u}](y) = v$.}
  Then $\stack'[x \mapsto v],\heap' \SAT{\emptyset} \rsh_{y}$ holds for some $y \in \VAR(\rsh)$ that is evaluated to $v$.
  However, since $\stack'[x \mapsto v],\heap' \SAT{\SRD} \sha$, this means that $\rsh_y \ENTAIL{\SRD} \psi$.
  This contradicts our assumption $\rsh_y \not \ENTAIL{\SRD} \sha$ for each $y \in \VAR(\rsh)$.%
  \footnote{Note that $\rsh_y$ is well--determined. Thus, one common model between $\sha$ and $\rsh_y$ is sufficient to prove the entailment $\rsh_y \ENTAIL{\SRD} \psi$.}
 \\
 \emph{``if''}\qquad
 Assume there exists $y \in \VAR(\rsh)$ such that $\rsh_y \ENTAIL{\SRD} \sha$.
 Thus, for each $\stack,\heap \SAT{\emptyset} \rsh_y$, we have $\stack,\heap \SAT{\SRD} \sha$.
 Furthermore, since $\rsh$ is established and $x$ occurs in $\rsh_y$ in the pure formula $x = y$ only, we know that
 \begin{align*}
  \stack(x) \in \stack(\FV{0}{\rsh}) \cup \DOM(\heap). \tag{$\dag$}
 \end{align*}
 In particular, $\stack,\heap \SAT{\emptyset} \rsh$ holds, because only one pure formula is added by $\rsh_y$.
 Moreover, since $x \notin \VAR(\rsh)$, $\restr{\stack}{\FV{0}{\rsh}},\heap \SAT{\emptyset} \rsh$ holds 
 by Lemma~\ref{thm:symbolic-heaps:fv-coincidence}.
 Then
 \begin{align*}
                 & \stack,\heap \SAT{\emptyset} \rsh_y \\
   ~\Rightarrow~ & \left[ \text{Assumption:}~\rsh_y \ENTAIL{\SRD} \sha \right] \\
                 & \stack,\heap \SAT{\emptyset} \sha \\
   ~\Rightarrow~ & \left[ \text{applying}~(\dag), v = \stack(x)  \right] \\
                 & \exists v \in \stack(\FV{0}{\rsh}) \cup \DOM(\heap) ~.~ (\restr{\stack}{\FV{0}{\rsh}})[x \mapsto v] = \stack ~\text{and}~ \stack,\heap \SAT{\emptyset} \sha \\
   ~\Rightarrow~ & \left[ (\restr{\stack}{\FV{0}{\rsh}})[x \mapsto v] = \stack \right] \\
                 & \exists v \in \stack(\FV{0}{\rsh}) \cup \DOM(\heap) ~.~ (\restr{\stack}{\FV{0}{\rsh}})[x \mapsto v],\heap \SAT{\emptyset} \sha \\
   ~\Rightarrow~ & \left[ \stack(\FV{0}{\rsh}) \cup \DOM(\heap) \subseteq \VAL \right] \\
                 & \exists v \in \VAL ~.~ (\restr{\stack}{\FV{0}{\rsh}})[x \mapsto v],\heap \SAT{\emptyset} \sha \\
   ~\Rightarrow~ & \left[ \text{SL semantics} \right] \\
                 & (\restr{\stack}{\FV{0}{\rsh}}),\heap \SAT{\emptyset} \exists x . \sha. \\
 \end{align*}
 Hence, each model of $\rsh$ is also a model of $\exists x . \sha$, i.e., $\rsh \ENTAIL{\SRD} \exists x \,.\, \sha$.
 \qed
\end{proof}
Note that the proof of the lemma from above works analogously if the additional free variable $x$ occurs at a different position in the tuple of free variables.
Following the previous lemma, the main idea is to nondeterministically guess some
variable $y$ of an unfolding $\rsh$ and verify that $\rsh : \{ x = y\}$ belongs
to $\USET{\sha}{\SRD}{\CENTAIL{\alpha}}$.
Since this construction increases the number of free variables, we remark that $y$ has to be chosen carefully such that a symbolic heap in $\SL{}{\alpha}$ is obtained.
In particular, if $y$ does not belong to the root node of an unfolding tree, the corresponding predicate must either have an arity smaller than $\alpha$, or $y$ is equal to some other free variable.
However, each $y$ occurring in the root node of an unfolding tree may be chosen.
Before we present a formal construction, we first define how a symbolic heap $\sh$ is modified in the following.
\begin{definition}
  Let $\T{x}$ be a tuple of variables with $\SIZE{\T{x}} = n$.
  Moreover, let $1 \leq k \leq n+1$. Then, the tuple of variables in which a fresh variable $x$ is placed at position $k$ is given by:
  \begin{align*}
     \T{x}\SQUEEZE{x}{k} ~\DEFEQ~ \PROJ{\T{x}}{1}\, \ldots \, \PROJ{\T{x}}{k-1} ~ x ~ \PROJ{\T{x}}{k+1} \, \ldots \, \PROJ{\T{x}}{n}.
  \end{align*}
  Now, let $\sh$ be a symbolic heap with $\NOFV{\sh} = \beta < \alpha$.
  Then, the symbolic heap $\sh\SQUEEZE{x}{k}$ is defined as $\sh$ except for the tuple of free variables being $\FV{0}{\sh}\SQUEEZE{x}{k}$.
  Furthermore, for some $1 \leq i \leq \NOCALLS{\sh}$ and $1 \leq \ell \leq \SIZE{\FV{i}{\sh}} + 1$,
  the symbolic heap $\sh\SQUEEZE{x}{i,\ell}$ is defined as $\sh$ except for the tuple of parameters of the $i$-th predicate call being $\FV{i}{\sh\SQUEEZE{x}{i,\ell}} = \FV{i}{\sh}\SQUEEZE{x}{\ell}$.
  Otherwise, $\sh\SQUEEZE{x}{k}$ as well as $\sh\SQUEEZE{x}{i,k}$ are undefined.
\end{definition}
Now, let $\HASH{\sha}$ be a heap automaton accepting $\USET{\sha}{\SRD}{\CENTAIL{\alpha}}$ and $\HATRACK$ be the tracking automaton introduced in Definition~\ref{def:zoo:track-automaton}.
We construct a heap automaton $\HASH{\exists z.\sha} = (Q,\SHCENTAIL{\alpha},\Delta,F)$ accepting $\USET{\exists x . \sha}{\SRD}{\CENTAIL{\alpha}}$, where $x$ is the $k$-th free variable of $\sha$, as follows:
\begin{align*}
   Q ~\DEFEQ~ Q_{\HASH{\sha}} \times \{0,1,\ldots,\alpha\} \times Q_{\HATRACK} 
   \qquad
   F ~\DEFEQ~ F_{\HASH{\sha}} \times \{k\} \times Q_{\HATRACK} 
\end{align*}
Moreover, for $\sh \in \SL{}{\CENTAIL{\alpha}}$ with $\NOCALLS{\sh} = m$, the transition relation $\Delta$ is given by:
\begin{align*}
                       & \MOVE{\HASH{\exists z. \sha}}{(p_0,q_0,r_0)}{\sh}{(p_1,q_1,r_m) \ldots (p_m,q_m,r_m)} \\
  \quad\text{iff}\quad &
                         \MOVE{\HASH{\sha}}{p_0}{\shb}{p_1 \ldots p_m}
                         ~\text{and}~
                         \MOVE{\HATRACK}{r_0}{\shb}{r_1 \ldots r_m}~,
\end{align*}
where $\shb$ adheres to one of the following cases:
\begin{enumerate}
 \item $q_0 = q_1 = \ldots = q_m = 0$ and $\shb = \sh$, or
 \item $q_0 = \ell > 0$ and there exists exactly one $1 \leq j \leq m$ such that $q_j = \ell' > 0$ and $\shb = (\sh\SQUEEZE{x}{\ell})\SQUEEZE{x}{j,\ell'}$, or
 \item $q_0 = \ell > 0$ and $\sum_{1 \leq i \leq m} q_i = 0$ and $\shb = (\sh\SQUEEZE{x}{\ell}) : \{ x = y \}$ for some $y \in \VAR(\sh)$.
\end{enumerate}
Here, $\sh : \{ y = x \}$ is a shortcut for the symbolic heap $\SYMBOLICHEAP{\sh} \cup \{ x = y \}$.
Moreover, the annotations of each predicate call $\CALLN{i}{\shb}$ are set to the sets contained in $r_i = (A_i,\Pi_i)$ instead of using existing annotations.
We remark that this construction is highly non-deterministic.
Furthermore, note that the automata rejects a symbolic heap if some $\shb$ does not belong to $\SL{}{\CENTAIL{\alpha}}$, because it is not a top-level formula and has more than $\alpha$ free variables.
\begin{lemma}
 $\HASH{\exists z.\sha}$ satisfies the compositionality property.
\end{lemma}
\begin{proof}
 Let $\sh \in \SL{}{\CENTAIL{\alpha}}$ with $\NOCALLS{\sh} = m$.
 Moreover, for each $1 \leq i \leq m$, let $\rsh_i \in \RSL{}{\CENTAIL{\alpha}}$ and $\rsh \DEFEQ \sh[\PS_1^{\sh} / \rsh_1, \ldots, \PS_m^{\sh} / \rsh_m]$.
 Assume $\OMEGA{\HATRACK}{(p_0,q_0,r_0)}{\rsh}$.
 By definition of $\Delta$, we either have $q_0 = 0$ or $q_0 > 0$.
 We proceed by case distinction.
 \emph{The case} $q_0 = 0$\qquad
 \begin{align*}
                      & \OMEGA{\HATRACK}{(p_0,0,r_0)}{\rsh} \\
    ~\Leftrightarrow~ & \left[ \text{Definition of}~\Delta~\text{(only first case applicable)} \right] \\
                      & \OMEGA{\HASH{\sha}}{p_0}{\rsh} ~\text{and}~ \OMEGA{\HATRACK}{r_0}{\rsh} \\
    ~\Leftrightarrow~ & \left[ \text{Compositionality of}~\HASH{\sha},\HATRACK \right] \\
                      & \exists p_1,\ldots,p_m \in Q_{\HASH{\sha}} ~.~ \exists r_1,\ldots,r_m \in Q_{\HATRACK} ~.~ \\
                      & \qquad \MOVE{\HASH{\sha}}{p_0}{\sh}{p_1 \ldots p_m} ~\text{and}~ \forall 1 \leq i \leq m ~.~ \OMEGA{\HASH{\sha}}{p_i}{\rsh_i} \\
                      & \qquad \text{and}~ \MOVE{\HATRACK}{r_0}{\sh}{r_1 \ldots r_m} ~\text{and}~ \forall 1 \leq i \leq m ~.~ \OMEGA{\HATRACK}{r_i}{\rsh_i} \\
    ~\Leftrightarrow~ & \left[ \text{choose}~q_1 = \ldots = q_m = 0 \right] \\
                      & \exists p_1,\ldots,p_m \in Q_{\HASH{\sha}} ~.~ \exists r_1,\ldots,r_m \in Q_{\HATRACK} ~.~ \\
                      & \qquad \MOVE{\HASH{\sha}}{p_0}{\sh}{p_1 \ldots p_m} ~\text{and}~ \forall 1 \leq i \leq m ~.~ \OMEGA{\HASH{\sha}}{p_i}{\rsh_i} \\
                      & \qquad \text{and}~ \MOVE{\HATRACK}{r_0}{\sh}{r_1 \ldots r_m} ~\text{and}~ \forall 1 \leq i \leq m ~.~ \OMEGA{\HATRACK}{r_i}{\rsh_i} \\
                      & \qquad \text{and}~ q_0 = q_1 = \ldots = q_m = 0 \\
    ~\Leftrightarrow~ & \left[ \text{Definition of}~\Delta \right] \\
                      & \exists (p_1,q_1,r_1),\ldots,(p_m,q_m,r_m) \in Q ~.~ \\
                      & \qquad \text{and}~ \MOVE{\HASH{\exists z. \sha}}{(p_0,q_0,r_0)}{\sh}{(p_1,q_1,r_1) \ldots (p_m,q_m,r_1)} \\
                      & \qquad \text{and}~ \forall 1 \leq i \leq m ~.~ \OMEGA{\HATRACK}{(p_i,q_i,r_i)}{\rsh_i}.
 \end{align*}
 \emph{The case} $q_0 = \ell > 0$\qquad
 \begin{align*}
                      & \OMEGA{\HASH{\exists z. \sha}}{(p_0,\ell,r_0)}{\rsh} \\
    ~\Leftrightarrow~ & \left[ \text{Definition of}~\Delta~\text{(only third case applicable)} \right] \\
                      &  \exists y \in \VAR(\rsh) ~.~ \OMEGA{\HASH{\sha}}{p_0}{(\rsh\SQUEEZE{x}{\ell}) : \{ x = y \}} \tag{$\dag$} \\
                      &  \qquad \text{and}~ \OMEGA{\HATRACK}{r_0}{(\rsh\SQUEEZE{x}{\ell}) : \{ x = y \}} \\
 \end{align*}
 Now, since $\VAR(\rsh)$ can be partitioned into $\VAR(\sh)$ and, for each $1 \leq k \leq m$, $\VAR(\rsh_k) \setminus \VAR(\sh)$, variable $y$ occurs in exactly one of these sets.
 First, assume $y \in \VAR(\sh)$. Then
 \begin{align*}
    (\dag) ~\Leftrightarrow~ & \left[ y \in \VAR(\sh) \right] \\
                             &  \exists y \in \VAR(\sh) ~.~ \OMEGA{\HASH{\sha}}{p_0}{(\rsh\SQUEEZE{x}{\ell}) : \{ x = y \}} \\
                             &  \qquad \text{and}~ \OMEGA{\HATRACK}{r_0}{(\rsh\SQUEEZE{x}{\ell}) : \{ x = y \}} \\
           ~\Leftrightarrow~ & \left[ \text{Compositionality of}~\HASH{\sha},\HATRACK \right] \\
                             &  \exists y \in \VAR(\sh) ~.~ \exists p_1,\ldots,p_m \in Q_{\HASH{\sha}} ~.~ \exists r_1, \ldots, r_m \in Q_{\HATRACK} ~.~ \\
                             &  \qquad \MOVE{\HASH{\sha}}{p_0}{(\sh\SQUEEZE{x}{\ell}) : \{ x = y \}}{p_1 \ldots p_m} \\
                             &  \qquad \qquad \text{and}~ \forall 1 \leq i \leq m ~.~ \OMEGA{\HASH{\sha}}{p_i}{\rsh_i} \\
                             &  \qquad \text{and}~ \MOVE{\HATRACK}{r_0}{(\sh\SQUEEZE{x}{\ell}) : \{ x = y \}}{r_1 \ldots r_m} \\
                             &  \qquad \qquad \text{and}~ \forall 1 \leq i \leq m ~.~ \OMEGA{\HATRACK}{r_i}{\rsh_i} \\
           ~\Leftrightarrow~ & \left[ \text{choose}~q_1 = \ldots = q_m = 0 \right] \\
                             &  \exists y \in \VAR(\sh) ~.~ \exists p_1,\ldots,p_m \in Q_{\HASH{\sha}} ~.~ \exists r_1, \ldots, r_m \in Q_{\HATRACK} ~.~ \\
                             &  \qquad \MOVE{\HASH{\sha}}{p_0}{(\sh\SQUEEZE{x}{\ell}) : \{ x = y \}}{p_1 \ldots p_m} \\
                             &  \qquad \qquad \text{and}~ \forall 1 \leq i \leq m ~.~ \OMEGA{\HASH{\sha}}{p_i}{\rsh_i} \\
                             &  \qquad \text{and}~ \MOVE{\HATRACK}{r_0}{(\sh\SQUEEZE{x}{\ell}) : \{ x = y \}}{r_1 \ldots r_m} \\
                             &  \qquad \qquad \text{and}~ \forall 1 \leq i \leq m ~.~ \OMEGA{\HATRACK}{r_i}{\rsh_i} \\
                             &  \qquad \text{and}~ q_1 = \ldots = q_m = 0 \\
           ~\Leftrightarrow~ & \left[ \exists x \exists y \equiv \exists y \exists x,~\text{Definition of}~Q \right] \\
                             &  \exists (p_1,q_1,r_1),\ldots,(p_m,q_m,r_m) \in Q ~.~ \exists y \in \VAR(\sh) ~.~ \\
                             &  \qquad \MOVE{\HASH{\sha}}{p_0}{(\sh\SQUEEZE{x}{\ell}) : \{ x = y \}}{p_1 \ldots p_m} \\
                             &  \qquad \qquad \text{and}~ \forall 1 \leq i \leq m ~.~ \OMEGA{\HASH{\sha}}{p_i}{\rsh_i} \\
                             &  \qquad \text{and}~ \MOVE{\HATRACK}{r_0}{(\sh\SQUEEZE{x}{\ell}) : \{ x = y \}}{r_1 \ldots r_m} \\
                             &  \qquad \qquad \text{and}~ \forall 1 \leq i \leq m ~.~ \OMEGA{\HATRACK}{r_i}{\rsh_i} \\
                             &  \qquad \text{and}~ q_1 = \ldots = q_m = 0 \\
           ~\Leftrightarrow~ & \left[ \text{Definition}~\Delta,~y \in \VAR(\sh) \right] \\
                             &  \exists (p_1,q_1,r_1),\ldots,(p_m,q_m,r_m) \in Q ~.~ \tag{$\clubsuit$} \\
                             &  \qquad \MOVE{\HASH{\exists z. \sha}}{(p_0,q_0,r_0)}{\sh}{(p_1,q_1,r_1) \ldots (p_m,q_m,r_m)} \\
                             &  \qquad \text{and}~ \forall 1 \leq i \leq m ~.~ \OMEGA{\HASH{\exists z. \sha}}{(p_i,r_i,q_i)}{\rsh_i} \\
                             &  \qquad \text{and}~ q_1 = \ldots = q_m = 0 \\
               ~\Rightarrow~ & \left[ A \wedge B \rightarrow A \right] \\
                             &  \exists (p_1,q_1,r_1),\ldots,(p_m,q_m,r_m) \in Q ~.~ \\
                             &  \qquad \MOVE{\HASH{\exists z. \sha}}{(p_0,q_0,r_0)}{\sh}{(p_1,q_1,r_1) \ldots (p_m,q_m,r_m)} \\
                             &  \qquad \text{and}~ \forall 1 \leq i \leq m ~.~ \OMEGA{\HASH{\exists z. \sha}}{(p_i,r_i,q_i)}{\rsh_i} \\
                             &  \qquad \text{and}~ q_1 = \ldots = q_m = 0.
 \end{align*}
 Note that equivalence holds up to the step marked with $(\clubsuit)$ only. The last step works is only in one direction.
 We deal with the backwards direction after considering the second case.
 Thus assume $y \in \VAR(\rsh_{j}) \setminus \VAR(\sh)$ for some $1 \leq j \leq m$. Then
 \begin{align*}
    (\dag) ~\Leftrightarrow~ & \left[ y \in \VAR(\rsh_{j}) \setminus \VAR(\sh) \right] \\
                             &  \exists y \in \VAR(\rsh_{j}) \setminus \VAR(\sh) ~.~ \OMEGA{\HASH{\sha}}{p_0}{(\rsh\SQUEEZE{x}{\ell}) : \{ x = y \}} \\
                             &  \qquad \text{and}~ \OMEGA{\HATRACK}{r_0}{(\rsh\SQUEEZE{x}{\ell}) : \{ x = y \}} \\
           ~\Leftrightarrow~ & \left[ \text{Compositionality of}~\HASH{\sha},\HATRACK \right] \\
                             &  \exists \ell' . \exists y \in \VAR(\rsh_{j}) \setminus \VAR(\sh) ~.~ \exists p_1,\ldots,p_m \in Q_{\HASH{\sha}} ~.~ \\
                             &  \exists r_1,\ldots,r_m \in Q_{\HATRACK} ~.~ \\
                             &  \qquad \MOVE{\HASH{\sha}}{p_0}{(\sh\SQUEEZE{x}{\ell})\SQUEEZE{x}{j,\ell'}}{p_1 \ldots p_m} \\
                             &  \qquad \text{and}~ \MOVE{\HATRACK}{r_0}{(\sh\SQUEEZE{x}{\ell})\SQUEEZE{x}{j,\ell'}}{r_1 \ldots r_m} \\
                             &  \qquad \text{and}~ \forall 1 \leq i \leq m ~.~ i \neq j ~\text{implies}~ \\
                             &  \qquad \qquad \qquad \OMEGA{\HASH{\sha}}{p_i}{\rsh_i} ~\text{and}~ \OMEGA{\HATRACK}{r_i}{\rsh_i} \\
                             &  \qquad \text{and}~ \OMEGA{\HASH{\sha}}{p_j}{\rsh_j\SQUEEZE{x}{\ell'}} ~\text{and}~ \OMEGA{\HASH{\sha}}{r_j}{\rsh_j\SQUEEZE{x}{\ell'}} \\
           ~\Leftrightarrow~ & \left[ \text{choose}~q_j ~\text{and for each}~i \neq j~\text{choose}~q_i = 0 \right] \\
                             &  \exists y \in \VAR(\rsh_{j}) \setminus \VAR(\sh) ~.~ \exists p_1,\ldots,p_m \in Q_{\HASH{\sha}} ~.~ \\
                             &  \exists r_1,\ldots,r_m \in Q_{\HATRACK} ~.~ \exists q_1,\ldots,q_m \in \{0,\ldots,\alpha\} \\
                             &  \qquad \MOVE{\HASH{\sha}}{p_0}{(\sh\SQUEEZE{x}{\ell})\SQUEEZE{x}{j,q_j}}{p_1 \ldots p_m} \\
                             &  \qquad \text{and}~ \MOVE{\HATRACK}{r_0}{(\sh\SQUEEZE{x}{\ell})\SQUEEZE{x}{j,q_j}}{r_1 \ldots r_m} \\
                             &  \qquad \text{and}~ \forall 1 \leq i \leq m ~.~ i \neq j ~\text{implies}~ \\
                             &  \qquad \qquad \qquad \OMEGA{\HASH{\sha}}{p_i}{\rsh_i} ~\text{and}~ \OMEGA{\HATRACK}{r_i}{\rsh_i} \\
                             &  \qquad \text{and}~ \OMEGA{\HASH{\sha}}{p_j}{\rsh_j\SQUEEZE{x}{q_j}} ~\text{and}~ \OMEGA{\HASH{\sha}}{r_j}{\rsh_j\SQUEEZE{x}{q_j}} \\
                             &  \qquad \text{and}~ q_j > 0 ~\text{and}~ \forall i \neq j ~.~ q_i = 0 \\
           ~\Leftrightarrow~ & \left[ \exists x \exists y \equiv \exists y \exists x,~\text{Definition of}~Q \right] \\
                             &  \exists (p_1,q_1,r_1),\ldots,(p_m,q_m,r_m) \in Q ~.~ \\
                             &  \exists y \in \VAR(\rsh_{j}) \setminus \VAR(\sh) ~.~ \\
                             &  \qquad \MOVE{\HASH{\sha}}{p_0}{(\sh\SQUEEZE{x}{\ell})\SQUEEZE{x}{j,q_j}}{p_1 \ldots p_m} \\
                             &  \qquad \text{and}~ \MOVE{\HATRACK}{r_0}{(\sh\SQUEEZE{x}{\ell})\SQUEEZE{x}{j,q_j}}{r_1 \ldots r_m} \\
                             &  \qquad \text{and}~ \forall 1 \leq i \leq m ~.~ i \neq j ~\text{implies}~ \\
                             &  \qquad \qquad \qquad \OMEGA{\HASH{\sha}}{p_i}{\rsh_i} ~\text{and}~ \OMEGA{\HATRACK}{r_i}{\rsh_i} \\
                             &  \qquad \text{and}~ \OMEGA{\HASH{\sha}}{p_j}{\rsh_j\SQUEEZE{x}{q_j}} ~\text{and}~ \OMEGA{\HASH{\sha}}{r_j}{\rsh_j\SQUEEZE{x}{q_j}} \\
                             &  \qquad \text{and}~ q_j > 0 ~\text{and}~ \forall i \neq j ~.~ q_i = 0 \\
           ~\Leftrightarrow~ & \left[ \text{Definition of}~\Delta~\text{(second case for $\sh$)} \right] \\
                             &  \exists (p_1,q_1,r_1),\ldots,(p_m,q_m,r_m) \in Q ~.~ \tag{$\spadesuit$}\\
                             &  \qquad \MOVE{\HASH{\exists z. \sha}}{(p_0,q_1,r_1)}{\sh}{(p_1,q_1,r_1) \ldots (p_m,q_m,r_m)} \\
                             &  \qquad \text{and}~ \forall 1 \leq i \leq m ~.~ \OMEGA{\HASH{\exists z. \sha}}{(p_i,q_i,r_i)}{\rsh_i} \\
                             &  \qquad \text{and}~ q_j > 0 ~\text{and}~ \forall i \neq j ~.~ q_i = 0 \\
           ~\Rightarrow~     & \left[ A \wedge B \rightarrow A \right] \\
                             &  \exists (p_1,q_1,r_1),\ldots,(p_m,q_m,r_m) \in Q ~.~ \\
                             &  \qquad \MOVE{\HASH{\exists z. \sha}}{(p_0,q_1,r_1)}{\sh}{(p_1,q_1,r_1) \ldots (p_m,q_m,r_m)} \\
                             &  \qquad \text{and}~ \forall 1 \leq i \leq m ~.~ \OMEGA{\HASH{\exists z. \sha}}{(p_i,q_i,r_i)}{\rsh_i}.
 \end{align*}
 Again, the last step is in one direction only.
 It remains to show the converse direction.
 Assume, for $q_0 = \ell > 0$,
 \begin{align*}
                             &  \exists (p_1,q_1,r_1),\ldots,(p_m,q_m,r_m) \in Q ~.~ \\
                             &  \qquad \MOVE{\HASH{\exists z. \sha}}{(p_0,q_1,r_1)}{\sh}{(p_1,q_1,r_1) \ldots (p_m,q_m,r_m)} \\
                             &  \qquad \text{and}~ \forall 1 \leq i \leq m ~.~ \OMEGA{\HASH{\exists z. \sha}}{(p_i,q_i,r_i)}{\rsh_i}.
 \end{align*}
 Then, by construction of $\Delta$ either $q_1 = q_2 = \ldots = q_m = 0$ or exactly one $q_j > 0$, for some $1 \leq j \leq m$.
 If $q_1 = q_2 = \ldots = q_m = 0$ then only the third case of $\Delta$ is applicable to $\sh$ (and there exists some suitable $y \in \VAR(\sh)$).
 Thus, $(\clubsuit)$ holds. As shown before, this is equivalent to $\OMEGA{\HATRACK}{q_0}{\rsh}$.
 If there exists exactly one $q_j$ with $q_j > 0$ then $(\spadesuit)$ holds.
 As shown above, this case is equivalent to $\OMEGA{\HATRACK}{q_0}{\rsh}$.
 Hence, in each case, $\HASH{\exists z.\sha}$ satisfies the compositionality property.
 \qed
\end{proof}
\begin{lemma}
 $L(\HASH{\exists z.\sha}) ~=~ \USET{\exists x . \sha}{\SRD}{\CENTAIL{\alpha}}$, where $x = \PROJ{\FV{0}{\sha}}{k}$ for some fixed position $k$.
\end{lemma}
\begin{proof}
  Let $\rsh \in \RSL{}{\CENTAIL{\alpha}}$. Then
  \begin{align*}
                       & \rsh \in L(\HASH{\exists z.\sha}) \\
     ~\Leftrightarrow~ & \left[ \text{Definition}~L(\HASH{\exists z.\sha}) \right] \\
                       & \exists (p_0,q_0,r_0) \in F ~.~ \OMEGA{\HASH{\exists z. \sha}}{(p_0,q_0,r_0)}{\rsh} \\
     ~\Leftrightarrow~ & \left[ \text{Definition of}~F \right] \\
                       & \exists p_0 \in F_{\HASH{\sha}} . \exists r_0 \in Q_{\HATRACK} ~.~ \OMEGA{\HATRACK}{(p_0,k,r_0)}{\rsh} \\
     ~\Leftrightarrow~ & \left[ \text{Definition of}~\Delta~\text{(only third case is applicable)} \right] \\
                       & \exists p_0 \in F_{\HASH{\sha}} ~.~ \exists r_0 \in Q_{\HATRACK} ~.~ \exists y \in \VAR(\rsh) ~.~ \\
                       & \qquad \OMEGA{\HASH{\sha}}{p_0}{\rsh\SQUEEZE{x}{k} : \{y = x\}} ~\text{and}~ \OMEGA{\HATRACK}{r_0}{\rsh\SQUEEZE{x}{k} : \{y = x\}} \\
     ~\Leftrightarrow~ & \left[ \exists x \exists y \equiv \exists y \exists x \right] \\
                       & \exists y \in \VAR(\rsh) ~.~ \exists p_0 \in F_{\HASH{\sha}} ~.~ \exists r_0 \in Q_{\HATRACK} ~.~ \\
                       & \qquad \OMEGA{\HASH{\sha}}{p_0}{\rsh\SQUEEZE{x}{k} : \{y = x\}} ~\text{and}~ \OMEGA{\HATRACK}{r_0}{\rsh\SQUEEZE{x}{k} : \{y = x\}} \\
     ~\Leftrightarrow~ & \left[ \text{for each}~\rsha~\text{there exists a}~r_0 \in Q_{\HATRACK} \right] \\
                       & \exists y \in \VAR(\rsh) ~.~ \exists p_0 \in F_{\HASH{\sha}} ~.~ \OMEGA{\HASH{\sha}}{p_0}{\rsh\SQUEEZE{x}{k} : \{y = x\}} \\
     ~\Leftrightarrow~ & \left[ \text{Definition of}~L(\HASH{\sha})  \right] \\
                       & \exists y \in \VAR(\rsh) ~.~ \left(\rsh\SQUEEZE{x}{k} : \{y = x\}\right) \in L(\HASH{\sha})  \\
     ~\Leftrightarrow~ & \left[ L(\HASH{\sha}) = \USET{\sha}{\SRD}{\CENTAIL{\alpha}}  \right] \\
                       & \exists y \in \VAR(\rsh) ~.~ \rsh\SQUEEZE{x}{k} : \{y = x\} \ENTAIL{\SRD} \sha  \\
     ~\Leftrightarrow~ & \left[ \text{Lemma}~\ref{thm:entailment:existential:structural-auxiliary} \right] \\
                       & \rsh \ENTAIL{\SRD} \exists x . \sha \\
     ~\Leftrightarrow~ & \left[ \text{Definition of}~\USET{\exists x .\sha}{\SRD}{\CENTAIL{\alpha}}  \right] \\
                       & \rsh \in \USET{\exists x . \sha}{\SRD}{\CENTAIL{\alpha}}. 
  \end{align*}
  \qed
\end{proof}


%
%
%
%
%
%
%
\section{Complexity of Entailment} \label{app:complexity}

This section provides a detailed complexity analysis of 
Algorithm~\ref{alg:entailment:decision-procedure}.
In particular, this includes a proof of Theorem~\ref{thm:entailment:complexity}

\subsection{Upper Complexity Bound for Entailments}

For the remainder of this section, we fix some notation used in
Algorithm~\ref{alg:entailment:decision-procedure}:
Let $\sh$, $\HASH{\sha}$, etc., be as in
Algorithm~\ref{alg:entailment:decision-procedure}.
Further, let $\SRD$ be an SID
such that the unfoldings of each predicate call can be accepted by a heap automaton over $\CENTAIL{\alpha}$. 
Moreover, let $k \DEFEQ \SIZE{\SRD} + \SIZE{\sh} + \SIZE{\sha}$
and $M \leq k$ denote the maximal number of predicate calls 
in $\sh$ and any symbolic heap in $\SRD$. 

We first analyze run-time of Algorithm~\ref{alg:entailment:decision-procedure}
for arbitrary SIDs.
\begin{lemma} \label{thm:entailment:general-upper-bound}
  Algorithm~\ref{alg:entailment:decision-procedure}
  decides whether $\sh \ENTAIL{\SRD} \sha$ holds in%
  \begin{align*}
   & \BIGO{ 2^{\text{poly}(k)} \cdot \left(2^{\SIZE{Q_{\HASH{\sha}}}}\right)^{2(M+1)}
         \cdot \SIZE{\Delta_{\HASH{\sha}}} },
  \end{align*}
  where $\text{poly}(k)$ denotes some polynomial function in $k$.
\end{lemma}
\begin{proof}
  Our previous complexity analysis of Algorithm~\ref{alg:on-the-fly-refinement} reveals that
  $\CALLSEM{\PS\T{x}}{\SRDALT} \cap L(\overline{\HASH{\sha}}) = \emptyset$ is decidable in
  \begin{align*}
    \BIGO{\SIZE{\SRDALT} \cdot \SIZE{Q_{\overline{\HASH{\sha}}}}^{M+1}
    \cdot \SIZE{\Delta_{\overline{\HASH{\sha}}}}}. \tag{$\clubsuit$}
  \end{align*}
  Regarding $\SIZE{\SRDALT}$, applying 
  the Refinement Theorem (Theorem~\ref{thm:compositional:refinement})
  to $\SRD \cup \{ \SRDRULE{\PS}{\sh} \}$ and $\HASAT$ (cf. Theorem~\ref{thm:zoo:sat:property})
  yields an SID $\SRDALT$ of size
  \begin{align*}
    \SIZE{\SRDALT} \leq c \cdot \SIZE{\SRD} \cdot 2^{\SIZE{\sh}^2} \cdot 2^{2\alpha^2 + \alpha} \leq 2^{\text{poly}(k)}~,
  \end{align*}
  for some positive constant $c$.
  Then $\SRDALT$ is computable in $\BIGO{2^{\text{poly}(k)}}$.
  Furthermore, $\overline{\HASH{\sha}}$ is obtained from complementation of $\HASH{\sha}$.
  Thus, by the construction to prove Lemma~\ref{thm:refinement:boolean}, we obtain that
  $\SIZE{Q_{\overline{\HASH{\sha}}}} \leq 2^{\SIZE{Q_{\HASH{\sha}}}}$ and
  that $\Delta_{\overline{\HASH{\sha}}}$ is decidable in
  $\left(2^{\SIZE{Q_{\HASH{\sha}}}}\right)^{M+1} \cdot \SIZE{\Delta_{\HASH{\sha}}}$.
  Putting both into $(\clubsuit)$ yields the result.
\qed
\end{proof}
Towards a more fine-grained analysis, recall from Definition~\ref{def:alpha-bounded} our assumption that SIDs are $\alpha$--bounded.
Further, we assume the arity of points-to assertions
$\PTS{x}{\T{y}}$, i.e., $\SIZE{\T{y}}$, to be bounded by some 
$\gamma \geq 0$.
Our next observation is that heap automata constructed 
for arbitrary determined symbolic heaps according to
Theorem~\ref{thm:entailment:top-level}
satisfy the same constraints.
Formally,
\begin{lemma} \label{thm:entailment:top-level-complexity}
  Let $\SRD$ be an $\alpha$--bounded SID and $\sha \in \SL{\SRD}{\CENTAIL{\alpha}}$.
  Then a heap automaton $\HASH{\sha}$ accepting $\USET{\sha}{\SRD}{\CENTAIL{\alpha}}$
  can be constructed
  such that $\Delta_{\HASH{\sha}}$ is decidable in $\BIGO{2^{\text{poly}(k)}}$
  and $\SIZE{Q_{\HASH{\sha}}} \leq 2^{\text{poly}(\alpha)}$.
\end{lemma}
\begin{proof}
  By induction on the structure of symbolic heaps, we show $\alpha$--boundedness
  for the heap automata constructed in the proof of
  Theorem~\ref{thm:entailment:top-level}.
  \qed
\end{proof}

%
We are now in a position to derive
an upper complexity bound on the entailment problem for a permissive symbolic heap fragment of separation logic with inductive predicate definitions.
\begin{lemma} \label{thm:entailment:double-exptime}
  \DENTAIL{\CENTAIL{\alpha}}{\SRD} is decidable in \CCLASS{2-ExpTime}
  for each $\alpha$--bounded SID $\SRD$.
\end{lemma}
\begin{proof}
By Lemma~\ref{thm:entailment:general-upper-bound}, we know that an entailment $\sh \ENTAIL{\SRD} \sha$
 can be discharged in
  $
     \BIGO{ 2^{\text{poly}(k)} \cdot \left(2^{\SIZE{Q_{\HASH{\sha}}}}\right)^{2(M+1)}
     \cdot \SIZE{\Delta_{\HASH{\sha}}} }.
  $
%
By Lemma~\ref{thm:entailment:top-level-complexity},
$\Delta_{\HASH{\sha}}$ is decidable in $\BIGO{2^{\text{poly}(k)}}$  and $\SIZE{Q_{\HASH{\sha}}} \leq 2^{\text{poly}(\alpha)}$.
Then it is easy to verify that $\sh \ENTAIL{\SRD} \sha$ is decidable in
$\BIGO{2^{2^{\text{poly}(k)}}}$:
  \begin{align*}
        & \BIGO{\SIZE{\SRDALT} \cdot \SIZE{Q_{\overline{\HASH{\sha}}}}^{M+1} \cdot \SIZE{\Delta_{\overline{\HASH{\sha}}}}} \\
    ~=~ &  \left[ \SIZE{\SRDALT} \leq 2^{\text{poly}(k)}, \SIZE{Q_{\overline{\HASH{\sha}}}} \leq 2^{\SIZE{Q_{\HASH{\sha}}}}, M \leq 2k \right] \\
        & \BIGO{ 2^{\text{poly}(k)} \cdot \left(2^{\SIZE{Q_{\HASH{\sha}}}}\right)^{2k} \cdot \SIZE{\Delta_{\overline{\HASH{\sha}}}}} \\
    ~=~ & \left[ \SIZE{\Delta_{\overline{\HASH{\sha}}}} \leq \left(2^{\SIZE{Q_{\HASH{\sha}}}}\right)^{M+1} \cdot \SIZE{\Delta_{\HASH{\sha}}} \right] \\
        & \BIGO{ 2^{\text{poly}(k)} \cdot \left(2^{\SIZE{Q_{\HASH{\sha}}}}\right)^{4k} \cdot \SIZE{\Delta_{\HASH{\sha}}}} \\
    ~=~ &  \left[ \SIZE{Q_{\HASH{\sha}}} \leq 2^{\text{poly}(k)} \right] \\
        & \BIGO{ 2^{\text{poly}(k)} \cdot \left(2^{2^{\text{poly}(k)}}\right)^{4k} \cdot \SIZE{\Delta_{\HASH{\sha}}}} \\
    ~=~ &  \left[ (a^b)^c = a^{bc},~ \SIZE{\Delta_{\HASH{\sha}}} \in \BIGO{2^{\text{poly}(k)}}  \right] \\
        & \BIGO{ 2^{\text{poly}(k)} \cdot 2^{4k \cdot 2^{\text{poly}(k)}} \cdot 2^{\text{poly}(k)}} \\
    ~=~ & \BIGO{ 2^{2^{\text{poly}(k)}} }
  \end{align*}
Hence, \DENTAIL{\CENTAIL{\alpha}}{\SRD} is in \CCLASS{2--ExpTime}.
\qed
\end{proof}
Since the maximal arity $\alpha$ of predicate symbols is fixed for any given SID,
we also analyze Algorithm~\ref{alg:entailment:decision-procedure} under the assumption that $\alpha$ is bounded by a constant.
This is a common assumption (cf.~\cite{brotherston2014decision,brotherston2016model}) that was considered in
Section~\ref{sec:zoo} already.
%
\begin{lemma} \label{thm:entailment:exptime}
  Let $\SRD$ be an $\alpha$--bounded SID for some constant $\alpha \geq 1$.
  Then the entailment problem $\DENTAIL{\CENTAIL{\alpha}}{\SRD}$ is in \CCLASS{ExpTime}.
\end{lemma}
\begin{proof}
By Lemma~\ref{thm:entailment:top-level-complexity},
$\SIZE{Q_{\HASH{\sha}}} \leq 2^{\text{poly}(\alpha)}$ and $\Delta_{\HASH{\sha}}$ is decidable in
$\BIGO{2^{\text{poly}(k)}}$.
Since $\alpha$ is bounded by a constant, so is $\SIZE{Q_{\HASH{\sha}}}$.
Then, by Lemma~\ref{thm:entailment:general-upper-bound}, we know that
$\DENTAIL{\CENTAIL{\alpha}}{\SRD}$ is decidable in
  \begin{align*}
   \BIGO{ 2^{\text{poly}(k)} \cdot \left(2^{\SIZE{Q_{\HASH{\sha}}}}\right)^{2(M+1)}
         \cdot 2^{\text{poly}(k)} }
   ~=~
   \BIGO{ 2^{\text{poly}(k)} }, ~
  \end{align*}
which clearly is in \CCLASS{ExpTime}.
\qed
\end{proof}
Then, the upper complexity bounds provided in Theorem~\ref{thm:entailment:complexity}
hold by Lemma~\ref{thm:entailment:double-exptime} and Lemma~\ref{thm:entailment:exptime}.
Further, $\CCLASS{ExpTime}$--completeness follows
directly from\cite[Theorem 5]{antonopoulos2014foundations}
and Appendix~\ref{app:entailment:lower}.

\subsection{Lower Complexity Bound for Entailments} \label{app:entailment:lower}
The proof of the \CCLASS{ExpTime}--lower bound in~\cite{antonopoulos2014foundations}
is by reducing the inclusion problem for nondeterministic finite tree automata (NFTA, cf. \cite{comon2007tree}) to
the entailment problem.
Their proof requires a constant (or free variable) for each symbol in the tree automatons alphabet.
In contrast, we prove their result by encoding the alphabet in a null-terminated singly-linked list.
Thus, a tree $a(b,a(b,b)$ is encoded by a reduced symbolic heap
\begin{align*}
  & \exists z_1 z_2 z_3 z_4 z_5 z_6 z_7 ~.~  \\
  & \quad  \PT{x}{z_1~z_2~\NIL}  \\
  & \quad  \SEP \PT{z_1}{\NIL~\NIL~z_3} \SEP \PT{z_3}{\NIL~\NIL~\NIL} \\
  & \quad  \SEP \PT{z_2}{z_4~z_5~\NIL}  \\
  & \quad  \SEP \PT{z_4}{\NIL~\NIL~z_6} \SEP \PT{z_6}{\NIL~\NIL~\NIL} \\
  & \quad  \SEP \PT{z_5}{\NIL~\NIL~ z_7} \SEP \PT{z_7}{\NIL~\NIL~\NIL},
\end{align*}
where the symbol $a$ is encoded by having $\NIL$ as third component in a points-to assertion and
symbol $b$ by a $\NIL$ terminated list of length one.
Now, given some NFTA $\HA{T} = (Q,\Sigma,\Delta,F)$ with $\Sigma = \{a_1,\ldots,a_n\}$,
we construct a corresponding $\SRD$.
Without less of generality, we assume that $\HA{T}$ contains no unreachable or unproductive states.
We set $\PRED(\SRD) \DEFEQ Q \cup \Sigma \cup \{I\}$, where each predicate symbol is of arity one.
Then, for each symbol $a_i \in \Sigma$ one rule of the form $\SRDRULE{a_{1}}{\PROJ{\FV{0}{}}{1}=\NIL}$
or, for $1 < i \leq n$,
\begin{align*}
  a_{i} ~\SRDARROW~ & \exists z_1~z_2~\ldots~z_{i-1} ~.~
                        \PT{\PROJ{\FV{0}{}{1}}}{\NIL~\NIL~z_1} \\
                      & \qquad \SEP \bigstar_{1 \leq j < i} \PT{z_j}{\NIL~\NIL~z_{j+1}} : \{ z_{i-1} = \NIL \}
\end{align*}
is added to $\SRD$.
Furthermore, for each $(p_1 \ldots p_m, a_i, p_0) \in \Delta$, $1 \leq i \leq n$, we add a rule
\begin{align*}
  p_0 ~\SRDARROW~ & \exists z_1 \ldots z_{m+1} ~.~
                        \PT{\PROJ{\FV{0}{}}{1}}{z_1~\ldots~z_{m+1}}  \\
                    & \qquad    ~\SEP~ a_i(z_{m+1})
                        ~\SEP~ \bigstar_{1 \leq i \leq m} p_i(z_i).
\end{align*}
Finally, we add rules $\SRDRULE{I}{p\PROJ{\FV{0}}{1} : \{\PROJ{\FV{0}{}}{1} \neq \NIL \}}$ for each $p \in F$.
Clearly $\SRD$ is established.
Moreover, it is easy to verify that, given two NFTAs $\HA{T}_1$ and $\HA{T}_2$ with distinct sets of states,
we have
\begin{align*}
  I_1 x \ENTAIL{\SRD_1 \cup \SRD_2} I_2x ~\text{iff}~ L(\HA{T}_1) \subseteq L(\HA{T}_2).
\end{align*}
Thus, following~\cite{antonopoulos2014foundations},
if $\USET{Ix}{\SRD}{\CENTAIL{1}}$ can be accepted by a heap automaton,
the entailment problem $\DENTAIL{\SRD}{\CENTAIL{\alpha}}$ is \CCLASS{ExpTime}--hard for certain SIDs $\SRD$
fixed $\alpha = 1$, and a fixed arity of points-to assertions $\gamma = 3$.
Such a heap automaton can easily be constructed.
Formally, let $\HA{T} = (Q,\Sigma,\Delta,F)$ be an NFTA as above and $Q = \{p_1,\ldots,p_k\}$ for some $k > 0$.
Furthermore, for each state $p_i$, let $t_i$ be some fixed finite tree that is accepted
by the tree automaton $\HA{T}_{i} = (Q,\Sigma,\Delta,\{p_i\})$ and $\rsh_i$ be the corresponding
encoding as a reduced symbolic heap.
One possible (not necessarily efficient)
 heap automaton $\HA{A} = (Q_{\HA{A}},\SHCENTAIL{1},\Delta_{\HA{A}},F_{\HA{A}})$
is given by:
\begin{align*}
  & Q_{\HA{A}} ~\DEFEQ~  \{ \rsh_i ~|~ 1 \leq i \leq k\} \cup \{ a_i ~|~ 1 \leq i \leq n \} \\
  & F_{\HA{A}} ~\DEFEQ~  F \\
  & \MOVE{A}{q_1 \ldots q_m}{\sh}{q_0} ~\text{iff}~ \sh\left[ \PS_1 / q_1, \ldots, \PS_m / q_m\right] \ENTAIL{\SRD} \PS\PROJ{\FV{0}{}}{1}~,
\end{align*}
where each $a_i$ corresponds to the reduced symbolic heap encoding symbol $a_i$ and $\PS$ is the predicate $p_i$ corresponding to reduced symbolic heap $\rsh_i$ as previously described.

\section{Construction of Heap Automata for Entailment}
\label{app:case-study}
This section presents a systematic way to construct heap automata to discharge entailments.
Further, we provide two example constructions.
\subsection{Systematic Construction of $\USET{\sh}{\SRD}{\CENTAIL{\alpha}}$}\label{sec:entailment:systematic}
Using heap automata to discharge entailments as presented in Theorem~\ref{thm:entailment:main}
requires the construction of suitable heap automata for each predicate symbol
of an SID.
%
We briefly present a systematic construction of such heap automata that is
similar to the well-known Myhill--Nerode Theorem for regular (tree) languages~\cite{comon2007tree,nerode1958linear}:
We partition the set of all reduced symbolic heaps into equivalence classes, where two formulas belong to the
same equivalence class if they can be extended in the same way to formulas entailing a predicate $\PS$ of interest.
%
%
\begin{definition} \label{def:subst-eq}
   Let $\PS \in \PRED(\SRD)$, $\SRD \in \SETSRD{}$.
   Then two symbolic heaps $\rsh,\rsha \in \RSL{}{\CENTAIL{\alpha}}$ with $\NOFV{\rsh} = \NOFV{\rsha} = \beta$
   are \emph{$\PS$--equivalent}, written $\rsh \EE{\PS} \rsha$,
   if for all $\sh \in \SL{}{}$ with exactly one predicate call $\PI$ of arity $\beta$,
   we have $\sh[\PI / \rsh] \ENTAIL{\emptyset} \PS\FV{0}{}$ iff
   $\sh[\PI / \rsha] \ENTAIL{\emptyset} \PS\FV{0}{}$.
\end{definition}
For example, all non-empty singly linked list segments
of the \texttt{sll} predicate from Section~\ref{sec:introduction} 
with $\IFV{1}$ as head and $\IFV{2}$ as tail are $\texttt{sll}$--equivalent.
\begin{theorem}
   \label{thm:eq:myhill}
   Let $\PS$ be a predicate symbol of an SID $\SRD$.
   Then there exists a heap automaton $\HASH{\PS}$ with
   $L(\HASH{\PS}) = \USET{\PS\FV{0}{}}{\SRD}{\CENTAIL{\alpha}}$ 
   iff
   the number of equivalence classes of $\EE{\PS}$ is finite.
\end{theorem}
\begin{proof}
  Assume there are only finitely many equivalence classes of $\EE{\PS}$.
  Furthermore, let $\EC{\rsh}{\PS}$ denote the equivalence class containing formula $\rsh$.
  Then a heap automaton $\HASH{\PS} = (Q,\SHCENTAIL{\alpha},\Delta,F)$ accepting
  $\USET{\PS\FV{0}{}}{\SRD}{\CENTAIL{\alpha}}$
  is given by
  $Q = \{ \EC{\rsh}{\PS} ~|~ \rsh \in \RSL{}{\CENTAIL{\alpha}} \}$,
  $F = \{ \EC{\rsh}{\PS} ~|~ \rsh \in \RSL{}{\CENTAIL{\alpha}}, \rsh \ENTAIL{\SRD} \PS\FV{0}{} \}$
  and
  \begin{align*}
    \MOVE{\HASH{\PS}}{\EC{\sh[\CALLN{1}{} / \rsh_1, \CALLN{m}{} / \rsh_m]}{\PS}}{\sh}{\EC{\rsh_1}{\PS} \ldots \EC{\rsh_m}{\PS}}~,
  \end{align*}
  for each symbolic heap $\sh \in \SHCENTAIL{\alpha}$ with $\NOCALLS{\sh} = m$.
  Then it is easy to verify that $\HA{A}$ satisfies the compositionality property
  and $L(\HASH{\PS}) = \USET{\PS\FV{0}{}}{\SRD}{\CENTAIL{\alpha}}$.
  The converse direction is straightforward.
  A full proof is found in \TECHNICALREPORT.
\qed
\end{proof}
Note that it suffices to represent each equivalence class by a single
reduced symbolic heap, e.g. the smallest one.
%
%
%
While the construction principle from above is generally applicable, it is, however,
often preferable to exploit other properties of the predicates of interest,
e.g., acyclicity (see Section~\ref{sec:zoo:acyclicity}), to reduce the number of equivalence classes.

\subsection{Example: Singly-linked List Segments} \label{app:more-data-structures}

Recall the SID for acyclic singly-linked list segments from
Example~\ref{ex:srd}. A heap automaton $\HA{A}$ for \texttt{sll} is
defined in Fig.~\ref{fig:sll-ha}.

\begin{figure}
$Q \DEFEQ \{q_{\qeq}, q_{\diff}, q_{\qrev}, q_{\qfst}, q_{\qsnd},
q_{\bot} \} \qquad F \DEFEQ \{ q_{\qeq}, q_{\diff}, q_{\qfst} \}$

\vspace{2mm}
\textbf{Transitions:}
\begin{align*}
  & (\EMPTYSEQ, \rsh, q_{\qeq})\in \Delta &\text{ iff }& \rsh \models \EMP : \{ \IFV{1} = \IFV{2}) \} \\
 &&& \text{or } \rsh \models \EMP \wedge \NOFV{\sh} = 2\\
  & (\EMPTYSEQ, \rsh, q_{\diff})\in \Delta &\text{ iff }& \rsh
    \models \PCDS{sll}{\IFV{1}\IFV{2}} : \{ \IFV{2} \neq \NIL \} \\
  & (\EMPTYSEQ, \rsh, q_{\qrev})\in \Delta &\text{ iff }& \rsh
    \models \PCDS{sll}{\IFV{2}\IFV{1}} : \{ \IFV{1} \neq \NIL \} \\
  & (\EMPTYSEQ, \rsh, q_{\qfst})\in \Delta &\text{ iff }& \rsh
    \models \PCDS{sll}{\IFV{1}\NIL} \\
  & (\EMPTYSEQ, \rsh, q_{\qsnd})\in \Delta &\text{ iff }& \rsh \models \PCDS{sll}{\IFV{2}\NIL}  \\
  & (\EMPTYSEQ, \rsh, q_{\bot}) \in \Delta &\text{ iff }& \rsh \not\models \PCDS{sll}{\FV{0}{}} \\
  & (\T{q}, \sh, q) \in \Delta & \text{ iff } &
    (\EMPTYSEQ, \sh[P_1/\rho_{\T{q}[1]},\ldots,P_m/\rho_{\T{q}[m]}], q) \in \Delta
\end{align*}

\vspace{3mm}
\textbf{Representations:}

$\begin{array}{lll}
 \rho_{q_{\qeq}} & \DEFEQ & \EMP : \{\IFV{1} = \IFV{2} \}\\
 \rho_{q_{\diff}} & \DEFEQ & \PT{\IFV{1}}{\IFV{2}}: \{ \IFV{2} \neq
                             \NIL, \IFV{2} \neq \IFV{1} \} \\
 \rho_{q_{\qrev}} & \DEFEQ & \PT{\IFV{2}}{\IFV{1}} : \{ \IFV{1} \neq
                             \NIL, \IFV{2} \neq \IFV{1} \} \\
 \rho_{q_{\qfst}} & \DEFEQ & \PT{\IFV{1}}{\NIL} \\
 \rho_{q_{\qsnd}} & \DEFEQ & \PT{\IFV{2}}{\NIL} \\
 \rho_{q_{\bot}} & \DEFEQ & \IFV{1} \neq \IFV{1} \\
\end{array}$

\caption{A heap automaton $\HA{A} = (Q,\CENTAIL{\alpha},\Delta,F)$ with $L(\HA{A}) = \USET{\PCDS{sll}{\T{x}}}{\SRD}{\CENTAIL{\alpha}}$, for acyclic singly-linked list fragments \texttt{sll} as defined in Ex.~\ref{ex:srd}; plus canonical representations $\rho_q$ for each state $q$.}
\label{fig:sll-ha}
\end{figure}

Observe that $\Delta$ is compositional.
Note further that we have
defined $\Delta$ in such way that $(\EMPTYSEQ, \tau, q) \in \Delta$
for $q \in F$ iff $\tau \models \PCDS{sll}{\FV{0}{}}$, i.e.,
$L(\HA{A}) = \USET{\PCDS{sll}{\T{x}}}{\SRD}{\CENTAIL{\alpha}}$.

Figure~\ref{fig:sll-ha} also shows the canonical representations of
each state, i.e., the minimal unfoldings of each state's
formula. These are the symbolic heaps that are substituted into the
predicate calls in symbolic heaps $\sh$ to obtain simple entailment
problems for deciding transitions $(\T{q}, \sh, q)$.

\subsection{Example: Trees with Linked Leaves}
We consider an example of the systematic construction of heap automata:
We determine the equivalence
classes $\EE{\mathtt{tll}}$, where the SID \texttt{tll} is defined as
in Example~\ref{ex:srd}, to obtain a heap automaton for the set of
well-determined reduced symbolic heaps that entail some tree with
linked leaves (TLL). Note that this SID is outside the scope of previous
decision procedures for entailment with (at most) exponential-time
complexity~\cite{berdine2004decidable,iosif2014entailment}.

For simplicity of presentation, we assume that 
all parameters of the \texttt{tll} predicate are different from $\NIL$
and consider only acyclic TLLs.\footnote{The only parameter that can
  actually be null is $\IFV{3}$, as the other two parameters are
  always equal to a variable that appears on the left-hand side of a
  points-to assertion. Likewise, $\IFV{3}$ is the only parameter that
  can introduce a cycle---it can be equal to $\IFV{1}$ or $\IFV{2}$.}
Recall that we can always check whether these conditions are satisfied
(see Sections~\ref{sec:zoo:track}, \ref{sec:zoo:acyclicity}).
%
%
%

We will represent each $\EE{\mathtt{tll}}$--equivalence class by an SID
whose unfoldings are exactly the reduced symbolic heaps contained in the respective class.
To this end, let \texttt{core} be a predicate specifying TLLs that lack the outgoing pointer of the right-most leaf:
\[
\begin{array}{l}
\PCDS{core} ~\SRDARROW~ \EMP : \{\IFV{1} = \IFV{2}\} \\
\PCDS{core} ~\SRDARROW~ \exists \ell\,r\,z ~.~ \PT{
  \IFV{1} }{ \ell\,r\,\NIL } \\ \qquad\qquad\quad \SEP ~ \PCDS{tll}{(\ell\,\IFV{2}\,z)} \SEP \PCDS{core}{(r\,z\,\IFV{3})}
\end{array}
\]
Here, the omitted pointer is reflected by the missing points-to assertion in the base case.
%
In the following, let $\PERMS$ be the set of all permutations of
$\FV{0}{}$.
  \newcommand{\rhTfv}{\T{y}_{0}}
  \newcommand{\rhfv}{\PIFV} 
\newcommand{\tllclasses}{
  \begin{enumerate}
  \item $\ITLLC{1}\ ~\SRDARROW~ \PT{ \rhfv{1} }{
      \rhfv{2}\,\rhfv{3}\,\NIL }$
  \item $\ITLLC{2}\ ~\SRDARROW~ \PT{ \rhfv{1} }{ \NIL\,\NIL\,\rhfv{3}
    } ~:~ \{ \rhfv{1} = \rhfv{2} \}$
  \item $\ITLLC{3}\ ~\SRDARROW~ \PT{ \rhfv{1} }{ \NIL\,\NIL\,\rhfv{3}
    } ~:~ \{ \rhfv{2} = \rhfv{3} \}$
  \item $\ITLLC{4}\ ~\SRDARROW~ \PT{ \rhfv{1} }{ \NIL\,\NIL\,\rhfv{3}
    } ~:~ \PUREA{}$

  where $\PUREA{} \DEFEQ{} \{ \rhfv{1} \neq \rhfv{2}, \rhfv{2} \neq \rhfv{3} \}$

  \item $\ITLLC{5}\ ~\SRDARROW~
    \PCDS{core}{\rhTfv}$ 
  \item $
      \ITLLC{6}\ ~\SRDARROW~ \exists \ell \, r \, z ~.~ \PT{
        \rhfv{1} }{ \ell\, r\, \NIL }$

      $\qquad\qquad ~\SEP\, \PCDS{tll}{(\ell\, \rhfv{2}\, z)} \SEP
      \PCDS{tll}{(r\, z\, \rhfv{3})}$

  \item  $\ITLLC{7}\ ~\SRDARROW~ \exists u \, \ell \, r \, z ~.~ \PTS{\rhfv{2}}{u} \SEP \PT{
        \rhfv{1} }{ \ell\, r\, \NIL }$

     $\qquad\qquad ~\SEP\, \PCDS{tll}{(\ell\,u\,z)} \SEP
      \PCDS{core}{(r\, z\, \rhfv{3})}$


  \item  $\ITLLC{8}\ ~\SRDARROW~ \exists u ~.~ \PT{\rhfv{2}}{(\NIL\,
      \NIL\, u)} \SEP \PCDS{tll}{(\rhfv{1}\, u\, \rhfv{3})}$
  \end{enumerate}
}
\begin{example} \label{ex:tll-classes} For each $\PERM \in \PERMS$,
  the relation $\EE{\mathtt{tll}}$ has 8 equivalence classes
  that can be extended to entail an unfolding of
  \texttt{tll}:

  \tllclasses{}
  We refer the reader to Figure~\ref{fig:fixed:tll} for an illustration of the unfoldings covered by each predicate. 
  Due to lack of space, this figure is found in the appendix.
  In addition, $\EE{\mathtt{tll}}$ has one equivalence class of
  symbolic heaps that cannot be extended to entail \texttt{tll}
  unfoldings, defined as the complement of the other classes:
  $\RSHCLASSFV{\alpha} \setminus \bigcup_{1 \leq i \leq 8, \PERM \in
    \PERMS} \ITLLC{i}$ 
\end{example}

\begin{figure*}
\begin{tikzpicture}[->,>=stealth',shorten >=1pt,auto,node distance=1.5cm,
                    semithick]
  \tikzstyle{every state}=[fill=gray!20,draw=none,text=black]
\matrix[column sep=0.1cm,row sep=5mm,column 1/.style={anchor=base west}] {
\node { $\ITLLC{1}\ ~\SRDARROW~ \PT{ \rhfv{1} }{
      \rhfv{2}\,\rhfv{3}\,\NIL }$};
\pgfmatrixnextcell \node {
\begin{tikzpicture}[->,>=stealth',shorten >=1pt,auto,node distance=1.5cm,
                    semithick]
  \tikzstyle{every state}=[fill=gray!20,draw=none,text=black]

  \node[state]  (1) {$1$};
  \node[state]  (2) [below left of = 1] {$2$};
  \node[state]  (3) [below right of = 1] {$3$};
  \path (1) edge (2)
            edge (3);
\end{tikzpicture}
}; \\
\node {$\ITLLC{2}\ ~\SRDARROW~ \PT{ \rhfv{1} }{ \NIL\,\NIL\,\rhfv{3}
    } ~:~ \{ \rhfv{1} = \rhfv{2} \}$};
\pgfmatrixnextcell \node {
\begin{tikzpicture}[->,>=stealth',shorten >=1pt,auto,node distance=1.5cm,
                    semithick]
  \tikzstyle{every state}=[fill=gray!20,draw=none,text=black]

  \node[state]  (1) {$1,2$};
  \node[state]  (3) [right of = 1] {$3$};
  \path (1) edge (3);
\end{tikzpicture}
}; \\
\node {$\ITLLC{3}\ ~\SRDARROW~ \PT{ \rhfv{1} }{ \NIL\,\NIL\,\rhfv{3}
    } ~:~ \{ \rhfv{2} = \rhfv{3} \}$};
\pgfmatrixnextcell \node {
\begin{tikzpicture}[->,>=stealth',shorten >=1pt,auto,node distance=1.5cm,
                    semithick]
  \tikzstyle{every state}=[fill=gray!20,draw=none,text=black]

  \node[state]  (1) {$1$};
  \node[state]  (3) [right of = 1] {$2,3$};
  \path (1) edge (3);
\end{tikzpicture}
}; \\
\node {$\ITLLC{4}\ ~\SRDARROW~ \PT{ \rhfv{1} }{ \NIL\,\NIL\,\rhfv{3}
    } ~:~ \PUREA{}$};
\pgfmatrixnextcell \node {
\begin{tikzpicture}[->,>=stealth',shorten >=1pt,auto,node distance=1.5cm,
                    semithick]
  \tikzstyle{every state}=[fill=gray!20,draw=none,text=black]

  \node[state]  (1) {$1$};
  \node[state]  (3) [right of = 1] {$3$};
  \node[state]  (2) [right of = 3] {$2$};
  \path (1) edge (3);
\end{tikzpicture}
}; \\
\node {$\ITLLC{5}\ ~\SRDARROW~
    \PCDS{core}{\rhTfv}$};
\pgfmatrixnextcell \node {
\begin{tikzpicture}[->,>=stealth',shorten >=1pt,auto,node distance=1.5cm,
                    semithick]
  \tikzstyle{every state}=[fill=gray!20,draw=none,text=black]

  \node[state]  (1) {$1$};
  \node[state]  (2) [below left of = 1] {$2$};
  \node[state]  (3) [below right of = 1] {$3$};
  \path (1) edge (2)
            edge (3)
        (2) edge (3);
\end{tikzpicture}
}; \\
\node {$\ITLLC{6}\ ~\SRDARROW~ \exists \ell \, r \, z ~.~ \PT{
        \rhfv{1} }{ \ell\, r\, \NIL }$};
\pgfmatrixnextcell \node {
\begin{tikzpicture}[->,>=stealth',shorten >=1pt,auto,node distance=1.2cm,
                    semithick]
  \tikzstyle{every state}=[fill=gray!20,draw=none,text=black]

  \node[state]  (1) {$1$};
  \node[state]  (2) [below left of = 1] {$2,\ell$};
  \node[state]  (3) [below right of = 1] {$r,z$};
  \node[state]  (4) [right of = 3] {$3$};
  \path (1) edge (2)
            edge (3)
        (2) edge (3)
        (3) edge (4);
\end{tikzpicture}
}; \\
\node {$\ITLLC{7}\ ~\SRDARROW~ \exists u \, \ell \, r \, z ~.~ \PT{\rhfv{2}}{u} \SEP \PT{
        \rhfv{1} }{ \ell\, r\, \NIL }$};
\pgfmatrixnextcell \node {
\begin{tikzpicture}[->,>=stealth',shorten >=1pt,auto,node distance=1.2cm,
                    semithick]
  \tikzstyle{every state}=[fill=gray!20,draw=none,text=black]

  \node[state]  (1) {$1$};
  \node[state]  (2) [below left of = 1] {$u,\ell$};
  \node[state]  (3) [below right of = 1] {$r,3$};
  \node[state]  (4) [left of = 2] {$2$};
  \path (1) edge (2)
            edge (3)
        (2) edge (3)
        (4) edge (2);
\end{tikzpicture}
}; \\
\node {$\ITLLC{8}\ ~\SRDARROW~ \exists u ~.~ \PT{\rhfv{2}}{\NIL\,
      \NIL\, u} \SEP \PCDS{tll}{\rhfv{1}\, u\, \rhfv{3}}$};
\pgfmatrixnextcell \node {
\begin{tikzpicture}[->,>=stealth',shorten >=1pt,auto,node distance=1.5cm,
                    semithick]
  \tikzstyle{every state}=[fill=gray!20,draw=none,text=black]

  \node[state]  (1) {$1,u$};
  \node[state]  (2) [left of = 1] {$2$};
  \node[state]  (3) [right of = 1] {$3$};
  \path (1) edge (3);
  \path (2) edge (1);
\end{tikzpicture}
}; \\
};
\end{tikzpicture}
\caption{Equivalence classes with graphical representatives of small unfoldings of each predicate.
         Here, numbers correspond to the index of tuple $\T{y}$  and letters to the respective variables in the predicates.}
\label{fig:fixed:tll}
\end{figure*}
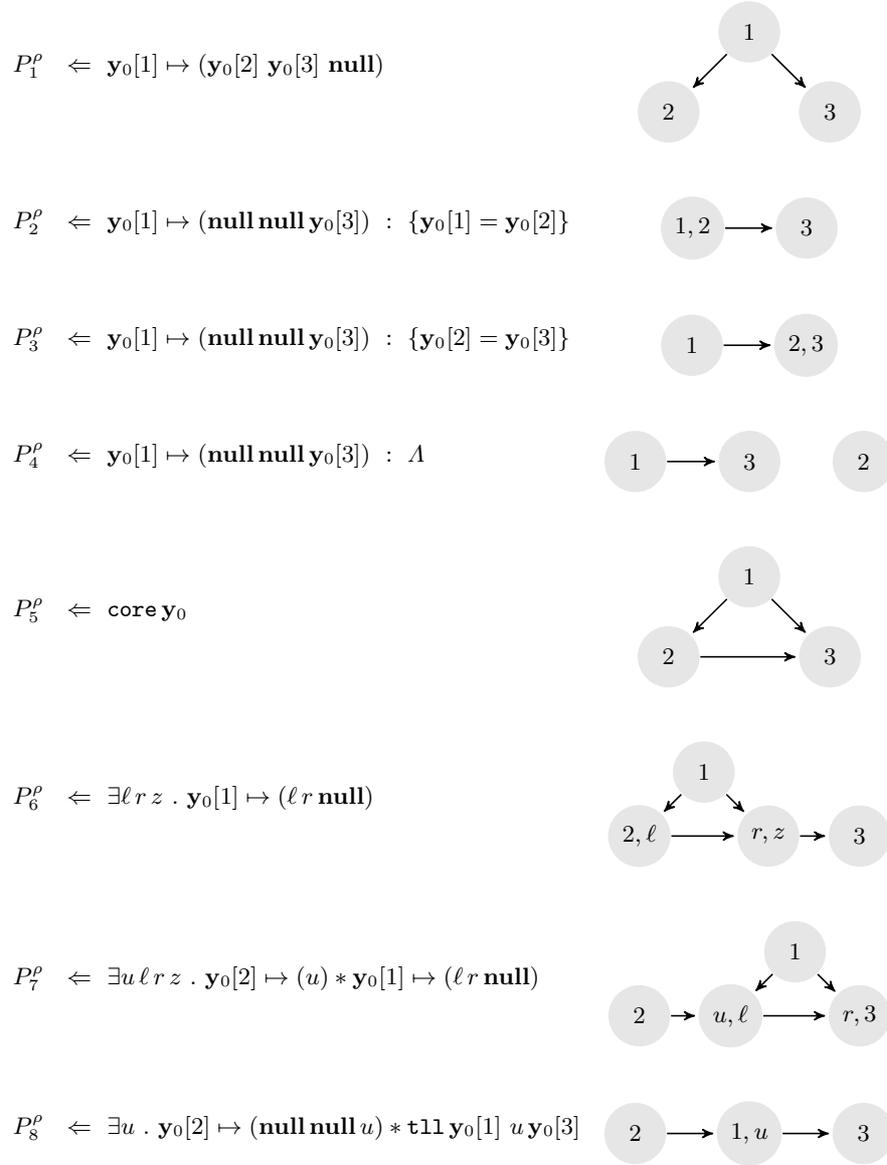

For example, $\ITLLC{7}$ specifies all trees with linked leaves that
consist of more than one pointer, lack the successor of the rightmost
leaf, but have an additional element at the head of the list of linked
leaves. All such symbolic heaps can be extended to heaps that entail
an unfolding of $\PCDS{tll}{\FV{0}{}}$ in the same way; for
example by adding the missing pointer from the last leaf to its successor
and adding a new root node that points to the additional leaf at the head
of the list as well as to the root of the
$\ITLLC{7}$--unfolding. Formally, for the identity permutation
$\fid \in \PERMS$, let
\[
  \begin{array}{l}
\sh^{\fid} ~\DEFEQ{}~ \exists u \, v ~.~ \PT{\IFV{1}}{\IFV{2}\, u\,
    \NIL} \SEP \ \APP{P_7^{\fid}}{(u\,\IFV{2}\, v)} 
    \SEP \ \PT{v}{\NIL\,\NIL\, \IFV{3 }}.
  \end{array}
\]
Then $\sh^{\fid}\,\FV{0}{} \ENTAIL{} \PCDS{tll}{\FV{0}{}}$. Note that if we
change the predicate call from $P_7^{\fid}$ to $\ITLLC{7}$ for a
different $\PERM$, the parameters in the predicate call have to
be reordered in accordance with $\PERM$ to obtain a valid entailment
$\sh^{\PERM} \ENTAIL{} \PCDS{tll}{\FV{0}{}}$; hence different permutations
induce different equivalence classes.
Further details are provided in \TECHNICALREPORT.
As there are six permutations of
$\FV{0}{}$, we conclude 
\begin{corollary}\label{cor:tll-automaton}
  Let $\SRD$ be the SID defining \textnormal{\texttt{tll}} from
  Example~\ref{ex:srd}.  There is a heap automaton
  $\HA{A}_{\mathtt{tll}}$ with $8 \cdot 6 + 1 = 49$ states that
  accepts the set $\USET{\mathtt{tll}\,\T{x}}{\SRD}{\CENTAIL{3}}$.
\end{corollary}

\begin{remark}
  One might be tempted to think that we have only moved the complexity
  into the transition relation and thus gained nothing:
  Instead of checking entailment w.r.t.~\texttt{tll}, we now have to
  verify entailments w.r.t.~several new SIDs.
  However, the states of our heap automaton coincide with
  the $\EE{\texttt{tll}}$-equivalence classes.
  Thus, it always suffices to evaluate the transition relation w.r.t.~ to small \emph{canonical
   representations} of each equivalence class.
  %
  Hence, instead of solving entailments of the
  form $P\, \T{x} \ENTAIL{} \mathtt{tll}\, \T{x}$ 
  it suffices to consider entailments
  $\rsh \ENTAIL{} \ITLLC{i}\, \T{x}$,
  where $\rsh$ is a \emph{reduced} symbolic heap.
  Moreover, since we assume well-determined symbolic heaps, each such
  $\rsh$ has a single tight model up to isomorphism.
  Then verifying $\rsh \ENTAIL{} \ITLLC{i}\, \T{x}$ boils down to the
  model-checking problem for fixed SIDs $\ITLLC{i}$. It follows by
  Remark~\ref{rem:model-checking} that
  $\Delta_{\HA{A}_{\mathtt{tll}}}$ is decidable in time
  $\BIGO{2^{\text{poly}(\SIZE{\SRD})}}$ (where $\SRD$ is the SID for
  $\mathtt{tll}$ from Example~\ref{ex:srd}).
  We thus apply
  Lemma~\ref{thm:entailment:exptime} to conclude that Algorithm~\ref{alg:entailment:decision-procedure}
  -- if fed with $\HA{A}_{\mathtt{tll}}$ -- is an \CCLASS{ExpTime} decision procedure for
  entailments containing \texttt{tll} predicates.
\end{remark}

\begin{remark}
  Heap automata allow for easy integration of additional syntactic conditions
  in order to obtain more compact state spaces. 
  For instance, one observation is that many states of $\HA{A}_{\mathtt{tll}}$ collapse if
  reordering of parameters is restricted.

  Another approach is to fix the SID in advance and construct an
  automaton
  $\HA{A}^{\SRD} = (Q, \SL{\SRD}{\CENTAIL{\alpha}}, F, \Delta)$ that
  is only defined on symbolic heaps in
  $\SL{\SRD}{\CENTAIL{\alpha}}$. Depending on $\SRD$, there may be
  many $\EE{\PS}$--equivalence classes $\EC{\tau}{\PS}$ with
  $\EC{\tau}{\PS} \cap \SL{\SRD}{\CENTAIL{\alpha}} = \emptyset$.
  Dropping such classes can lead to significantly smaller
  automata. 
  By restricting heap automata in this way, we obtain decision procedures
  for (established) symbolic heaps with fixed (as opposed to user-defined)
  SIDs, such as the original decidable symbolic heap
  fragment introduced by Berdine et al.~\cite{berdine2004decidable}.
\end{remark}


\subsection{Proof of Theorem~\ref{thm:eq:myhill}} \label{app:eq:myhill}

\begingroup
\def\thetheorem{\ref{thm:eq:myhill}}
\begin{theorem}[Myhill--Nerode~\cite{nerode1958linear}]
   Let $\PS$ be a predicate symbol of an SID $\SRD$.
   Then there exists a heap automaton $\HASH{\PS}$ with
   $L(\HASH{\PS}) = \USET{\PS\FV{0}{}}{\SRD}{\CENTAIL{\alpha}}$ 
   iff
   the number of equivalence classes of $\EE{\PS}$ is finite.
\end{theorem}
\begin{proof}
  Assume there are only finitely many equivalence classes of $\EE{\PS}$.
  Furthermore, let $\EC{\rsh}{\PS}$ denote the equivalence class containing formula $\rsh$.
  Then a heap automaton $\HASH{\PS} = (Q,\SHCENTAIL{\alpha},\Delta,F)$ accepting
  $\USET{\PS\FV{0}{}}{\SRD}{\CENTAIL{\alpha}}$
  is given by
  $Q = \{ \EC{\rsh}{\PS} ~|~ \rsh \in \RSL{}{\CENTAIL{\alpha}} \}$,
  $F = \{ \EC{\rsh}{\PS} ~|~ \rsh \in \RSL{}{\CENTAIL{\alpha}}, \rsh \ENTAIL{\SRD} \PS\FV{0}{} \}$
  and
  \begin{align*}
    \MOVE{\HASH{\PS}}{\EC{\sh[\CALLN{1}{} / \rsh_1, \CALLN{m}{} / \rsh_m]}{\PS}}{\sh}{\EC{\rsh_1}{\PS} \ldots \EC{\rsh_m}{\PS}}~,
  \end{align*}
  for each symbolic heap $\sh \in \SHCENTAIL{\alpha}$ with $\NOCALLS{\sh} = m$.
  Now, a straightforward induction on the height of unfolding trees reveals for each
  unfolding tree $t$ and each $q \in Q$, we have
  \begin{align*}
    \OMEGA{\HASH{\PS}}{q}{\UNFOLD{t}} ~\text{iff}~ \UNFOLD{t} \in q.
  \end{align*}
  Then, by definition of $F$, we have
  \begin{align*}
    \rsh \ENTAIL{\SRD} \PS\FV{0}{}
    ~\text{iff}~ \bigvee_{q \in F} \rsh \in q
    ~\text{iff}~ \bigvee_{q \in F} \OMEGA{\HASH{\PS}}{q}{\rsh}
    ~\text{iff}~ \rsh \in L(\HASH{\PS}).
  \end{align*}
  For the converse direction, let $\HASH{\PS}$ be a heap automaton
  with $L(\HASH{\PS}) = \USET{\PS\FV{0}{}}{\SRD}{\CENTAIL{\alpha}}$.
  Without loss of generality, we assume that $\HASH{\PS}$ is \emph{deterministic},
  i.e., for each $\rsh \in \RSL{}{\CENTAIL{\alpha}}$ there exists exactly
  one $q \in Q_{\HASH{\PS}}$ such that $\OMEGA{\HASH{\PS}}{q}{\rsh}$ holds.
  Otherwise, we may apply the (subset) construction similar to the proof
  of Lemma~\ref{thm:refinement:complement} to obtain such an automaton (with an exponentially larger state space).
  Now, for each $q \in Q_{\HASH{\PS}}$ let
  \begin{align*}
    L_q \DEFEQ \{ \rsh \in \RSL{}{\CENTAIL{\alpha}} ~|~ \OMEGA{\HASH{\PS}}{q}{\rsh} \}.
  \end{align*}
  Then, by assumption, we have $L_{p} \cap L_{q} = \emptyset$ for each $p,q \in Q_{\HASH{\PS}}$ with $p \neq q$.
  Moreover,
  \begin{align*}
     \RSL{}{\CENTAIL{\alpha}} ~\DEFEQ~ \bigcup_{q \in Q} L_{q}.
  \end{align*}
  Hence, $\textit{Part}(\HASH{\PS}) \DEFEQ \{ L_{q} ~|~ q \in Q \}$ is a partition of $\RSL{}{\CENTAIL{\alpha}}$.
  Further, $\rsh,\rsha \in L_{q}$ clearly implies $\rsh \EE{\PS} \rsha$, i.e., our partition $\textit{Part}(\HASH{\PS})$
  refines the partition induced by $\EE{\PS}$.
  Thus, $\SIZE{\textit{Part}(\HASH{\PS}}) \leq \SIZE{Q_{\HASH{\PS}}}$ is an upper bound on the number of equivalence class
  of $\EE{\PS}$.
  Since $Q_{\HASH{\PS}}$ is a finite set, this means that $\EE{\PS}$ has only finitely many equivalence classes.
  \qed
\end{proof}
\addtocounter{theorem}{-1}
\endgroup

\section{Details on Heap Automata for Entailment}

\subsection*{Correctness Proof for
  the Construction in Example~\ref{ex:tll-classes}}\label{app:tll-classes}

\paragraph{Induced automata.}
Recall that formulas can only be $\PS$-equivalent if they have
the same number of free variables $\beta$. A heap automaton based on
the Myhill--Nerode construction (cf.~Theorem~\ref{thm:eq:myhill}) for
$\SHCLASSFV{\alpha}$ therefore has states for the equivalence classes
for each $\beta \leq \alpha$.  In Example~\ref{ex:tll-classes} we,
however, restricted ourselves to the classes with exactly $\alpha = 3$
free variables.
We first justify this
choice by noting that the equivalence classes of formulas with exactly
$\alpha$ free variables induce all classes for $\beta < \alpha$. This
is true in general, not only for the \texttt{tll} case study.
\begin{remark}
  Let $\HA{A}^{=\alpha}$ be a heap automaton obtained by the
  Myhill--Nerode construction of Theorem~\ref{thm:eq:myhill}
  restricted to symbolic heaps with exactly $\alpha$ parameters
  and restricted to acyclic models.
  Moreover, let each equivalence class be represented by an SID.
  Then $\HA{A}^{=\alpha}$ induces an
  automaton $\HA{A}^{=\beta}$ for symbolic heaps with exactly $\beta$
  variables for each $\beta < \alpha$.
\end{remark}


\begin{proof}[Proof sketch]
  $\HA{A}^{=\beta}$ is obtained by
  \begin{enumerate}
  \item removing all classes that are defined by SIDs that use more
    than $\beta$ parameters\footnote{Or, more precisely, more than
      $\beta$ parameters that are pairwise not definitely
      equal} in points-to assertions, and
  \item for the remaining classes, defining new SIDs by adding the
    closure of the pure formulas, dropping those parts of the pure
    formulas that do not refer to variables that occur in points-to
    assertions, and then projecting the remaining parameters onto the
    first $\beta$ free variables.\footnote{Note that this is only
      sufficient when we assume acyclicity. Otherwise, variables that
      only occur on the right-hand side of points-to assertions can
      also be identified with variables on the left-hand side, thus
      also reducing the number of free variables in the symbolic
      heap.} 
  \end{enumerate}
  \qed
\end{proof}

\begin{example}
  The automaton derived from Example~\ref{ex:tll-classes} is the
  automaton $\HA{A}_{\mathtt{tll}}^{=3}$ for acyclic symbolic heaps
  with exactly $3$ free variables. The only classes of
  $\HA{A}_{\mathtt{tll}}^{=3}$ that do not use all $3$ parameters in
  points-to assertions are the classes $\ITLLC{2}$, $\ITLLC{3}$, and
  $\ITLLC{4}$.

  For each of these classes, the projection yields the same new
  classes, defined by the SIDs
\[
  \begin{array}{l}
R_1\,\IFV{1}\,\IFV{2} ~\SRDARROW~ \PT{ \IFV{1} }{
    \NIL\,\NIL\,\IFV{2}} \\
R_2\,\IFV{1}\,\IFV{2} ~\SRDARROW~ \PT{ \IFV{2} }{ \NIL\,\NIL\,\IFV{1}}.
  \end{array}
\]

Thus, $\HA{A}_{tll}^{=2}$ has only three states: The two states
defined by $R_1$ and $R_2$ and one state for all other RSHs with two free
variables.

The automaton $\HA{A}_{tll}^{=1}$ has only a single state, because all
SIDs in the definition of $\HA{A}_{tll}^{=3}$ use at least two
variables in points-to assertions.
\end{example}

As $\HA{A}^{=\alpha}_{\mathtt{tll}}$ induces $\HA{A}^{=\beta}_\mathtt{tll}$ for $\beta<\alpha$, it
is sufficient to prove the correctness of the automaton $\HA{A}^{=3}_{\mathtt{tll}}$, which we will do in the following.


\paragraph*{Overview of the correctness proof.}

We first recall the definition of \texttt{tll} from Section~\ref{ex:srd}.
 \begin{align*}
   \textnormal{\texttt{tll}} ~\SRDARROW~ & \PT{\PROJ{\FV{0}{}}{1}}{\NIL\,\NIL\,\PROJ{\FV{0}{}}{3}} : \{ \PROJ{\FV{0}{}}{1} = \PROJ{\FV{0}{}}{2} \} \\
   \textnormal{\texttt{tll}} ~\SRDARROW~ & \exists \ell\,r\,z ~.~ \PT{\PROJ{\FV{0}{}}{1}}{\ell\,r\,\NIL} \\
                                           & \qquad \SEP \textnormal{\texttt{tll}}\,\ell\,\PROJ{\FV{0}{}}{2}\,z \SEP \textnormal{\texttt{tll}}\,r\,z\,\PROJ{\FV{0}{}}{3}. 
 \end{align*}

 Towards a correctness proof for Example~\ref{ex:tll-classes}, we have
 to show that the predicates defined there (and repeated below)
 correspond exactly to the $\EE{\mathtt{tll}}$-equivalence classes
 that contain all reduced symbolic heaps $\rsh$ that can be extended
 to reduced symbolic heaps $\rsh'$ that entail a \texttt{tll}
 unfolding, i.e., $\rsh' \ENTAIL{} \mathtt{tll}\,\T{x}$. For brevity,
 we will henceforth call such symbolic heaps \emph{partial models} of
 the \texttt{tll} predicate.\footnote{As we are in a well-determined
   setting, there is a one-to-one correspondence between \emph{reduced
     symbolic heap} and \emph{model}, so we sometimes use these terms
   interchangeably. Further, let $\rsh \in \RSL{}{}$ with
   $\SIZE{\FV{0}{\rsh}} = \beta$ and $\PS \in \PRED(\SRD)$. $\rsh$ is
   a \emph{partial model} of $\PS$ if there exists a
   $\sh \in \SL{\SRD}{}$ with exactly one predicate call $\PI$ of
   arity $\beta$, such that $\sh[\PI / \rsh] \ENTAIL{\SRD} P\FV{0}{}$.
 } In the following, recall that $\PIFV{i} \DEFEQ{} \PERM(\IFV{i})$.


\tllclasses{}

An illustration of each predicate, together with a graphical representation of a corresponding
unfolding, is found in Figure~\ref{fig:fixed:tll}.

The correctness proof consists of the following steps.
\begin{enumerate}
\item \emph{Partition.} For each pair of predicates
  $\ITLLC{i}, P_j^{\PERMA}$,
  $1 \leq i,j \leq \NUMTLL, \PERM,\PERMA\in \PERMS$, we show that the sets
  $\CALLSEM{\ITLLC{i}\,\T{x}}{}$ and $\CALLSEM{P_j^{\PERMA}\,\T{x}}{}$ are
  disjoint.
\item \emph{Equivalence.} We prove for each $\PERM$ and each predicate
  $\ITLLC{i}, 1 \leq i \leq \NUMTLL$, that for all
  $\rsh,\rsha \in \CALLSEM{\ITLLC{i}\,\T{x}}{}$, $\rsh \EE{\mathtt{tll}}
  \rsha$ holds.
\item \emph{Completeness.} We show that the predicates $\ITLLC{i}$
  define all equivalence classes of partial models, i.e., for each
  partial model $\rsh$ of a \texttt{tll} unfolding there exist an $i$
  and a $\PERM \in \PERMS$ such that $\rsh \in \ITLLC{i}$.
\end{enumerate}
Once we have established these properties, we immediately obtain a
heap automaton for \texttt{tll}-entailment due to
Theorem~\ref{thm:eq:myhill}.

\paragraph{Partition.}
We first show that the sets $\CALLSEM{\ITLLC{i}}{}$ partition the set
of all partial models of \texttt{tll}.
To this end, we define for each $1 \leq i \leq \NUMTLL$ and $\PERM \in
\PERMS$ a set of
formulas $\{\CTXT{i}{1}, \ldots \CTXT{i}{n_i} \}$ with a single predicate
call $\PI$ of arity $3$ such that
\begin{enumerate}
\item For all $1 \leq k \leq n_i$ and for all
  $\rsh \in \CALLSEM{\ITLLC{i}\,\T{x}}{}$,
  $\CTXT{i}{k}[\PI/\rsh] \ENTAIL{} \PCDS{tll}{\FV{0}{}}$,
\item For all elements $\rsha \in \CALLSEM{P_j^{\PERMA}\,\T{x}}{}$ with $\ARITY(P_j^{\PERMA}) = \NOFV{\PCTXT{\PERM}{i}{k}}$,
  where $j \neq i$ and/or $\PERM \neq \PERMA$, there exists a
  $1 \leq k \leq n_i$ such that
  $\PCTXT{\PERM}{i}{k}[\PI/\rsha] \not\ENTAIL{} \PCDS{tll}{\FV{0}{}}$.
\end{enumerate}
In other words, we provide a \emph{distinguishing context} for each $\ITLLC{i}$.
Permuting the free variables according to $\PERM$ will lead to a
permutation of the parameters of the parameter calls in the
distinguishing context. To express this in a uniform way, we write
$\LIFT{u_1\, u_2\, u_3}$ to denote a call to $\PI$ where the
parameters $u_1\, u_2\, u_3$ are reordered in accordance with
$\PERM$.\footnote{E.g., if $\PERM(\IFV{i}) = 3 - i$, then
  $\LIFT{u_1 u_2 u_3} \DEFEQ \PI\,u_3\,u_2\,u_1$.}
\begin{itemize}



\item The (single) member of $\LTLLC{1}$ is a binary tree with three
  nodes\footnote{Because we assume acyclicity, the three nodes are
    definitely different.} and without linked leaves. The smallest
  formula yielding a valid \texttt{tll} thus needs to add the two link
  edges from the left child and the right child and from the right child to
  its successor.

$\CTXT{1}{1}\,\T{x} \DEFEQ{} \exists u ~.~
  \PT{\PIFV{2}}{\NIL\, \NIL\, u} \SEP \PT{u}{\NIL\, \NIL\, \PIFV{3}}
  \SEP \LIFT{\PIFV{1}\, \PIFV{2}\, u}$

\item The (single) member of $\LTLLC{2}$ is the smallest valid
  \texttt{tll}, consisting only of a root (which is thus the leftmost
  leaf) and its successor.
  The smallest context yielding a
  \texttt{tll} is thus the symbolic heap that contains nothing but the
  predicate call. To distinguish $\ITLLC{2}$ from
  $\ITLLC{3}$ and $\ITLLC{4}$, we add the pure formula $\PIFV{1} = \PIFV{2}$.

  $\CTXT{2}{1}\,\T{x} \DEFEQ{} \LIFT{\T{x}}  : \{ \PIFV{1} = \PIFV{2} \}$

  Alternatively, we obtain a \texttt{tll} by using the pointer
  allocated in $\ITLLC{3}$ as the pointer from the leftmost source to an
  inner leaf.

  $\CTXT{2}{2}\,\T{x} \DEFEQ{} \exists u \, v ~.~ \PT{\PIFV{1}}{\PIFV{2}\, v\,
    \NIL} \SEP \LIFT{\PIFV{2}\, u\, v} \SEP \PT{v}{ \NIL\,\NIL\,
    \PIFV{3} } $

  Note that the pure formula in the definition of $\ITLLC{2}$ enforces
  $u = \PIFV{2}$. The only reason to quantify over $u$ (rather than
  just writing $\PIFV{2}$) is reusability of the formula in
  the treatment of $P_3$ and $P_4$ below 

\item We get the context for $\LTLLC{3}$ and $\LTLLC{4}$ in the
  same way as for
  $\LTLLC{2}$, i.e.,

  $\CTXT{3}{1}\,\T{x} = \LIFT{\T{x}}  : \{ \PIFV{2} = \PIFV{3} \}$

  $\CTXT{4}{1}\,\T{x} = \LIFT{\T{x}}  : \{ \PIFV{1} \neq \PIFV{2},
  \PIFV{2} \neq \PIFV{3} \}$

  $\CTXT{3}{2} = \CTXT{4}{2} \DEFEQ{} \CTXT{2}{2}$

\item The partial models generated by $\ITLLC{5}$ differ from
  \texttt{tll}s only in the absent pointer from the rightmost leaf to
  its successor.

$\CTXT{5}{1}\,\T{x} \DEFEQ{} \exists u ~.~ \LIFT{\PIFV{1}\, \PIFV{2}\, u} \SEP
  \PT{u}{\NIL\,\NIL\, \PIFV{3}}$
\item $\ITLLC{6}$ describes all \texttt{tll}s that are obtained by applying
  the second rule of the \texttt{tll}-SID at least once. All $\rsh \in
  \LTLLC{6}$ thus already entail a \texttt{tll} unfolding
  without extension.

$\CTXT{6}{1}\,\T{x} \DEFEQ{} \LIFT{\T{x}}$ 
\item $\ITLLC{7}$ describes all \texttt{tll}-like heaps that consist of more
  than one pointer, lack the successor of the rightmost leaf, but
  have an additional predecessor of the leftmost leaf. To get a
  \texttt{tll}, we thus have to both add that last pointer and add a
  tree-structure that points to the additional leaf at the left as
  well as the $\ITLLC{7}$-root.

  $\CTXT{7}{1}\,\T{x} \DEFEQ{} \exists u \, v ~.~ \PT{\PIFV{1}}{\PIFV{2}\, u\,
    \NIL} \SEP \LIFT{u\,\PIFV{2}\,v} \SEP \PT{v}{ \NIL\,\NIL\,
    \PIFV{3} }$

\item $\ITLLC{8}$ describes \texttt{tll}-like graphs with an additional
  predecessor of the leftmost leaf.

$\CTXT{8}{1}\,\T{x} \DEFEQ{} \exists y ~.~ \PT{\PIFV{1}}{\PIFV{2}\, y\,
    \NIL} \SEP \LIFT{y\,\PIFV{2}\, \PIFV{3}}$
\end{itemize}
First observe that for a fixed permutation $\PERM$, each pair of
classes is separated by at least one formula. Indeed, this separation
already occurs for the minimal models of each $\ITLLC{i}$ and can thus
easily be verified by hand; we omit these tedious
constructions.
For fixed $\PERM$, this shows the pairwise disjointness of
the $\ITLLC{i}$ classes.

Now let
$\rsha \in \LTLLC{i}, \rshb \in \CALLSEM{P_i^{\PERMA}\,\T{x}}{}$,
for $\PERM \neq \PERMA \in \PERMS$. We have
$\PCTXT{\PERM}{i}{1}[\PI/\rsha] \ENTAIL{}
\PCDS{tll}{\PERM(\FV{0}{})}$,
and
$\PCTXT{\PERMA}{i}{1}[\PI/\rshb] \ENTAIL{}
\PCDS{tll}{\PERMA(\FV{0}{})}$,
but
$\CALLSEM{\mathtt{tll}\,\PERM(\T{x})}{} \cap
\CALLSEM{\mathtt{tll}\,\PERMA(\T{x})}{} = \emptyset$.
Thus we also have that, for fixed $i$ and $\PERM \neq \PERMA \in \PERMS$,
$\LTLLC{i} \cap \CALLSEM{P_i^{\PERMA}\,\T{x}}{} = \emptyset$:




Putting this together with the arguments above, we obtain a proof
for each $i, j \in \{1, \ldots, 8\}, \PERM,\PERMA \in \PERMS$, that if
$i \neq j$ or $\PERM \neq \PERMA$, then
$\LTLLC{i} \cap \CALLSEM{P_{j}^{\PERMA}\,\T{x}}{} = \emptyset$.

\paragraph{$\EE{\textnormal{\texttt{tll}}}$-Equivalence.}

For $1 \leq i \leq 4$, the set $\CALLSEM{P_i^{\rho}}{}$ is a singleton, i.e. there is nothing to show.
Thus, let $i > 4$.
Formally, the remaining cases are proven individually by coinduction:
We assume for each $1 \leq i \leq 8$ that all reduced symbolic heaps in
$\CALLSEM{P_i^{\rho}}{}$ are already known to be equivalent and then, for each rule with left-hand side $P_i^{\rho}$,
 show that the unfolding obtained from replacing each predicate call by one of these reduced symbolic heaps
is again equivalent.
We omit these tedious calculations and only
provide a rough intuition for each case.
Let $\rsh \in \CALLSEM{P_i^{\rho}}{}$ and $\sh$ be a symbolic heap with one predicate call $I\T{x}$ such that
$\sh[I / \rsh] \ENTAIL{\SRD} \texttt{tll}\FV{0}{}$.
Such a symbolic heap always exists, because each unfolding can be extended to a TLL.
We proceed by case distinction to see that $\sh$ and $\rsh$ are always of the same form -- except for the size of the respective TLL.
\begin{itemize}
  \item For $i=5$, it is easy to verify that exactly one pointer -- the one to the rightmost leaf -- is missing in $\rsh$.
        Thus, $\sh$ is always a TLL with an additional pointer to the rightmost leaf in which the rightmost subtree,
        $\rsh$, is missing (except for the additional pointer),
        e.g.  $\PCTXT{\PERM}{i}{5}$.
  \item For $i=6$, $\rsh$ is a non-empty TLL. Thus, $\sh$ always corresponds to one of the $\texttt{tll}$ rules with the single predicate being $\texttt{tll}\T{x}$,
        e.g.  $\PCTXT{\PERM}{i}{6}$.
  \item For $i=7$, $\rsh$ is a TLL that lacks the pointer to the rightmost leaf and has an additional pointer to the leftmost leaf.
        Thus, $\sh$ always corresponds to a tree-structure lacking a pointer to the rightmost leaf and an additional isolated pointer,
        e.g.  $\PCTXT{\PERM}{i}{7}$.
  \item For $i=8$, $\rsh$  is a TLL with an additional incoming pointer at the leftmost leaf.
        Thus, $\sh$ always corresponds to a TLL lacking the rightmost subtree,
        e.g.  $\PCTXT{\PERM}{i}{8}$.

\end{itemize}
%
%
%
%


\paragraph{Completeness.}

We now show that the $8 \cdot 6$ predicates
$\ITLLC{1} \ldots \ITLLC{\NUMTLL}$, $\PERM \in \PERMS$, cover all
partial models of \texttt{tll} unfoldings.

The basic idea is to show that every partial model of \texttt{tll} is
either in one of the sets $\LTLLC{i}$ or has at least four free
variables.\footnote{Where, as before, we do not count $\IFV{0}$ as
  free variable.}  To this end, we first develop a sufficient
condition for concluding that a partial model needs at least four free
variables.
%
%
Recall that we assume reduced symbolic heaps under consideration to be well-determined.
Since there exists exactly one tight model, i.e. a stack-heap pair, up to isomorphism for each such symbolic heap,
we will argue about the structure of such a stack-heap pair instead.

Intuitively, a location that occurs in two separated heaps, i.e. $\heap \uplus \heap'$, is called a \emph{shared location}.
Formally, for a heap $\heap$, we write $\LOC(\heap)$ to denote all locations allocated or referenced -- pointed to -- to in $h$, i.e. $\LOC(\heap) = \DOM(\heap) \cup co\DOM(\heap)$.
\begin{definition}[Shared Location]
  Let $(\stack,\heap) \in \STATES$ be a stack-heap pair
  and $\heap'$ be a heap such that $\heap \uplus \heap'$ is well-defined.
  Then, a location $\ell \in \LOC(\heap)$ is a \emph{shared location} of $\heap$ w.r.t. $\heap'$ if $\ell \in \LOC(\heap')$ or $\ell \in co\DOM(\stack)$.
\end{definition}
%


\begin{lemma}
  Let $(\stack,\heap_1 \uplus \heap_2)$ be the unique tight model up to isomorphism
  of $\sh[I / \rsh]$ for some partial model $\rsh$ of $\PS$
  such that $(\stack,\heap_1)$ is a model of $\rsh$.
  If $\heap_1$ has $\alpha$ shared locations w.r.t. $\heap_2$
  then $\rsh$ has at least $\alpha$ free variables.
\end{lemma}
\begin{proof}[Proof sketch]
  Since $\rsh$ is established, equalities and reachability between symbolic heaps
  and unfoldings of their predicate calls have to be propagated through
  predicate calls
  (cf. Lemma~\ref{obs:zoo:establishment:equality-propagation}
   and Lemma~\ref{obs:zoo:reachability:propagation}).
  Thus, if $\heap_1$ has $\alpha$ shared locations w.r.t. $\heap_2$, $\rsh$ has at least $\alpha$
  free variables.
\end{proof}

To show the completeness of the equivalence classes $\LTLLC{i}$, it
thus suffices to show that every well-determined partial model is
either in one of the classes or its unique tight model
has at least four shared locations.

\begin{lemma} \label{lem:four-extensions} Every
  $\rsh \in \RSL{}{} \setminus \bigcup_{1 \leq i \leq \NUMTLL, \PERM
    \in \PERMS} \ITLLC{i}$
  that is a partial model of \textnormal{\texttt{tll}} has a unique tight model with at least four shared locations.
\end{lemma}
\begin{proof}
  Note that \texttt{tll}-unfoldings contain only two types of
  points-to assertions: $\PT{x}{y\,z\,\NIL}$ and
  $\PT{x}{\NIL\,\NIL\,y}$ (for some variables $x,y,z$).
  Let $(\stack,\heap)$ be the unique tight model of some partial model $\rsh$ of $\textnormal{\texttt{tll}}$.
  Due to the one-to-one correspondence between $\rsh$ and $(\stack,\heap)$,
  we will not distinguish between $(\stack,\heap)$ and $\rsh$ in the remainder of the proof.

  We refer to heap entries corresponding to the first kind of points-to assertion, i.e. $\heap(u) = (\_,\_,\NIL)$,
  as $i$-pointers (short for inner pointers).
  Moreover, we refer to heap entries corresponding to the second kind of points-to assertion, i.e. $\heap(u) = (\NIL,\NIL,\_)$,
  as $\ell$-pointers (short for leaf pointers).
  In both cases, given $\heap(u) = (v_1,v_2,v_3)$, we refer to $u$ as $i$-source and $\ell$-source
  and to $v_1,v_2,v_3$ as $i$-target and $\ell$-target, respectively.


  Recall that we restricted our attention to those calls of
  $\PCDS{tll}{\T{x}}$ whose parameters are non-null and pairwise
  definitely unequal, where the latter point is a consequence of the
  assumption of acyclicity.\footnote{Unless we are dealing with the
    base case, the \texttt{tll} that consists of a single pointer,
    where $\IFV{1} = \IFV{2}$ holds.}

  We make several observations regarding the shared locations of
  partial \texttt{tll} models $\rsh$ based on the properties of
  \texttt{tll} unfoldings.
\begin{enumerate}
\item In a \texttt{tll}, there is at most one $i$-source that is no $i$-target (the
  root).
  Thus, if $\rsh$ contains $n$ distinct $i$-sources that are not
  $i$-targets, this requires $n-1$ distinct shared locations.

\item In a \texttt{tll}, every $i$-target that is not an $i$-source
  is an $\ell$-source.
  Thus, if $\rsh$ contains $n$ distinct $i$-targets that are neither
  $i$-sources nor $\ell$-sources, this requires $n$ shared locations.

\item In a \texttt{tll}, if there is at least one $i$-location, then all
$\ell$-locations are $i$-targets.
Thus, in such $\rsh$, all $\ell$-sources that are not $i$-targets
are shared locations.

\item In a \texttt{tll}, there is only one $\ell$-source that is no $\ell$-target (the
  leftmost leaf), and there is only one $\ell$-target that is no $\ell$-location (the
  successor of the rightmost leaf)
Thus, if $\rsh$ contains $n$ distinct $\ell$-sources that are not
  $\ell$-targets, this requires $n-1$ distinct shared locations.
Likewise, if $\rsh$ contains $n$ distinct $\ell$-targets that are not
  $\ell$-sources, this requires $n-1$ distinct shared locations.
(Note that we have to take care not to count these shared locations
twice, because of possible overlap with the previous point.)

\item In a \texttt{tll}, all roots of the trees induced by considering only $i$-sources
  must be referenced by free variables (the locations in the corresponding model belong to $co\DOM(\stack)$).
Thus, $\rsh$ has one shared location per such $i$-tree.

\item In a \texttt{tll}, in a sequence of linked $\ell$-sources, both the first
  $\ell$-source and the last $\ell$-target must be referenced by
  free variables.

Thus, the leftmost $\ell$-source and the successor of the
  rightmost $\ell$-source are shared locations for every maximal
  linked list of $\ell$-sources in $\rsh$.
\end{enumerate}

Based on these observations, we make a case distinction over the
possible structure the unique tight model of a partial model $\rsh$
interpreted as a graph and
show that each graph either corresponds to a $\rsh \in \LTLLC{i}$ for
some $i, \PERM$ or corresponds only to $\rsh$ with at least four
shared locations .
\begin{itemize}
\item $\rsh$ contains only $\ell$-sources.
  \begin{itemize}
  \item If it contains only a single source, then $\rsh \in
    \CALLSEM{P_i}{}$ for $i \in \{ 2,3,4 \}$.
  \item If it consists of two sources, then $\rsh \in
  \CALLSEM{P_7}{}$.
  \item If it contains $n \geq 3$ sources, then $\rsh$ has $n$
    shared locations plus the one for the successor of
    the last leaf,
    i.e, at least four.
  \end{itemize}
\item $\rsh$ contains only $i$-sources
  \begin{itemize}
  \item If it contains only a single source, $\rsh \in
    \CALLSEM{P_1}{}$
  \item If it contains at least two sources, then it contains at
    least three $i$-targets that are neither $i$-sources nor
    $\ell$-sources; in addition, it contains at least one root,
    which also has to be a shared location; yielding a total of at
    least four shared locations
  \end{itemize}
\item If it contains both $i$- and $\ell$-sources, it needs one
  shared location per root and two shared locations per maximal list
  of linked leaves.\footnote{This is only true because we assumed the last node in a TLL must not be equal to null nor create a cycle.}
  \begin{itemize}
  \item If it contains two roots, it thus has at least four
    shared locations
  \item If it contains two unconnected lists of linked
    $\ell$-sources, it has at least five shared locations
  \item If it contains only one list of linked leaves, but one of the
    inner nodes of these leaves is no $i$-target, it has at least four
    shared locations
  \item If it has at most three shared locations, it thus has exactly
    one root, and at most its first and its last $\ell$-source are
    not $i$-targets (but each of them can be). Thus, $\rsh \in
    \CALLSEM{P_i}{}$ for some $5 \leq i \leq 8$. 
  \end{itemize}
\end{itemize}
\qed
\end{proof}

\begin{corollary}
  The set $\bigcup_{1 \leq \leq \NUMTLL, \PERM \in \PERMS}
  \CALLSEM{\ITLLC{i}}{}$ contains all well-determined partial
  \texttt{tll} models that can be expressed with at most three free
  variables.
\end{corollary}

\begin{corollary}
  The set
  $\{ \ITLLC{i} \mid 1 \leq i \leq \NUMTLL, \PERM \in \PERMS, \PUREA
  \in \PURES \}$
 of equivalence classes of well-determined partial \texttt{tll} models is complete.
\end{corollary}

\subsection*{Heap Automata for Singly-linked List Segments} \label{app:more-data-structures}

Recall the SID for acyclic singly-linked list segments from
Example~\ref{ex:srd}. A heap automaton $\HA{A}$ for \texttt{sll} is
defined in Fig.~\ref{fig:sll-ha}.

\begin{figure}
$Q \DEFEQ \{q_{\qeq}, q_{\diff}, q_{\qrev}, q_{\qfst}, q_{\qsnd},
q_{\bot} \} \qquad F \DEFEQ \{ q_{\qeq}, q_{\diff}, q_{\qfst} \}$

\vspace{2mm}
\textbf{Transitions:}
\begin{align*}
  & (\EMPTYSEQ, \rsh, q_{\qeq})\in \Delta &\text{ iff }& \rsh \models \EMP : \{ \IFV{1} = \IFV{2}) \} \\
 &&& \text{or } \rsh \models \EMP \wedge \NOFV{\sh} = 2\\
  & (\EMPTYSEQ, \rsh, q_{\diff})\in \Delta &\text{ iff }& \rsh
    \models \PCDS{sll}{\IFV{1}\IFV{2}} : \{ \IFV{2} \neq \NIL \} \\
  & (\EMPTYSEQ, \rsh, q_{\qrev})\in \Delta &\text{ iff }& \rsh
    \models \PCDS{sll}{\IFV{2}\IFV{1}} : \{ \IFV{1} \neq \NIL \} \\
  & (\EMPTYSEQ, \rsh, q_{\qfst})\in \Delta &\text{ iff }& \rsh
    \models \PCDS{sll}{\IFV{1}\NIL} \\
  & (\EMPTYSEQ, \rsh, q_{\qsnd})\in \Delta &\text{ iff }& \rsh \models \PCDS{sll}{\IFV{2}\NIL}  \\
  & (\EMPTYSEQ, \rsh, q_{\bot}) \in \Delta &\text{ iff }& \rsh \not\models \PCDS{sll}{\FV{0}{}} \\
  & (\T{q}, \sh, q) \in \Delta & \text{ iff } &
    (\EMPTYSEQ, \sh[P_1/\rho_{\T{q}[1]},\ldots,P_m/\rho_{\T{q}[m]}], q) \in \Delta
\end{align*}

\vspace{3mm}
\textbf{Representations:}

$\begin{array}{lll}
 \rho_{q_{\qeq}} & \DEFEQ & \EMP : \{\IFV{1} = \IFV{2} \}\\
 \rho_{q_{\diff}} & \DEFEQ & \PT{\IFV{1}}{\IFV{2}}: \{ \IFV{2} \neq
                             \NIL, \IFV{2} \neq \IFV{1} \} \\
 \rho_{q_{\qrev}} & \DEFEQ & \PT{\IFV{2}}{\IFV{1}} : \{ \IFV{1} \neq
                             \NIL, \IFV{2} \neq \IFV{1} \} \\
 \rho_{q_{\qfst}} & \DEFEQ & \PT{\IFV{1}}{\NIL} \\
 \rho_{q_{\qsnd}} & \DEFEQ & \PT{\IFV{2}}{\NIL} \\
 \rho_{q_{\bot}} & \DEFEQ & \IFV{1} \neq \IFV{1} \\
\end{array}$

\caption{A heap automaton $\HA{A} = (Q,\CENTAIL{\alpha},\Delta,F)$ with $L(\HA{A}) = \USET{\PCDS{sll}{\T{x}}}{\SRD}{\CENTAIL{\alpha}}$, for acyclic singly-linked list fragments \texttt{sll} as defined in Ex.~\ref{ex:srd}; plus canonical representations $\rho_q$ for each state $q$.}
\label{fig:sll-ha}
\end{figure}

Observe that $\Delta$ is compositional.
Note further that we have
defined $\Delta$ in such way that $(\EMPTYSEQ, \tau, q) \in \Delta$
for $q \in F$ iff $\tau \models \PCDS{sll}{\FV{0}{}}$, i.e.,
$L(\HA{A}) = \USET{\PCDS{sll}{\T{x}}}{\SRD}{\CENTAIL{\alpha}}$.

Figure~\ref{fig:sll-ha} also shows the canonical representations of
each state, i.e., the minimal unfoldings of each state's
formula. These are the symbolic heaps that are substituted into the
predicate calls in symbolic heaps $\sh$ to obtain simple entailment
problems for deciding transitions $(\T{q}, \sh, q)$.


\subsection*{Optimizations when the Left-Hand Sides of Entailments are Restricted} \label{app:ha-size-fixed-srd}

\paragraph*{Restrictions on parameter reordering.}
First of all, recall why the automaton from
Corollary~\ref{cor:tll-automaton} had $8 \cdot 6 + 1$ states:
\begin{itemize}
\item For a fixed parameter ordering, there are $8$ equivalence
  classes of \emph{partial models}
\item For each class, we have to differentiate between all
  $3 \cdot 2 \cdot 1 = 6$ permutations of the three free variables
  $\IFV{1}$, $\IFV{2}$, $\IFV{3}$: While the permutation of free
  variables does not change \emph{that} a symbolic heap can be
  extended to entail a \texttt{tll} unfolding, it changes \emph{how}
  it can be extended, because the extension may need to reorder
  parameters based on the permutation.
\item Hence there are $8 \cdot 6$ $\EE{\texttt{tll}}$-equivalence classes
  (and thus $48$ states in $\HA{A}_{\mathtt{tll}}$) for partial
  models.
\item There is one $\EE{\texttt{tll}}$ equivalence class for symbolic
  heaps that cannot be extended to entail a \texttt{tll} unfolding,
  which corresponds to a sink state.
\end{itemize}
It will, however, often not be necessary to deal with all such
permutations in practice. For example, just by requiring that all
predicates used in the program that define tree-like structures use
their first parameter to identify the root, merely two possible
permutations remain.  We can thus specify a \texttt{tll} automaton
with just $8 \cdot 2 + 1 = 17$ states that guarantees a sound analysis
as long as the predicates used in the program satisfy the
aforementioned assumption.

\paragraph*{Fixing SIDs in advance.}

In practice, the fact that heap automata are independent of specific
SIDs may also lead to a needless blow-up of the state space.
On the one hand, this independence is a useful feature that enables
reusability of automata across analyses. On the other hand, every
program analysis on a fixed program will usually be performed w.r.t.~a
fixed SID $\SRD$.  It may thus make sense to perform SID-specific
reductions of heap automata to boost the performance of the analysis
on specific programs of interest.

Such optimizations are based on the observation that the set of reduced symbolic heaps
that can be derived from a fixed SRD $\SRD$ will often not intersect
every single $\EE{\PS}$ equivalence class.  Formally, given an
equivalence class $\EC{\rsh}{\PS}$, where we write $\EC{\rsh}{\PS}$ to
denote the equivalence class that contains $\rsh$
(cf.~Theorem~\ref{thm:eq:myhill}), it may turn out that
$\EC{\rsh}{\PS} \cap \SL{[\SRD]}{\SRDCLASS} = \emptyset$.

If we know that all formulas that occur on the left-hand side of
entailments only contain predicates from $\SRD$, we can drop all the
states that correspond to equivalence classes $\EC{\rsh}{\PS}$ with
$\EC{\rsh}{\PS} \cap \SL{[\SRD]}{\SRDCLASS} = \emptyset$.
A heap automaton will therefore often be much smaller if we build it
for entailment w.r.t.~a fixed SID $\SRD$ rather than for arbitrary
formulas in $\SHCENTAIL{\alpha}$.

This becomes particularly clear if we want to reimplement decision
procedures for fixed SIDs such as the one by Berdine et
al.~\cite{berdine2004decidable} within our framework.  If, for
example, formulas on the left-hand side of the entailment may only
refer to the \texttt{sll} predicate (as is the case in the work by
Berdine et al.~\cite{berdine2004decidable}), we can simplify the
automaton $\HA{A}$ from Figure~\ref{fig:sll-ha} as follows.  We can
drop the states $q_{\qrev}$ and $q_{\qsnd}$, because we know in
advance that there is no predicate in our SID that defines
singly-linked lists in reverse order.
Even in this simple case, we can thus reduce the size of the state
space by a third. In more complex cases, even larger reductions are
possible.


%
\end{document}